\documentclass[12pt, oneside,reqno]{amsart}
\pdfoutput=1
\usepackage[margin=1in]{geometry}
\usepackage{abstract}
\usepackage{amssymb, amsmath, amsthm, mathrsfs, amscd, amsaddr}
\usepackage[bbgreekl]{mathbbol}
\usepackage{subfig, color, graphicx}
\usepackage[dvipsnames]{xcolor}
\usepackage{microtype, siunitx, fancyhdr, polynom}
\usepackage{tikz, tikz-cd, physics, multicol, multirow}
\usepackage{hyperref, makecell}
\usepackage[all, cmtip]{xy}

\renewcommand{\arraystretch}{1.5}

 \usepackage{stackrel,nicefrac}
 \usepackage{mathtools}

\hypersetup{colorlinks,linkcolor=blue,citecolor=cyan}

\usetikzlibrary{decorations.markings}

\newcommand{\m}[0]{\mathbb}
\newcommand{\C}[0]{\mathscr}
\newcommand{\B}[0]{\mathcal}
\newcommand{\G}[0]{\mathfrak}

\DeclareMathOperator{\SU}{SU}
\DeclareMathOperator{\SL}{SL}
\DeclareMathOperator{\GL}{GL}
\DeclareMathOperator{\Log}{Log}

\allowdisplaybreaks

\tikzset{->-/.style={decoration={
  markings,
  mark=at position #1 with {\arrow{>}}},postaction={decorate}}}

\numberwithin{equation}{section}

\def\CN{\mathcal{N}}
\def\CB{\mathcal{B}}
\def\tilde{\widetilde}

\newtheorem{mythm}[equation]{Theorem}
\newtheorem{mydef}[equation]{Definition}
\newtheorem{mylem}[equation]{Lemma}
\newtheorem{myprop}[equation]{Proposition}

\newtheorem{myconj}[equation]{Conjecture}
\newtheorem{myquest}[equation]{Question}
\newtheorem{myprob}[equation]{Problem}

\title{Going to the Other Side via the Resurgent Bridge}

\author{Ovidiu Costin$^1$, Gerald V. Dunne$^2$, Angus Gruen$^3$, Sergei Gukov$^{3,4}$}

\address{$^1$Mathematics Department, The Ohio State University, 231 W. 18th Avenue,  Columbus, Ohio 43210, USA}
\address{$^2$Physics Department, University of Connecticut, 196 Auditorium Road,
Storrs, CT 06269, USA}
\address{$^3$Division of Physics, Mathematics and Astronomy, California Institute of Technology, 1200~E.~California Blvd., Pasadena, CA 91125, USA}
\address{$^4$Dublin Institute for Advanced Studies, 10 Burlington Rd, Dublin, Ireland}

\email{costin@math.osu.edu}
\email{gerald.dunne@uconn.edu}
\email{gukov@math.caltech.edu}

\begin{document}

\maketitle

\begin{abstract}

Using resurgent analysis we offer a novel mathematical perspective on a curious bijection (duality) that has many potential applications ranging from the theory of vertex algebras to the physics of SCFTs in various dimensions, to $q$-series invariants in low-dimensional topology that arise {\it e.g.} in Vafa-Witten theory and in non-perturbative completion of complex Chern-Simons theory. In particular, we introduce explicit numerical algorithms that efficiently implement this bijection. This bijection is founded on preservation of relations, a fundamental property of resurgent functions. Using resurgent analysis we find new structures and patterns in complex Chern-Simons theory on closed hyperbolic 3-manifolds obtained by surgeries on hyperbolic twist knots. The Borel plane exhibits several intriguing hints of a new form of integrability. An important role in this analysis is played by the twisted Alexander polynomials and the adjoint Reidemeister torsion, which help us determine the Stokes data. The method of singularity elimination enables extraction of geometric data even for very distant Borel singularities, leading to detailed non-perturbative information from perturbative data. We also introduce a new double-scaling limit to probe 0-surgeries from the limiting $r\to\infty$ behavior of $\frac{1}{r}$ surgeries, and apply it to the family of hyperbolic twist knots.

\end{abstract}

\setcounter{tocdepth}{2}
\tableofcontents

\section{Introduction and Summary}

An alternative title for this paper could be ``Analysis meets mathematical physics and quantum topology.'' Indeed, one of our main goals is to employ the methods of analysis --- resurgence analysis, to be precise --- to tackle problems that arise at the intersection of quantum topology and mathematical physics.

One large class of such problems emerged recently almost around the same time in several different areas of mathematics and mathematical physics: trace formulae, low-dimensional topology, quantum algebra, and BPS state counting. All these different disciplines provide their own interpretation and shed a new light on the following phenomenon that we first illustrate by a concrete example. Imagine that we are given a $q$-series (a ``counting function'') of the form
\begin{equation}
f (q) = q^{\Delta} \left( c_0 + c_1 q + c_2 q^2 + \ldots \right)
\quad \in \, q^{\Delta} \mathbb{Z} [[q]]
\label{generalqform}
\end{equation}
with integer-valued coefficients $c_j\in \mathbb Z$, and such that $f(q)$ converges inside the unit disk, $|q|<1$, possibly except at a finite set of points.\footnote{This latter condition can be generalized further by requiring convergence only outside the set of measure zero in $|q|<1$, with respect to the flat Euclidean metric on the $q$-plane.} For concreteness, let us take the following example
\begin{eqnarray}
f(q) & = & q^{\frac{1}{2}} \left(
1 - q - q^5 + q^{10} - q^{11} + q^{18} + q^{30} - q^{41} + \ldots \right) \label{Sigma237} \\
& = & q^{\frac{1}{2}} \sum_{n=0}^{\infty} \frac{(-1)^n q^{\frac{n(n+1)}{2}}}{(q^{n+1};q)_n} \nonumber
\end{eqnarray}
where $(x;q)_n := \prod\limits_{i=0}^{n-1} (1-x q^{i-1})$ is the $q$-Pochhammer symbol. The last equality in \eqref{Sigma237} expresses the $q$-series in a $q$-hypergeometric form. In this form, imagine replacing $q$ by $q^{-1}$ in each term in the sum over $n$, and then multiplying both the numerator and denominator by a power of $q$, so that the result is also naturally a $q$-series (rather than a series in $q^{-1}$). In our example, we get
\begin{eqnarray}
f(q)^{\vee} & = & q^{-\frac{1}{2}} \left(1 + q + q^3 + q^4 + 
q^5 + 2 q^7 + q^8 + 2 q^9 +\ldots \right) \label{minusSigma237} \\
& = & q^{- \frac{1}{2}} \sum_{n=0}^{\infty} \frac{q^{n^2}}{(q^{n+1};q)_n} \nonumber
\end{eqnarray}
It is easy to see that the ``dual'' $q$-series \eqref{minusSigma237} is also of the form \eqref{generalqform} and converges inside the unit disk, $|q|<1$.\footnote{In this example, the $q$-series $f(q)^{\vee}$ is the order 7 mock theta function conventionally denoted $\mathcal F_0(q)$ \cite{GM12}. Also note that, while the small $q$ expansions of $f(q)$ and $f(q)^\vee$ both have integer-valued coefficients, the expansion of $f(q)$ has the further special property of being a {\it unary} $q$-series, with coefficients taking only the values $\pm 1$ and $0$. Neither this property nor mock modularity will play any role in our analysis, which aims to develop tools applicable to more general dual pairs $f(q)$ and $f(q)^\vee$. In fact, from the perspective of BPS state counting / 3-manifold invariants, such properties are extremely special, as one can easily see from many examples in the following sections and references provided.}

What makes this operation fascinating is that {\it a priori} it had no right to exist, meaning that for a randomly chosen $q$-series $f(q)$ the ``dual'' $q$-series $f(q)^{\vee}$ might not exist and might not be well-defined (unique). For example, it is clear that we do not wish to restrict this phenomenon only to special $q$-hypergeometric forms --- that we used in the above example for illustration purposes --- since a $q$-hypergeometric form may not exist or may not be unique. Yet, it is remarkable that in so many different contexts the operation
\begin{equation}
f(q)
\quad {{\footnotesize{\text{``}q \leftrightarrow 1/q\text{''}}} \atop \overleftrightarrow{\phantom{transformat}}} \quad
f(q)^{\vee}
\label{ffotherside}
\end{equation}
not only carries a non-trivial meaning, but also leads to the same dual $f(q)^{\vee}$ for a given $f(q)$:
\begin{itemize}

\item In low-dimensional topology, the operation \eqref{ffotherside} is simply the operation of orientation reversal (i.e. parity operation). It describes how $q$-series invariants of 3-manifolds and 4-manifolds behave under orientation reversal.

\item In quantum algebra, the operation \eqref{ffotherside} can be interpreted as the relation between characters of logarithmic Vertex Operator Algebras (VOAs) in positive and negative Kazhdan-Lusztig zones.

\item In 3d $\mathcal{N}=2$ quantum theories with 2d $\mathcal{N} = (0,2)$ boundary conditions, the operation \eqref{ffotherside} describes a non-trivial duality of 2d $\mathcal{N} = (0,2)$ boundary conditions that accompanies the parity reversal in 3d $\mathcal{N}=2$ theories.

\end{itemize}
The first and the last items on this list are related by the 3d-3d correspondence. We provide further context and describe each of these applications more fully in Section \ref{sec:other-side}. For now, we just remark that each of these contexts offers its own method (or, in some cases, several methods) for computing $f(q)^{\vee}$ starting with a given $f(q)$, see {\it e.g.} \cite{CCFGH} for a method based on Rademacher sums or {\it e.g.} \cite{Park21} for a method based on surgery formulae. The existing methods, however, have various limitations and, as a result, provide little insight to the question: What should one mean by the operation \eqref{ffotherside} ?

In other words, the operation \eqref{ffotherside} is clearly asking for a pair of $q$-series, $f(q)$ and $f(q)^{\vee}$, that have the form \eqref{generalqform} and converge inside the unit disk, $|q|<1$. However, how does one formulate ``$q \leftrightarrow 1/q$'' mathematically? As a first approximation to the answer, one can say that $f(q)^{\vee}$ is supposed to be the unique continuation of $f(q)$ across its natural boundary, such that certain additional conditions are met (and, ideally, specify $f(q)^{\vee}$ unambiguously).\footnote{Since in all of the above mentioned contexts $f(q)$ and $f(q)^{\vee}$ have interpretation as ``counting functions,'' integrality of the coefficients and integrality of the powers (up to an overall shift $\Delta$) should certainly represent part of the conditions. It would be nice, however, if integrality is not part of the primary conditions, but rather emerges as a consequence of more fundamental properties that are to be identified.} Identifying these additional conditions is, in a sense, the key element in addressing this question and understanding the deeper meaning of the magical operation \eqref{ffotherside}.

The way we described it here, the operation \eqref{ffotherside} of going to the other side of the natural boundary is clearly a kind of problem that should be well suited to resurgent analysis. Indeed, as anticipated earlier \cite{CCFGH}, in this paper we use resurgent analysis to formulate more precisely the sought after additional properties of $f(q)^{\vee}$ and the deeper structure represented by \eqref{ffotherside}. Along with this conceptual understanding comes strong computational power which, among many other applications, allows to settle the following long-standing question that we later review from the perspectives of topology, quantum algebra, and physics.

\begin{myquest}
\label{quest:dualtheta}
Does an infinite family of false theta-functions $\widetilde \Psi_p^{(a)} (q)$, defined in (\ref{falsetheta}), and labeled by integers $p$ and $a$, enjoy the property \eqref{ffotherside}? If so, what is the explicit form of $\widetilde \Psi_p^{(a)} (q)^{\vee}$?
\end{myquest}

Our first example \eqref{Sigma237} is actually a linear combination of four $\widetilde \Psi_p^{(a)} (q)$'s, namely
$$
q^{\frac{83}{168}} \left( \tilde \Psi_{42}^{(1)} - \tilde \Psi_{42}^{(13)} - \tilde \Psi_{42}^{(29)} + \tilde \Psi_{42}^{(41)} \right)
$$
In fact, the majority of examples \eqref{ffotherside} that appeared in various areas of mathematics and mathematical physics involve a linear combination of several $\widetilde \Psi_p^{(a)} (q)$'s. For instance, characters of log-VOAs and BPS $q$-series invariants typically involve pairs, such as 
$\widetilde \Psi_p^{(a)} (q) + \widetilde \Psi_p^{(p-a)} (q)$, as illustrated in Section \ref{sec:other-side} with an infinite family of dual pairs \eqref{ffotherside} that come from surgery formulae for 3-manifold invariants. Hence, based on the previous work in any of the fields --- topology, algebra, and physics --- it was not clear whether $\widetilde \Psi_p^{(a)} (q)$ individually should have a dual in the sense of \eqref{ffotherside}, for general $a$ and $p$.

In this paper, we answer Question~\ref{quest:dualtheta} in the affirmative and provide a systematic way to compute as many terms in the $q$-expansion of $\widetilde \Psi_p^{(a)} (q)^{\vee}$ as the computer power allows. The method is numerical and, therefore, does not give a closed form of the dual $q$-series. But, on the other hand, it does not rely on any special properties ({\it e.g.} any type of modularity) of the original $q$-series $f(q)$. As such, it can be applied in a variety of different contexts: to $q$-series invariants of 4-manifolds and 3-manifolds, characters of log-VOAs, {\it etc.}

In this paper, we mainly focus on one application of our method, to the BPS $q$-series invariants that provide a non-perturbative completion for the complex Chern-Simons theory. Other applications are briefly discussed in Section~\ref{sec:other-side}, though. Below we present a more detailed account of the questions from BPS state counting and complex Chern-Simons theory that motivated our work, introducing along the way some of the key notations used in the rest of the paper, and summarizing the main results and surprises that we found.

\subsection{Resurgence at 228 loops}

One motivation for our work comes from Quantum Field Theory in general, and gauge theory in particular. While the perturbative formulation of such quantum theories is available (though in many cases not readily computable), the non-perturbative formulation of a general QFT / gauge theory remains one of the major challenges in modern theoretical physics.

This state of affairs is well illustrated by complex Chern-Simons theory, which arguably is the simplest non-trivial representative in this class. As in more general quantum field theories, its perturbative formulation starts with the Feynman path integral,
\begin{equation}
\int_{\B{A}} DA \, e^{- \frac{1}{\hbar} S(A)}
\label{Feynman-int}
\end{equation}
which can be systematically evaluated by the saddle point method to produce a (formal) power series in the small (``coupling constant'') parameter $\hbar$:
\begin{equation}
\B{Z}_{\alpha}^{\text{pert}}(\hbar) \; = \; \sum_{n = 0}^{\infty} a_n^{\alpha} \hbar^{n + c_{\alpha}}
\label{eq: perturbative expansion}
\end{equation}
Here, $\alpha$ denotes the choice of a critical point of the action functional $S (A)$ and $a_n^{\alpha}$ are the so-called perturbative coefficients. In a general QFT, their explicit computation relies on Feynman diagrams and becomes exponentially difficult with the perturbative order (``loop number'') $n$. However, if one can compute sufficiently many perturbative coefficients $a_n^{\alpha}$, then the magic of resurgent analysis allows to extract detailed quantitative information about other saddle points $\beta \ne \alpha$ from a finite set of $a_n^{\alpha}$'s. In other words, it provides an opportunity to understand the non-perturbative structure (and, hopefully, one day can lead to a mathematical definition) of the Feynman path integral \eqref{Feynman-int}.

One of the key elements of the resurgent analysis is the analytic continuation, $B_{\alpha} (\xi)$, of the Borel transform\footnote{For a series $\sum_{n = 0}^{\infty} a_n \hbar^{n + c}$ the Borel transform is defined as $\sum_{n = 1}^{\infty} \frac{a_n}{\Gamma(n + c)} \xi^{n + c-1}$.} of the perturbative series
\begin{equation}
B^{\text{pert}}_{\alpha} (\xi) = B\B{Z}^{\text{pert}}_{\alpha} (\xi)
\end{equation}
In resurgent (path) integrals, the function $B_{\alpha} (\xi)$ is expected to have singularities only at $S_{\alpha}$, the critical values of the action functional. What makes complex Chern-Simons theory a good model for testing this, and other predictions of resurgent analysis, is that it is possible to compute perturbative coefficients $a_n^{\alpha}$ to relatively high loop order with relatively little work. In particular, for closed 3-manifolds in this paper we typically truncate the perturbative expansion at the order $n=228$.

Another aspect well illustrated by complex Chern-Simons theory is that, in general, in gauge theories the integration domain $\B{A}$ in \eqref{Feynman-int} is not simply-connected. It consists of gauge connections on $M_3$ modulo gauge equivalence, and the latter quotient is responsible for the non-trivial $\pi_1 (\B{A}) \cong \mathbb{Z}$. For example, in Chern-Simons theory with gauge group $G$, the action functional $S(A)$ is the Chern-Simons functional
\begin{equation}
CS(A) = \frac{1}{4 \pi} \int_{M_3} \tr(A \wedge dA + \frac{2}{3}A \wedge A \wedge A)
\label{CSfunctional}
\end{equation}
and its critical points are $G$-flat connections
\begin{equation}
dA + A \wedge A = 0
\label{flatconn}
\end{equation}
If we work with the universal cover of $\B{A}$, gauge equivalent connections can have different actions so it is important to differentiate between an element
\[
\bbalpha \in \pi_0(\B{M}_{\text{flat}}(M_3, G)) \times \m{Z}
\]
and its gauge equivalence class
\[
\alpha \in \pi_0(\B{M}_{\text{flat}}(M_3, G)).
\]
With this notation,
\begin{equation}
\bbalpha = (\alpha, CS(\bbalpha)), \quad \quad CS(\bbalpha) \in \m{Z} + CS(\alpha)
\label{alphalift}
\end{equation}
remembers the exact Chern-Simons value in $\m{R}$, whereas $\alpha$ remembers it in $\m{R}/\m{Z}$.

\subsection{Non-perturbative complex Chern-Simons}

When the gauge group $G$ is compact, Chern-Simons theory admits a non-perturbative \cite{Wit89} --- in fact, a mathematically rigorous \cite{RT90} --- formulation, in part due to the fact that the spaces of states $\mathcal{H} (\Sigma)$ are all finite-dimensional and allow exact computations via cutting-and-gluing of 3-manifolds. On the other hand, when $G$ is complex, {\it e.g.} $G = SL(2,\mathbb{C})$ that will be our default choice in this paper, the spaces $\mathcal{H} (\Sigma)$ are infinite-dimensional, as in most QFTs of physical interest.

The detailed computations in complex Chern-Simons theory go back at least 20 years \cite{Guk05}, when they were used to explain and generalize the volume conjecture and the analogues of the MMR expansion \cite{MM,Roz1} around complex $SL(2,\mathbb{C})$ flat connections. This quickly led to a variety of exact perturbative techniques that allow to compute \eqref{eq: perturbative expansion} at all loops, and even to non-perturbative calculations for cusped 3-manifolds. However, the quantitative non-perturbative formulation of the theory that extends to arbitrary {\it closed} 3-manifolds and behaves well under cutting-and-gluing remained elusive until recently. A candidate for non-perturbative complex Chern-Simons proposed in \cite{GPV,GPPV} comes from embedding it into string theory, building on a large body of prior work \cite{Wit92,BT,BSV,GV,OV,GSV}, and formulating the problem in terms of $Q$-cohomology or, equivalently, BPS state counting. It associates to a closed 3-manifold $M_3$ a collection of $q$-series of the form
\begin{equation}
\widehat{Z}_b (M_3,q) = q^{\Delta_b} \left( c_0^{(b)} + c_1^{(b)} q + c_2^{(b)} q^2 + \ldots \right)
\quad \in \, q^{\Delta_b} \mathbb{Z} [[q]]
\label{Zhatexpansion}
\end{equation}
that behave well under cutting-and-gluing (surgery) formulae and, moreover, have integer coefficients $c_i^{(b)} \in \mathbb{Z}$ that have enumerative meaning and are expected to admit a categorification. Unlike the formal series \eqref{eq: perturbative expansion}, its non-perturbative counterpart \eqref{Zhatexpansion} is conjectured to be an actual function, well-defined inside the unit disk $|q|<1$, such that the variables in the two expansions are related via
\begin{equation}
q \; = \; e^{\hbar}
\label{qvsh}
\end{equation}
when $q \to 1$. While the physical definition of the invariants \eqref{Zhatexpansion} applies, at least in principle, to arbitrary 3-manifolds, various mathematical definitions have been proposed --- {\it e.g.} based on quantum groups at generic $q$ \cite{Park20,Park21}, on the geometry of affine Grassmannian \cite{GHNPPS}, on curve counting \cite{EGGKPS}, and even on resurgent analysis combined with other tools \cite{GM} ---  some of which are fully developed for large infinite classes of {\it closed} 3-manifolds. In particular, some of these alternative perspectives and definitions make it clear that the set of labels $b$ in \eqref{Zhatexpansion} should be the set of Spin$^c$ structures on $M_3$; in other words, as a non-perturbative theory, complex Chern-Simons is a Spin-TQFT \cite{CGP,Chae,Jag}. This conclusion is noteworthy because the naive formulation of the Feynman path integral \eqref{Feynman-int} as well as the perturbative version of the theory \eqref{eq: perturbative expansion} do not require a choice of Spin or Spin$^c$ structure on $M_3$.

In light of these developments, another natural motivation for our work is the study of a precise and quantitative relation between the perturbative $\hbar$-series \eqref{eq: perturbative expansion} and the non-perturbative $q$-series \eqref{Zhatexpansion} for complex Chern-Simons theory on {\it closed} 3-manifolds. As noted earlier, such 3-manifolds present a greater challenge than a much better understood class of 3-manifolds with toral boundaries (that include knot and link complements). Since on a complex $q$-plane, the expansions \eqref{Zhatexpansion} and \eqref{eq: perturbative expansion} are around $q=0$ and $q=1$, respectively, resurgent analysis is a natural tool to study their relation. It has been successfully employed for studying this relation for various closed 3-manifolds \cite{GMP,Chun17,CCFGH,CFG19,Wu20,FP20}, including some infinite families \cite{Chung20,AM22}. Clearly, the relation between the perturbative $\hbar$-series \eqref{eq: perturbative expansion} and the non-perturbative $q$-series \eqref{Zhatexpansion} is intimately related to our main problem of developing systematic methods for going to the other side \eqref{ffotherside}.

In various parts of our analysis we also aim to be fairly general and, in particular, perform the analysis for infinite families of closed hyperbolic 3-manifolds, that we choose to be surgeries on twist knots:
\begin{equation}
M_3 \; = \; S^3_{p/r} (K_N)
\label{Msurgery}
\end{equation}
In other words, this class of examples is labeled by three integer numbers $(p,r,N)$, such that $p$ and $r$ are relatively prime and defined up to an overall sign change $(p,r) \sim (-p,-r) $. Even though for small finite set of values of $\frac{p}{r} \in \mathbb{Q}$ and $N \in \mathbb{Z}$ the resulting 3-manifolds in \eqref{Msurgery} are non-hyperbolic,\footnote{These are called exceptional surgeries.} we still include them in our analysis since many questions that we study exhibit (sometimes unexpected) regularity in $\frac{p}{r}$ and in $N$.

Including in this family of surgeries a special value $\frac{p}{r}=0$ brings us to another motivation, which also highlights the role of {\it closed} 3-manifolds and involves the 0-surgeries on knots
\begin{equation}
S^3_0 (K)
\label{0surgeryonK}
\end{equation}
that play an important role in topology, including the generalized property R, the smooth Poincar\'e conjecture in 4 dimensions, the slice-ribbon conjecture, {\it etc.} In order to understand these challenging problems via 0-surgeries, it would help to connect categorifiable quantum invariants of 3-manifold with cobordism / concordance invariants. Curiously, a candidate for such relation can come precisely from a relation between the perturbative $\hbar$-series \eqref{eq: perturbative expansion} and the non-perturbative $q$-series \eqref{Zhatexpansion} in complex Chern-Simons theory.

Indeed, as noted earlier, the invariants \eqref{Zhatexpansion} are associated with quantum groups at generic $|q|<1$ and are conjectured to be categorifiable for a general 3-manifold $M_3$. On the other hand, in section~\ref{sec:twisted-Alexander} we establish a new explicit relation between the perturbative $\hbar$-series \eqref{eq: perturbative expansion} for surgeries on a knot $K$ and the twisted Alexander polynomial of $K$. In a nutshell, the relation comes from a closer look at the perturbative $\hbar$-series \eqref{eq: perturbative expansion} which, for a general complex flat connection $\alpha$ is conjectured to take the following form \cite{Guk05,GM08}:
\begin{equation}
e^{- \frac{1}{\hbar} S_{\alpha}}
\B{Z}_{\alpha}^{\text{pert}}(\hbar) \; = \;
e^{- \frac{4 \pi^2}{\hbar} CS( \alpha )} \, 
\sqrt{\tau^{adj}_{M_3}(\alpha)} \;
\hbar^{\delta^{(\alpha)} / 2} 
\left(1 + \sum_{n = 1}^{\infty} a_n^{\alpha} \, \hbar^n \right)
\label{Zpertcomplex}
\end{equation}
This is the complex Chern-Simons analogue of the Witten's conjecture \cite{Wit89} for Chern-Simons theory with compact gauge group.\footnote{The two versions are especially close when $\alpha = 0$ is the trivial flat connection. However, even in that case a careful argument is needed to show that both versions have the same right-hand side.}
Here, $CS (\alpha)$ is the value of the functional \eqref{CSfunctional}, $\tau^{adj}_{M_3}(\alpha)$ sometimes abbreviated as $\tau (\alpha)$ is the adjoint torsion twisted by $\alpha$, and 
$\delta^{(\alpha)} = h^1 - h^0$ is a simple cohomological invariant of a complex flat connection that depends on its stabilizer, $\text{Stab}_G (\alpha) \subseteq G$. Specifically, for $G = \SL(2, \m{C})$, $\delta^{(\alpha)} = 0$ if the connection $\alpha$ is irreducible, $\delta^{(\alpha)} = 1$ if $\alpha$ is abelian, and $\delta^{(\alpha)} = 3$ if $\alpha$ is central \cite{GM08}.
As we explain in section~\ref{sec:twisted-Alexander}, for 3-manifolds given by surgeries on knots, {\it cf.} 
\eqref{Msurgery}--\eqref{0surgeryonK}, $\tau^{adj}_{M_3}(\alpha)$ is determined by the twisted Alexander polynomial of $K$. Besides potential applications to topology, this has important direct applications to resurgence that we discuss next.

\subsection{Symmetries and integrability from the Borel plane}

In resurgence, one of the most important pieces of data is that of the Stokes coefficients, which in our problem we denote by $\mathcal{S}_{\bbalpha}^{\bbbeta}$. Specifically, starting with a perturbative expansion near a saddle point $\bbalpha$, a lift of $\alpha$ to the universal cover in the space of fields described around \eqref{alphalift}, in a resurgent (path) integral we expect to see other critical values $S_{\bbbeta}$ as singularities of $B_{\bbalpha} (\xi)$, the analytic continuation of the Borel transform $B\B{Z}_{\alpha}^{\text{pert}}(\hbar)$.
Therefore, generically, near $\xi = - S_{\bbbeta} = - 4 \pi^2 CS (\bbbeta)$ we expect
\begin{equation}
B_{\bbalpha} (\xi) \; = \; \frac{\mathcal{S}_{\bbalpha}^{\bbbeta}}{\xi + 4 \pi^2 CS (\bbbeta)} + \text{less singular terms}
\label{BalphaS}
\end{equation}
Our first group of results in this paper is to verify the following conjecture that follows directly from the structure of \eqref{Zpertcomplex}.
\begin{myconj}
\label{conj:Storsion}
\begin{equation}
\mathcal{S}_{\bbalpha}^{\bbbeta} \; \in \;
\frac{1}{2 \pi i} \sqrt{\frac{\tau (\beta)}{\tau (\alpha)}} \, \mathbb{Z}
\end{equation}
\end{myconj}
\noindent
Based on this, we can define
\begin{equation}
m_{\bbalpha}^{\bbbeta} \; := \;
2 \pi i \, \mathcal{S}_{\bbalpha}^{\bbbeta} \, \sqrt{\frac{\tau (\alpha)}{\tau (\beta)}}
\end{equation}
which is expected to be an integer.

Unlike the standard Picard-Lefschetz theory --- where Stokes coefficients can be interpreted as intersection numbers between Lefschetz thimbles and, therefore, are (skew) symmetric --- this is no longer the case in gauge theory. In particular,

\begin{mythm}[\cite{GMP}]
\label{thm:stabilizer}
In complex Chern-Simons theory the Stokes coefficients in general are asymmetric in $\bbalpha$ and $\bbbeta$; in particular,
\begin{equation}
\mathcal{S}_{\bbalpha}^{\bbbeta} = 0
\quad \text{whenever} \quad
\dim \text{Stab}_G (\bbalpha) < \dim \text{Stab}_G (\bbalpha)
\label{stabab}
\end{equation}
while $\mathcal{S}^{\bbalpha}_{\bbbeta}$ does not need to vanish.
\end{mythm}
This general vanishing theorem does not prevent non-degenerate (Gaussian) saddles to appear as transseries in the Borel resummation of a degenerate (non-Gaussian) saddles. (It says the converse cannot happen.) Therefore, in fairly generic examples, like the ones considered in this paper, one might expect all non-degenerate saddles to behave similarly and, in particular, ``light up'' as singularities on the Borel plane for a degenerate saddle.\footnote{In the context of $SL(2,\mathbb{C})$ Chern-Simons theory on $M_3 = S^3_{p/r} (K)$ with $|p|=1$, the only degenerate saddle is the trivial flat connection, usually denoted $\alpha = 0$. In gauge theory literature, it is also sometimes denoted $\alpha = \theta$.} When this does not happen, it certainly draws our attention to such special instances, and so we give them a name.

\begin{mydef}[phantom saddles]
We call a saddle $\bbbeta$ a {\it phantom saddle} (relative to $\bbalpha$) when $\mathcal{S}_{\bbalpha}^{\bbbeta} = 0$ that is not enforced by the Theorem~\ref{thm:stabilizer}.
\end{mydef}

\noindent
In other words, phantom saddles are the true saddles of the path integral \eqref{Feynman-int} that do not show up on the Borel plane.

Curiously, we find strong (numerical) evidence for such saddles already in the simplest members of the family \eqref{Msurgery}. It would be interesting to uncover the precise condition that trigger this phenomenon; clearly, it must be more subtle than \eqref{stabab}. We do not address this question in the present paper, but expect that such phantom saddles can be explained by an extra grading (``height'') assigned to saddle points that is not directly visible in the path integral formulation \eqref{Feynman-int}. Then, $\mathcal{S}_{\bbalpha}^{\bbbeta} = 0$ would be a consequence of a strict inequality between the gradings of $\bbalpha$ and $\bbbeta$, much as in \eqref{stabab}. That should lead to a new vanishing theorem, a refinement of Theorem~\ref{thm:stabilizer}.

In the process of studying the Borel plane we find an array of new and somewhat surprising structures that point to some kind of ``integrability'' of complex Chern-Simons theory on hyperbolic manifolds.
First, in section \ref{sec: Observations, Torsions and CS values} we observe a number of surprising relations among Chern-Simons values and among values of the adjoint torsion that appear to hold in general (for a general 3-manifold). These relations can be viewed as a first hint that the Borel plane 
has more hidden structure than one could have expected.

 In contexts that are mathematically well understood, two phenomena that point to integrability are:  (i) the {\em decoupling of singularities},  and (ii) the absence of ``expected'' Borel plane singularities.

  Singularity decoupling is simply  illustrated in the world of ODEs  by the basic equation $y''-y-x^{-1}=0$. Its Borel plane singularities, $\mp\frac12 (p\pm 1)^{-1}$, are decoupled: the Laurent expansions at either singularity has infinite radius of convergence and does not ``see'' the other singularity. As a result, this second order ODE decouples: the general solution is $\frac12(y_++y_-)$ where $y_{\pm}$ are the general solutions of {\em $1^{st}$ order equations}: $y'\pm y=x^{-1}$. In quantum field theory a very familiar example of this decoupling is the Euler-Heisenberg effective action \cite{Hei36,Dun04}, where the Borel transform is meromorphic, and all pole singularities are decoupled but are in fact identical up to a simple rescaling.
  A possibly related structure has been found in recent work \cite{Ma21,Baj21} analyzing the resurgent structure of the Bethe ansatz solution to certain integrable 2d QFTs, where the leading singularity is "pure" (without fluctuations) and so can be cleanly separated from the others.

  Vanishing of singularities (more generally exact compensation of the effects of singularities on the Stokes phenomenon) plays a crucial role in eigenvalue problems. A simple example is the QM time-independent harmonic oscillator in one dimension, $-\frac12 \psi''+\frac12 x^2 \psi-E \psi=0$, where, for simplicity, we take $E\in (0,3/2)$. Formal WKB shows that the \'Ecalle critical time\footnote{The variable we have to pass to in order to do a proper Borel analysis \cite{Ec81,Co08}.} is $t=x^2/2$. The equation becomes $- tu''-\frac12 u'+u(t-E)=0$.

  To Borel analyze the series part of the transseries, we make the further substitution $u=e^{-t} t^{1/2} v$ ($t^{1/2}$ chosen for convenience), followed by a Borel transform, $V=\mathcal{B} v$ to get $$p(p+2) V'+\frac12(p+2E+1)V=0$$ with general solution $V(p)=C p^{-\frac{m}{2}-\frac{1}{4}} (p+2)^{\frac{m}{2}-\frac{1}{4}}$. For general $E$, the Borel plane of $V$ exhibits two coupled singularities, $0$ and $-2$, but for the special value $m=1/2$ we simply get $V(p)=C p^{-1/2}$, hence $v=C\pi^{1/2} t^{-1/2}$. After undoing the changes of variables, this yields the first eigenvalue $E_1=1/2$ and eigenfunction $(2\pi)^{-1/2}e^{-x^2/2}$. Writing the equation as a second order system,  a change of variables at $m=1/2$ splits the system into two independent first order equations. (A similar analysis yields the full spectrum of the problem.) A more general version of this phenomenon occurs for Schr\"odinger equations with {\it reflectionless} (or more generally, {\it finite gap}) potentials \cite{DMN76}, for which the Schr\"odinger operator factorizes and the generic factorial divergent asymptotic expansions of energy levels reduce to convergent expansions. This, in turn, is closely related to the "Cheshire Cat" phenomenon in which a Borel singularity appears and disappears when an external parameter is tuned to a special (e.g. integer) value \cite{Dun16,Koz16}. Another interesting example appears in the analysis of the large-order terms of the Weyl expansion for quantum billiards, where certain geodesics are not seen in the large order behavior of the Weyl series \cite{BeHo94,HoTr99,Tr98}.

  In the realm of $q$-series dependent on parameters one may also find  cases where Borel plane singularities vanish when an external parameter is tuned to special values. For example, in the study of mock modular forms, one encounters the family of Appell-Lerch sums \cite{Zwe08}
  \begin{equation}
    \label{eq:q-ser}
    S_q(a,z;q)=\sum_{n\in\mathbb{Z}} \frac{e^{-t(n^2-a n z)}}{\cosh(t(n-z))}
    \qquad, \quad q:=e^{-t}
  \end{equation}
  depending on two parameters: $a$ and $z$. 
  The value $a=2$ is special: the Borel plane becomes independent of $z$. This can be seen from  a straightforward Poisson summation argument. Following the expectation that this is a sign of of integrability, we look first for symmetries and then use them to extract integrals of motion.

  Indeed, we discover that $ S_q(2,z;q)$ is an eigenfunction of two commuting operators: $F(t;z)\mapsto F(t;z+1)$ and $F(t;z)\mapsto F(t;z+\pi i/t)$.\footnote{Note that, up to changes of variables, these are $SL(2,\mathbb{Z})$ transformations.} Defining $S_1=e^{-t z^2} S_q(2,z;q)$,  we get (a) $S_1(t;z+1)=S_1(t;z)$ and (b) $S_1(t;z+\pi i/t)=-S_1(t;z)$. We first seek entire functions $F$ of $z$ in the form $\sum_{k\in\mathbb{Z}}c_k(t) e^{2\pi i k z}$ satisfying (a) and (b). Identifying the Fourier coefficients, we see that their general form is $F(z;t)=B(t) \vartheta_3(\pi(z+1/2), e^{-\pi^2/t})$. Proceeding by variation of parameters, we look for $S_1$ in the form $\vartheta_3\Big(\pi(z+1/2), e^{-\pi^2/t}\Big)g(t;z)$. Some further algebra and analysis at singular points give $S_1(t,z)=a(t)-(2\pi t)^{-1}e^{\pi^2/(4t)} \wp(z,1,\pi i/t)$, where $\wp$ is the Weierstrass elliptic function. All in all, we get the conserved-quantity relation
  \begin{equation}
  \label{eq:finform}
 \frac{S(t;z)}{\vartheta _3\left(\pi  ( z+\tfrac12),e^{-\frac{\pi
        ^2}{t}}\right)} +\frac{e^{\pi^2/(4t)}}{2\pi t} \wp(z,1,\pi i/t)=\frac{S(t;0)}{\vartheta _3\left(   \tfrac{\pi}2,e^{-\frac{\pi
        ^2}{t}}\right)} +\frac{e^{\pi^2/(4t)}}{2\pi t} \wp(0,1,\pi i/t)
\end{equation}
In particular, \eqref{eq:finform} shows that $z$ necessarily belongs to the purely nonperturbative terms, and, also as expected, it appears in closed form.

This conclusion receives further support in section \ref{sec:borel} where we perform detailed numerical analysis of the Borel plane for various hyperbolic surgeries. In particular, in section \ref{sec:decoupling 4_1} we find a peculiar decoupling of the leading singularity that is more subtle than \eqref{stabab}. In Section~\ref{sec:thooft-twist} we observe a set of curious integrality properties that appear to hold for surgeries on all twist knots, and that may or may not be directly related to the decoupling and ``integrability'' discussed above. We expect these interesting integrality properties to be related to the integrality of BPS invariants in the physical realization of the mathematical problem at hand. It would be interesting to study whether integrality is related to integrability, and extending the analysis of Section~\ref{sec:thooft-twist} to a larger class of knots could be a very useful step in this direction.

\subsection{A trace formula for complex Chern-Simons}

A seemingly different line of motivation for this work comes from the Selberg trace formula, its many variants, and resurgent analysis for the heat kernel. In particular, in a recent work \cite{Dun21} it was shown that the Borel plane for the heat kernel has a very distinct form. Quite surprisingly, we will find the same structure of the Borel plane in a completely different problem in the context of the complex Chern-Simons theory, {\it cf.} Figure~\ref{fig:41-modified-borel-poles}. The heat kernel enjoys various symmetries, also discussed in \cite{Dun21}, including the short-time vs. long-time behavior, $t \to \frac{4 \pi^2}{t}$, and the analytic continuation $t \to -t$ relating heat kernels on spaces of negative and positive curvature. Curiously, we find analogs of all these symmetries in the context of complex Chern-Simons theory, where the role of $t$ is played by $\hbar$.

In general, on a manifold $M$, the trace of the heat kernel $K (t,x,y)$ can be expressed in terms of (the discrete part of) the spectrum of the Laplace operator,
\begin{equation}
\Tr e^{- t \Delta} \; = \; \int_M K (t,x,x) \, dx \; = \; \sum_{\lambda_j} \exp (- \lambda_j t)
\label{heattrace}
\end{equation}
where $\frac{\partial}{\partial t} K = - \Delta_x K$ and $K (0,x,y) = \delta (x-y)$. For example, on a real line, $M = \mathbb{R}$, we have $K^{\mathbb{R}} (t,x,y) = \frac{1}{\sqrt{4\pi t}} \exp \left( - \frac{|x-y|^2}{4t} \right)$, and on a circle $K^{S^1} (t,x,x) = \sum_{n \in \mathbb{Z}} K^{\mathbb{R}} (t,x+n,x)$. The deep insight of Selberg, later extended by R.~Langlands in various directions, is that on a symmetric space of the form $M = H / \Gamma$ the spectrum \eqref{heattrace} can be expressed as a sum over geodesics on $M$ or, equivalently, as a sum over hyperbolic conjugacy classes in $\Gamma$.

The Selberg trace formula has many generalizations and applications, which range from representation theory to dynamical systems. All such variants schematically look like
\begin{equation}
\sum \{ \text{spectral terms} \}
\; = \;
\sum \{ \text{geometric terms} \}
\label{trace-schematic}
\end{equation}
In physics, one can think of a problem that involves the motion of a free particle on $M$. The quantum Hamiltonian of this system is precisely the Laplace operator $\Delta$ on $M$, and its spectrum gives the partition function, which is the spectral (quantum) side of the trace formula. In the (semi-)classical approach to the same problem, the partition function is given by the sum over classical trajectories ({\it cf.} geodesics, or saddles) and gives the geometric (classical) side of the trace formula. A well-known concrete example is the heat kernel for the motion of a particle moving on a compact Lie group manifold \cite{Schul68,Dowk70}.

\begin{table}[h!]
\centering
\begin{tabular}{|c|c|}
                 \hline 
				\makecell{Topology} & 
                \makecell{Resurgence}\\
    \hline \hline
				{\rm flat connection} & {\rm path integral saddle} 
    \\
				\hline 
				{\rm Chern-Simons invariant} & {\rm Borel singularity} 
    \\
				\hline 
				{\rm Adjoint Reidemeister torsion} & {\rm residue}
    \\
\hline
\end{tabular}
\vspace{3mm}
\caption{Identification of the topological data with that obtained from the Chern-Simons path integral.}
\label{tab:data}
\end{table}

Already at this stage, the reader can probably recognize a parallel between the trace formula and the relation between the perturbative expansion \eqref{Zpertcomplex} in complex Chern-Simons theory and its non-perturbative formulation \eqref{Zhatexpansion} given by the sum over a BPS spectrum. In this parallel, the sum over complex flat connections (saddles) is clearly an analogue of the sum over classical/geometric terms (geodesics), whereas the spectral side corresponds to the BPS spectrum. This nice parallel, though, is only the tip of a much richer iceberg.

Further work on trace formulae, that led to many above-mentioned generalizations, highlighted several important features. First, it re-established a symmetry between the two sides of the trace formula. Although this symmetry is not at all manifest in the standard way of writing the trace formula {\it \`a la} \eqref{trace-schematic}, it can in fact be traced back to the simplest instance of the trace formula, namely to Poisson summation. Indeed, if $M$ is a torus, {\it i.e.} a quotient of $\mathbb{R}^n$ by a lattice $\Gamma$, then the corresponding theta-function can be written in two different ways, which are precisely the two sides of the trace formula for $M = \mathbb{R}^n / \Gamma$. The Poisson summation indeed relates perturbative and non-perturbative complex Chern-Simons for a large class of 3-manifolds \cite{GPPV}, called plumbed manifolds. In Section \ref{sec:other-side} we explain the origin of this relation in terms of resurgent analysis.

Moreover, the work of Langlands on generalizations of the trace formula led him to several deep (and still largely open) conjectures and ideas, which include the principle of functoriality, theory of endoscopy, and what we now call the Langlands program. In the Langlands program, the Galois side is usually compared to the representation of the fundamental group, the idea that is fully realized in the geometric version of the Langlands program and that in the context of complex Chern-Simons theory would correspond to a representation
\begin{equation}
\rho : \pi_1 (M_3) \to G
\label{repintoG}
\end{equation}
The other side of the Langlands correspondence involves automorphic representations or, more precisely, the so-called $L$-packets of automorphic representations of the Langlands dual group $G^{\vee}$. This subtle but important feature parallels an equally delicate phenomenon in complex Chern-Simons theory which, in part, is due to Theorem \ref{thm:stabilizer}. Namely, under $\hbar \to - \frac{1}{\hbar}$, which is indeed related to the geometric Langlands program \cite{KW}, the complex flat $G$-connections map to linear combinations (a superposition) of complex flat connections for the Langlands dual group $G^{\vee}$ \cite{DG,GPV}. The precise map is not known for a general flat $G$-connection $\alpha$, but in the abelian case was conjectured to be of the form \cite[eq.(2.67)]{GPPV}:
\begin{equation}
\rho_a \mapsto \rho_a^{\vee} = \sum_{{b \atop \text{abelian}}} S^{ab} \big( \rho_b + \sum_{{c \atop \text{non-abelian}}} m^{c}_b \; \rho_{c} \big)
\label{Langl}
\end{equation}
where $m_{\alpha}^{\beta}$ are the transseries coefficients that we saw earlier and $S_{ab}$ is a very explicit matrix of coefficients, which happens to coincide with the $S$-matrix of a logarithmic vertex algebra. This structure plays an important role in the relation between different types of ``partition functions'' on $M_3$, namely $Z_{\alpha} (M_3, q)$, which can be understood as transseries (integrals over Lefschetz thimbles in complex Chern-Simons theory), and the $q$-series invariants $\hat Z_b (M_3,q)$:
\begin{equation}
\widehat{Z}_b = \sum_{{a \atop \text{abelian}}} S^{ab} \big( \mathcal{S} Z^{\text{pert}}_a + \sum_{{\beta \atop \text{non-abelian}}} m^{\beta}_a \, \mathcal{S} Z^{\text{pert}}_{\beta} \big)
\label{ZSmS}
\end{equation}
On the right-hand side, we tacitly suppressed the sum over integral lifts of $\beta$, which makes $m^{\beta}_a$ into the generating $\tilde q$-series of the Stokes coefficients. We shall return to this later and make it explicit {\it e.g.} in the case of homology spheres \eqref{longeqZSmS}, where the first sum on the right-hand side of \eqref{ZSmS} simplifies.

To summarize, the parallel with the trace formula, Poisson summation, and the Langlands program tells us to view the relation between perturbative complex Chern-Simons and its non-perturbative completion as another instance of the (generalized) trace formula.

\section{Expected Borel plane from knot polynomials}
\label{sec:CSandAlex}

The main goal of this section is to describe the basic topological invariants of closed 3-manifolds $M_3 = S^3_{p/r} (K)$ that we expect to match with the position and strength of singularities on the Borel plane.

The computation of these invariants echoes the construction of $M_3 = S^3_{p/r} (K)$ itself; namely, the relevant invariant of a closed 3-manifold is obtained via a suitable surgery formula from the corresponding invariant of the knot complement, $S^3 \setminus K$, and the surgery coefficient $\frac{p}{r}$. The invariant of a knot complement, in turn, can be expressed in terms of a suitable knot polynomial. For the position of singularities on the Borel plane and their associated Stokes coefficients, the relevant knot polynomials are respectively the A-polynomial and the twisted Alexander polynomial. While the former has already appeared in the study of complex Chern-Simons theory, the twisted Alexander polynomial so far did not play a prominent role in the non-perturbative formulation of the theory, based on $Q$-cohomology and BPS spectra, and one of our main goals here to bring it into the spotlight.

The significance of relating the non-perturbative complex Chern-Simons theory with the twisted Alexander polynomial is that they belong to two different worlds, yet both are related to 4-dimensional topology in a non-trivial way. The non-perturbative formulation of complex Chern-Simons theory via $Q$-cohomology and BPS spectra provides a path toward categorification of quantum group invariants for general 3-manifolds, a major open problem in quantum topology that is expected to provide a generalization of Rasmussen's $s$-invariant and its use in a purely combinatorial proof of the Milnor conjecture \cite{Ras}. Any such homological invariants and spectral sequences that they enjoy are naturally related to smooth 4-manifold topology. Twisted Alexander polynomials also provide obstructions to sliceness \cite{HKL,KL99a,KL99b,FV}, though in a somewhat different way; in particular, unlike deeply quantum and non-perturbative BPS $q$-series, twisted Alexander polynomials are almost ``classical.'' In our story they determine the first term in the perturbative expansion \eqref{eq: perturbative expansion}. To the best of our knowledge, there are no direct connections between these two types of invariants. One of our main results is that the resurgent analysis provides such a connection.

Through surgeries \eqref{Msurgery}, deep 4-dimensional questions about {\it knots} (such as sliceness) can be translated into analogous questions about {\it 3-manifolds} and {\it 4-manifolds}. As mentioned above and will be discussed further in Section~\ref{sec:thooft-small}, among such questions those involving 0-surgeries \eqref{0surgeryonK} often turn out to be some of the most subtle and non-trivial ones.

\subsection{Position of singularities from the A-polynomial}
\label{sec:A-polynomial}

As already mentioned in the introduction, the critical points of the Chern-Simons functional \eqref{CSfunctional} are $G$-flat connections \eqref{flatconn}. This follows directly from the Euler-Lagrange equations and holds true for any gauge group $G$ which, in particular, can be either compact or complex. One should remember, though, the lifts to the universal cover \eqref{alphalift}.

Cutting and gluing along a 2-manifold $\Sigma$ in this theory is controlled by the ``Hilbert space of states'' $\mathcal{H} (\Sigma)$ which, in turn, is obtained by quantizing the theory on $\mathbb{R} \times \Sigma$. Since the action functional \eqref{CSfunctional} is first order in derivatives, in particular, in the time derivative parametrizing $\mathbb{R}$ here, the space $\mathcal{H} (\Sigma)$ is obtained by (geometric) quantization of the space of classical solutions on $\mathbb{R} \times \Sigma$, invariant under translations along $\mathbb{R}$. These are again the flat $G$-connections, $\B{M}_{\text{flat}}(\Sigma, G)$, this time on $\Sigma$. The space of flat $G$-connections on $\Sigma$ comes equipped with the Atiyah-Bott symplectic form that also follows directly from \eqref{CSfunctional}. It is easy to see that $\B{M}_{\text{flat}}(\Sigma, G)$ is finite-dimensional, which means that quantization of (complex) Chern-Simons theory is a simple quantum mechanical problem, with the classical phase space $\B{M}_{\text{flat}}(\Sigma, G)$.

{}From the view point of the quantization of $\B{M}_{\text{flat}}(\Sigma, G)$, the only difference between complex and compact gauge group is that in the latter case the classical phase space $\B{M}_{\text{flat}}(\Sigma, G)$ is also compact and, hence, the space of states $\mathcal{H} (\Sigma)$ is finite-dimensional. On the other hand, when $G$ is complex, the phase space $\B{M}_{\text{flat}}(\Sigma, G)$ is non-compact and, correspondingly, $\mathcal{H} (\Sigma)$ is infinite-dimensional. This is the key feature of complex Chern-Simons theory that makes it both interesting and very non-trivial at the same time.

Of particular interest is $\Sigma = T^2$. Indeed, according to a Theorem of Lickorish and Wallace, any 3-manifold can be obtained by performing a sequence of cutting-and-gluing (surgery) operations along knots and links in $S^3$. Since a complement of every knot (resp. link) has $T^2$ as its boundary (resp. disjoint copies of $T^2$), this cutting and gluing construction involves only manifolds with toral boundaries. This, in part, is the motivation why in the present paper we are especially interested in infinite families of surgeries \eqref{Msurgery} on various knots.

The case $\Sigma = T^2$ is also special from the quantization perspective. Indeed, this is the only instance of $\Sigma$, such that $\pi_1 (\Sigma)$ is non-trivial and abelian. As a result, the corresponding space of representations, analogous to \eqref{repintoG}, is essentially flat. More precisely, the (holomorphic) Atiyah-Bott symplectic form is flat in the logarithmic coordinates, which for $\SL(2, \m{C})$ takes a simple form
\begin{equation}
\omega = \frac{dy}{y} \wedge \frac{dx}{x}
\end{equation}
expressed in terms of the holonomy eigenvalues for the two generators of $\pi_1(T^2) = \m{Z}_l \oplus \m{Z}_m$ called the longitude and the meridian,
\begin{align*}
Hol_A: & \pi_1(M_3) \to G
\\
l & \mapsto 
\begin{bmatrix}
y & * \\
0 & y^{-1}
\end{bmatrix}
\\
m & \mapsto
\begin{bmatrix}
x & * \\
0 & x^{-1}
\end{bmatrix}
\end{align*}
To summarize, $\B{M}_{\text{flat}}(T^2, \SL(2, \m{C})) = \frac{\m{C}^* \times \m{C}^*}{\m{Z}_2}$, where the quotient by $\m{Z}_2$ is associated with the Weyl symmetry of $\SL(2, \m{C})$, which acts on $(x,y) \in \m{C}^* \times \m{C}^*$ as $(x, y) \mapsto (x^{-1}, y^{-1})$.

Now, for every 3-manifold with a toral boundary we can consider the image of its representation variety to that of the $T^2$ boundary. In particular, for a knot complement, $S^3 \setminus K$, the flat $\SL(2, \m{C})$ connections on the $T^2$ that can be extended to the knot complement are described by an algebraic curve \cite{CCGLS},
\begin{equation}
\B{A}_K = \{(x, y) \in (\m{C}^*)^2 \mid A_K(x, y) = 0\}
\end{equation}
which plays the role of a ``spectral curve'' in complex Chern-Simons theory \cite{Guk05} and is analogous to the Seiberg-Witten curve of the corresponding supersymmetric QFT related to the knot complement via 3d-3d correspondence \cite{GGP14}.

Conceptually, the affine variety\footnote{As long as one remembers the $\m{Z}_2$ quotient, and all the ingredients are properly invariant under the Weyl group action, it sometimes can be omitted at the intermediate stages to avoid clutter.} $\B{A}_K \subset (\m{C}^*)^2$ should, in fact, be thought of as the holomorphic Lagrangian subvariety. That way, it naturally represents the classical limit of a state in $\mathcal{H} (T^2)$ associated with the knot complement. Indeed, one can explicitly verify that $\omega$ vanishes when restricted to the image of $\B{M}_{\text{flat}}(M_3, G)$ in $\B{M}_{\text{flat}}(\Sigma, G)$, for more general $M_3$ with boundary $\Sigma = \partial M_3$ and for $G$ of higher rank. Therefore, in the WKB approximation, the state associated to the knot complement is a function (more precisely, a half-density) on $\B{A}_K \subset (\m{C}^*)^2$ obtained by integrating the primitive 1-form $d^{-1} \omega \vert_{\B{A}_K}$ along a path on $\B{A}_K$ that connects the point of interest $(x,y)$ to some reference point. This has been discussed in great detail throughout the history of the subject, see {\it e.g.} \cite{KK,GM08,GMP} and references therein.

In particular, for a general surgery {\it \`a la} \eqref{Msurgery}, flat connections on $S^3 \setminus K$ that extend to $S^3_{p/r}(K)$ satisfy $y^r x^p = 1$ since the Dehn filling has the effect of annihilating the element $l^r m^p$ in homology. In practice, almost all points in the intersection of $A_K(x, y) = 0$ and $y^r x^p = 1$ correspond to true extendable flat connections, and the WKB integral along a path on the $A$-polynomial curve provides an easy method that allows to compute all Chern-Simons invariants $CS(\alpha)$ for simple surgeries.

\begin{mylem}\label{lem: Computing CS Values}
Let $\rho_1: \pi_1(S^3 \setminus K) \to \SL(2, \m{C})$ be a non-parabolic representation which extends to a flat connection $\alpha$ on $S^3_{p/r}(K)$. Then there exists a path $\rho_t$ of non-parabolic representations with $\rho_0 = 1$ and
\begin{align}
\label{eq: Computing CS Values, explicit}
2\pi^2 CS(\alpha) & = \int_{\gamma} \frac{\log(y)}{x} \ dx + \frac{vp}{2}\log((\rho_1)_x)^2 + \frac{sr}{2}\log((\rho_1)_y)^2 
\\ & \nonumber \hspace{2cm} - vr \log((\rho_1)_x)\log((\rho_1)_y)
\end{align}
where $v, s$ is a pair of integers satisfying $ps - rv = 1$.
\end{mylem}
Note that, due to branching, $\gamma$ should be thought of as a path in the Riemann surface associated to the map $(x, y) \mapsto (\log(x), \log(y))$. Specifically, for 0-surgeries this formula simplifies to
\begin{align}
\label{eq: Computing CS Values, 0-surg}
CS(\rho_1) & = \frac{1}{2\pi^2}\left(\int_{\gamma} \frac{\log(y)}{x} \ dx + \log((\rho_1)_x)\log((\rho_1)_y)\right)
\end{align}
and for the so-called ``small'' $\frac{1}{r}$-surgeries it gives 
\begin{align}
\label{eq: Computing CS Values, 1/r-surg}
CS(\rho_1) & = \frac{1}{2\pi^2}\left(\int_{\gamma} \frac{\log(y)}{x} \ dx + \frac{j}{2}\log((\rho_1)_y)^2\right).
\end{align}
The strength of Lemma~\ref{lem:  Computing CS Values} is that given a path in $\B{A}_K$, provided a lift to a path in the representation variety of $S^3 \setminus K$ exists, we can apply equation \eqref{eq: Computing CS Values, explicit} without ever having to explicitly construct the lift. For small knots, these lifts will always exist and so the problem of computing Chern-Simons values boils down to finding appropriate paths in $\B{A}_K$.
One minor caveat is that different paths $\gamma$ between $\rho_0$ and $\rho_1$ may lead to $(\rho_1)_x$ and $(\rho_1)_y$ ending on different branches of $\log$ and this can produce different CS values. Thus this method requires a careful treatment of branches.
    		
As we want to start all paths at the trivial flat connection $A=0$, {\it i.e.} at $(x,y) = (1, 1)$, the first problem we encounter is how to travel off the abelian branch $\B{A}_K^{ab} = (\m{C}^*)_x \times \{1\}$ that arises from abelian representations which behave identically for all knots as $H_1(S^3 \setminus K) = \m{Z}$. Indeed, the variety $\B{A}_K$ decomposes into the union of two subvarieties, $\B{A}_K^{ab}$ and $\B{A}_K^{irred}$, called the abelian and irreducible branches, respectively. The irreducible branch is the closure of what remains after the abelian branch is removed, $\B{A}_K^{irred} =  \overline{\B{A}_K \backslash \B{A}_K^{ab}}$. Correspondingly, the $A$-polynomial factors into polynomials representing the two branches
\begin{equation}
A_K(x, y) = (y - 1)A_K^{irred}(x, y).
\end{equation}
While points connecting the abelian and irreducible branches in $\B{A}_K$ are easy to find,\footnote{They correspond to solutions of $A_K^{irred}(x, 1) = 0$.} the issue is that we need to find a branch point which, when lifted to the representation variety of the knot complement, lifts to a path connecting the abelian and irreducible branches.

\begin{mylem}[\cite{CCGLS}: Section 6]
In the representation variety of the knot complement, the abelian and irreducible branches meet along non-abelian reducible representations. In $\B{A}_K$, these reducible representations map surjectively to points $(x, 1)$ where $x^2$ is a root of the Alexander polynomial of the knot $K$.
\end{mylem}
    		
We provide a brief sketch of the proof here. Non-abelian reducible representations are representations $\rho: \pi_1(S^3 \setminus K) \to SL_2(\m{C})$ landing in the upper triangular subgroup. Given such an element
\begin{equation*}
\rho(g) = \begin{bmatrix}
x_g & z_g \\
0 & x_g^{-1}
\end{bmatrix}
\end{equation*}
the path from the abelian branch to the irreducible branch is given by
\begin{equation}
\rho_t(g) = \begin{bmatrix}
x_g & t z_g \\
0 & x_g^{-1}
\end{bmatrix}.
\end{equation}
It remains to understand when $\rho$ exists.
    		
If we project $\rho$ onto $\B{A}_K$ we find that the image is heavily constrained. As the longitude $l$ lies in the second commutator subgroup of $\pi_1(S^3 \setminus K)$ and all second commutators of $2 \times 2$ upper triangular matrices are trivial, $\rho$ projects onto a point $(x, 1)$. Then through careful analysis \cite{DR} it can be shown that $\rho$ exists if and only if $x^2$ is a root of the Alexander polynomial.\footnote{This is a special case of a more general representation deformation problem which is solved by the twisted Alexander polynomial \cite{Wada94}.}
    		
Altogether, this discussion guarantees that we can always move off the abelian branch and onto the irreducible branch. From here, we move along the irreducible branch until we reach our desired point, taking care when moving between different sheets of the irreducible branch.

\subsection{Residues from the twisted Alexander polynomial}
\label{sec:twisted-Alexander}

Based on the earlier work in complex Chern-Simons theory summarized around eq.\eqref{Zpertcomplex}, we expect the residues $\mathcal{S}_{\bbalpha}^{\bbbeta}$ in \eqref{BalphaS} to be related to a 3-manifold invariant called torsion. (See also Conjecture~\ref{conj:Storsion}.) This invariant has both an algebraic and analytic description. The analytic description is directly relevant to the way it appears in \eqref{Zpertcomplex}, as a ratio of one-loop determinants in complex Chern-Simons theory. However, it is not as computationally friendly as the algebraic formulation \cite{Fre,Por15}, which will be our main focus here.

Given two bases $\alpha, \beta$ of a vector space $V$ over a field $\m{F}$, let $\eta$ be the unique change of basis matrix satisfying $\alpha_i = \sum_j \eta_{ij} \beta_j$. Define
\[
[\alpha, \beta] = \det(\eta) \in \m{F}^{\times}.
\]
\begin{mydef}[Reidemeister Torsion]
Let
\[
C_*: 0 \to C_n \xrightarrow{\partial} C_{n-1} \xrightarrow{\partial} \cdots \xrightarrow{\partial} C_1 \xrightarrow{\partial} C_0 \to 0
\]
be a chain complex of finite-dimensional vector spaces over $\m{F}$ with homology $H_*(C_*, \partial)$. Fix a pair of bases, $\textbf{c}$ for the chain complex and $\textbf{h}$ for the homology such that $\textbf{c}_i$ and $\textbf{h}_i$ are bases of $C_i$ and $H_i$ respectively. Let $k_j$ be the rank of $\partial: C_j \to C_{j - 1}$ and choose a collection of $k_j$ elements $\textbf{s}_j = \{s_{j, i}\} \subset C_j$ such that $\partial \textbf{s}_j = \{\partial s_{j, i}\}$ spans $\Im(\partial)$. For each homology basis $\textbf{h}_j$, choose a lift $\hat{\textbf{h}}_j \subset C_j$. Then the Reidemeister Torsion, $\tau_{C_*, \textbf{c}, \textbf{h}}$ is given by:
\[
\tau_{C_*, \textbf{c}, \textbf{h}} = \prod_{j = 0}^n \left[\left\{\partial \textbf{s}_{j + 1}, \textbf{s}_j, \hat{\textbf{h}}_j\right\}, \textbf{c}_j\right]^{(-1)^{j + 1}} \in \m{F}^{\times}
\]
\end{mydef}
There are a couple of observations to make:
\begin{itemize}
\item Despite appearances, choices of $s_j$ and $\hat{h}_j$ do not affect the value of $\tau$.
\item If we change our basis for $C_*$ and $H_*$, then $\tau$ changes as
\begin{equation} \label{eq: Torsion Basis Change}
\frac{\tau_{C_*, \textbf{c}', \textbf{h}'}}{\tau_{C_*, \textbf{c}, \textbf{h}}} = \prod_{i} \left(\frac{[\textbf{c}'_i,\textbf{c}_i]}{[\textbf{h}'_i,\textbf{h}_i]}\right)^{(-1)^i}
\end{equation}
\end{itemize}
        		
Currently, this discussion is purely algebraic but there is a natural link to geometry. Let $X$ be a CW space, $\tilde{X}$ its universal cover and $\rho:\pi_1(X) \to \GL(V_{\rho})$ a linear representation. By construction, $\tilde{X}$ has an induced $CW$ structure upon which $\pi_1(X)$ acts freely. Hence each term of the cellular chain complex $C(\tilde{X}, \m{Z})$ is a free $\m{Z}[\pi_1(X)]$ module and so we can define the twisted chain complex
\[
C_*(X; V_{\rho}) = C(\tilde{X}) \otimes_{\m{Z}[\pi_1(X)]} V_{\rho}
\]
        		
There is a canonical\footnote{Naively, this basis is not canonical as it depends on the chosen cellular structure and a basis for the representation space, but the torsion is independent of both these choices.} basis for $C_j(X; V_{\rho})$ of the form
\[
\textbf{c} = \{\sigma_1 \otimes v_1, \sigma_1 \otimes v_2, \cdots , \sigma_{k_q} \otimes v_n\}.
\]
where $\{v_1, \cdots, v_n\}$ is a basis of $V_{\rho}$ and $\sigma_1, \cdots \sigma_{k_q}$ are q-cells giving a basis (as a module) of $C(\tilde{X}, \m{Z})$. Additionally, if there is homology, pick a basis $\textbf{h}$. 
        		
\begin{mydef}
The Reidemeister torsion $\tau_{X, \textbf{h}}(\rho)$ is
\[
\tau_{X, \textbf{h}}(\rho) = |\tau_{C_*(X; V_{\rho}), \textbf{c}, \textbf{h}}|.
\]          
\end{mydef}
                
While there are choices in the basis $\textbf{c}$, it is well known \cite{Joh, Mil66} that
        
\begin{mythm}
$\tau_{X, \textbf{h}}(\rho)$ is a piecewise linear invariant of $X, \textbf{h}$ and $\rho$.
\end{mythm}
        		
Observe that for a given group homomorphism $\rho: \pi_1(X) \to \SL(2, \m{C})$, there is a family of Reidemeister torsions corresponding to the family of irreduicble representations of $\SL(2, \m{C})$. In the literature the majority of attention has been on studying the Reidemeister torsion corresponding to the standard representation but it is essential here that we work with the adjoint one. To distinguish this choice, we will label $\tau$ as $\tau^{st}$ or $\tau^{adj}$ depending on the context.

Often\footnote{In particular this is the case for manifolds which are surgeries on knot complements.}, after quotienting by the conjugation action, the space of representations $\pi_1(X) \to \SL(2, \m{C})$ is finite and for each representation $\rho$, $C_*(X; V_{\rho})$ is acyclic. In these cases, we get a finite collection of $\m{C}$-valued torsion invariants for $X$ which are algebraic and so can be assembled into an rational polynomial called the torsion polynomial $\sigma_{X}(t)$ by\footnote{It is conventional to remove denominators to get a non-monic polynomial in $\m{Z}(t)$.}
\begin{equation} \label{eq: Torsion Polynomial}
\sigma_{X}(t) = \prod_{\rho: \pi_1(X) \to \SL(2, \m{C})} (t - \tau_X(\rho)) \in \m{Q}(t).
\end{equation}
In the literature, $\sigma^{st}_{X}(t)$ has been computed for surgeries on torus knots and the $4_1$ knot \cite{Joh, Kit2016} but $\sigma^{adj}_{X}(t)$ has not appeared.

Given a surgery manifold $S^3_{p/r} (K)$, the particular Reidemeister torsion relevant to the residues $\mathcal{S}_{\bbalpha}^{\bbbeta}$ is expected to be the adjoint torsion $\tau^{adj}_{S^3_{p/r} (K)}(\rho)$. Ideally we would like to find a method to compute this systematically for surgeries $\frac{p}{r}$ as opposed to attempting to proceed via first principles in each case. To do this we need to study the torsion associated to knot complements. In the literature \cite{Kit15, Tra15}, this problem has mainly been studied in the standard representation however there is a marked difference between the standard and adjoint representations in this case as the adjoint representations do not lead to acyclic complexes.

For knot complements\footnote{Sometimes in the literature it is also denoted as $S^3 \setminus N(K)$ or $S^3 \setminus \nu K$, where $N(K)$ or $\nu K$ denotes the tubular neighborhood of $K$. Of these different notations, we choose the most compact one.} $S^3 \setminus K$, a general description of the twisted homology groups was given by Porti \cite{Por15}.

\begin{mylem}[\cite{Por15}, Appendix B] \label{lem: Homology Knot Complement Torsion}
For a generic representation $\rho: \pi_1(S^3 \backslash K) \to \SL(2, \m{C})$, let $V$ denote the adjoint representation of $\SL(2, \m{C})$ and $V_{\rho}$ the induced representation of $\pi_1(K)$. Then
\[
H_i(S^3 \backslash K; V_{\rho}) = \begin{cases}
						\m{C} & i = 1, 2 \\
						0 & i = 0, 3
					\end{cases}.
\]
These groups can be realised as
\[
H_1(S^3 \backslash K; V_{\rho}) = \langle i_*(a\otimes [\gamma]) \rangle \quad \text{ and } \quad H_2(S^3 \backslash K; V_{\rho}) = \langle i_*(a\otimes [T^2]) \rangle
\]
where $i_*: H_1(T^2, V_{\rho}) \to H_1(S^3 \setminus K, V_{\rho})$ is the map induced from the boundary inclusion map, $[T^2]$ is a fundamental class, $[\gamma]$ is any non-zero element in $H^1(T^2)$, and $a$ is the unique invariant vector in $V_{\rho \mid_{T^2}}$.
\end{mylem}

Hence a choice of $[\gamma] \in H^1(T^2)$ determines the adjoint torsion and so we define
\[
    \tau^{adj}_{S^3 \backslash K, \ [\gamma]}(\rho) = \tau^{adj}_{S^3 \backslash K, \ \{i^*(a \otimes [\gamma]), i^*(a \otimes [T^2])\}}(\rho).
\]
For comparison, in the standard case we simply have $\tau^{st}_{S^3 \backslash K}(\rho)$ as the twisted complex is acyclic. 
We prove the following lemma relating $\tau^{adj}_{S^3_{p/r} (K)}(\rho)$ and $\tau^{adj}_{S^3 \backslash K, [l]}(\rho)$. 

\begin{mylem}\label{lem: Reidemeister Torsion for Surgeries}
Let $x$ and $y$ denote eigenvalues of $\rho(m)$ and $\rho(l)$ viewed in the standard representation corresponding to a common eigenvector. Then\footnote{It is interesting to compare this to the corresponding surgery formula for the standard torsion:
\[ 
\tau^{st}_{S^3_{p/r} (K)}(\rho) = \frac{\tau^{st}_{S^3 \backslash K}(\rho)}{2 - y - y^{-1}}
\]}
\[
\tau^{adj}_{S^3_{p/r} (K)}(\rho) = \frac{\big(p \frac{y}{x} \frac{dx}{dy} + r\big)\tau^{adj}_{S^3 \backslash K, [l]}(\rho)}{2 - y^{2} - y^{-2}}.
\]
\end{mylem}

We start by computing the adjoint torsion in the simple cases, $S^1$ and $T^2$.
            
\begin{myprop}
For a generic representation $\rho$, $H_*(S_1, V_{\rho}) \cong H_0 \oplus H_1 \cong \langle [a\otimes p], [a\otimes x] \rangle$ where $p, x$ are generators of $H_0(S^1), H_1(S^1)$ and $a \in V_{\rho}$ is an invariant vector. Then
\[
\tau^{adj}_{S^1}(\rho) = \tau^{adj}_{S^1, \ \{[a\otimes p], [a \otimes x]\}}(\rho) = \frac{1}{2 - y^{2} - y^{-2}}
\]
where $y$ is an eigenvalue of $\rho(x) \in \SL(2, \m{C})$ viewed in the standard representation. Similarly for $T^2$, $H_*(T^2, V_{\rho}) \cong \langle [a\otimes p], [a\otimes m], [a\otimes l], [a\otimes T^2] \rangle$ and
\[
\tau^{adj}_{T^2}(\rho) = \tau^{adj}_{S^1, \ \{[a\otimes p], [a \otimes m], [a \otimes l], [a \otimes T^2]\}}(\rho) = 1.
\]
\end{myprop}

\begin{proof}
We explicitly demonstrate the case of $S^1$. Choosing a basis which diagonalizes $\rho(x)$ we find
\[
\rho(x) = \begin{pmatrix}
y & 0 \\
0 & y^{-1} \\
\end{pmatrix}
\]
and so, passing to the adjoint representation of $\SL(2, \m{C})$, $\rho$ acts as
\[
\rho^{adj}(x) = \begin{pmatrix}
1 & 0 & 0 \\
0 & y^2 & 0 \\
0 & 0 & y^{-2} \\
\end{pmatrix}.
\]
The invariant vector is clearly $a = \left(\begin{smallmatrix} 1\\ 0 \\ 0          \end{smallmatrix}\right)$ and so the induced map $\partial^*: C^1(S^1, V_{\rho})/\langle a\otimes x \rangle \to C^0(S^1, V_{\rho})/\langle a\otimes p \rangle$ is given by
\[
                    \partial^* = I_2 - \rho^{adj}(x)|_{\{2, 3\}} = \begin{pmatrix}
                        1 - y^2 & 0 \\
                        0 & 1 - y^{-2} \\
                    \end{pmatrix}.
\]
Hence
\[
                    \tau^{adj}_{T^2}(\rho) = \frac{1}{\det(\partial^*)} = \frac{1}{2 - y^{2} - y^{-2}}.
\]
The proof for $\tau^{adj}_{T^2}(\rho)$ is identical.
\end{proof}
            
The main tool we need to prove Lemma \ref{lem: Reidemeister Torsion for Surgeries} is the following Theorem of Milnor:

\begin{mythm}[\cite{Mil66}: Theorem 3.2]
Suppose that
\[
0 \to C'_* \xrightarrow{i} C_* \xrightarrow{j} C''_* \to 0
\]
is a short exact sequence of chain complexes giving rise to the long exact sequence of homology
\[
H_* = \cdots \to H_1(C''_*) \xrightarrow{\partial} H_0(C'_*) \xrightarrow{i} H_0(C_*) \xrightarrow{j} H_0(C''_*) \to 0.
\]
For each $k$, choose compatible\footnote{Denoting the chosen elements as $\textbf{c}'_k$, $\textbf{c}_k$, $\textbf{c}''_k$, compatible means that $\textbf{c}_k = i(\textbf{c}'_k)\wedge \textbf{b}_k$ with $j(\textbf{b}_k) = \textbf{c}''_k$.} volume elements in $C'_k, C_k, C''_k$ such that the torsion of the short exact sequence is $1$. Then
\[
\tau_C = \tau_{C'}\tau_{C''}\tau_{H}.
\]
\end{mythm}
Given a surgery $S^3_{p/r} (K) = (S^3 \setminus K) \cup_{T^2} (S^1 \times D^2)$ we have a corresponding Mayer-Vietoris-like sequence
\[
0 \to C_*(T^2, V_{\rho}) \xrightarrow{i} C_*(S^3 \setminus K, V_{\rho}) \oplus C_*(S^1, V_{\rho}) \xrightarrow{j} C_*(S^3_{p/r} (K), V_{\rho}) \to 0
\]
and the above theorem yields
\[
\tau^{adj}_{S^3_{p/r} (K)}(V_{\rho}) = \frac{\tau^{adj}_{S^3 \setminus K, \textbf{h}}(V_{\rho})\tau^{adj}_{S^1, \textbf{h}'}(V_{\rho})}{\tau^{adj}_{T^2, \textbf{h}''}(V_{\rho})\tau^{adj}_{H}} = \frac{\tau^{adj}_{S^3 \setminus K, [\gamma]}(V_{\rho})\tau^{adj}_{S^1}(V_{\rho})}{\tau^{adj}_{H}}.
\]
With a careful choice of $\gamma$, we can force $\tau_{H} = 1$ yielding
\begin{myprop}[\cite{Por15}: Proposition 4.23]
Fix a group homomorphism $\rho: \pi_1(S^3_{p/r} (K)) \to \SL(2, \m{C})$. Then
\begin{align*}
\tau^{adj}_{S^3_{p/r} (K)}(\rho) & = \tau^{adj}_{S^3 \setminus K, \ p[m] + r[l]}(\rho) \tau^{adj}_{S^1}(\rho)
\end{align*}
\end{myprop}

As $H_1(S^3 \setminus K)$ is one-dimensional, equation \eqref{eq: Torsion Basis Change} shows
\[
\tau^{adj}_{S^3 \setminus K, \ p[m] + r[l]}(\rho) = p\ \tau^{adj}_{S^3 \setminus K, [m]}(\rho) + r\ \tau^{adj}_{S^3 \setminus K, [l]}(\rho)
\]
and we can express $\tau^{adj}_{S^3 \backslash K, [m]}(\rho)$ in terms of $\tau^{adj}_{S^3 \backslash K, [l]}(\rho)$ via the following lemma.
\begin{mylem}[\cite{Por95}] \label{lem: Change of Coords}
Let $x$ and $y$ denote eigenvalues of the meridian and longitude in the standard representation. Then
\[
    				\tau^{adj}_{S^3 \backslash K, [m]}(\rho) = \pm \frac{y}{x} \frac{dx}{dy} \tau^{adj}_{S^3 \backslash K, [l]}(\rho).
\]
\end{mylem}
This $\pm$ disappears if we force $x, y$ to be eigenvalues of a common eigenvector. Combining this all with our previous computation of adjoint torsion for the circle completes the proof of Lemma \ref{lem: Reidemeister Torsion for Surgeries}.
			
Let us turn now to studying how to explicitly compute these torsions, illustrating the computations with examples of surgeries on hyperbolic knots $4_1$ and $5_2$. The main difficulty is in computing $\tau^{adj}_{S^3 \backslash K, [l]}(\rho)$ for which we will need to take a brief detour to discuss twisted Alexander polynomials.

\subsubsection{Twisted Alexander Polynomials}

Let $\alpha: \pi_1(S^3 \setminus K) \to \m{Z} = \langle t\rangle$ denote the abelianization homomorphism. Then any linear representation $\rho: \pi_1(S^3 \setminus K) \to \SL(n, \m{C})$ can be lifted to a representation $\rho \otimes \alpha: \pi_1(S^3 \setminus K) \to \GL(n, \m{C}(t))$.
		    
\begin{mydef}[\cite{Kit96}]
For generic t, $C_*(S^3 \setminus K, V_{\rho \otimes \alpha})$ is acyclic letting us define
\[
\Delta_{K, \rho}(t) = \tau_{S^3 \setminus K}(\rho \otimes \alpha).
\]
This family of invariants are known as the twisted Alexander polynomials.
\end{mydef}
When $\rho(g) = 1$ is the trivial representation, $(1 - t)\Delta_{K, 1}(t)$ is the Alexander polynomial, justifying the name. These invariants were initially described in a different context in \cite{Lin01, Wada94} before being related to the Reidemeister torsion in \cite{Kit96}.
Assuming that $\rho$ lands in a non trivial irreducible representation $\textbf{n}$ of $\SL(2, \m{C})$, consider what happens in the limit as $t \to 1$.
		    
\begin{mythm}[\cite{Yam08}]
There are two possible cases
\begin{itemize}
\item If $H_*(S^3 \setminus K, V_{\rho}) = 0$, then
\[
\lim_{t \to 1} \Delta_{K, \rho}^{\textbf{n}}(t) = \tau_{S^3 \setminus K}^{\textbf{n}}(\rho)
\]
\item If $H_*(S^3 \setminus K, V_{\rho}) \neq 0$, then
\begin{equation} \label{eq: AlexPoly to Torsion}
\lim_{t \to 1} \frac{\Delta_{K, \rho}^{\textbf{n}}(t)}{t - 1} = \tau^{\textbf{n}}_{S^3 \backslash K, [l]}(\rho)
\end{equation}
\end{itemize}
\end{mythm}
		    
Hence, we can compute $\tau^{adj}_{S^3 \backslash K, [l]}(\rho)$ via the computation of $\Delta_{K,\rho}^{adj}(t)$. This is easier as the acyclic situation is far simpler to work with.
Computations of $\Delta_{K, \rho}^{\textbf{n}}(t)$ have been done for certain knots and representations in the literature \cite{Tra13, Tra2015}, and here we show explicitly how it works for the adjoint representation for our class of manifolds. 
First, recall the definition of the Fox derivative.
\begin{mydef}
Given a free group $F$ with generators $g_i$ the Fox derivative is the function $\frac{\partial}{\partial g_i}: \m{Z}[F] \to \m{Z}[F]$ defined by
\begin{align*}
\frac{\partial}{\partial g_i}g_j & = \delta_{ij}
\\ \frac{\partial}{\partial g_i}e & = 0
\\ \frac{\partial}{\partial g_i}(uv) & = \frac{\partial}{\partial g_i}(u) + u \frac{\partial}{\partial g_i}(v)
\end{align*}
\end{mydef}
Fix a Wirtinger presentation\footnote{With a little care this definition can be extended to work with any deficiency $1$ representation.}  of the knot group
\[
\pi_1(S^3 \backslash K) = \langle g_1, \cdots, g_n \mid r_1, \cdots, r_{n - 1}\rangle
\]
where each $g_i$ is a meridian and, given a representation $\rho: \pi_1(S^3 \setminus K) \to \SL(2, \m{C})$, denote $X_i = \rho(g_i)$. Then it is well known that $S^3 \setminus K$ retracts onto a $2$-complex with one $0$-cell, $n$ $1$-cells labelled $g_1, \cdots, g_n$ and $(n - 1)$ $2$-cells with attaching maps given by $r_1, \cdots, r_{n - 1}$. From this we can compute the chain complex $C_*(S^3 \setminus K; V_{\rho})$ to be
\[
0 \to V^{\oplus (n - 1)} \xrightarrow{\partial_2} V^{\oplus n} \xrightarrow{\partial_1} V \to 0
\]
where\footnote{Note that $\rho(\frac{\partial r_i}{\partial g_j})$ means computing $\frac{\partial r_i}{\partial g_j}$ in $\m{Z}[F]$ and then taking the natural quotient $\m{Z}[F] \to \m{Z}[G]$ before applying $\rho$.}
\[
\partial_2 = A = \begin{pmatrix}
\rho(\frac{\partial r_1}{\partial g_1}) & \cdots & \rho(\frac{\partial r_1}{\partial g_n})
\\ \vdots & \ddots & \vdots 
\\ \rho(\frac{\partial r_{n - 1}}{\partial g_1}) & \cdots & \rho(\frac{\partial r_{n - 1}}{\partial g_n})
\end{pmatrix}
\]
and
\[
\partial_1 = \begin{pmatrix}
 \rho(g_1) - I \\ \vdots \\ \rho(g_n) - I
\end{pmatrix}
\]
Let $A_i$ denote the square matrix where we have removed the $i$'th column from $A$. Then a careful computation shows\footnote{While this was initially proven by \cite{Joh}, a clearer proof was given by Kitano in \cite{Kit94}, Theorem 2.1.}
\begin{mythm}[Johnson] \label{thm: Johnson Acyclic Knot Complement Torsion}
Assuming $C_*(S^3 \setminus K; V_{\rho})$ is acyclic, there exists an $i$ such that $\det(X_i - I)$ and $\det A_i$ are both non-zero and, using this, Reidemeister torsion is given by
\[
\tau(S^3 \setminus K, V_{\rho}) = \frac{\det\left(A_i\right)}{\det(\rho(g_i) - I)}.
\]
This is independent of choice of $i$, up to overall factors of $t$.
\end{mythm}

This is easy to compute for any knot and is particularly simple when the knot group admits a presentation with $2$ generators and $1$ relation, as is the case for the $4_1$ and $5_2$ knots which we consider next.

\begin{table}[h]
    \centering
            \begin{tabular}{|c|c|c|}
                \hline
				\makecell{Twist Knot} & \makecell{Rolfsen Table Label} & \makecell{Alexander Polynomial}
				\\ \hline \hline
				$K_0$ & $0_1$ & $1$
				\\ \hline
				$K_1$ & $3_1$ & $t^{-1} - 1 + t$
				 \\ \hline
				$K_{-1}$ & $4_1$ & $t^{-1} - 3 + t$
`               \\ \hline
				$K_2$ & $5_2$ & $2t^{-1} - 3 + 2t$
				\\ \hline
				$K_{-2}$ & $6_1$ & $2t^{-1} - 5 + 2t$
				\\ \hline
				$K_3$ & $7_2$ & $3t^{-1} - 5 + 3t$
				\\ \hline
				$K_{-3}$ & $8_1$ & $3t^{-1} - 7 + 3t$
				\\ \hline
			\end{tabular}
\vspace{3mm}
\caption{Twist knot conventions.}
\label{tab:twist knots}
\end{table}

\subsubsection{Riley polynomials}
\label{sec:41_52_computations}

The $K_n$ twist knot has knot group
\[
\langle g, h \mid h^{-1} \omega^n g \omega^{-n} \rangle \quad \quad \omega = hg^{-1}h^{-1}g
\]
where $h, g$ are two meridians and the longditude is 
\[
\lambda = \overleftarrow{\omega}^n\omega^n \text{ where } \overleftarrow{\omega} = gh^{-1}g^{-1}h.
\]
Given a representation $\rho: \pi_1(S^3 \setminus K_n) \to \SL(2, \m{C})$, as the generators are conjugate we can assume that, up to conjugation
\[
\rho(g) = \begin{pmatrix}
x & x^{-1} \\ 0 & x^{-1}
\end{pmatrix} \text{ and } \rho(h) = \begin{pmatrix}
x & 0 \\ -xu & x^{-1}
\end{pmatrix}.
\]
If we compute $\rho(\omega^n g) - \rho(h\omega^n)$ we find that
\[
\rho(\omega^{-1} g) - \rho(h\omega^{-1}) =  \begin{bmatrix}
0 & x^{-1} \phi_{n}(x, u) \\ 
u x \ \phi_{n}(x, u) & 0
\end{bmatrix}
\]
for a polynomial $\phi_n(x, u)$ known as the Riley polynomial \cite{Ril}. To simplify notation, let
$$
\theta_m := x + x^{-1}
$$
denote the trace of the meridian, and
$$
\theta_{m, j} := x^j + x^{-j}
$$
be higher order functions of the trace. Then, for the knots $4_1$ and $5_2$ we have $n = -1$ and $n=2$, respectively, so that
\begin{align*}
\phi_{-1}(x, u) & = u^2 - (u + 1)(\theta_{m, 2} - 3) \\
\phi_{2}(x, u) & = u^3 + (3 - 2\theta_{m, 2})(u^2 + 1) + (6 - 3\theta_{m, 2} + \theta_{m, 4})u.
\end{align*}
Similarly, we can compute the longitude in each of these cases to get
\begin{align*}
\rho(l_{4_1}) & = \begin{pmatrix}
y_{4_1}(x, u) & -\frac{\theta_m (\theta_{m, 2} - 3 - 2u)}{x} \\ 0 & y_{4_1}(x^{-1}, u)
\end{pmatrix}
\\ y_{4_1}(x, u) & = \frac{1 - 2x^2 - ux^2 - x^4 + x^6 + ux^6}{x^4}
\end{align*}
and
\begin{align*}
\rho(l_{5_2}) & = \begin{pmatrix}
y_{5_2}(x, u) & \frac{\theta_m(u\theta_{m, 8} - (2 + 2u + u^2)\theta_{m, 6} + (1 + 2u)\theta_{m, 4} - \theta_{m, 2} - 1)}{x} \\ 0 & y_{5_2}(x^{-1}, u)
\end{pmatrix}
\\ y_{5_2}(x, u) & = 1 - x^2(1 + u)(2 - x^4 + x^6) + (1 + u)^2x^4 - u^2 x^8 + ux^{10}
\end{align*}
Note that $y(x^{-1}, u) = y(x, u)^{-1}$. (One needs the corresponding Riley polynomial to realise this relation.) From this we can read off the torsion corresponding to the gluing torus (for $y \neq \pm 1$) to be
\begin{align*}
    \tau^{adj}_{S^1, 4_1} & = \frac{1}{2 - y_{4_1}^2 - y_{4_1}^{-2}} = \frac{1}{(\theta_{m, 2} + 2)(\theta_{m, 2} + 1)(\theta_{m, 2} - 2)(\theta_{m, 2} - 3)}
\end{align*}
and
\begin{align*}
    \tau^{adj}_{S^1, 5_2} & = \frac{1}{2 - y_{5_2}^2 - y_{5_2}^{-2}}
    \\ & = \Big(-u\theta_{m, 20} + (2 + 3u + u^2)\theta_{m, 18} - (3 + 3u + u^2)\theta_{m, 16} - (3u + u^2)\theta_{m, 14}
    \\ & \quad \quad  + (7 + 8u + 3u^2)\theta_{m, 12} - (4 + u - u^2)\theta_{m, 10} - 3(2 + 3u + u^2)\theta_{m, 8}
    \\ & \quad \quad + (6 + 3u - u^2)\theta_{m, 6} + (2 + 5u + u^2)\theta_{m, 4} - 2(2 + u)\theta_{m, 2}\Big)^{-1}
\end{align*}

We can also compute the derivative $\frac{dx}{dy}$. In practice, it is easier to compute $\frac{dy}{dx}$, which can be done in the following two-step process. We first use the Riley polynomial to compute $\frac{du}{dx}$ and then we can differentiate the expressions $y(x, u)$ given above. We get
\begin{align*}
\frac{du_{4_1}}{dx} & = - \frac{2(x - x^{-1})\theta_m(1 + u)}{x(\theta_{m, 2} - 3 - 2u)} \\
\frac{du_{5_2}}{dx} & = - \frac{2(x - x^{-1})\theta_m(2u\theta_{m, 2} - 2 - 3u - 2u^2)}{x(\theta_{m, 4} - (3 + 4u)\theta_{m, 2} + 3(2 + 2u + u^2))}
\end{align*}
and, using this we find\footnote{There are many equivalent expressions here, depending on different simplification procedures. This one is chosen for simplicity and to make the $x \to x^{-1}$ Weyl symmetry manifest.} 
\begin{align*}
\frac{y_{4_1}}{x} \frac{dx}{dy_{4_1}} & = \frac{3 + 2u - \theta_{m, 2}}{2(2\theta_{m, 2} - 1)}
\\ \frac{y_{5_2}}{x} \frac{dx}{dy_{5_2}} & = \frac{19 + 29u + 5u^2 - (19 + 11u + u^2)\theta_{m, 2} - (2 + u)\theta_{m, 4} + 2\theta_{m, 6}}{2(39 + 105u + 7u^2) - 2(58 + 49u + 21u^2)\theta_{m, 2} + 2(21u - 4)\theta_{m, 4}}
\end{align*}

Finally we need to compute the twisted Alexander polynomial in both these cases. Passing to the adjoint representation, we find that our matrices become
\[
\rho^{adj}(g) = \begin{pmatrix}
					1 & 0 & x^{-2} \\ -2 & x^2 & x^{-2} \\ 0 & 0 & x^{-2}
\end{pmatrix} \text{ and } \rho^{adj}(h) = \begin{pmatrix}
					1 & u x^2 & 0 \\ 0 & x^2 & 0 \\ -2u & -u^2x^2 & x^{-2}
\end{pmatrix}.
\]
in the basis $\{h, e, f\}$. For the generic twist knot $K_n$,
\[
    A_2 = \rho\left(\frac{\partial (h^{-1}\omega^n g\omega^{-n})}{\partial g}\right) = \rho\left(-h^{-1} + (h^{-1} - 1) \frac{\partial \omega^n}{\partial \omega}(1 - hg^{-1}h^{-1})\right).
\]
where
\[
\frac{\partial \omega^n}{\partial \omega} = \begin{cases}
1 + \omega + \cdots + \omega^{n - 1} & n \geq 0
\\ -\omega^{-1} - \cdots - \omega^{-n} & n < 0.
\end{cases}
\]
For our two examples of ${\bf 4_1}$ and ${\bf 5_2}$ knots, this simplifies to
\begin{align*}
A_2^{4_1} & = \rho\left(-h^{-1} + (h^{-1} - 1)(g^{-1} - \omega^{-1})\right) \\
			    A_2^{5_2} & = \rho\left(-h^{-1} + (h^{-1} - 1)(1 + \omega)(1 - hg^{-1}h^{-1})\right)
\end{align*}
and applying Theorem \ref{thm: Johnson Acyclic Knot Complement Torsion} we get
\begin{align*}
\Delta^{adj}_{4_1, \rho}(t) & = (1 - t)(2t(\theta_{m, 2}) - 1 + t - t^2)\\
\Delta^{adj}_{5_2, \rho}(t) & = (1 - t)\Big(2(1 + t^2)(\theta_{m, 2}^2u - \theta_{m, 2} - (1 + \theta_{m, 2})(1 + u + u^2))
\\ & \quad \quad + (3 + (4 + \theta_{m, 4})u + (1 - \theta_{m, 2})(2 + 3u + u^2))t\Big)
\end{align*}
Therefore, eq.\eqref{eq: AlexPoly to Torsion} gives
\begin{align*}
    \tau^{adj}_{S^3 \backslash 4_1, [l]}(\rho) & = (2\theta_{m, 2} - 1) \\
    \tau^{adj}_{S^3 \backslash 5_2, [l]}(\rho) & = 1 - 10\theta_{m, 2} + (11 - 7\theta_{m, 2} + 5\theta_{m, 4})u - (3 + 5\theta_{m, 2})u^2
\end{align*}
Putting it all together we find
\begin{eqnarray}
\tau^{adj}_{S^3_{p/r}(4_1)}(\rho) & =& \frac{\frac{p}{2} (3 + 2u - \theta_{m, 2}) + r\ (2\theta_{m, 2} - 1)}{(\theta_{m, 2} + 2)(\theta_{m, 2} + 1)(\theta_{m, 2} - 2)(\theta_{m, 2} - 3)}
\label{eq:41-torsion}
\\ \tau^{adj}_{S^3_{p/r}(5_2)}(\rho) & =& \frac{p \ \frac{y_{5_2}}{x} \frac{dx}{dy_{5_2}}\tau^{adj}_{S^3 \backslash 5_2, [l]}(\rho) + r\ \tau^{adj}_{S^3 \backslash 5_2, [l]}(\rho)}{\tau^{adj}_{S^1, 5_2}}
\label{eq:52-torsion}
\end{eqnarray}

\subsection{Examples: surgeries on small twist knots} \label{sec: algorithm}

Combining the previous two subsections we end up with a simple algorithm to compute the Chern-Simons values and torsions corresponding to each flat connection on a general knot surgery $M_3 = S^3_{p/r} (K)$ and for surgeries on twist knots \eqref{Msurgery} in particular.

\begin{enumerate}
\item Find all intersections between the curves $A_{K}^{irred}(x, y) = 0$ and $y^rx^p = 1$ with $x, y \neq 0, 1, -1$. These come in pairs $(x, y), (x^{-1}, y^{-1})$. For the sake of consistency, pick solutions so that $\Im(x) > 0$ or $\Im(x) = 0$ and $|x| > 1$. 
\item For each solution $(x^*, y^*)$ determine $u$ by solving $\phi(x^*, u) = 0, y(x^*, u) = y^*$ and use this to compute $\tau^{adj}_{S^3_{p/r} (K)}$. 
\item Fix a root of $\Delta_{4_1}(x^2)$; the choice is irrelevant but for simplicity we choose the root $\chi$ which minimises $|\log(\chi)|$.
\item For each solution $(x, y)$, find a path from $(\chi, 1)$, the intersection with the abelian branch, to $(x, y)$. This will involve computing which sheet $(x, y)$ lies on and may involve passing through intersection points. Additionally, determine the appropriate continuous extension of the Log function.
\item Apply \eqref{eq: Computing CS Values, explicit} to determine the Chern-Simons value and normalise to get an answer in the interval $[-\frac{1}{2}, \frac{1}{2}]$.
\end{enumerate}

After computing the complete set of invariants for a given knot, we additionally give normalisations which are important for comparisons to numerical results in the Borel Plane.

\begin{itemize}
\item The normalized CS invariants are given by $\frac{\text{CS}_{\alpha}}{\text{CS}_{\rm leading}}$ where $\text{CS}_{\rm leading}$ denotes the smallest invariant. These give a clearer picture of the relative magnitude of these invariants. This is important when identifying the CS invariants  with Borel singularities, as their relative distance from the origin is an important physical consideration, the closest ones being related to more dominant non-perturbative effects, whereas the non-perturbative influence of more distant Borel singularities is exponentially suppressed.
\item The residues corresponding to Borel Singularities, (Stokes Constants, ${\mathcal S}_{\alpha} := {\mathcal S}_0^{\alpha}$), are related to the Adjoint Reidemeister Torsion $\tau_{\alpha}$ via:
\begin{eqnarray}
{\mathcal S}_0^{\alpha} = \frac{1}{\sqrt{4 \pi \tau_{\alpha} \ (-4\pi^2 \text{CS}_{\rm leading})^3}}
\label{eq:stokes-torsion}
\end{eqnarray}
\end{itemize}

\subsubsection{Surgeries on the \texorpdfstring{$4_1$}{4\_1} knot}
\label{sec:41-Apoly}

To illustrate the computational procedure, we start with the $4_1$ knot, the simplest hyperbolic knot. The irreducible $A$-polynomial is
\begin{equation}
A_{4_1}^{irred}(x, y) = (-x^4 + (1 - x^2 - 2 x^4 - x^6 + x^8) y - x^4 y^2) 
\end{equation}
As this is quadratic in $y$, it forms a $2$ sheeted branched covering space over $(\m{C}^*)_x$ with local\footnote{It is impossible to continuously extend $y_1(x)$ and $y_2(x)$ to all of $(\m{C}^*)_x$ as the roots of a polynomial equation form an unordered set.} sections $y_1(x)$ and $y_2(x)$.
In this simple case, at each root of $\Delta_{4_1}(x^2)$, both irreducible branches meet the Abelian branch. This means we can always choose straight line paths from roots of $\Delta_{4_1}(x^2)$ to our desired points and we will not have to move between different irreducible branches.
From the $4_1$ Alexander polynomial
\[
\Delta_{4_1}(x^2) = -x^2 - x^{-2} + 3
\]
we set our intersection point with the abelian branch to be $(x, y)=\left(\frac{1 + \sqrt{5}}{2}, 1\right)$.

Let us start by looking at the $0$-surgery. Setting $y = 1$, we find
\begin{equation}
A_{4_1}^{irred}(x, 1) = x^2(1 + x^2)^2(x^2 + x^{-2} - 3) =- x^2(1 + x^2)^2\Delta_{4_1}(x^2).
\end{equation}
Hence there is a single interesting pair of roots located at $(\pm i, 1)$ with multiplicity $2$. Visualising our branches $y_1(x), y_2(x)$ along the straight line from $x = \frac{1}{2}(1 + \sqrt{5})$ to $x = i$ produces Figure \ref{fig: Path to 0 sugery 41}. While in $\B{A}_{4_1}$ it appears that these branches join at $(1, i)$, when lifted to the representation variety these branches separate. Hence, the two possible CS values for the $0$-surgery correspond to travelling to $(1, i)$ along each branch. As both of the loops in Figure \ref{fig: Path to 0 sugery 41} do not contain $0$, we stay on the initial logarithm branch and so the boundary term in \eqref{eq: Computing CS Values, 0-surg} vanishes. Thus we find\footnote{Up to a precision of $10^{-100}$.} that the Chern-Simons invariants of the upper and lower branches are $-\frac{1}{5}$ and $+\frac{1}{5}$, respectively, {\it cf.} \cite{KK}.

\begin{figure}
\centering
\includegraphics{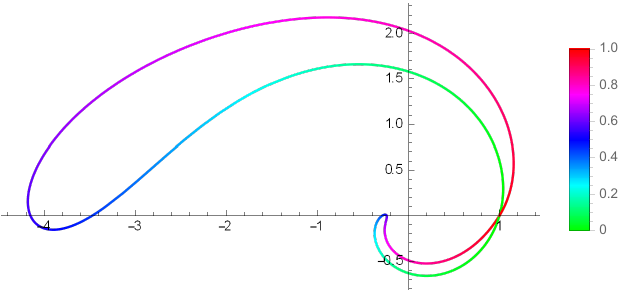}
\caption{Branches of $\B{A}_{4_1}$ along the line $x(t) = \frac{1 - t}{2}(1 + \sqrt{5}) + t \ i$. The colour indicates the $t$ value and the position is the complex value of $y_1(x(t))$ or $y_2(x(t))$. Intersection points where the colours align correspond to sheet intersection points in $\B{A}_{4_1}$ and so this diagram shows  two distinct paths from $(\frac{1 - t}{2}(1 + \sqrt{5}), 1)$ to $(i, 1)$.}
\label{fig: Path to 0 sugery 41}
\end{figure}

Let us now move to the more interesting case of $-\frac{1}{2}$ surgery. This surgery enforces $x = y^2$ and so flat connections correspond to roots of
\begin{eqnarray}
A_{4_1}(y^2, y) &=& y (1 + y)^2 \times\\
&&\hskip -3cm   \left(1 - 2 y + 3 y^2 - 4 y^3 + 4 y^4 - 4 y^5 + 4 y^6 - 5 y^7 + 4 y^8 - 4 y^9 + 4 y^{10} - 4 y^{11} + 3 y^{12} - 2 y^{13} + y^{14}\right) \nonumber
\end{eqnarray}
The $x$ values of these roots are plotted in Figure \ref{fig: 41 -1/2 surg intersections}. Note that several flat connections land almost directly on an intersection with the abelian branch.
\begin{figure}
\centering
\includegraphics[width=\linewidth]{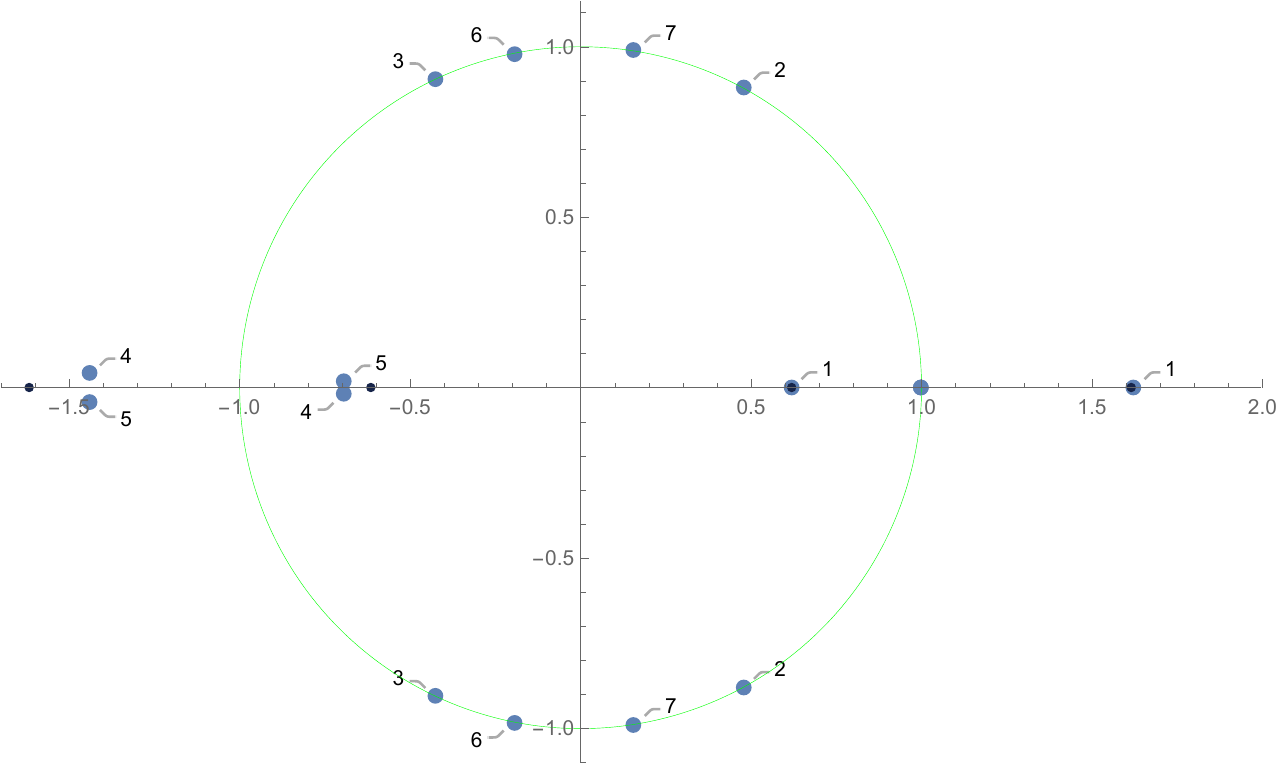}
\caption{Intersections of the irreducible branch with $y^2 = x$ for the $4_1$ knot. Intersections with the abelian branch are in black.}
\label{fig: 41 -1/2 surg intersections}
\end{figure}
Modding out by the Weyl symmetry, there are $7$ pairs of intersection points and a fixed point of multiplicity $2$ at $(1, -1)$. 
From here we simply apply the remainder of the algorithm described at the start of this section to compute all CS and torsion invariants. The results are shown in Table \ref{tab: 4_1 knot}.

Let us illustrate the computational procedure using the example of the leading Chern-Simons invariant, which corresponds to $x = 1.622$, as in the first row of Table \ref{tab: 4_1 knot}. As in Section \ref{sec:41_52_computations}, define $\theta_{m, 2}^{\rm leading}\equiv (x^{\rm leading})^2+(x^{\rm leading})^{-2}=3.012$. Then the A-polynomial expression (\ref{eq: Computing CS Values, 1/r-surg}) for the Chern-Simons invariant can be expressed\footnote{While in general the integration contour needs to stay on $\B{A}_{4_1}$, for small contours like this one it doesn't matter.} (after some simple integrations-by-parts) as:
\begin{eqnarray}
\text{CS}_{\rm leading} &=&\frac{1}{8\pi^2}\left(\left[\log\left(x^{\rm leading}\right)\right]^2 \right. \\
&& \left. -\frac{1}{2} \int_3^{\theta_{m, 2}^{\rm leading}} d\theta\, \log\left(\frac{\theta+\sqrt{\theta^2-4}}{\theta-\sqrt{\theta^2-4}}\right)\frac{2(2\theta-1)}{\sqrt{(\theta+1)(\theta^2-4)(\theta-3)}}\right)
\\
 &=& -0.0029434014775824953 ....
\end{eqnarray}
This is the $CS$ invariant value listed in the first row of Table \ref{tab: 4_1 knot}. The lower limit on the $\theta$ integration comes from the chosen reference intersection point, $x^{\rm ref}=\frac{1}{2}(1+\sqrt{5})$, for which $\theta_{m, 2}^{\rm ref}\equiv (x^{\rm ref})^2+(x^{\rm ref})^{-2}=3$, as described in Section \ref{sec:A-polynomial}.

Similarly, from  the Adjoint Reidemeister torsion analysis in Section \ref{sec:41_52_computations}, we use expression (\ref{eq:41-torsion}) with $p=-1$, $q=2$,  $u$ evaluated at the vanishing of the Riley polynomial, $\phi_{-1}(x, u)=0$, all evaluated at $\theta_{m, 2}^{\rm leading}$:
\begin{eqnarray}
\tau^{adj}_{S^3_{-1/2}(4_1)}({\rm leading}) & =&
\left[\frac{-\frac{1}{2} \sqrt{\left(\theta_{m, 2}-3\right)\left(\theta_{m, 2}+1\right)} +2(2\theta_{m, 2}-1) }{(\theta_{m, 2} + 1)(\theta_{m, 2} + 2)(\theta_{m, 2} - 2)(\theta_{m, 2} - 3)}\right]_{\theta_{m, 2}=\theta_{m, 2}^{\rm leading}}
\\
&=& 41.6374502692239...
\end{eqnarray}
This is the Torsion value listed in the first row of Table \ref{tab: 4_1 knot}. Note that both of these invariants can be evaluated to any desired precision and the procedure is similar for the other intersection pairs in Table \ref{tab: 4_1 knot}.

Note that since $4_1$ is an amphichiral knot, for $+\frac{1}{2}$ surgery on $4_1$ the corresponding results for the Chern-Simons and torsion invariants only differ by simple overall minus signs.

As a brief cross check, the intersections lying on the unit circle (labelled $\alpha = 2, 3, 6, 7$ in Table \ref{tab: 4_1 knot}) can be realised in $\SU(2)$ and in those cases our CS results match \cite{KK}.  The Chern-Simons invariants for the other saddles $\alpha$ in Table \ref{tab: 4_1 knot} are new. More discussion of observations from these computations can be found below in Section \ref{sec: Observations, Torsions and CS values}. We also note that this example of the $4_1$ knot is representative of a large class of hyperbolic twist knots, as discussed in Section \ref{sec:thooft-small}.

\begin{table}
\centering
\begin{tabular}{|c|c|c|c|c|c|c|}
                \hline
$\alpha$ & $x$ & $y$ & CS Invariant & 
Torsion & \makecell{Normalized \\ CS Invariant} & \makecell{Stokes \\ Constant} \\
    			\hline \hline
1 & $1.622270086$ & $1.273683668$ & $-0.002943401$ & $41.637450$ & $1$ & $1.1036698$\\
    			\hline
    			2 & $0.48 + 0.88 i$ & $-0.86 - 0.51 i$ & $-0.485874320$ &  $5.7891940$ & $165.072391$ & $2.9598667$\\
    			\hline
    			3 & $-0.42 + 0.91 i$ & $-0.54 - 0.84 i$ & $0.053933576$ & $2.6968674$ & $-18.323554$ & $4.3366204$ \\
    			\hline
    			4 & $-1.44 + 0.04 i$ & $-0.02 - 1.20 i$ & \multirow{2}{*}{\shortstack{$0.123303626$ \\ $\pm 0.03542464 i$}} & \multirow{2}{*}{\shortstack{$-1.975859$ \\ $\mp 0.076259 i$}} & \multirow{2}{*}{\shortstack{$-41.891542$ \\ $\mp 12.03527 i$}} &
    			\multirow{2}{*}{\shortstack{$-0.097680$ \\ $\pm 5.06362 i$}} \\
    			5 & $-0.69 + 0.02 i$ & $-0.01 - 0.83 i$ & & & & \\
    			\hline
    			6 & $-0.19 + 0.98 i$ & $0.64 + 0.77 i$ & $0.235159766$ & $3.5102399$ & $-79.893881$ & $3.8011307$\\
    			\hline
    			7 & $0.16 + 0.99 i$ & $0.76 + 0.65 i$ & $-0.171882873$ & $6.3179654$ & $58.3960000$ & $2.8333002$\\
\hline
\end{tabular}
\vspace{3mm}
\caption{Chern-Simons and Adjoint Reidemeister Torsion invariants for $- \frac{1}{2}$ surgery on the $4_1$ knot. The normalized CS invariants are the CS invariants divided by the CS invariant of smallest magnitude, and the Stokes constants are related to the torsions via expression (\ref{eq:stokes-torsion}).
See also Section \ref{sec: Observations, Torsions and CS values}. 
These invariants can be evaluated to essentially any degree of precision. }
\label{tab: 4_1 knot}
\end{table}

\subsubsection{Surgeries on the \texorpdfstring{$5_2$}{5\_2} knot}

Surgeries on another interesting class of hyperbolic twist knots is illustrated by surgeries on the $5_2$ knot.
For the $5_2$ knot we have
\begin{equation}
A_{5_2}^{irred}(x, y) = x^{14} + (x^4 - x^6 + 2 x^{10} + 2 x^{12} - x^{14}) y + (-1 + 2 x^2 + 2 x^4 - x^8 + x^{10}) y^2 + y^3
\end{equation}
and so $\B{A}_{5_2}^{irred}$ forms a 3 sheeted branched covering space over $(\m{C}^*)_x$ with local sections $y_1(x), y_2(x), y_3(x)$. There are now several types of intersection points in $\B{A}_{5_2}$:
\begin{itemize}
\item Exactly two sheets of the irreducible branch intersect\footnote{These also occur the $4_1$ case but were not important there.}.
\item A single sheet of the irreducible branch intersects the abelian branch.
\item All three sheets of the irreducible branch intersect.
\item All three sheets of the irreducible branch intersect the abelian branch.
\end{itemize}
We focus on the first two possibilities as the last two are either all spurious or correspond to parabolic representations. This leaves us with a collection of points where exactly $2$ sheets of $\B{A}_{5_2}$ meet, as shown in Figure \ref{fig: 52 branch points}. As
\[
\Delta_{5_2}(x^2) = 2x^2 + 2x^{-2} - 3
\]
we set our intersection point with the abelian branch to be $\left(\frac{1}{2}\sqrt{3 + i \sqrt{7}}, 1\right)=\left(\frac{1}{2\sqrt{2}}(\sqrt{7}+i), 1 \right)$.

\begin{figure}
\centering
\includegraphics[width=0.7\linewidth]{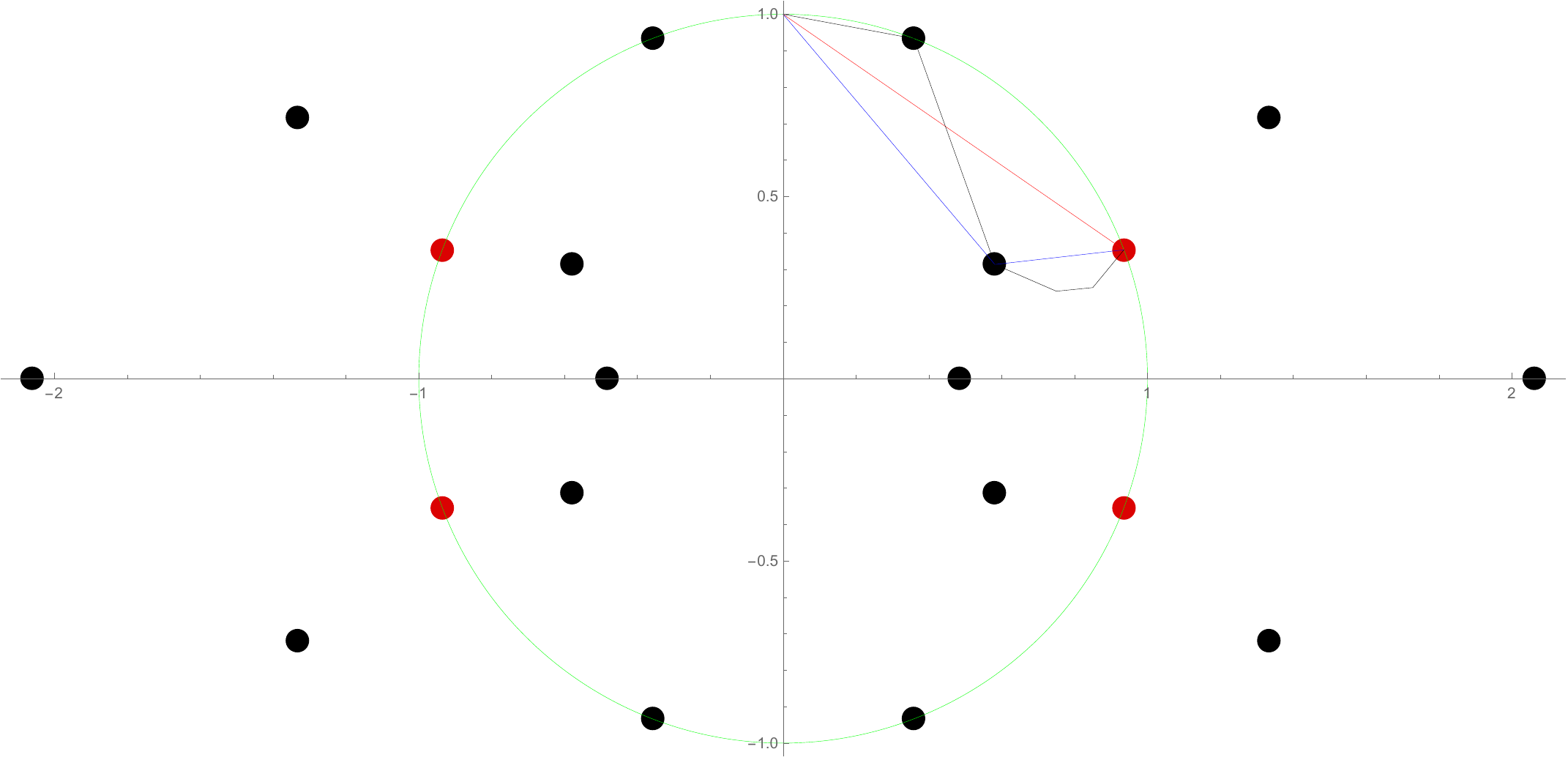}
\caption{A plot of the $x$ values of branch points where exactly $2$ sheets meet, for the $5_2$ knot. Branch points involving the abelian branch are in red. The lines show three paths from the abelian branch to the point $(1, i)$ each arriving on a different sheet of the irreducible branch.}
\label{fig: 52 branch points}
\end{figure}

\begin{figure}
\centering
\includegraphics[width=.3\textwidth]{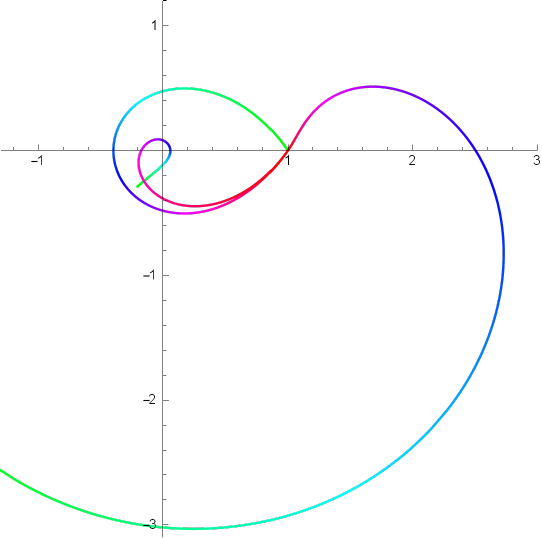}\hfill
\includegraphics[width=.3\textwidth]{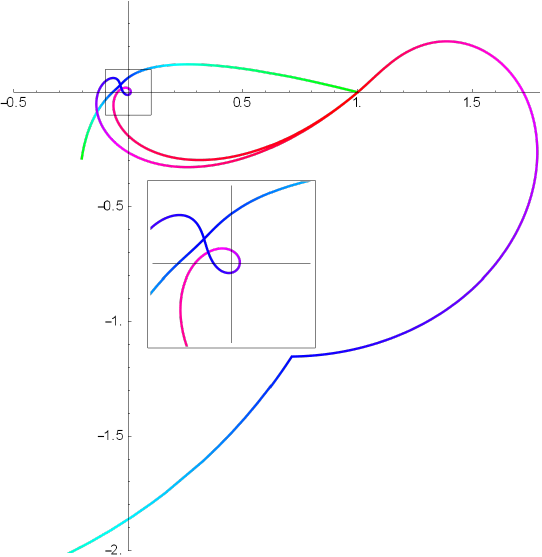}\hfill
\includegraphics[width=.3\textwidth]{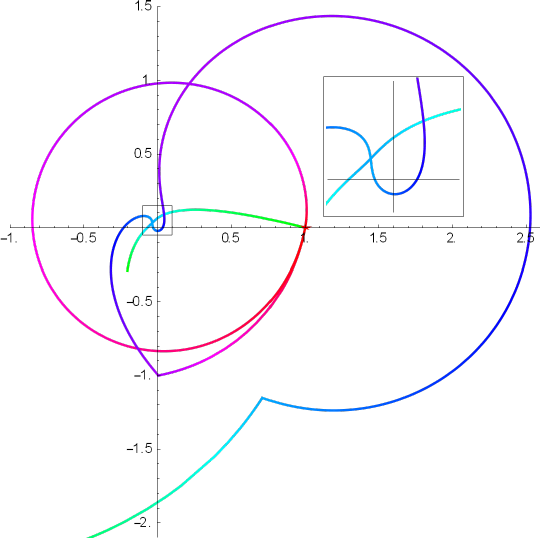}
\caption{Plots of the three branches along respectively the red, blue and black paths in Figure \ref{fig: 52 branch points}. Intersection points where colours line up correspond to intersections of irreducible sheets. The shading is the same as in Figure \ref{fig: Path to 0 sugery 41}, with green corresponding to $t$ close to $0$, and red to $t$ close to $1$.}
\label{fig: Paths to 0 5_2 Knot}
\end{figure}

We again start with analysing $0$-Surgeries as this gives a simpler setting to explain the differences from the $4_1$ case. Setting $y = 1$ the polynomial for the irreducible branch factors as
\begin{equation}
A_{5_2}^{irred}(x, 1) = x^4(1 + x^2)^3\Delta_{5_2}(x^2)
\end{equation}
and so we get three Chern-Simons values corresponding to approaching the point $x = i$ along each of the three sheets. As only one of the irreducible sheets intersects the abelian branch we need to use some intersection points to shift between the irreducible sheets. The simplest path for arriving on each sheet is shown in Figure \ref{fig: 52 branch points}, and the plot of what the three sheets look like along these paths is given in Figure~\ref{fig: Paths to 0 5_2 Knot}.

Looking at Figures \ref{fig: 52 branch points} and  \ref{fig: Paths to 0 5_2 Knot}, we see that along the red and black paths, $y$ circles $0$ once so $\log((\rho_1)_y) = 2\pi i$, and along the blue path it circles twice Equation \eqref{eq: Computing CS Values, 0-surg}, we find that the three CS values are $-\frac{1}{7}, -\frac{2}{7}$ and $-\frac{4}{7}$, {\it cf.} \cite{CGPS}.

Finally we move on to the more interesting example, $\pm \frac{1}{2}$-surgery on the $5_2$ knot. Using the work in the preceding subsections, it is straightforward to extend our algorithm to the $\pm \frac{1}{2}$-surgeries. Intersection points of the irreducible branch with the curves $x = y^{\mp 2}$ correspond to roots of the polynomials 
\begin{eqnarray}
&& A_{5_2}^{irred}(y^{-2}, y) =  \nonumber\\
&& \quad = y^{-28}(1 + y)^3 \left(1 - 4 y + 9 y^2 - 16 y^3 + 25 y^4 - 34 y^5 + 43 y^6 - 52 y^7 + 61 y^8 \right.
\label{eq:52p-apoly}\\
&& \left. \quad - 68 y^9 + 74 y^{10} - 79 y^{11} + 83 y^{12} - 86 y^{13} + 87 y^{14} - 86 y^{15} + 83 y^{16} - 79 y^{17} + 74 y^{18} \right.
\nonumber\\
&& \left. \quad - 68 y^{19} + 61 y^{20} - 52 y^{21} + 43 y^{22} - 34 y^{23} + 25 y^{24} - 16 y^{25} + 9 y^{26} - 4 y^{27} + y^{28}\right) \nonumber\\
&& A_{5_2}^{irred}(y^2, y) =  \nonumber\\
&& \qquad = -y^2(1 + y)^3 (1 - y + y^2 - y^3 + y^4) (1 - 2 y + y^2 - y^4 + y^6 - 2 y^7 + y^8)
\label{eq:52m-apoly}\\
&&  \qquad (1 - y + 2 y^2 - 2 y^3 + 2 y^4 - 3 y^5 + 3 y^6 - 3 y^7 + 2 y^8 - 2 y^9 + 2 y^{10} - y^{11} + y^{12}) \nonumber
\end{eqnarray}
These are shown\footnote{Ignoring spurious points located at $x = 1$.} in Figure \ref{fig: 52 -1/2 surg intersections} along with the branch and sheet intersection points. From this we find that there are $12$ and $14$ pairs respectively. Using these intersection points and applying the general algorithm described above, we obtain the Chern-Simons invariants and Adjoint Reidemeister torsions summarized in Tables \ref{tab: 1/2 surg 52} and \ref{tab: -1/2 surg 52}. Note that these invariants can be computed to essentially any desired precision.

\begin{figure}
\centering
\includegraphics[width=.4\linewidth]{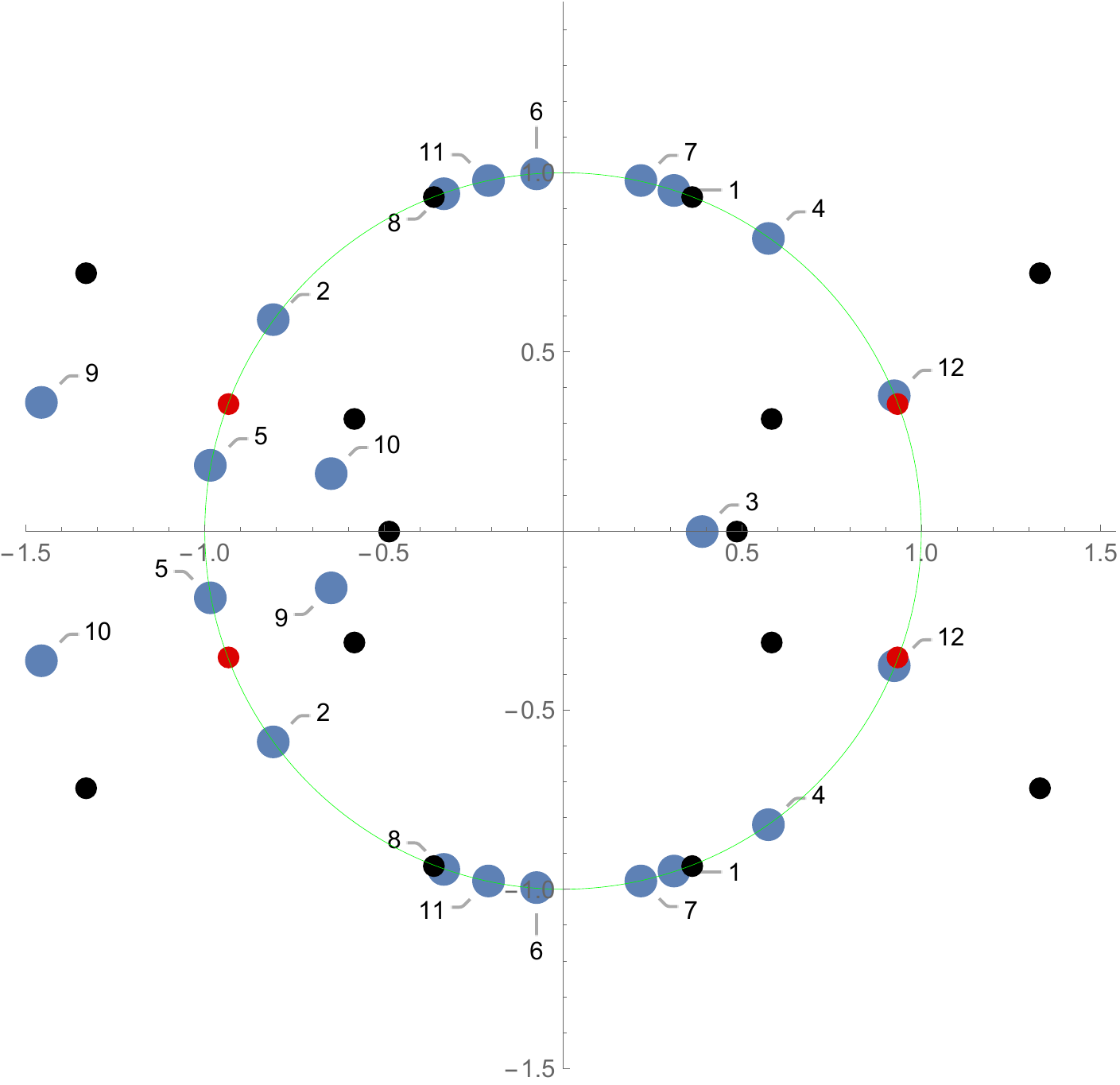}
\includegraphics[width=.4\linewidth]{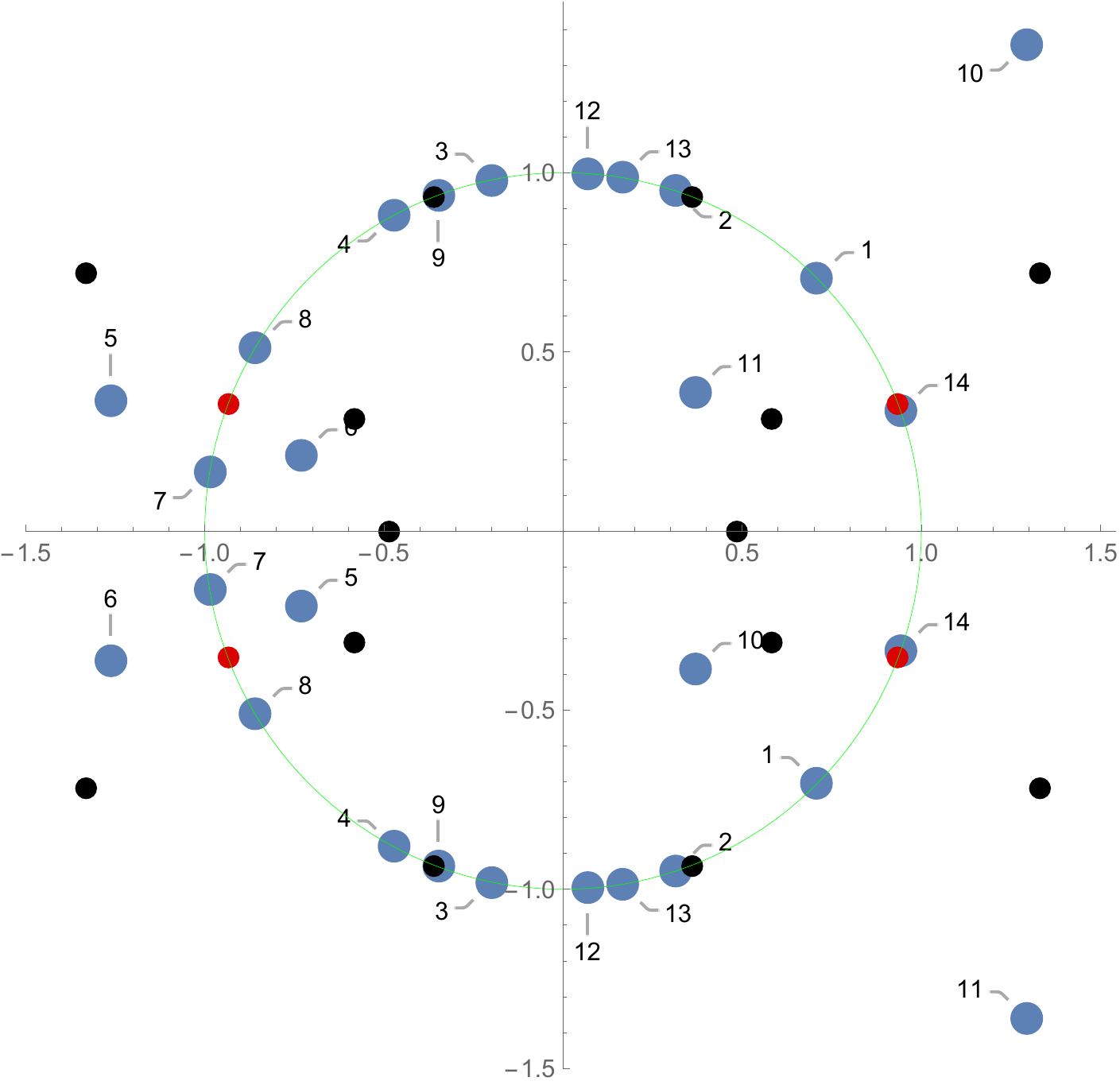}
\caption{Intersections of the irreducible branch of the $5_2$ knot with $x = y^{2}$ ($-\frac{1}{2}$ surgery) and $x = y^{-2}$ ($+\frac{1}{2}$ surgery), respectively. Red and black points are sheet intersection points as in Figure \ref{fig: 52 branch points}. The corresponding CS values and Adjoint Reidemeister Torsions appear in Tables \ref{tab: 1/2 surg 52} and \ref{tab: -1/2 surg 52}.}
\label{fig: 52 -1/2 surg intersections}
\end{figure}

\begin{table}
\begin{tabular}{|c|c|c|c|c|c|c|}
\hline
$\alpha$ & $x$ & $y$ & CS Invariant & Torsion & \makecell{Normalized \\ CS Invariant} & \makecell{Stokes \\ Constant} \\
				\hline \hline
$1$ & $0.71 + 0.70i$ & $-0.92 - 0.38i$ & $-0.41474938$ & $34.60612077$ & $-266.36303$ & $3.14638085$ \\ \hline
$2$ & $0.31 + 0.95i$ & $-0.81 - 0.59i$ & $-0.17678762$ & $17.97484582$ & $-113.53768$ & $4.36571281$ \\ \hline
				$3$ & $-0.20 + 0.98i$ & $-0.63 - 0.77i$ & $0.302315372$ & $16.00200688$ & $194.154930$ & $4.62701090$ \\ \hline
				$4$ & $-0.47 + 0.88i$ & $-0.51 - 0.86i$ & $0.034034396$ & $7.362013101$ & $21.8577894$ & $6.8216504$ \\ \hline
				$5$ & $-1.26 + 0.36i$ & $-0.16 - 1.14i$ & \multirow{2}{*}{\shortstack{$-0.03760350$ \\ $\mp 0.0609846 i$}} & \multirow{2}{*}{\shortstack{$-1.86400966$ \\ $\mp 1.3885247 i$}} & \multirow{2}{*}{\shortstack{$-24.149960$ \\ $\mp 39.16589 i$}} & \multirow{2}{*}{\shortstack{$3.82037990$ \\ $\pm 11.52378 i$}} \\ 
				$6$ & $-0.73 + 0.21i$ & $-0.12 - 0.86i$ & & & & \\ \hline
				$7$ & $-0.99 + 0.17i$ & $-0.08 - 1.00i$ & $-0.14661662$ & $9.505081961$ & $-94.161070$ & $6.00357329$ \\ \hline
				$8$ & $-0.86 + 0.51i$ & $0.27 + 0.96i$ & $-0.09186365$ & $6.260695062$ & $-58.997265$ & $7.39735522$ \\ \hline
				$9$ & $-0.35 + 0.94i$ & $0.57 + 0.82i$ & $0.101277095$ & $3.284475775$ & $65.0428294$ & $10.2130326$ \\ \hline
				$10$ & $1.29 + 1.36i$ & $1.26 + 0.54i$ & \multirow{2}{*}{\shortstack{$0.126589040$ \\ $\pm 0.0255284 i$}} & \multirow{2}{*}{\shortstack{$-6.30070938$ \\ $\mp 2.4528860 i$}} & \multirow{2}{*}{\shortstack{$81.2988306$ \\ $\pm 16.39501 i$}} & \multirow{2}{*}{\shortstack{$1.31374918$ \\ $\pm 6.995940 i$}} \\ 
				$11$ & $0.37 + 0.39i$ & $0.67 + 0.29i$ & & & & \\ \hline
				$12$ & $0.07 + 1.00i$ & $0.73 + 0.68i$ & $-0.45846249$ & $24.92967061$ & $-294.43674$ & $3.70705884$ \\ \hline
				$13$ & $0.17 + 0.99i$ & $0.76 + 0.65i$ & $0.257772219$ & $14.36839939$ & $165.548138$ & $4.88296436$ \\ \hline
				$14$ & $0.94 + 0.33i$ & $0.99 + 0.17i$ & $-0.00155708$ & $267.5361287$ & $-1$ & $1.13160936$ \\ \hline
\end{tabular}
\vspace{3mm}
\caption{Chern-Simons invariant and Adjoint Reidemeister Torsion for $+\frac{1}{2}$ surgery on the $5_2$ knot. The normalized CS invariants are the CS invariants divided by the CS invariant of smallest magnitude, and the Stokes constants are related to the torsions via expression (\ref{eq:stokes-torsion}).
See also Section \ref{sec: Observations, Torsions and CS values}. 
These invariants can be evaluated to essentially any degree of precision, as discussed in the text.}
\label{tab: 1/2 surg 52}
\end{table}

\begin{table}
\begin{tabular}{|c|c|c|c|c|c|c|}
\hline
$\alpha$ & $x$ & $y$ & CS Invariant & Torsion & \makecell{Normalized \\ CS Invariant} & \makecell{Stokes \\ Constant}\\
				\hline \hline
$1$ & $e^{\frac{\mathstrut 2 \pi i}{5}}$ & $e^{\frac{\mathstrut 9 \pi i}{5}}$ & $\tfrac{\mathstrut 1}{\mathstrut 6}$ & $5$ & $94.4345662$ & $6.85952149$ \\ \hline
				$2$ & $e^{\frac{\mathstrut 4 \pi i}{5}}$ & $e^{\frac{\mathstrut 3 \pi i}{5}}$ & $\tfrac{\mathstrut 1}{\mathstrut 6}$ & 5 & $94.4345662$ & $6.85952149$ \\ \hline
				$3$ & $2.57746915$ & $0.62287839$ & $\tfrac{\mathstrut 5}{\mathstrut 48}$ & $2 \left(7-2 \sqrt{2}\right) $ & $59.0216038$ & $5.31023704$ \\ \hline
				$4$ & $0.57 + 0.82i$ & $-0.89 + 0.46i$ & $-\tfrac{\mathstrut 19}{\mathstrut 48}$ & $2 \left(7+2 \sqrt{2}\right)$ & $-224.28209$ & $3.45956754$ \\ \hline
				$5$ & $-0.98+0.19i$ & $0.09 - 1.00i$ & $\tfrac{\mathstrut 5}{\mathstrut 48}$ & $2 \left(7-2 \sqrt{2}\right)$ & $59.0216038$ & $5.31023704$ \\ \hline
				$6$ & $-0.07 + 1.00i$ & $0.68 - 0.73i$ & $\tfrac{-\mathstrut 19}{\mathstrut 48}$ & $2 \left(7+2 \sqrt{2}\right)$ & $-224.28209$ & $3.45956754$ \\ \hline
				$7$ & $0.22 + 0.98i$ & $-0.78 + 0.63i$ & $-0.13406551$ & $21.59143663$ & $-75.962512$ & $3.30094321$ \\ \hline
				$8$ & $-0.33 + 0.94i$ & $-0.58 + 0.82i$ & $0.388460431$ & $7.558096488$ & $220.104554$ & $5.57920869$ \\ \hline
				$9$ & $-1.45 + 0.36i$ & $0.10 - 0.81i$ & \multirow{2}{*}{\shortstack{$0.211341064$ \\ $\mp 0.0564034 i$}} & \multirow{2}{*}{\shortstack{$-1.67330521$ \\ $\mp 0.6108315 i$}} & \multirow{2}{*}{\shortstack{$119.747410$ \\ $\mp 31.95860 i$}} & \multirow{2}{*}{\shortstack{$2.00099436$ \\ $\pm 11.31681 i$}} \\
				$10$ & $-0.65 + 0.16i$ & $0.15 - 1.22i$ & & & & \\ \hline
				$11$ & $-0.21 + 0.98i$ & $0.63 + 0.78i$ & $0.321158064$ & $7.764552467$ & $181.970534$ & $5.50453463$\\ \hline
				$12$ & $0.93 + 0.38i$ & $0.98 - 0.19i$ & $0.001764890$ & $171.9325248$ & $1$ & $1.16976818$ \\ \hline
\end{tabular}
\vspace{3mm}
\caption{Chern-Simons invariant and Adjoint Reidemeister Torsion for $-\frac{1}{2}$ surgery on the $5_2$ knot. The normalized CS invariants are the CS invariants divided by the CS invariant of smallest magnitude, and the Stokes constants are related to the torsions via expression (\ref{eq:stokes-torsion}).
See also Section \ref{sec: Observations, Torsions and CS values}. 
These invariants can be evaluated to essentially any degree of precision, as discussed in the text.}
\label{tab: -1/2 surg 52}
\end{table}

\subsection{Curious observations}
\label{sec: Observations, Torsions and CS values}

Looking at the complete set of geometric invariants for $\pm \frac{1}{2}$ surgery for the $4_1$ and $5_2$ knots, as shown in Tables \ref{tab: 4_1 knot}, \ref{tab: 1/2 surg 52} and  \ref{tab: -1/2 surg 52}, some interesting structures emerge:
\begin{itemize}
\item While the Chern-Simons values in Tables \ref{tab: 4_1 knot}, \ref{tab: 1/2 surg 52} and  \ref{tab: -1/2 surg 52} appear to be generically irrational, we observe numerically  that if we take the sum of all the Chern-Simons values (for a given surgery) we get a rational number (to several hundred digits of precision):
\begin{eqnarray}
S^3_{\pm \frac{1}{2}}(4_1)&:& \sum_{\alpha=1}^7 {\rm CS} (\alpha) =
\pm \frac{1}{8} \nonumber \\
S^3_{+\frac{1}{2}}(5_2)&:&  \sum_{\alpha=1}^{14} {\rm CS} (\alpha) =-\frac{5}{12} \label{CSsum} \\
S^3_{-\frac{1}{2}}(5_2)&:&  \sum_{\alpha=1}^{12} {\rm CS} (\alpha) = \frac{3}{4} \nonumber
\end{eqnarray}
This hints at some form of integrability. Recall, that Chern-Simons values for a surgery on a knot are given by integrals of a Liouville 1-form along open paths on the A-polynomial curve. In general, there are many inequivalent homotopy types of paths that correspond to different integral lifts of the flat connections and the corresponding Chern-Simons values. Consistency of the theory imposes stringent constraints on period integrals of this 1-form along {\it closed} 1-cycles on the A-polynomial curve, which basically express independence on the lift for $SU(2)$ Chern-Simons theory and other similar specializations \cite{Guk05}. We expect the relations \eqref{CSsum} to be of a similar nature. In other words, we expect that the sum over flat connections in \eqref{CSsum} can be expressed as a concatenation of paths on the A-polynomial curve, so that the total integral has especially simple form.

\item It is clear from their definition that torsions are algebraic numbers and so can be encoded as the roots of an integral polynomial, the {\it torsion polynomial} \eqref{eq: Torsion Polynomial}:
\begin{align}
\sigma^{adj}_{S^3_{-\frac{1}{2}}(4_1)}(t) & = -7215127 + 2828784 t^2 - 417832 t^3 - 272624 t^4 + 83296 t^5 - 7168 t^6 + 128 t^7 \label{eq:41m-torsion-polynomial}\\
\sigma^{adj}_{S^3_{+\frac{1}{2}}(5_2)}(t) & =  5554214481270856813 - 833790268928570748 t^2 - 163953024020455456 t^3
                    \nonumber\\ & \quad \quad + 128473842183215536 t^4 - 11213038799872672 t^5 - 2369935480771328 t^6
                    \nonumber\\ & \quad \quad + 398092105583488 t^7 + 8781117136640 t^8 - 6166471077376 t^9 \label{eq:52p-torsion-polynomial} \\
                     & \quad \quad + 584070450176 t^{10} - 26306131968 t^{11} + 607559680 t^{12}
                    \nonumber\\ & \quad \quad - 6316032 t^{13} + 16384 t^{14}  \nonumber \\
                \sigma^{adj}_{S^3_{-\frac{1}{2}}(5_2)}(t) & =  \big(-5 + t\big)^2 \big(164 - 28 t + t^2\big)^2 \label{eq:52n-torsion-polynomial}
                    \\ & \quad \quad \big(44241255 + 32803272 t + 695124 t^2 - 2966904 t^3 + 386592 t^4 - 13152 t^5 + 64 t^6\big) \nonumber
\end{align}
Intriguingly, in each of these cases, the torsions are all {\it algebraic half-integers}. This follows immediately from observing that the leading term has the form $2^i t^j$ for $i \leq j$.
\item Additionally, observe that the sum of the inverse torsions vanishes. Writing the polynomial as $\sum_i a_it^i$ we see that in all of these cases
\begin{eqnarray}
\sum_{\alpha=1}^{\alpha_{\rm max}} \frac{1}{\tau (\alpha)} = \frac{a_1}{a_0} = 0
\end{eqnarray}
From the perspective of the Borel plane, this translates to a relation between the squares of the monodromies.

\item In general, $\SL(2, \m{C})$ flat connections naturally split into $3$ categories, $\SL(2, \m{R})$, $\SU(2)$ and full $\SL(2, \m{C})$ connections. Given a point $(x, y) \in \B{A}_K$, the flat connection lies in $\SL(2, \m{R})$ if $x, y \in \m{R}$ and in $\SU(2)$ if $x, y \in S^1$. We stress that the properties described above only hold when we consider all $3$ categories together. In particular, for each knot surgery we find that
\begin{eqnarray}
\sum_{SU(2)} \frac{1}{{\rm torsion}}+\sum_{SL(2, \mathbb R)}\frac{1}{{\rm torsion}} =- \sum_{SL(2, \mathbb C)}\frac{1}{{\rm torsion}}
\end{eqnarray}
            
\item For connections in $\SL(2, \m{R})$ or $\SU(2)$, the corresponding torsions are positive real numbers. Hence in order for the sum of the inverse torsions to be $0$, the $\SL(2, \m{C})$ contribution, needs to cancel the $\SL(2, \m{R})$ and $\SU(2)$ contributions. For some surgeries (such as $-1$ on the $4_1$ knot), all torsions lie in $\SL(2, \m{R})$ or $\SU(2)$. Hence for these surgeries the sum of the inverse torsions must be non-zero. That being said, the sum appears to always be integral. We explore this further in Appendix \ref{app: Torsion Polynomial, Twist Knot Surgeries}.

\item In the case of the $-\frac{1}{2}$ surgery on $5_2$, half the CS invariants are simple rational numbers, with corresponding closed-form torsions. These simple values are associated with the factorization (for $-\frac{1}{2}$ surgery) of the irreducible A-polynomial $A_{5_2}(y^{2}, y)$ in (\ref{eq:52m-apoly}), and of the torsion polynomial $\sigma^{adj}_{S^3_{-1/2}(5_2)}(t)$ in (\ref{eq:52n-torsion-polynomial}). As can be seen in Appendix \ref{app: Torsion Polynomial, Twist Knot Surgeries}, this turns out to be the first indicator of a more general structure concerning the $-\frac{1}{n}$ surgery on the $K_n$ twist knots.
\end{itemize}

In the next section we show how the Chern-Simons and Torsion values in Tables \ref{tab: 4_1 knot}, \ref{tab: 1/2 surg 52} and \ref{tab: -1/2 surg 52} can be obtained {\it numerically} via resurgent analysis applied to a truncated perturbative expansion of the Chern-Simons partition function. The basic strategy is to explore the Borel plane using as input this truncated perturbative expansion of the Chern-Simons partition function in powers of $\hbar$. We convert this to a truncated Borel transform, and then explore its singularity structure. The singularities of the Borel transform describe the non-perturbative physics. The goal is to find: (i) the location, (ii) the exponent, and (iii) the strength, of each singularity. As explained above, each of these has a physical meaning.  In increasing order of precision we can do the following:
    \begin{enumerate}
        \item Simple ratio tests, combined with Richardson acceleration, can provide information about {\it leading} Borel singularities.
        \item To probe subleading (i.e., more distant) Borel singularities, a useful (but rough) picture can be easily obtained  using a Pad\'e approximant for the Borel transform of the truncated series.
        \item Given this rough picture, a more precise analysis can be achieved by making a suitable conformal map, and {\bf then} making a Pad\'e approximation. Even if this conformal map only takes into account the leading singularities, this step still leads to significant improvements in precision \cite{CD20a}.
        \item 
        Much more precise information can be achieved using a suitable uniformizing map instead of the conformal map. The method of {\it singularity elimination} provides a powerful method to probe any given singularity with extremely high precision \cite{CD20b}.
    \end{enumerate}
    We provide explicit examples of these numerical methods in the next Section.

\section{Resurgent analysis for hyperbolic surgeries}
\label{sec:borel}

From the perspective of complex Chern-Simons theory, our analysis in section~\ref{sec:CSandAlex} is essentially classical, or at best semi-classical, as it relies on the Gaussian approximation near each saddle point of the Feynman path integral~\eqref{Feynman-int}.
The goal of this section is to probe deeper into the structure of the full quantum theory by applying powerful techniques of resurgent analysis to a very high loop order of the perturbative expansion \eqref{eq: perturbative expansion}. This should teach us about the non-perturbative formulation of the theory, in particular allowing a direct comparision with the BPS $q$-series \eqref{Zhatexpansion} that provides a candidate for the non-perturbative completion that behaves well under cutting-and-gluing operations.

For a given choice of the 3-manifold \eqref{Msurgery}, the starting point of this analysis is the analytic continuation, $B_{\alpha} (\xi)$, of the Borel transform of the perturbative series \eqref{eq: perturbative expansion}. For a generic choice of the $\SL(2, \m{C})$ flat connection $\alpha$, the computation of \eqref{eq: perturbative expansion} can be carried out by a variety of different methods (such as the explicit computation of Feynman diagrams, topological recursion, etc.), but for $\alpha = 0$ we can use a shortcut. Indeed, at the perturbative level, the Feynman path integral \eqref{Feynman-int} is analytic in $A$, so that all perturbative coefficients in \eqref{eq: perturbative expansion} should be the same in theories with gauge groups $SU(2)$ and $\SL(2, \m{C})$. This important feature was discussed in detail in \cite{Guk05}, where it was also used to explain the volume conjecture and to produce its various generalizations. In particular, one consequence of this non-trivial fact is that the perturbative expansion in complex Chern-Simons theory near the trivial flat connection $\alpha = 0$ is given by the ``Laplace transform'' ({\it cf.} \cite{BBL,GMP}):
\begin{equation}
\B{Z}^{\text{pert}}_{\alpha=0} (S^3_{p/r}(K)) \; \simeq \;
\B{L}^{(0)}_{\frac{p}{r}}
\Big(
(x^{\frac{1}{2r}} - x^{-\frac{1}{2r}})
\underbrace{(x^{\frac{1}{2}} - x^{-\frac{1}{2}}) \sum_{m = 0}^{\infty} C_{m} (K;q) (qx)_m (qx^{-1})_m }_{F(q,x)}
\Big)
\label{ZpertLaplace}
\end{equation}
where $q$ is related to $\hbar$ as in \eqref{qvsh}, $C_{m} (K;q)$ are the cyclotomic coefficients for the knot $K$, and the operation $\B{L}^{(a)}_{\frac{p}{r}}$ is defined as
\begin{equation}
\B{L}^{(a)}_{\frac{p}{r}} (q^i x^j) = \begin{cases}
q^{i -j^2 \frac{r}{p}} & rj - a \in p\m{Z}\\
0 & \text{otherwise}.
\end{cases}
\label{Laplace}
\end{equation}
It is not an accident that the same structure appears in the computation of the non-perturbative invariants \eqref{Zhatexpansion}. For this reason, in \eqref{ZpertLaplace} one can replace $F(q,x)$ by its non-perturbative counterpart \cite{GM}, re-expanded in $\hbar$. Either way, for $\alpha=0$ the computation of the perturbative series in complex Chern-Simons theory drastically simplifies and can be expressed in terms of simpler objects familiar from the $SU(2)$ Chern-Simons theory.\footnote{Clearly, this can not be the case for more general $\alpha$; after all, even the notion of a complex flat connection itself may not be meaningful in a theory with $SU(2)$ gauge group.} For example, the cyclotomic coefficients can be written explicitly for a general twist knot $K_n$ \cite{Mas}:
\begin{equation}
C_{m}(K_n; q) = q^m \sum_{j = 0}^m (-1)^j q^{j(j+1)n + j(j-1)/2}(1 - q^{2j + 1})\frac{(q;q)_m}{(q;q)_{m+j+1}(q;q)_{m-j}}
\end{equation}
and in the special cases $p=-1$ and $p=2$ reduce to rather compact expressions for the knots $4_1$ and $5_2$, which we use as our prime examples to peform surgeries on\footnote{It is currently not known whether there is a way to express the BPS $q$-series of a knot complement, $F_K (x,q) : = \sum_{b \in \mathbb{Z}} \widehat{Z}_b (S^3 \setminus K; q)$, in terms of cyclotomic coefficients, for a general knot $K$.}
\begin{equation}
K_{-1} = 4_1: \qquad C_{m}(4_1; q) = (-1)^mq^{-\frac{m(m+1)}{2}}
\end{equation}
\begin{equation}
K_{2} = 5_2: \qquad C_{m}(5_2; q) = q^m\sum_{j = 0}^m (-1)^j q^{j\frac{3 + 5j}{2}}(1 - q^{2j + 1})\frac{(q;q)_m}{(q;q)_{m+j+1}(q;q)_{m-j}}
\end{equation}
Substituting these into \eqref{ZpertLaplace} gives an efficient way of computing the perturbative series to a very high loop order. This provides perturbative data from which we use resurgent analysis to extract non-perturbative information about the Chern-Simons theory. Therefore, our next goal is to analyze the corresponding Borel plane for hyperbolic surgeries on twist knots.\footnote{In the special case of $- \frac{1}{2}$ surgery on the figure-eight knot, the resurgent analysis around the saddle $\alpha=0$ (trivial flat connection) was independently carried out in the recent work \cite{Whee}. It is not clear, though, what non-perturbative completion, if any, was assumed in \cite{Whee} and how it relates to the Spin$^c$ decorated TQFT that provides a non-perturbative completion based on $Q$-cohomology (BPS spectra).}

\subsection{Surgeries on ${\bf 4_1}$ knot}
\label{sec:41borel}

For the figure-eight knot $K = 4_1$, the surgeries \eqref{Msurgery} with opposite coefficients $+\frac{p}{r}$ and $- \frac{p}{r}$ are directly related to our main question \eqref{ffotherside} since the resulting 3-manifolds are related by orientation reversal. Indeed, in general the orientation reversal of $M_3 = S^3_{+ \frac{p}{r}} (K)$ also admits a surgery presentation $S^3_{- \frac{p}{r}} (\bar K)$ on the mirror knot $\bar K$. Since the figure-eight knot is amphichiral, it suffices to flip the sign of the surgery coefficient to reverse the orientation of $M_3$. In Section~\ref{sec:other-side} we discuss the effect of this operation on the BPS $q$-series, relating \eqref{minushalfqseries} and \eqref{plushalfqseries} below, and proposing an answer to the general question \eqref{ffotherside}.

For $-\frac{1}{2}$ surgery on the figure-eight knot $K=4_1$, using the procedure outlined above we expand the perturbative partition function \eqref{ZpertLaplace} to order $\hbar^{228}$. The first few terms are given here:
\begin{eqnarray}
\B{Z}^{\text{pert}}_{\alpha=0} (S^3_{-\frac{1}{2}} (4_1))&=& 1+\frac{97 \hbar}{8}
+\frac{33985 \hbar^2}{128} 
+\frac{24726817 \hbar^3}{3072} 
+\frac{30753823105 \hbar^4}{98304}
+\frac{58360349239777
   \hbar^5}{3932160}+\dots 
    	     \nonumber\\
    	    &:=& \sum_{n=0}^\infty a_n\, \hbar^n
\label{eq:z41}
\end{eqnarray}
The coefficients $a_n$ are all rational and positive.
The full set of perturbative coefficients is presented in an accompanying supplementary data file.
As mentioned above, the $4_1$ knot is amphichiral, and so the perturbative series for the $+\frac{1}{2}$ surgery can be obtained simply by replacing $\hbar \to - \hbar$:
\begin{eqnarray}
\B{Z}^{\text{pert}}_{\alpha=0} (S^3_{+\frac{1}{2}} (4_1))&=& 1-\frac{97 \hbar}{8}
+\frac{33985 \hbar^2}{128} 
-\frac{24726817 \hbar^3}{3072} 
+\frac{30753823105 \hbar^4}{98304}
-\frac{58360349239777
   \hbar^5}{3932160}+\dots 
    	     \nonumber\\
    	    &:=& \sum_{n=0}^\infty (-1)^n a_n\, \hbar^n
\label{eq:z41plus}
\end{eqnarray}

While the two perturbative expansions \eqref{eq:z41} and \eqref{eq:z41plus} are related by an obvious symmetry, their respective non-perturbative completions can be obtained from the surgery formulae \cite{GM,Park21} and the results turn out to be highly asymmetric, as functions of $q$:
\begin{eqnarray}
\widehat{Z} \big( S^3_{-\frac{1}{2}} (4_1) \big) &=& 
- q^{-1/2} (1-q + 2 q^3 - 2q^6 + q^9 + 3q^{10} + q^{11}
+ \dots \nonumber\\
&& + 2335418615q^{1600} + \dots )
\end{eqnarray}
\label{minushalfqseries}
\begin{equation}
\widehat{Z} \big( S^3_{+\frac{1}{2}} (4_1) \big)
= q^{3/2} (1 - 2q^2 + q^3 - 3 q^4 + 4 q^5 - q^6 + q^7 + 5 q^8 + \dots)
\label{plushalfqseries}
\end{equation}
In a similar way, one can obtain explicit expressions for surgeries on twist knots as well as general positive braid knots \cite{Park20}, satellite knots \cite{Chae23}, and connect sums \cite{Chae21}.

\subsubsection{Leading Borel structure from the Perturbative Coefficients}

Given these formal series (\ref{eq:z41})-(\ref{eq:z41plus}), the first interesting physical observation is that they are factorially divergent. This can be seen clearly from a simple ratio test, which shows that 
\begin{eqnarray}
    \frac{a_{n+1}}{a_n}\sim (8.6058\dots)\times \left(n+\frac{3}{2}\right)
    \label{eq:41-ratio}
\end{eqnarray}
\begin{figure}[htb]
    \centerline{\includegraphics[scale=.8]{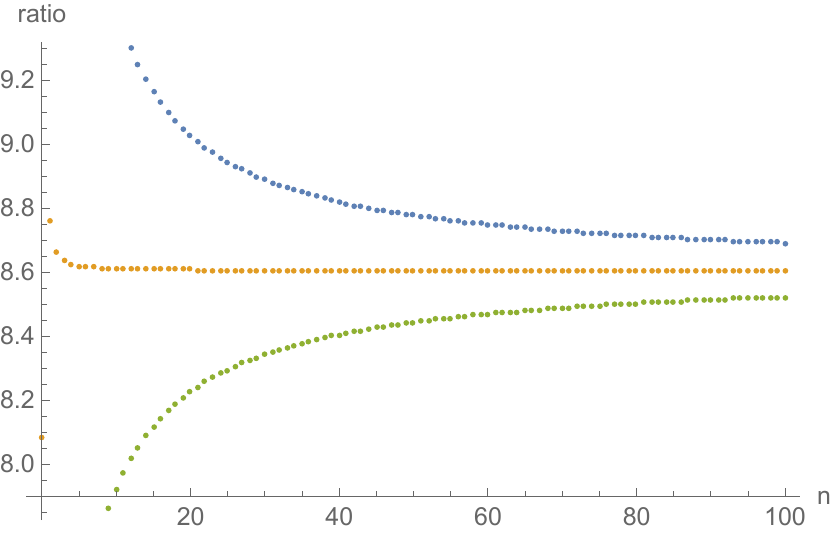}}
\caption{Ratio test for the coefficients $a_n$ of the perturbative series expansion of the Chern-Simons partition function $\B{Z}^{\text{pert}}_{\alpha=0} (S^3_{-\frac{1}{2}} (4_1))$ in (\ref{eq:z41}). The curves plot the ratios $\frac{a_{n+1}}{a_n\,\left(n+\frac{1}{2}\right)}$ [blue dots], $\frac{a_{n+1}}{a_n\,\left(n+\frac{3}{2}\right)}$ [orange dots], and $\frac{a_{n+1}}{a_n\,\left(n+\frac{5}{2}\right)}$ [green dots], as a function of the perturbative order $n$. The $1/\left(n+\frac{3}{2}\right)$ factor is clearly preferred, as can be confirmed by further Richardson accelerations. This carries information about the {\it nature} of the leading Borel singularity.}
    \label{fig:41-ratio}
\end{figure}
See Figure \ref{fig:41-ratio}, which illustrates that the leading growth rate of this ratio is $\left(n+\frac{3}{2}\right)$, rather than some other offset from $n$.
The overall constant factor in (\ref{eq:41-ratio}) can be determined to extremely high precision using high order Richardson extrapolation \cite{BO} of the ratio $\frac{a_{n+1}}{a_n\,\left(n+\frac{3}{2}\right)}$. The inverse of this overall constant gives the radius of convergence of the corresponding Borel transform (discussed below in Section \ref{sec:pade-borel}):
\begin{eqnarray}
\text{radius}_{4_1}=
0.11620083270928446726565248388502115693767811655351206400791649983959 \dots 
\nonumber\\
\label{eq:41radius}
\end{eqnarray}
With our 228 coefficients $a_n$ as perturbative input, this radius of convergence can be computed to 140 stable digits. This determines to very high precision the Chern-Simons invariant (with conventional normalization) for the leading non-trivial saddle:
\begin{eqnarray}
CS(\alpha_1)&=&-\frac{\text{radius}_{4_1}}{4\pi^2}\\
&=& -0.0029434014775824953073213809724952529218074079609198\dots
    \label{eq:cs41-leading}
\end{eqnarray}
This agrees to all 140 stable digits with the leading Chern-Simons invariant computed using the A-polynomial method in the previous Section: see the first row of Table \ref{tab: 4_1 knot}.

The leading factorial divergence of the expansion coefficients is therefore of the form:
    \begin{eqnarray}
    a_n\sim {\mathcal S_{4_1}} \frac{\Gamma\left(n+\frac{3}{2}\right)}{(\text{radius}_{4_1})^n}\qquad, \quad n\to\infty
    \label{eq:41-leading}
    \end{eqnarray}
    The Stokes constant ${\mathcal S_{4_1}}$, can also be determined to extremely high precision using high order Richardson extrapolation. We find
    \begin{eqnarray}
    {\mathcal S_{4_1}}=1.1036697620938896715472771709343445316179658869626659342
\dots
\label{eq:41-stokes}
\end{eqnarray}
This Stokes constant agrees to all 140 stable digits with the Stokes constant associated with the leading Chern-Simons invariant listed in Table \ref{tab: 4_1 knot}. Recall that in Table \ref{tab: 4_1 knot} this value was derived from the Adjoint Reidemeister torsion using the A-polynomial method, with the identification in (\ref{eq:stokes-torsion}), and can be computed to arbitrary precision. 
This is the first confirmation of the identifications (as shown in Table \ref{tab:data}) of the geometric data, the Chern-Simons invariant and the Adjoint Reidemeister torsion, with the perturbative data derived directly from the formal perturbative expansion of the partition function. The perturbative expansion is an expansion about the trivial saddle point, but the simple analysis above demonstrates that this expansion encodes precise information about the {\it location} of, the {\it exponent} of, and the {\it strength} of the closest (most dominant) non-perturbative Chern-Simons saddle. We show below that information about more distant saddles may also be decoded from the perturbative series, but that this requires more sophisticated methods than simple ratio tests.
    
With 228 terms of the formal series it is also possible to extract {\it subleading} power-law corrections to the leading large-order factorial growth in (\ref{eq:41-leading}):
\begin{eqnarray}
    a_n\sim {\mathcal S_{4_1}} \frac{\Gamma\left(n+\frac{3}{2}\right)}{(\text{radius}_{4_1})^n}\left[1-\frac{(0.0572609835...)}{\left(n+\frac{1}{2}\right)}+\dots\right]+\dots 
    \qquad, \quad n\to\infty
    \label{eq:41-subleading}
\end{eqnarray}
It is straightforward to extract further subleading power-law corrections and this has important applications, as discussed below in Section \ref{sec:decoupling 4_1}. 
    
\subsubsection{Pad\'e-Borel and Pad\'e-Conformal-Borel Analysis}
\label{sec:pade-borel}
In fact, we also expect further {\it exponentially suppressed} corrections to the large order growth in (\ref{eq:41-subleading}). These exponentially suppressed corrections are associated with more distant Borel singularities, which in turn are identified with Chern-Simons invariants of greater magnitude. These are difficult to resolve with ratio tests and root tests, because the exponentially suppressed corrections are swamped by the power-law corrections. However, some of these further Borel singularities can be resolved using simple Pad\'e and conformal mapping methods in the Borel plane. Due to the large values of the Chern-Simons invariants in Table \ref{tab: 4_1 knot}, to resolve the most distant Borel singularities we need to use more refined techniques, such as {\it singularity elimination \cite{CD20b}}, described below in Section \ref{sec:elimination}.
    
The first step is to regularize the divergent formal series (\ref{eq:z41}) by transforming to the Borel plane. We define the Borel transform by dividing out the factorial growth of the coefficients:
\begin{eqnarray}
    {\mathcal B}_{4_1}(\xi):=\sum_{n=0}^\infty \frac{a_n\, }{\Gamma(n+1)} \, (\text{radius}_{4_1})^n\, \xi^n
    \label{eq:41borel}
    \end{eqnarray}
    This is now a convergent series. We choose to multiply the Borel variable $\xi$ by a factor of $(\text{radius}_{4_1})$, as this leads to a convenient normalization in which the Borel radius of convergence is 1.
    The formal perturbative series (\ref{eq:z41}) is reconstructed term-by-term by the Laplace-Borel integral 
    \begin{eqnarray}
        Z_{4_1}(\hbar)=\frac{(\text{radius}_{4_1})^{3/2}}{\Gamma(3/2) \hbar^{3/2}} \int_0^\infty d\xi \sqrt{\xi}\, e^{-\text{radius}_{4_1} \xi/\hbar} {\mathcal B}_{4_1}(\xi) 
        \label{eq:41borel-inverse}
\end{eqnarray}
The non-perturbative features of the Chern-Simons partition function $Z_{4_1}(\hbar)$ are encoded in the singularities of the Borel transform ${\mathcal B}_{4_1}(\xi)$ in (\ref{eq:41borel}). Given only a {\it finite number} of terms of the perturbative series for $Z_{4_1}(\hbar)$ (and therefore also for ${\mathcal B}_{4_1}(\xi)$), the technical challenge is to extract as much information as possible about the physical Borel singularities. This information can be extracted with remarkable precision using appropriate combinations of elementary methods such as Pad\'e approximants, conformal maps, and uniformizing maps \cite{CD20a,CD20b}.
    \begin{figure}[htb]
    \centering
    \includegraphics[scale=.8]{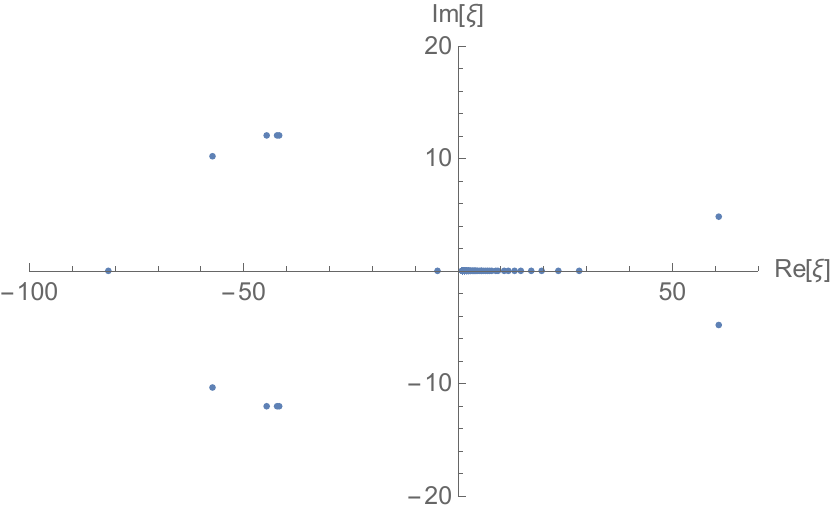}
\caption{The Pad\'e-Borel poles from an order $[114,114]$ diagonal Pad\'e approximant to the 228 term truncated Borel transform in (\ref{eq:41borel}). The Borel variable $\xi$ is normalized so that the leading singularity is at $+1$. Figure \ref{fig:41borel-poles-zoom} shows a zoomed-in view of the poles accumulating to $\xi=1$ on the positive Borel axis.}
    \label{fig:41borel-poles}
    \end{figure}
    \begin{figure}[htb]
    \centering
    \includegraphics[scale=.8]{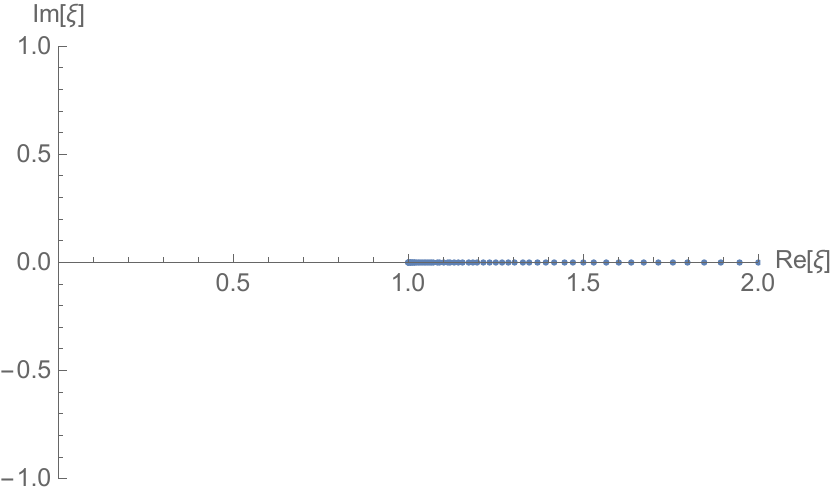}
    \caption{A zoomed-in view of the Pad\'e-Borel poles accumulating on the positive Borel axis, from Figure \ref{fig:41borel-poles}. We see a line of poles accumulating to $\xi=1$, which is how Pad\'e attempts to represent a branch cut, with a branch point at the accumulation point. }
    \label{fig:41borel-poles-zoom}
\end{figure}

Pad\'e approximants \cite{BGM,BO} provide a simple rough overview of the singularity structure in the Borel plane.  Given our 228 term truncation of the Borel transform, we first construct a diagonal  Pad\'e approximant of order $[114,114]$, and compute its poles in the Borel plane.\footnote{We also compute near-diagonal Pad\'e approximants, to filter out spurious Pad\'e poles.} The Pad\'e-Borel poles are shown in Figure \ref{fig:41borel-poles}, and a zoomed-in view of the neighborhood of the leading singularity is shown in Figure \ref{fig:41borel-poles-zoom}. 
This simple Pad\'e analysis confirms that the leading Borel singularity is indeed at $\xi=1$, consistent with our normalization convention in (\ref{eq:41borel}). Recall that since Pad\'e is by construction an approximation by {\it rational} functions, its only possible singularities are poles. Pad\'e represents branch points as the accumulation points of arcs of poles, according to the electrostatic interpretation of Pad\'e as a minimizer of an associated capacitor \cite{GS58,St97,Sa10,CD20a,CD20b}. 

Thus, this elementary Pad\'e-Borel construction decodes the {\it leading} Chern-Simons invariant, and by plotting the Pad\'e-Borel transform as it approaches the leading singularity we can also extract a rough numerical estimate of the associated Stokes constant, which tells us the associated Adjoint Reidemeister torsion. 
However, to obtain higher precision, and more importantly to decode the other Chern-Simons invariants, and their associated Reidemeister torsions, we need further tools, beyond Pad\'e-Borel. Here we can take advantage of the fact that the other Chern-Simons invariants have much larger magnitude  (see column 6 in Table \ref{tab: 4_1 knot}). Hence we expect the other Borel singularities to be far separated from the leading one. This means that the Borel branch cut starting at $\xi=1$ is dominant, and we can therefore build an {\it approximate} conformal map based on it. This conformal map significantly improves the precision of the Pad\'e analytic continuation of the truncated Borel transform (this improvement can be quantified \cite{CD20a}). This enables a more precise numerical probe of the leading singularity, but more importantly it resolves more cleanly the more distant Borel singularities.
       \begin{table}
            \centering
            \begin{tabular}{|c|c|c|c|c|}
                \hline
				 \makecell{Exact \\
				 CS Invariant} & \makecell{Normalized \\
				 CS Invariant} & 
				 \makecell{Pad\'e-Borel }
				 & 
				\makecell{Pad\'e-Conformal\\-Borel} &
				 \makecell{Singularity\\ Elimination}\\
			\hline
			  $-0.002943401$ & $1$ & $1$ & 1 & 1\\
				\hline
				  $-0.485874320$ & $165.072391$ & not resolved & not resolved & $161.05$ \\
				\hline
				 $0.053933576$ & $-18.323554$ & not resolved & absent & absent\\
				\hline
				\multirow{2}{*}{\shortstack{$0.123303626$ \\ $\pm 0.03542464 i$}} &
				  \multirow{2}{*}{\shortstack{$-41.891542$ \\ $\mp 12.03527 i$}} &
				  {\multirow{2}{*}{\shortstack{ $-42\mp 12i$ \\  ~}}} & 
				  \multirow{2}{*}{\shortstack{$-41.8814$ \\ $\mp 12.0371 i$}} &
				  \multirow{2}{*}{\shortstack{$-41.891542$ \\ $\mp 12.03527 i$}} \\
				 & & & & \\
				\hline
				  $0.235159766$ & $-79.893881$ & not resolved & not resolved & $-79.89$\\
				\hline
				  $-0.171882873$ & $58.3960000$ & not resolved & $58.3754$ & $58.3960000$ \\
				\hline
\end{tabular}
\bigskip
\caption{CS Invariants for $- \frac{1}{2}$ surgery on the $4_1$ knot, obtained from different analysis methods using the same perturbative input. The first column has the exact CS invariants, from Table \ref{tab: 4_1 knot}. All subsequent columns list the normalized values, which are obtained from the exact ones by dividing by the leading CS invariant $-0.002943401$. These subsequent columns show the CS invariants extracted using the Pad\'e-Borel, Pad\'e-Conformal-Borel, or singularity elimination methods to analyze the Borel transform based on the finite order perturbative expansion in (\ref{eq:z41}). Note that the more precise methods exclude, with very high numerical precision, the existence of a Borel singularity near $-18.323554$. This illustrates the power of the singularity elimination method in resolving even very distant singularities.}
            \label{tab: 4_1 CS pade}
\end{table}
Concentrating on this leading cut in the Borel plane 
we map the cut plane into the unit disk via the invertible conformal map
\begin{eqnarray}
    \xi=\frac{4z}{(1+z)^2} \quad \longleftrightarrow\quad
    z=\frac{1-\sqrt{1-\xi }}{1+\sqrt{1-\xi }}
    \label{eq:41conformal-map}
\end{eqnarray}
The Pad\'e-Conformal-Borel procedure is to re-expand the mapped truncated Borel series ${\mathcal B}\left(\frac{4z}{(1+z)^2}\right)$ to the same order (this is optimal \cite{CD20b}) in $z$, and then make a Pad\'e approximant in $z$, and finally to map back to the original Borel $\xi$ plane.  The resulting landscape of singularities in the Borel plane is shown in Figure \ref{fig:41zpoles}. Here the black dots show the Pad\'e poles in the conformal $z$ plane when they are mapped back to the original Borel $\xi$ plane. The red dots show the Chern-Simons invariants computed from the A-polynomial approach, as listed in Table \ref{tab: 4_1 knot}. Compared to the Pad\'e-Borel results in Figure \ref{fig:41borel-poles}, we now see much more clearly and precisely the complex conjugate pair of singularities at $\xi=-41.8814 \pm  12.0371 i$, and we also resolve a singularity at $\xi=58.3754...$, close to an exact Chern-Simons invariant at $58.3960000$. These values are shown in Table \ref{tab: 4_1 CS pade}. We stress that these Pad\'e-Conformal-Borel results are obtained from exactly the same perturbative input used for the Pad\'e-Borel results shown in Figure \ref{fig:41borel-poles}, but just processed differently.
   
\begin{figure}[h!]
\centering
\includegraphics[scale=1]{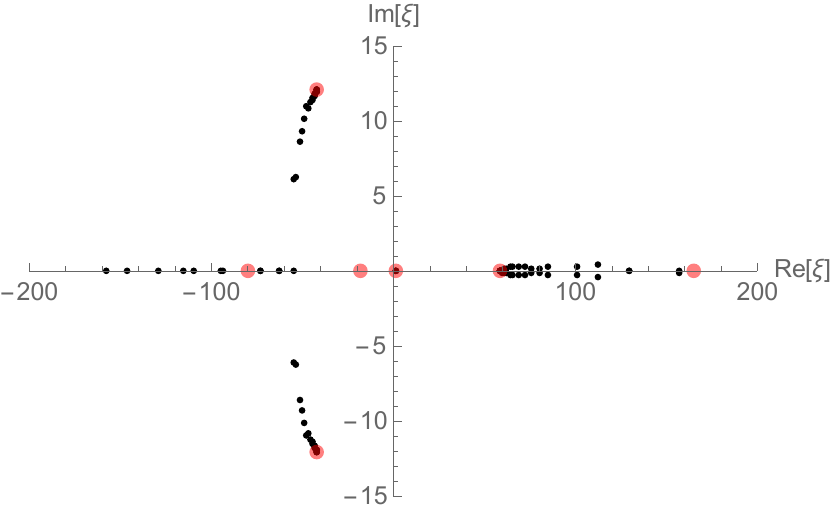}
\caption{
Pad\'e-Conformal-Borel analysis of the Borel singularities for $-\frac{1}{2}$ surgery for the $4_1$ knot case. The black dots show the inverse conformal map images of the poles of the Pad\'e approximant made in the conformal $z$ plane. The opaque red dots show the Chern-Simons invariants (see Table \ref{tab: 4_1 knot}) for this $-\frac{1}{2}$ surgery on the figure-eight knot $4_1$, computed using the A polynomial method in Section \ref{sec:A-polynomial}. Notice that the Borel singularity near $\xi=59$ is now resolved: compare with Figure \ref{fig:41borel-poles} where this singularity is not resolved. }
       \label{fig:41zpoles}
\end{figure}
            
However, this Pad\'e-Conformal-Borel analysis is still not able to resolve the expected more distant singularities at $\xi= -79.893881$ and $\xi=165.072391$. Nevertheless, it is interesting to note that this analysis does rule out the existence of a Borel singularity at $\xi=-18.3236$ to a high degree of precision (this is made even more precise in the next subsection below, using the singularity elimination method). This leads to the prediction that this Chern-Simons invariant (labelled pair number 3 in Table \ref{tab: 4_1 knot}) is disconnected from the perturbative saddle.
            
\subsubsection{Singularity Elimination Analysis of the Borel Structure}
\label{sec:elimination}

A more refined method to probe the Borel singularities is to use ``singularity elimination'' \cite{CD20b}. The essential idea underlying this method is to note that an isolated singularity (in this context a Borel singularity) can be removed by a suitably chosen sequence of maps. The first map is an invertible linear operator (essentially a fractional derivative) that converts the exponent of the singularity to $\frac{1}{2}$. A second map, $\xi\to 2z-z^2$ for example, makes the function analytic at the location of the original singularity (assuming $\xi$ has been normalized to place the singularity at $\xi=1$). In this current problem we are fortunate because each singularity already has a half integer exponent \cite{GMP,GM}. Thus, only the second map is needed. This singularity elimination procedure is extremely sensitive, in part because there are simple procedures to distinguish analyticity from non-analyticity. For example, if the location of the original singularity is only known approximately, this location can be iteratively tuned so that the singularity is fully removed, in the sense that the residue at the mapped location truly vanishes. 
 
The procedure for the second map is as follows. Suppose the function $F(\xi)$ under consideration has a regular expansion about $\xi=0$ with a finite radius of convergence, due to a square root non-analyticity at a point that we normalize to lie at $\xi=1$:
 \begin{eqnarray}
  F(\xi)&=&\sqrt{1-\xi}\, A(\xi)+B(\xi)
  \qquad\qquad, \quad \xi\to 1
  \\
  &=&
  \sqrt{1-\xi}\, \sum_{k=0}^\infty a_k(1-\xi)^k+\sum_{k=0}^\infty b_k(1-\xi)^k
  \label{eq:sing1}
 \end{eqnarray}
Here $A(\xi)$ and $B(\xi)$ are analytic at $\xi=1$. Now compose with the map $\xi=2z-z^2$. Then, since $\xi=0$ maps to $z=0$, $F(2z-z^2)$ is also analytic at $z=0$, and since $\xi=1$ maps to $z=1$, the location of the non-analyticity is not moved. However, now $F(2z-z^2)$ is analytic at $z=1$. In fact, near $z=1$:
  \begin{eqnarray}
  F(2z-z^2) = b_0+a_0(1-z)+b_1(1-z)^2+a_1(1-z)^3+b_2(1-z)^4 + \dots
  \label{eq:sing2}
 \end{eqnarray}
The coefficients $a_k$ and $b_k$ in (\ref{eq:sing2}) are exactly those in the original expansion (\ref{eq:sing1}). These can be determined for example by making a Pad\'e approximation to the re-expansion of $F(2z-z^2)$ about $z=0$, to the same order as the original expansion (this truncation order is optimal \cite{CD20b}). The coefficient $a_0$ determines the Stokes constant at $\xi=1$, and the $a_n$ coefficients determine the fluctuations about the leading singularity.
 
Applying this to the new leading singularity (after elimination of the closest singularity, normalized to be at $\xi=+1$) we now obtain 19 digits of precision, compared to 2 digits of precision from Pad\'e-Borel, and 5 digits of precision from Pad\'e-Conformal-Borel. Note that the input data is exactly the same, but the singularity elimination method is much more precise \cite{CD20b}. Furthermore, from the local behavior near the removed singularity we can extract the Stokes constant to much higher precision, and we can also analyze more cleanly the more distant singularities. 
For example, applying this procedure to probe the more distant Borel singularities associated with the Chern-Simons invariants at $\xi=-41.891542\pm 12.03527 i$, and at $\xi=58.396$, we find the following leading singular behavior in the neighborhood of each Borel singularity:
\begin{eqnarray}
{\mathcal B}_{4_1}(\xi) &\sim & \,  \frac{(0.0976802126206 \mp 5.0636173816082\, i)\Gamma\left(\frac{3}{2}\right)}{((-41.891542\pm 12.03527 i)-\xi)^{3/2}}+ \dots \\
{\mathcal B}_{4_1}(\xi) &\sim & \,  \frac{(2.833300156162\, i)\Gamma\left(\frac{3}{2}\right)}{(58.396-\xi)^{3/2}}+ \dots
\label{eq:further-stokes}
\end{eqnarray}
This analysis simultaneously determines the location of the Borel singularity and its Stokes constant to very high precision.
The values in (\ref{eq:further-stokes}) are obtained to a relative error $O(10^{-21})$. The even more distant singularities at $-79.89$ and $165$ are obtained with relative errors of $0.3\%$ and $3\%$, respectively.

To summarize: applying the singularity elimination method to the Borel transform of the perturbative series (\ref{eq:z41}) we resolve all the remaining (normalized) Chern-Simons invariants, as shown in the last column of Table \ref{tab: 4_1 knot}, except for the one at $-18.323554...$. We are able to rigorously rule out (to more than 100 digits of precision) the appearance of a Borel singularity near $-18.323554...$. This numerical result strongly suggests that this Chern-Simons saddle is disconnected from the others, for a symmetry reason that remains to be understood.

\subsubsection{Numerical Evidence for Decoupling of the Leading Borel Singularity} \label{sec:decoupling 4_1}

We conclude this Borel analysis of the $4_1$ knot case by presenting two pieces of numerical evidence that the leading Borel singularity "decouples" from the other Borel singularities. This evidence is two-fold:
       
\begin{itemize} 

\item
We computed the first 20 subleading power-law corrections to the leading factorial growth of the perturbative expansion coefficients shown in (\ref{eq:41-subleading}). This only requires elementary methods of Richardson extrapolation. Invoking resurgence \cite{Ec81,Co08,MS16,ABS19}, these subleading coefficients encode the fluctuations about the leading Borel singularity, and in a typical Borel landscape this fluctuation expansion would itself be divergent. However, these 20 subleading coefficients indicate that this fluctuation expansion is in fact convergent, not divergent. Therefore, the fluctuations about the leading Borel singularity are not resurgently coupled to the other Borel singularities.
       
\item
Using the more advanced method of singularity elimination, as described above, we computed 50 terms of the expansion of the Borel transform function about its leading singularity, and an analysis of these coefficients shows even stronger evidence of convergence,  even hinting at potentially being entire.

\end{itemize}

Moreover, we numerically explored, with reliably high numerical accuracy, two further sheets of the Borel plane Riemann surface, and determined that neither has new singularities. This is a fingerprint of linear problems and an indicator of integrability.  The singularity separation noted above is another hallmark of linearity (and integrability). One can simply subtract out the Laplace transform of the decoupled singularity which then becomes a standalone transseries term  (one which must be an exact  solution of the homogeneous part of the underlying problem).
After subtraction,  the new Borel plane is freed of this particular singularity and, as a result, the underlying equation factors, a strong form of integrability. 
Together, these pieces of numerical evidence suggest on the geometric side that the flat  connection corresponding to this leading Borel singularity is special. This deserves to be studied further, both from the geometric and topological perspective and from the quantum field theory perspective.

\subsection{Surgeries on ${\bf 5_2}$ Knot}
\label{sec:52borel}

We now apply the same numerical procedures to the $\mp \frac{1}{2}$ surgeries on the $5_2$ knot. The $5_2$ knot is interestingly different from the $4_1$ knot, so it is not clear in advance what to expect. In particular, since the $5_2$ knot is not amphichiral, its $+ \frac{1}{r}$ and $- \frac{1}{r}$ surgeries are not related by the orientation reversal\footnote{Instead, for $K = 5_2$ we have $S^3_{+1/r} (K) = - S^3_{-1/r} (\bar K)$ and $S^3_{+1/r} (\bar K) = - S^3_{-1/r} (K)$.} and represent completely different 3-manifolds. Therefore, the formal $\hbar$ series of the partition function for the $\mp \frac{1}{2}$ surgeries will be different. Indeed, we show below that the Borel plane structure is quite different for the $\mp \frac{1}{2}$ surgeries.

As before, we expand the partition function as a perturbative series in powers of $\hbar$. 
     The first terms are for the $\pm\frac{1}{2}$ surgery are:
      	\begin{eqnarray}
    	    Z_{5_2}^{-\frac{1}{2}}(\hbar)&=& 
    	    1-\frac{191 \hbar}{8}
    	    +\frac{107137 \hbar^2}{128} 
    	    -\frac{127522367 \hbar^3}{3072}  
    	    +\frac{261703390465 \hbar^4}{98304} 
    	    -\frac{822656668343231 \hbar^5}{3932160} +\dots
    	     \nonumber\\
    	    &:=& \sum_{n=0}^\infty b_n^{-\frac{1}{2}}\, \hbar^n
    	    \label{eq:z52minus}
    	\end{eqnarray}
    	\begin{eqnarray}
    	    Z_{5_2}^{+\frac{1}{2}}(\hbar)&=& 
    	    -1-\frac{183 \hbar}{8}
    	    -\frac{122577 \hbar^2}{128}
    	   -\frac{56438733 \hbar^3}{1024} 
    	    -\frac{133022451595 \hbar^4}{32768} 
    	    -\frac{477942803207261 \hbar^5}{1310720}-\dots
    	     \nonumber\\
    	    &:=& \sum_{n=0}^\infty b_n^{+\frac{1}{2}}\, \hbar^n
    	    \label{eq:z52plus}
    	\end{eqnarray}
    	The full list of expansion coefficients is contained in accompanying supplementary data files. For $-\frac{1}{2}$ surgery we generated 228 terms, while for $+\frac{1}{2}$ surgery we generated 188 terms. 
  Note that the perturbative expansion for the $-\frac{1}{2}$ surgery case is alternating in sign, while for the $+\frac{1}{2}$ surgery case it is non-alternating.\footnote{This sign pattern correlates with the sign of the surgery in the opposite way compared to the $4_1$ case. This sign pattern change is characteristic of the two different classes of hyperbolic twist knots: see Section \ref{sec:thooft-twist}.}

\subsubsection{Leading Borel structure from the Perturbative Coefficients}
   
    We observe that the formal series (\ref{eq:z52minus}) and (\ref{eq:z52plus}) are factorially divergent. Ratio tests combined with Richardson extrapolation determine 
    the leading growth as:
             \begin{eqnarray}
    b_n^{-\frac{1}{2}}\sim {\mathcal (-1)^n \, S_{5_2^{-\frac{1}{2}}}} \frac{\Gamma\left(n+\frac{3}{2}\right)}{(\text{radius}_{5_2^{-1/2}})^n}\qquad, \quad n\to\infty
    \label{eq:52m-leading}
    \end{eqnarray}
    \begin{eqnarray}
    b_n^{+\frac{1}{2}}\sim -{\mathcal S_{5_2^{+\frac{1}{2}}}}\, \frac{\Gamma\left(n+\frac{3}{2}\right)}{(\text{radius}_{5_2^{+1/2}})^n}\qquad, \quad n\to\infty
    \label{eq:52p-leading}
    \end{eqnarray}
The offset $\frac{3}{2}$ of the factorial growth term is clearly seen in a ratio test. Therefore the expansion coefficients $b_n^{\mp \frac{1}{2}}$ of both series have the same factorial divergence $\Gamma\left(n+\frac{3}{2}\right)$, which is also the same as for the $\pm \frac{1}{2}$ surgeries on $4_1$ discussed in Section \ref{sec:41borel}.
    
    From (\ref{eq:52m-leading}) and (\ref{eq:52p-leading}) we can extract the radii of convergence of the associated Borel transforms:
       \begin{eqnarray}
    \text{radius}_{5_2^{-1/2}}=
0.06967508334205362331643137281436160974803084178134507\dots
    \label{eq:52m-radius}
    \end{eqnarray}
     \begin{eqnarray}
    \text{radius}_{5_2^{+1/2}}=
0.06147117938868975855184395044865487683233400560419384\dots
    \label{eq:52p-radius}
    \end{eqnarray}
    These Borel radii of convergence are close in magnitude, but clearly distinct, for the two $\mp\frac{1}{2}$ surgeries. They are also approximately half the magnitude of the Borel radius of convergence for the $4_1$ knot case in (\ref{eq:41radius}). These features may be understood from the A-polynomial perspective: see Section \ref{sec:small-cs-values}.
    
    This defines, for each surgery, the leading Chern-Simons invariant:\footnote{Recall the normalization convention that the Chern-Simons invariant is equal to {\it minus} the Borel singularity.} 
     \begin{eqnarray}
    CS(\alpha_2)_{5_2^{-1/2}}&=&\frac{\text{radius}_{5_2^{-1/2}}}{4\pi^2}\\
    &=& 
0.0017648904786488511307396258970947779330492530820971\dots
    \label{eq:cs52m-leading}
    \end{eqnarray}
    \begin{eqnarray}
    CS(\alpha_2)_{5_2^{+1/2}} &=&-\frac{\text{radius}_{5_2^{+1/2}}}{4\pi^2}\\
    &=& 
-0.0015570831638881308832537780298728018909607838655340\dots
    \label{eq:cs52p-leading}
    \end{eqnarray}
    These leading Borel singularities match precisely the Chern-Simons invariants of smallest magnitude derived using the A-polynomial method: compare with the last row of Tables \ref{tab: -1/2 surg 52} and \ref{tab: 1/2 surg 52}.
  
   The Stokes constants  ${\mathcal S_{5_2^{\mp 1/2}}}$, the 
    overall coefficients in (\ref{eq:52m-leading}) and (\ref{eq:52p-leading}),  can also be extracted with very high precision:
       \begin{eqnarray}
    {\mathcal S_{5_2^{-1/2}}}=
1.16976817659921816048007806514112639659319054525915404\dots
  \label{eq:52m-stokes} 
    \end{eqnarray}
    \begin{eqnarray}
    {\mathcal S_{5_2^{+1/2}}}=
1.13160935822828864782306399594919914503998537214989941
    \dots \label{eq:52p-stokes} 
    \end{eqnarray}
 
    The Stokes constants in (\ref{eq:52m-stokes}) and (\ref{eq:52p-stokes}) agree to more than 100 digits of precision with the corresponding Stokes constants in Tables \ref{tab: -1/2 surg 52} and \ref{tab: 1/2 surg 52}, based on the identification of the Stokes constant with the Adjoint Reidemeister torsion in (\ref{eq:stokes-torsion}).
    Note that these two Stokes constants for the $\mp\frac{1}{2}$ surgery cases have a similar order of magnitude, but are clearly distinct. Furthermore, they also have a similar order of magnitude compared to the Stokes constant of the leading singularity for the $4_1$ knot in (\ref{eq:41-stokes}). Similarly to the Chern-Simons invariants, these features of the associated Stokes constants may also be understood from the A-polynomial perspective: see Section \ref{sec:small-cs-values}.
    
    Thus, this elementary series analysis once again
    determines to very high precision three important physical and geometric quantities: (i) the leading power factor in the growth rate, which determines the {\it location} of the leading Borel singularity, which in turn determines the leading non-trivial Chern-Simons invariant; (ii) the offset of the factorial growth, which determines the {\it nature} of the leading Borel singularity; (iii) the overall Stokes constant which determines the Adjoint Reidemeister torsion associated with the leading Chern-Simons invariant. This non-perturbative information about more distant Borel singularities (and therefore about other nontrivial flat connections) is indeed encoded in the formal asymptotic expansions about the trivial flat connection/Chern-Simons saddle.
  \begin{figure}[htb]
   \centering
   \includegraphics[scale=.8]{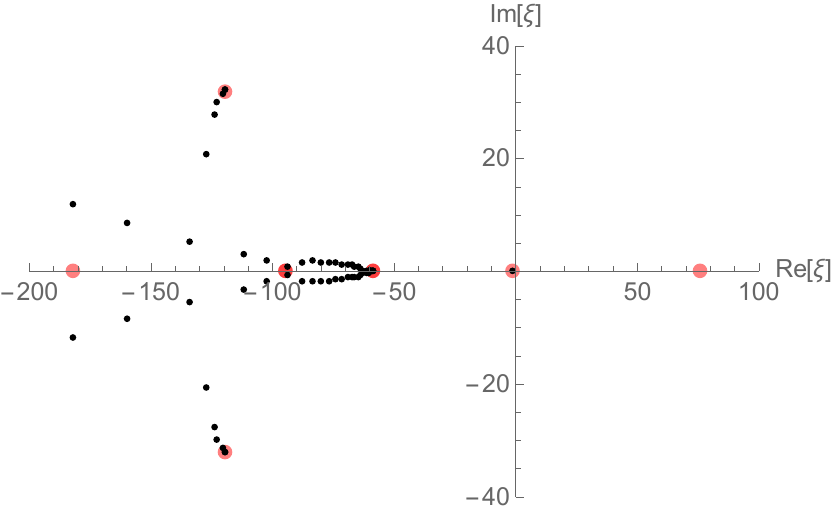}
    \caption{The Borel plane structure for the $-\frac{1}{2}$ surgery on the $5_2$ knot, resolved by the Pad\'e-Conformal-Borel method. The black dots show the inverse conformal map images of the poles of the Pad\'e approximant made in the conformal $z$ plane. The opaque red dots show the Chern-Simons invariants computed using the A-polynomial method, shown in Table \ref{tab: -1/2 surg 52}.}
   \label{fig:52-minus-borel}
 \end{figure}    
      \begin{figure}[htb]
   \centering
    \includegraphics[scale=.8]{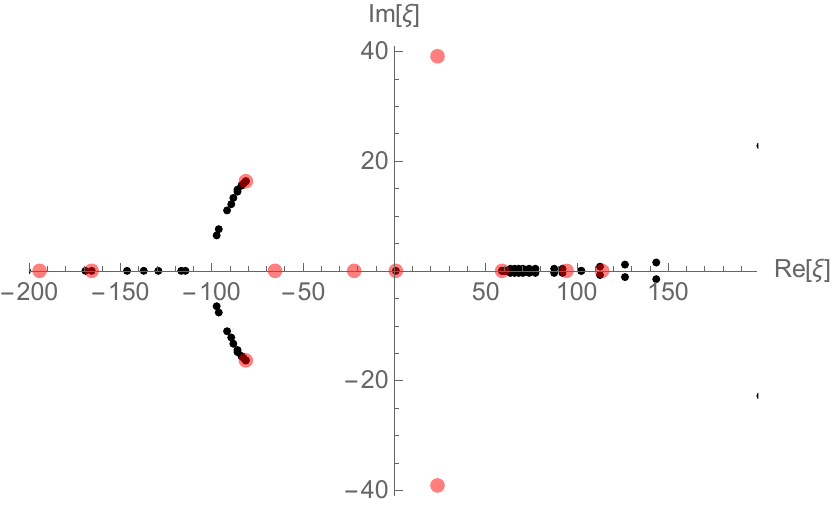}
    \caption{The Borel plane structure for the $+\frac{1}{2}$ surgery on the $5_2$ knot, resolved by the Pad\'e-Conformal-Borel method. The black dots show the inverse conformal map images of the poles of the Pad\'e approximant made in the conformal $z$ plane. The opaque red dots show the Chern-Simons values computed using the A-polynomial methodshown in Table \ref{tab: 1/2 surg 52}.}
   \label{fig:52-plus-borel}
 \end{figure} 
 
\subsubsection{Pad\'e-Conformal-Borel Analysis}
    
To probe this more deeply, we now turn to the higher-precision method of Borel analysis.
Given the similar form of the leading growth of the coefficients, compared to the $4_1$ knot case discussed in the previous section, we define analogous Borel transforms
\begin{eqnarray}
    {\mathcal B}_{5_2^{\mp 1/2}}(\xi):=\sum_{n=0}^\infty \frac{b_n^{\mp 1/2}}{\Gamma(n+1)}\,\left(\text{radius}_{5_2^{\mp 1/2}}\right)^n\, \xi^n
    \label{eq:52minus-plus-borel}
\end{eqnarray}
from which the formal perturbative series (\ref{eq:z52minus}) and (\ref{eq:z52plus}) are reconstructed by the Laplace-Borel integral as in (\ref{eq:41borel-inverse}). Note that we have normalized with the appropriate radius of convergence, so that the leading Borel singularity is at $\xi=\mp 1$ for the $\mp \frac{1}{2}$ surgeries.
As described in Section \ref{sec:41borel}, we use the Pad\'e-Conformal-Borel procedure to extract information about the Borel plane singularities.\footnote{We omit the lower-resolution Pad\'e-Borel method.}  
The results of this Pad\'e-Conformal-Borel analysis are shown in Figures \ref{fig:52-minus-borel} and \ref{fig:52-plus-borel}. In the plots the black dots show the Pad\'e poles in the conformal $z$ plane when they are mapped back to the original Borel $\xi$ plane. The red dots show the Chern-Simons invariants computed from the A-polynomial approach, as listed in Tables \ref{tab: -1/2 surg 52} and \ref{tab: 1/2 surg 52}.
Note that in Figures \ref{fig:52-minus-borel} and \ref{fig:52-plus-borel}, the Borel plane variable $\xi$ has been normalized by the radius of convergence in order to place the leading Borel singularity at $\xi=\mp 1$ for the $\mp \frac{1}{2}$ surgery. 
  
In Figure \ref{fig:52-minus-borel} we see that for the $-\frac{1}{2}$ surgery, the Pad\'e-Conformal-Borel procedure identifies with good precision the leading singularity, together with 2 more distant ones on the negative Borel axis, in addition to another complex conjugate pair with negative real part. More distant Borel singularities at $\xi\approx -182$ and $\xi\approx -220$ are not resolved with the available data. However, this analysis rules out a singularity near $\xi\approx 76$, implying that the corresponding flat connection is decoupled from the trivial flat connection.
  
In Figure \ref{fig:52-plus-borel} we see that for the $+\frac{1}{2}$ surgery, the Pad\'e-Conformal-Borel procedure identifies with good precision the leading singularity, together with a more distant one near $59$ on the positive Borel axis, in addition to another complex conjugate pair with negative real part. There are hints of singularities at $\xi\approx 94$ and $\xi\approx 113$, as well as at $\xi\approx -165$. Furthermore, this analysis rules out singularities near $\xi\approx -22$ and $\xi\approx -65$, as well as a complex conjugate pair near $\xi\approx 24\pm 39i$, implying that the corresponding flat connections are decoupled from the trivial flat connection.
  
It is interesting to note that the $+\frac{1}{2}$ surgery case has a more complicated Borel plane structure, consistent with the fact that the associated A-polynomials and torsion polynomials in (\ref{eq:52p-apoly}) and (\ref{eq:52p-torsion-polynomial}) do not factorize like they do for the $-\frac{1}{2}$ surgery cases in (\ref{eq:52m-apoly}) and (\ref{eq:52n-torsion-polynomial}). These differences, as well as the patterns of decoupled flat connections, deserve further study with more perturbative data and using more advanced analysis methods \cite{CD20b}.  

\subsubsection{Evidence for Decoupling of the Leading Borel Singularities}
\label{sec:decoupling}

It is straightforward to extract subleading power-law corrections to the large-order factorial growth in (\ref{eq:52m-leading}) and (\ref{eq:52p-leading}):
     \begin{eqnarray}
    b_n^{-1/2}\sim {\mathcal (-1)^n\, S_{5_2^{-1/2}}}  \frac{\Gamma\left(n+\frac{3}{2}\right)}{(\text{radius}_{5_2^{-1/2}})^n}\left[1-\frac{(0.10941587902...)}{\left(n+\frac{1}{2}\right)}+\dots\right]
    \quad, \quad n\to\infty
    \label{eq:52m-subleading}
    \end{eqnarray}
       \begin{eqnarray}
    b_n^{+1/2}\sim -{\mathcal S_{5_2^{+1/2}}} \frac{\Gamma\left(n+\frac{3}{2}\right)}{(\text{radius}_{5_2^{+1/2}})^n}\left[1-\frac{(0.09757544883...)}{\left(n+\frac{1}{2}\right)}+\dots\right]
    \quad, \quad n\to\infty
    \label{eq:52p-subleading}
    \end{eqnarray}
Analysis of the first 20 subleading corrections to this large-order growth indicates that for the each of the $\mp \frac{1}{2}$ surgeries in this $5_2$ knot case, these subleading coefficients do not grow factorially. This suggests, as in the $4_1$ case, that the flat connections associated with these leading Chern-Simons invariants are special. This observation deserves to be studied more closely in future work.

\section{Resurgent continuation to the other side}
\label{sec:other-side}

Recent work on $q$-series in the 3d-3d correspondence suggests a canonical ``duality'' operation on BPS $q$-series (``orientation reversal''): where under $q\to \frac{1}{q}$ both $F(q)$ and its dual $F(q)^{\vee}$ are $q$-series with integer powers of $q$ (up to an overall factor $q^{\Delta}$) and with integer coefficients (up to an overall normalization factor). Furthermore, the functions $F(q)$ and their duals $F (q)^{\vee}$ have the ``same'' perturbative expansions near $q = e^{\hbar} \approx 1$, that only differ by a factor of $(-1)^n$ in the perturbative coefficients $\sum_n \hbar^n a_n$, i.e. $a_n \mapsto (-1)^n a_n$. Since these expansions are generally divergent, it is natural to use Borel-\'Ecalle methods of resurgent asymptotics \cite{Ec81,Co08}.

In this Section we argue that the Borel representation encodes all the information necessary to construct an explicit and unique mapping relating $F(q)$ and its dual $F (q)^{\vee}$ across the boundary. This argument relies on the fact that the expansions on both sides of the boundary are divergent, and when these are expressed in terms of resurgent transseries using Borel-\'Ecalle methods, one can invoke uniqueness of  continuation (in the Borel plane) and the associated preservation of properties for resurgent functions.

{\bf Example:} We illustrate the problem with a simple and well-known example. Consider $M_3 = M(-2; \tfrac{1}{2}, \tfrac{1}{3}, \tfrac{1}{2})$. This manifold has two Spin$^c$ structures modulo $\mathbb{Z}_2$ and, therefore, two BPS $q$-series invariants $\widehat{Z}_b (M_3,q)$ for which ``going to the other side'' \eqref{ffotherside} has been studied in detail \cite{CCFGH}. Specifically, one of the BPS $q$-series $\widehat{Z}_b (M_3,q)$ can be written as
\begin{equation}
\widehat{Z}_0 (M_3,q) = \frac{1}{2} q^{1/24} F(q)
\end{equation}
where $F (q)$
is the false theta-function:\footnote{The standard $q$-Pochhammer symbol is defined as $(z;q)_n:=\prod_{j=0}^{n-1}(1-z\, q^j)$.}
\begin{eqnarray}
F (q) &=& 1 - \sum_{n \ge 1} \frac{(-1)^n q^{\frac{n(n-1)}{2}}}{(-q;q)_n}
\nonumber\\
&=& 2(1-q+q^2-q^5+q^7-q^{12}+q^{15}-q^{22}+q^{26} + \ldots)
\label{eq:F1}
\end{eqnarray}
This example is special because we have an explicit $q$-hypergeometric expression for $F(q)$, which means that 
one can formally perform the map \eqref{ffotherside} in one line.
Indeed, replacing $q$ by $q^{-1}$ in every term and multiplying both numerator and denominator by the same overall power of $q$, for the manifold with reverse orientation, we find
\begin{equation}
\widehat{Z}_0 (- M_3,q) = \frac{1}{2} q^{- 1/24} F(q)^{\vee}
\end{equation}
where $F(q)^{\vee}$ can be expressed in $q$-Pochhammer form:
\begin{eqnarray}
F (q)^{\vee} &=& 1 - \sum_{n \ge 1} \frac{(-1)^n q^{n}}{(-q;q)_n} \\
&=& 1+q-2 q^2+3 q^3-3 q^4+3 q^5-5 q^6+7 q^7 -\dots 
\label{eq:Fv1}
\end{eqnarray}
We recognize this $q$-series for $F(q)^{\vee}$ as that of the celebrated order-3 mock theta-function $f(q)$ of Ramanujan \cite{Wat,GM12}. This result can also be confirmed by more advanced techniques, {\it e.g.} with the use of Rademacher sums.
\medskip

\noindent\underline{\bf Comment:} We would like to identify a {\bf unique} pair of $q$-series:
    \begin{eqnarray}
        \text{unary q-series}\,\,F(q)
        \qquad\longleftrightarrow\qquad
        \text{integer-coefficient q-series}\,\, F (q)^{\vee}
    \end{eqnarray}
    Uniqueness is clearly important, both physically and mathematically. However, note that even the simple example above is delicate, because there exists a different $q$-Pochhammer expression for $F(q)^{\vee}$ which generates the same $q$-series as in (\ref{eq:Fv1}), but under the formal term-by-term replacement $q\to q^{-1}$ in the reverse direction it generates a different $q$-series for $F(q)$: 
    \begin{eqnarray}
F (q)^{\vee} &=&  \sum_{n=0}^{\infty} \frac{q^{n^2}}{(-q;q)_n^2} \\
&=& 1+q-2 q^2+3 q^3-3 q^4+3 q^5-5 q^6+7 q^7 - \dots
\\
F(q)&=& \sum_{n=0}^\infty \frac{q^n}{(-q; q)_n^2}
\\
&=&
1+q-q^2+2 q^3-4 q^4+5 q^5-6 q^6+7 q^7-
\dots
\label{eq:Fv2}
\end{eqnarray}
In particular, note that the $q$-series for $F(q)$ in (\ref{eq:Fv2}) has integer coefficients, but is not unary, in contrast to (\ref{eq:F1}). The difference between two such $F(q)$ expressions can be written in terms of theta functions \cite{HM14,BFR12}. 

\begin{quote}
    {\bf Goal:} In general we are interested in developing a method for (uniquely) mapping $F(q)\leftrightarrow F(q)^\vee$ which does not rely on knowledge of explicit $q$-Pochhammer representations, and furthermore which does not rely on explicit knowledge of modular properties.
\end{quote}    

\subsection{The other side of log-VOAs}
\label{sec:VOA}

At this point, it may be helpful to interrupt our discussion with the explanation of the meaning of the operation ``going to the other side'' in physics, vertex algebra, and low-dimensional topology, elaborating on numerous connections to other subjects briefly mentioned in the introduction. Not only this clarifies our motivation, but such alternative perspectives on the problem can tell us what should be expected.

We start with the interpretation of \eqref{ffotherside} in algebra, based on the Kazhdan-Lusztig correspondence. The latter relates representation theory of quantum groups to that of affine Lie algebras and vertex operator algebras (VOAs). It is usually formulated for negative values of the ``level,'' see {\it e.g.} \cite{KT1,KT2}. Among many variants of this correspondence that have been studied over the years, of particular interest to us is the one that involves quantum groups at roots of unity and logarithmic VOAs, log-VOAs for short. Characters of such VOAs are linear combinations of false theta-functions 
(\ref{falsetheta}) and other functions of $q$ that are relevant to us here.

\begin{figure}[ht]
	\centering
	\includegraphics[width=3.2in]{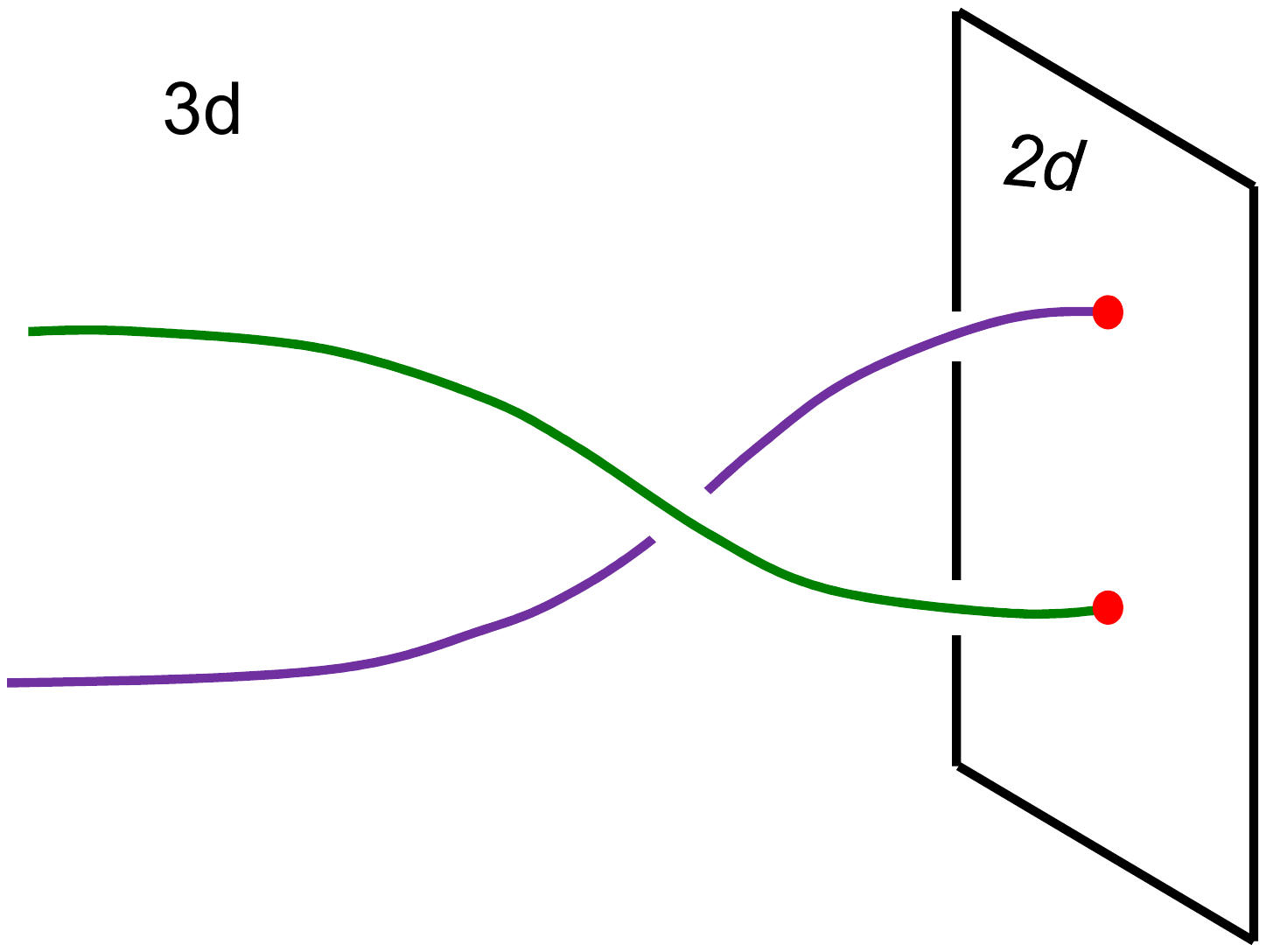}
	\caption{Vertex operators of a 2d VOA (CFT) are end-points of line operators in 3d theory. A monodromy of vertex operators in 2d corresponds to braiding of line operators in 3d.}
	\label{fig:braidbdry}
\end{figure}

Since in such algebras (think of a lattice VOA/CFT) the level defines the overall magnitude of the quadratic form, a naive change of its sign leads to a disaster because energy becomes unbounded from below. For characters, this operation looks as if $q$ is replaced by $q^{-1}$. Up to charge conjugation, the braiding structure is expected to remain\footnote{This can be justified via a connection between $Q$-cohomology in 2d-3d coupled systems and logarithmic VOAs, initiated in \cite{CCFGH}. See also section~\ref{sec:QFT}.} the same, however, much like in vertex algebras with complementary central charges $c + \bar c = 26$ \cite{FZ} that describe ``matter'' and ``gravity.'' A natural question, then, is whether there exists a less naive version of logarithmic VOAs in the positive Kazhdan-Lusztig zone:

\begin{myprob}
Construct counterparts of the familiar log-VOAs, such as {\it triplet} and {\it singlet} log-VOAs, in the positive KL zone.
\end{myprob}

We hope that applying \eqref{ffotherside} to characters of familiar log-VOAs using the tools of resurgence can help solving this problem. We hope it will connect to the interesting recent work \cite{AG,Gai,Liu,Fu} where the semi-infinite cohomology plays an important role.

\subsection{Low-dimensional topology}
\label{sec:topology}

In applications to low-dimensional topology, more precisely to TQFTs in dimensions $d=3$ and $d=4$, the operation \eqref{ffotherside} corresponds to orientation reversal on the underlying $d$-manifold,
\begin{equation}
M_d \; \to \; - M_d
\label{Mparity}
\end{equation}

Perturbatively, in complex Chern-Simons theory and its close cousins (such as the Teichmuller TQFT), orientation reversal on a 3-manifold $M_3$ is equivalent to $\hbar \to - \hbar$, which then implies \eqref{ffotherside} via \eqref{qvsh}. This is an important and highly non-trivial step. What makes it possible in complex Chern-Simons theory, is the interpretation of the non-perturbative completion as a BPS counting problem, where $q$ rather than $\hbar$ plays a primary role.

There are other important TQFTs that admit interpretations as BPS counting problems and where the orientation reversal \eqref{Mparity} is equivalent to \eqref{ffotherside}, {\it e.g.} Vafa-Witten theory and its variants.

In all such theories, we can imagine an independent {\it definition} of \eqref{ffotherside} by considering BPS $q$-series invariants $Z (M_d)$ of manifolds $M_d$ and $- M_d$:
\begin{equation}
Z(M_d; q) = f(q)
\quad {{\footnotesize{\text{``}q \leftrightarrow 1/q\text{''}}} \atop \overleftrightarrow{\phantom{transformat}}} \quad
f(q)^{\vee} = Z (-M_d; q)
\label{ZZff}
\end{equation}
In this paper, we mainly consider applications to the BPS $q$-series invariants $\widehat Z_b (M_3,q)$ of 3-manifolds, but it would be interesting to consider other topological invariants. We leave this to future work.

In the context of the BPS $q$-series $\widehat Z_b (M_3,q)$, it was observed early on that the invariants of $M_3$ and $- M_3$ appear rather different and, in particular, one is usually much harder to compute than the other. For example, the general formula for negative-definite plumbed manifolds \cite{GPPV} heavily relies on the negative-definite condition and the situation is similar for surgeries on knots and links \cite{GM}. Recently, various proposals to extend the formulation of these invariants to positive definite manifolds started to emerge, including positive surgery formulae proposed by Park~\cite{Park21}. It allows to construct many interesting families of examples of \eqref{ZZff}. Here we consider only one such infinite family, based on the simplest version of the surgery formula, namely the version for $+1$ surgeries.

There exists an infinite class of $q$-series for which the duals are known, based on a particular linear combination of two particular false theta functions. This result is motivated from topology. We have

\begin{myprop}
For a 2-parameter family of functions,
\begin{equation}
F_{a,b} (q) = \tilde \Psi^{(4a+b)}_{4a^2+2ab} - \tilde \Psi^{(b)}_{4a^2+2ab}
\label{eq:Fab}
\end{equation}
parametrized by $a,b \in \mathbb{Z}_+$, we have
\begin{equation}
F_{a,b} (q)^{\vee} = 
\frac{q^{\frac{4a + 3b}{8a + 4b}}}{(q)_{\infty}} \sum_{{j \ge 0 \atop |k|>j}} (-1)^k
q^{\frac{b}{8a} (2j+1)^2 + \frac{k}{2} (3k+1)} \,
(q^{-j^2} - q^{- (j+1)^2}) \,
\psi_{4a}^{(1)} (2j+1)
\end{equation}
\end{myprop}

\begin{proof}
Follows from \cite{Park21} and \cite[Lemma 7.3]{GM} (In the notations of \cite{GM}, the family in question is obtained by setting $c=2a$ and $r=1$, so that $d=1$, $e=2ab$, and $v = - \frac{b}{8a + 4b}$.) Namely, the $(-1)$-surgery formula applied to
\begin{equation}
F (x,q) = \frac{1}{2} \sum_{m=0}^{\infty} \big( x^{\frac{m}{2}} - x^{-\frac{m}{2}} \big) \, f_m (q) 
\label{FFFF}
\end{equation}
with $f_m (q) = q^{v + \frac{e}{4c^2} m^2}$, gives a linear combination of false theta-functions
\begin{equation}
\tilde \Psi_{4a^2 + 2ab}^{(4a + b)} - \tilde \Psi_{4a^2 + 2ab}^{(b)} = \B{L}^{(0)}_{-1} \Big[ \big( x^{\frac{1}{2}} - x^{-\frac{1}{2}} \big) F (x,q) \Big]
\end{equation}
where $\B{L}^{(0)}_{-1}$ is the Laplace transform \eqref{Laplace} familiar from \eqref{ZpertLaplace}. This is precisely the right-hand side of \eqref{eq:Fab}. To obtain $F_{a,b} (q)^{\vee}$, we apply the $(+1)$-surgery formula proposed in \cite{Park21}.
\end{proof}

\noindent{\bf Comments:}
\begin{enumerate}
    \item 
    The functions $F_{a,b} (q)$ and their duals $F_{a,b} (q)^{\vee}$ have the ``same'' perturbative expansions near $q = e^{\hbar} \approx 1$ that only differ by $(-1)^n$ factor in the perturbative coefficients $\sum_n \hbar^n a_n$, i.e. $a_n \mapsto (-1)^n a_n$.
    \item
    The combination of two false theta-functions \eqref{eq:Fab} appears as a building block in many examples, including surgeries on torus knots, see {\it e.g.} \cite{GMP,CCFGH,GM}. What was unclear, however, is whether individual functions 
    $\tilde \Psi_{p}^{(a)}$ for all $a$ and $p$ should have duals ``on the other side.'' Low-dimensional topology did not offer any insights, possibly suggesting that the answer might be ``no'' because this entire family does not come up as $q$-series invariants of any known class of 3-manifolds.
    \item This question is also motivated by the work of Cheng and Duncan on optimal Mock Jacobi theta functions \cite{Ch20}.
    \item Heuristic methods for obtaining $q\to 1/q$ dualities, using Appell-Lerch sums, have been studied in \cite{HiMo14,Mort14}.
\end{enumerate}

Interesting invariants of smooth 4-manifold are expected to be very asymmetric under orientation reversal:
\begin{eqnarray}
M_4 & \; \longleftrightarrow \; & - M_4 \\
b_2^- & \; \longleftrightarrow \; & b_2^+ \nonumber
\end{eqnarray}

\subsection{Physics of 2d-3d coupled systems}
\label{sec:QFT}

Both incarnations of `going to the other side' in algebra and in topology, described in sections \ref{sec:VOA} and \ref{sec:topology} respectively, can be put under one umbrella of BPS state counting ($Q$-cohomology) of 3d supersymmetric theories with 2d $(0,2)$ boundary conditions. This will give us yet another interpretation of this phenomenon and will provide an independent justification to some of the claims made in the previous two subsections.

\begin{figure}[ht]
	\centering
	\includegraphics[trim={0.3in 0.3in 0.3in 0.3in},clip,width=3.0in]{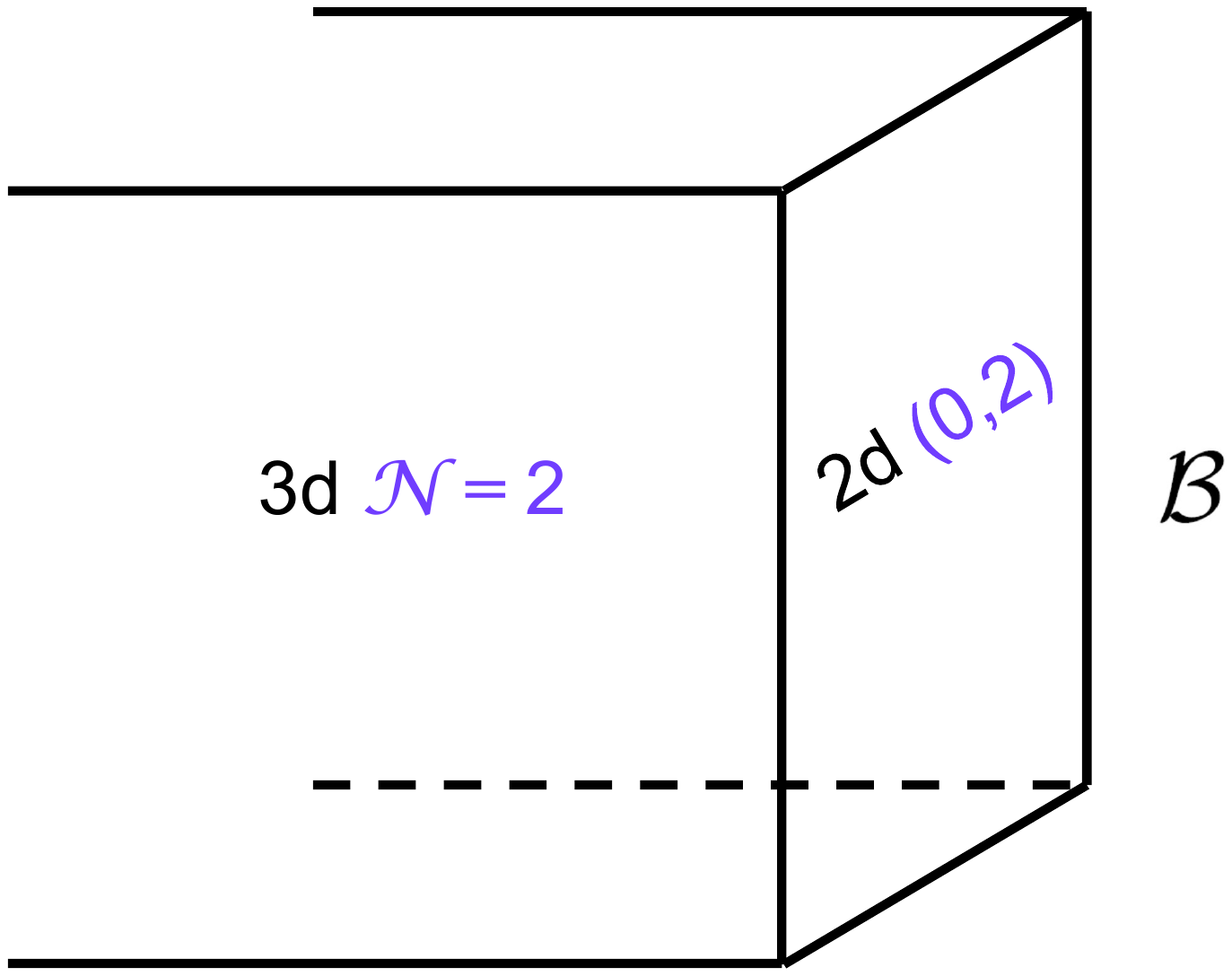}
	\caption{A coupled system of a 3d $\CN=2$ theory and a 2d $(0,2)$ boundary condition.}
	\label{fig:slab1}
\end{figure}

Consider a 3d $\CN=2$ theory with a 2d $\CN = (0,2)$ boundary condition $\CB$. In the context of 3d-3d correspondence and $\widehat{Z}$-invariants, the 3d $\CN=2$ theory is $T[M_3]$ and the boundary condition is a very particular one, labeled by Spin$^c$ structure on $M_3$. Similarly, in the context of 4-manifolds glued along 3-manifolds, the 3d $\CN=2$ theory is also $T[M_3]$, while the boundary condition (or interface) is $T[M_4]$. In either of these situations, as well as in their close cousins, reversing orientation on $M_3$ and $M_4$ is equivalent to the parity operation in 3d $\CN=2$ theory. For the ``bulk'' theory itself, this is usually not a complicated operation; it basically flips the signs of all Chern-Simons levels. However, its effect on the 2d $\CN=(0,2)$ boundary condition is a lot more interesting.

On a 2d boundary of a 3d theory, parity should result in exchanging left and right sectors of the 2d boundary theory. However, in our setup the left and right sectors are very asymmetric: the left sector is non-supersymmetric, whereas the right sector enjoys $\CN=2$ supersymmetry, which in particular includes supercharge $Q$ whose cohomology gives the invariants we are interested in. Therefore, the effect of the 3d parity operation on boundary conditions is a rather non-trivial bijection / duality:
\begin{equation}
\text{Parity}: \qquad \CB \quad \longleftrightarrow \quad \CB^{\vee}
\label{BBdual}
\end{equation}
that, roughly speaking, exchanges left and right sectors, while supersymmetrizing the former and de-supersymmetrizing the latter:
$$
\xymatrixcolsep{9pc}\xymatrix{
\boxed{~\phantom{\oint} {\text{Left} \atop \text{$\CN=0$}} \phantom{\oint}~} \quad \ar@/^/[r]^{\text{supersymmetrize}} &
\quad \boxed{~\phantom{\oint} {\text{Right} \atop \text{$\CN=2$}} \phantom{\oint}~}  \ar@/^/[l]^{\text{desupersymmetrize}}}
$$

Depending on the precise nature of these two ``parity-dual'' boundary conditions, one may find that
\begin{equation}
Q\text{-cohomology} (\CB) \; \ne \; Q\text{-cohomology} (\CB^{\vee}) \,,
\end{equation}
whereas the braiding of vertex operators and modular data are the same. This usually happens when $\CB$ and $\CB^{\vee}$ are Lorentz-invariant boundary conditions that transform well under $SL(2,\mathbb{Z})$ modular group, as {\it e.g.} in the large class of $T[M_4]$ theories.\footnote{Note that most boundary conditions used in the study of 3d indices and $\widehat{Z}$-invariants, such as Nahm pole boundary conditions and boundary conditions labeled by complex flat connections, are {\it not} of this type. As the other extreme, however, there can be boundary conditions $\CB$ and $\CB^{\vee}$ with the same $Q$-cohomology, or even such that $\CB$ and $\CB^{\vee}$ coincide as 2d $\CN=(0,2)$ theories. In the context of $T[M_4]$ theories this happens {\it e.g.} when $M_4$ has a boundary $M_3$, such that $M_4 = - M_4$ and $M_3 = - M_3$.}
Then, via anomaly inflow, $\CB$ and $\CB^{\vee}$ cancel same anomalies (up to a sign) of the 3d $\CN=2$ theory and, together, define a consistent 2d $\CN=(0,2)$ theory obtained by placing 3d $\CN=2$ theory on a slab with boundary condition $\CB$ on one side and $\CB^{\vee}$ on the other, as illustrated in Figure \ref{fig:slab2}.

\begin{figure}[ht]
	\centering
	\includegraphics[trim={0.3in 0.3in 0.3in 0.3in},clip,width=3.0in]{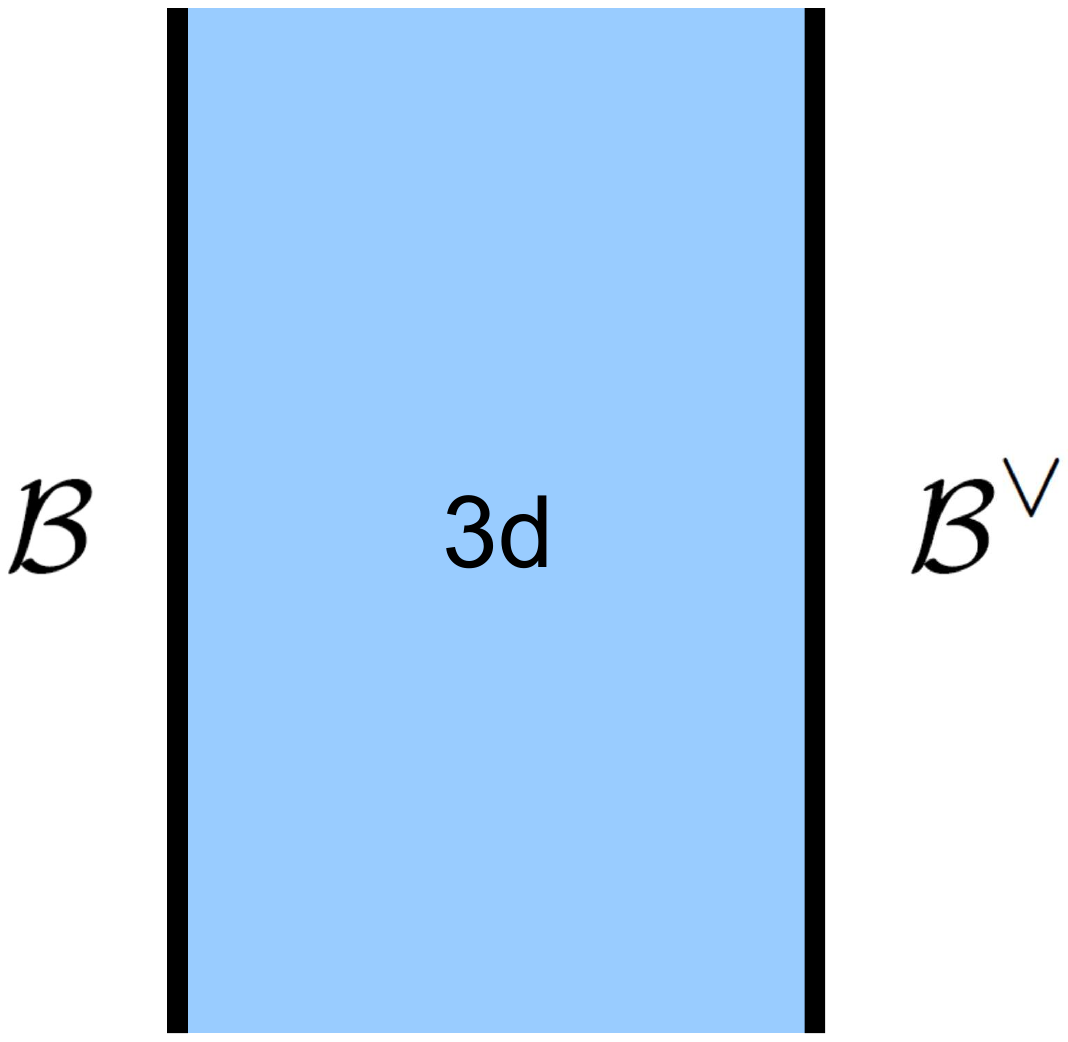}
	\caption{3d $\CN=2$ theory on a slab with 2d $(0,2)$ boundary conditions $\CB$ and $\CB^{\vee}$.}
	\label{fig:slab2}
\end{figure}

A prototypical example of such pairing is realized in 2d $\CN=(0,2)$ SQCD. In that case, the left sector is a WZW model, whereas the right sector is the $\CN=2$ Kazama-Suzuki coset model. They can be paired together to produce a modular-invariant partition function via a non-trivial automorphism (level-rank duality). This suggests that other dual pairs of $\CB$ and $\CB^{\vee}$ might be related in a similar way.

To gain some intuition about the physics of `going to the other side' in the context of 2d $\CN=(0,2)$ boundary conditions, let us consider a couple of toy examples. In fact, as a our first example, let us consider a 1d quantum system rather than 2d quantum field theory. A copy of a free harmonic oscillator (one quantum boson) has spectrum $E_n = n + \frac{1}{2}$, illustrated in Figure~\ref{fig:spectrum}. Its contribution to the elliptic genus is a $q$-series $\sum_{n} q^{E_n}$ that under \eqref{ffotherside} transforms as
\begin{equation}
\frac{q^{1/2}}{1-q} \quad \longleftrightarrow \quad  \frac{q^{-1/2}}{1-q^{-1}} = - \, \frac{q^{1/2}}{1-q}
\label{qqoscillator}
\end{equation}
Since the result is invariant up to a sign, it tells us that a bosonic oscillator is almost self-dual. In particular, it did not turn into a fermion, but it did experience a shift in the homological grading, resulting in the extra minus sign on the right-hand side of \eqref{qqoscillator}. This can be considered as a first indication that the operation \eqref{BBdual} is best understood in the derived setting.

We can easily upgrade this 1d toy example to a simple 2d boson. Its contribution to the elliptic genus is given by the (inverse) Dedekind eta-function, which under \eqref{ffotherside} transforms as
\begin{equation}
\frac{1}{\eta (q)}  \quad \longleftrightarrow \quad  \frac{1}{q^{-1/24} \prod_{n=1}^{\infty} (1-q^{-n})} = \frac{1}{\eta (q)}
\end{equation}
where we used $\zeta$-function regularization for the sum $1 + 2 + 3 + \ldots$. As in our previous toy example, we find that 2d chiral boson is simply self-dual under \eqref{BBdual}. More precisely, $Q$-cohomology is self-dual, if we wish to promote this statement to the level of 2d $\CN=(0,2)$ theories. Either way, we learn that bosons did not turn into fermions, and we would come to a similar conclusion had we started with fermions.

\begin{figure}[ht]
	\centering
	\includegraphics[trim={0 0 0 0},clip,width=2.0in]{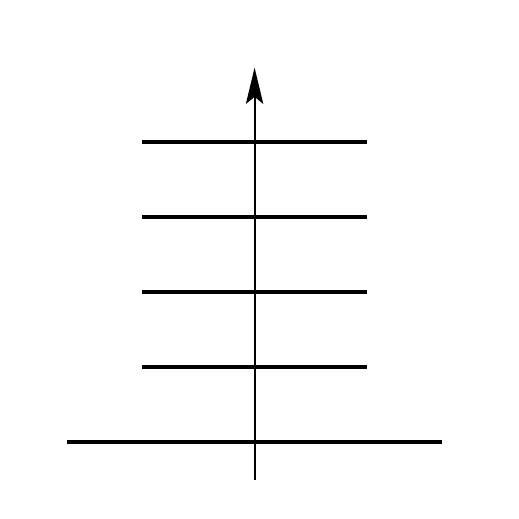}
	\caption{Energy states of a harmonic oscillator.}
	\label{fig:spectrum}
\end{figure}

Note, for resurgent analysis, which is our main tool in this paper, it is in fact helpful that functions $f(q)$ and $f(q)^{\vee}$ have non-trivial $\hbar$-expansions near $q = e^{\hbar} \approx 1$:
\begin{equation}
\sum_n a_n \hbar^n
\end{equation}
These expansions, up to $a_n \to (-1)^n a_n$, are equal for both $f(q)$ and $f(q)^{\vee}$ and serves as a bridge in relating the two functions. In physical applications to the 3d-3d correspondence, the non-triviality of this ``perturbative'' expansion is often directly related to the non-triviality of the 3d $\CN=2$ ``bulk'' theory, which controls the failure of the half-index of the combined system to be modular. In other words, for the techniques of this paper, it is actually helpful to have $f(q)$ lacking classical modular properties. For this reason, we expect that, in applications to Vafa-Witten theory and to $T[M_4]$, the techniques of this paper can be more useful for 4-manifolds with $b_2^+ = 1$ (or, $b_2^- = 1$). It would be interesting to explore these applications further.

\subsection{Example: false theta-functions}

An important special class of examples is the false theta functions \cite{Bring15,GMP,CCFGH,HLSS22}, for $p, a \in \mathbb Z$:
\begin{equation}
\tilde \Psi^{(a)}_p (q) \;  :=  \; \sum_{n=0}^\infty \psi^{(a)}_{2p}(n) q^{\frac{n^2}{4p}} \qquad \in q^\frac{a^2}{4p}\,\mathbb{Z}[[q]]
\label{falsetheta}
\end{equation}
$$
\psi^{(a)}_{2p}(n)  =  \left\{
\begin{array}{cl}
\pm 1, & n\equiv \pm a~\mod~ 2p\,, \\
0, & \text{otherwise}.
\end{array} \right.
$$
\noindent{\bf Conjecture:} There exists  a unique dual pair
\begin{equation}
\tilde \Psi^{(a)}_p (q)
\quad {{\footnotesize{\text{``}q \leftrightarrow 1/q\text{''}}} \atop \overleftrightarrow{\phantom{transformat}}} \quad
\tilde \Psi^{(a)}_p (q)^{\vee}
\end{equation}

\medskip
\begin{quote}
{\bf Goal:} For each given $p$ and $a$, write an explicit $q$-series $\tilde \Psi^{(a)}_p (q)^{\vee}$ which is dual to $\tilde \Psi^{(a)}_p (q)$ under $q\to q^{-1}$.
\end{quote}

In certain cases, an approach to this question, and to finding $F(q)^{\vee}$ more generally, is based on Rademacher sums. For example, with this approach one can reproduce Ramanujan's classical example (an order 7 mock theta function):
\begin{eqnarray}
F(q) &=& {q^{\frac{1}{168}}}   
\sum_{n\geq 0} 
\frac{(-1)^n q^{\frac{n(n+1)}{2}}}{(q^{n+1};q)_{n}}
\label{eq:mock71}
\\
&=&q^{\frac{1}{168}} (1-q-q^5+q^{10}-q^{11}+q^{18}+q^{30}-q^{41}+q^{43}-q^{56}-q^{76}+q^{93}-q^{96}+\dots)
\nonumber\\
F(q)^{\vee} &=& {q^{-\frac{1}{168}}} \sum_{n\geq 0}\frac{ q^{n^2} }{(q^{n+1};q)_{n}}
\label{eq:mock7v}
\\
&=& q^{-\frac{1}{168}} (1+q+q^3+q^4+q^5+2 q^7+q^8+2 q^9+q^{10}+2 q^{11}+q^{12}+3 q^{13}+\dots)
\nonumber
\end{eqnarray}
The pair $(F(q), F(q)^\vee)$ in (\ref{eq:mock71})-(\ref{eq:mock7v}) are related by the formal replacement $q \to q^{-1}$ in each term of the explicit $q$-hypergeometric expressions. However, as noted above, this formal replacement operation suffers from ambiguities.  The Rademacher sum approach, as well as other techniques available at present, resolve this ambiguity but are quite labor-intensive. In particular, they allow to construct $F(q)^{\vee}$ in a systematic manner, but the amount of work is substantial and only allows to treat one function $F(q)$ at a time. 

{\bf Comment:}
$F(q)$ in (\ref{eq:mock71}) can be expressed as a linear combination of four different false theta functions:
\begin{eqnarray}
    F(q)=\tilde{\Psi}_{42}^{(1)}(q)-\tilde{\Psi}_{42}^{(13)}(q)+\tilde{\Psi}_{42}^{(41)}(q)-\tilde{\Psi}_{42}^{(29)}(q)
    \label{eq:mock7-decomp}
\end{eqnarray}
The first two terms in this expression form an example of {\bf Proposition 4.15}, for $(a,b) =(3,1)$, in which case $4a^2+2ab=42$, and $4a+b=13$:
\begin{eqnarray}
   - F_{3,1}(q)&\equiv& \tilde{\Psi}_{42}^{(1)}(q)-\tilde{\Psi}_{42}^{(13)}(q)\\
    &=& q^{\frac{1}{186}}(1-q+q^{30}-q^{41}+q^{43}-q^{56}+q^{143}-q^{166}+q^{170}-q^{195}+\dots)
    \label{eq:f31a}
\end{eqnarray}
The remaining two false theta functions in (\ref{eq:mock7-decomp}) give the remaining terms in the expansion (\ref{eq:mock71}):
\begin{eqnarray}
   \tilde{\Psi}_{42}^{(41)}(q)-\tilde{\Psi}_{42}^{(29)}(q)=
   q^{\frac{1}{186}}(-q^5+q^{10}-q^{11}+q^{18}-q^{76}+q^{93}-q^{96}+q^{115}-q^{231}+\dots)
   \label{eq:f31b}
\end{eqnarray}
Note that the upper indices in (\ref{eq:f31b}) are $41=42-1$ and $29=42-13$.
In this sense, the $q$-series of $F_{a,b}(q)$ in (\ref{eq:Fab}) is roughly speaking ``half of'' the large $q$ expansion of a mock theta function. 

In the next Section we will discuss how resurgence and Borel summation can be used to address ``going to the other side'' for more general $q$-series, for example for a {\bf single} false theta function $\tilde \Psi^{(a)}_p(q)$. 

\subsection{Going to the other side with resurgence}
\label{sec:resurgence-bridge}

Consider a Borel-Laplace integral of the form
\begin{eqnarray}
    J(\hbar)=-\frac{1}{\hbar}\int_0^\infty B(u)\, e^{u^2/ \hbar}\, du
    \label{eq:borel}
\end{eqnarray}
where $\hbar<0$, so that $q=e^\hbar<1$.
Modulo a change of integration variable $u^2=\zeta$, $J(\hbar)$ is a Borel-Laplace representation with Borel kernel $ \zeta^{-1/2}B(\sqrt{\zeta})$.
When no confusion is possible, we will speak of the Borel plane in either variable, $u$ or $\zeta$.

The small $u$ expansion of the Borel transform function $B(u)$ generates a formal small $\hbar$ expansion of $J(\hbar)$. If the Maclaurin series of $B(u)$ has a non-zero finite radius of convergence then this small $\hbar$ expansion is divergent. Given the Borel integral (\ref{eq:borel}), we can also immediately deduce an expansion for large $\hbar$. This follows from the basic Fourier identity \cite{Grig98}:
\begin{eqnarray}
e^{u^2/\hbar}=2\sqrt{\frac{-\hbar}{\pi}}\int_0^\infty dv\, e^{v^2 \hbar}\, \cos(2v u) 
\label{eq:fourier}
\end{eqnarray}
Therefore we obtain a dual integral representation for $J(\hbar)$
\begin{eqnarray}
J(\hbar)=\frac{1}{\sqrt{-\pi\, \hbar}} \int_0^\infty dv\, e^{v^2\, \hbar}\, \tilde{B}(v)
\label{eq:mbdual}
\end{eqnarray}
where $\hbar<0$ and $\tilde{B}(v)$ is the Fourier transform of $B(u)$, assuming it exists. Expression (\ref{eq:mbdual}) is just the familiar Fourier-Poisson transformation.\footnote{This asymptotic analysis can also be approached using Mellin transforms: see Appendix \ref{app:mellin}.}

This defines a new Borel plane, the $v$ plane, associated with the large $\hbar$ expansion of $J(\hbar)$, coming from the small $v$ expansion of the Fourier transform $\tilde{B}(v)$. If the Maclaurin series of $B(v)$ has a non-zero finite radius of convergence then the large $\hbar$ expansion is also divergent. 

This phenomenon where both the small $\hbar$ and large $\hbar$ expansions are divergent\footnote{In many other examples in quantum field theory, quantum mechanics and string theory, the weak coupling expansion is divergent while the strong coupling expansion is convergent, or vice versa \cite{LeG90}.} occurs for example for the class of Borel functions $B(u)$ that are rational functions of $e^{-u}$. We refer to this as the Mordell-Borel class, since the Borel integrals can be decomposed into a sum of Mordell integrals. The Mordell-Borel class is particularly interesting because the Fourier transform $\tilde{B}(v)$ is also a rational function of $e^{-v}$, so the dual integral (\ref{eq:mbdual}) is again within the Mordell-Borel class.
A specific concrete class of such problems involves:
\begin{eqnarray}
B^{(s)}_{(p,a)}(u)&=&\frac{\sinh((p-a)u)}{\sinh(p u)}
\label{eq:sinh-borel}\\
B^{(c)}_{(p,a)}(u)&=&\frac{\cosh((p-a)u)}{\cosh(p u)}
\label{eq:cosh-borel}
\end{eqnarray}
for $p, a \in \mathbb Z$ and $a\in 1, 2, ..., (p-1)$. The corresponding Borel integrals appear in the construction of Mock Theta functions \cite{Wat,GM12} and the False theta functions in (\ref{falsetheta}). We focus on the sinh kernels since the analysis of the cosh ones is very similar. These will be discussed in detail below in Section \ref{sec:pa}.

\noindent{\bf Comments:} 
\begin{enumerate}
    \item 
    In terms of $q$, the transformation between small and large $\hbar$ corresponds to:
    \begin{eqnarray}
    q=e^{\hbar}\qquad \longleftrightarrow \qquad \tilde{q}=e^{\frac{\pi^2}{\hbar}}
    \label{eq:qqh}
    \end{eqnarray}
    \item 
    The dual Borel integral (\ref{eq:mbdual}) is significant because it plays an important role in the decomposition of the Borel integral (\ref{eq:borel}) into $q$-series and $\tilde{q}$-series, as is explained below. 
    \item 
    We stress that the mapping from small to large $\hbar$ via the dual integral representation is not restricted to the Mordell-Borel class, where $B(u)$ is a rational function of $e^{-u}$. It is sufficient for the Fourier transform $\tilde{B}(v)$ to exist and for the dual Borel integral to be well-defined. Therefore, it is possible to use these duality methods (for example, numerically) beyond the Mordell-Borel class of integrals.
  
\end{enumerate}

The natural starting point of the resurgent analysis of the Borel integral $J(\hbar)$ is the neighborhood of a Stokes line, where $\hbar\in \mathbb{R}^+$. This is perhaps counter-intuitive, because this is where the integral (\ref{eq:borel}) appears to break down. However, this is the most delicate regime and is precisely the regime in which resurgent analysis is most powerful.  Near the Stokes line, $J(\hbar)$ has a unique decomposition into a $q$-series and a linear combination of $\tilde{q}$-series.
This uniqueness is crucial because it implies that one can uniquely identify a unary $q$-series $F(q)$ with the real part of the Borel integral on the Stokes line. This separation into $q$-series and $\tilde{q}$-series on the Stokes line is preserved under Borel analytic continuation, thereby leading to a unique separation into $q$-series and $\tilde{q}$-series after $q\to \frac{1}{q}$. This therefore leads to a unique dual $q$-series $F(q)^\vee$ for $F(q)$.

Note that the Borel integral representation of a function (when it exists) is unique, when properly normalized and written in terms of the \'Ecalle critical variable  \cite{Ec81,Co08}. Fundamentally this is due to the injectivity of the Laplace operator. Moreover, two Borel  integral representations (written in terms of the \'Ecalle critical variable) coincide iff they agree trivially, that is, iff the Borel functions are identical. This is unlike $q$ series decompositions of Mordell integrals, among which there exist nontrivial identities, leading to non-uniqueness, for example as noted above in the difference between (\ref{eq:F1}) and (\ref{eq:Fv2}).

To analytically continue $J(\hbar)$ to the Stokes line we rotate $\hbar$ through the upper (or lower) half plane and simultaneously rotate the contour of the $u$ integration through the lower (or upper, resp.) half plane. When $\hbar$ reaches the fourth quadrant, the contour of $u$ integration of $J(\hbar)$ will thus be in the second quadrant. We write
\begin{eqnarray}
 J(\hbar) = \frac12(J+J^-) +\frac12(J-J^-)=\mathrm{PV} [J] +\pi i\,\, {\rm res}
\label{eq:split}
\end{eqnarray}
where $J^-$ is the same integral through the second quadrant, PV is the Cauchy principal value and $res$ are the residues of the integrand along $\mathbb R^+$. 

\noindent{\bf Comments on Uniqueness:}
\begin{enumerate}
    \item 
    The expression (\ref{eq:split}) gives a {\bf unique} decomposition of the Borel integral $J(\hbar)$ into a real and imaginary part when we continue to the Stokes line $\hbar\in \mathbb R^+$. This is also a unique decomposition into unary $q$-series and unary $\tilde{q}$-series, up to numerical prefactors and prefactor powers of $q$ and $\tilde{q}$, respectively.
    \item
    The real part follows from a residue analysis of the dual Borel integral (\ref{eq:mbdual}), while the imaginary part follows from a residue analysis of the original Borel integral (\ref{eq:borel}). 
  
      \item  In terms of analytic properties, for $\hbar>0$ we have
  \begin{equation}
    \label{eq:Stokes4} 
   \sqrt{\hbar}\, J (\hbar+i 0)
   = g_1(q)+i \sqrt{\frac{\pi}{\hbar}}\, g_2(\tilde{q})
  \end{equation}
where $g_1$ and $g_2$ are real-valued and have convergent Puiseux series. The decomposition into a pair $(g_1,g_2)$ with these properties is  manifestly unique.
The continuation is smooth and unique, even though  $g_1$ and $g_2$ separately have $\hbar\in i\mathbb R$ as a natural boundary.
\item 
We note that $J$ is a Borel sum of a resurgent asymptotic series. Assume that for $\hbar<0$ there also is a unique pair $(f_1(1/q),f_2(1/{\tilde q}))$ such that $f_1(1/q)+\sqrt{\frac{\pi}{-\hbar}} f_2(1/\tilde{q})
=\sqrt{-\hbar}\, J(\hbar)$ and $f_1,f_2$ have convergent Puiseux series with the same structure as those of $(g_1,g_2)$. When such a decomposition exists and is unique it is natural to identify $f_1,f_2$ as the continuation of $g_1,g_2$ across the boundary.  Uniqueness  and the properties of Borel summation guarantee that the continuation map across the boundary is  property-preserving (an extended isomorphism).  

\item In Section \ref{S-uniq-omega} we present an explicit proof of uniqueness for the class of order 3 Mock Theta functions, and conjecture that a similar approach should yield uniqueness more generally.

 \item 
    The decomposition produces {\it unary} series because they arise from a residue analysis of the Borel and dual Borel functions.
    
\end{enumerate}

In what follows we will see many instances\footnote{See {\it e.g.} \eqref{eq:31}, or \eqref{eq:511}, or \eqref{eq:paq-sym1}.} of the decomposition \eqref{eq:Stokes4} which, in turn, is a special instance of \eqref{ZSmS}. When $M_3$ is a homology sphere the first sum on the right-hand side of \eqref{ZSmS} simplifies and we can write it more explicitly as
\begin{multline}
\label{longeqZSmS}
\widehat{Z}_0 (M_3,q)
= \sum_{\alpha} \mathcal{S}_{0}^{\alpha} (\tilde q)
\; \mathcal{S} Z^{\text{pert}}_{\alpha} (q)
= \sum_{\alpha} c_{\alpha} \left( \sum_{n} m_0^{\bbalpha} \; \tilde q^{n+\text{CS} (\alpha)} \right)
\mathcal{S} Z^{\text{pert}}_{\alpha} (q)
\\
= \mathcal{S} Z^{\text{pert}}_{0} (q) + \sum_{\alpha \ne 0}
\mathcal{S}_{0}^{\alpha} (\tilde q)
\; \mathcal{S} Z^{\text{pert}}_{\alpha} (q)
= \mathcal{S} Z^{\text{pert}}_{0} (q) +
\sum_{\alpha \ne 0} c_{\alpha} \left( \sum_{n} m_0^{\bbalpha} \; \tilde q^{n+\text{CS} (\alpha)} \right)
\mathcal{S} Z^{\text{pert}}_{\alpha} (q)
\end{multline}
where $\bbalpha = (\alpha , n)$ denotes the integral lift of $\alpha$. In the context of non-perturbative complex Chern-Simons theory, \eqref{longeqZSmS} describes the transseries structure of the BPS $q$-series $\widehat{Z}_0 (M_3,q)$ in the case $H_1 (M_3, \mathbb{Z}) = 0$. The second line is the exact copy of the first line, with the contribution of the trivial flat connection $\alpha = 0$ singled out.

\subsection{Resurgent transseries for the Mordell-Borel class}
\label{sec:pa}

In this section we concentrate on the sinh-like Mordell-Borel class, and summarize various important identities for the Borel transform and dual Borel transform functions, which are important for the decomposition of the Borel integrals into $q$-series and $\tilde{q}$-series. In the following Sections we show how these decompositions can be achieved numerically.

We introduce the following notation for the relevant Mordell-Borel building blocks \cite{GM12}:
\begin{eqnarray}
J_{(p,a)}(\hbar):=\frac{1}{(-\hbar)} \int_0^\infty du\, e^{p u^2/\hbar} \, \frac{\sinh[(p-a)u]}{\sinh[p u]} 
\label{eq:js}
\end{eqnarray}
where $p, a\in \mathbb Z$, with $1\leq a <p$. Note that in this defining expression $\hbar<0$, corresponding to $q=e^\hbar<1$. We will subsequently analytically continue by rotating to $\hbar>0$.
The Fourier transform is known, giving the dual Borel integral \cite{GM12}:
\begin{eqnarray}
{J}_{(p,a)}(\hbar)= \frac{\sin\left(\frac{a\pi}{p}\right) }{\sqrt{p\pi (- \hbar) }} 
   \int_0^\infty dv  \, e^{p v^2 \hbar /(\pi^2)} \frac{1}{\cosh[2v] -\cos\left(\frac{a \pi}{p}\right)}
   \label{eq:jsd}
\end{eqnarray}
Notice the $\hbar\to\frac{\pi^2}{\hbar}$ transformation in the exponent of the Gaussian factors in the Borel integrands. 
We also note the following basic trigonometric identities relating the Borel transform and its dual:
\begin{eqnarray}
\frac{\sinh[(p-a)u]}{\sinh[p u]}&=&\frac{1}{p}\sum_{b=1}^{p-1} \sin\left( \frac{a b \pi}{p}\right)\,
 \frac{\sin\left(\frac{b\pi}{p}\right)}{\cosh[u]-\cos\left(\frac{b\pi}{p}\right)}
\nonumber \\
 \frac{\sin\left(\frac{a\pi}{p}\right)}{\cosh[u]-\cos\left(\frac{a\pi}{p}\right)}
 &=& 2\sum_{b=1}^{p-1} \sin\left( \frac{a b \pi}{p}\right)\,\frac{\sinh[(p-b)u]}{\sinh[p u]}
 \label{eq:identities}
 \end{eqnarray}
These identities imply that under the duality transformation $\hbar\to 4\pi^2/\hbar$ we obtain a finite linear combination of exactly the same Mordell-Borel building block integrals:
\begin{eqnarray}
 J_{(p,a)}(\hbar)= \left(\frac{2\pi}{(-\hbar)}\right)^{3/2}\,  \sqrt{\frac{2}{p}}\, \sum_{b=1}^{p-1} \sin\left(\frac{a b \pi}{p}\right)J_{(p,b)}\left(\frac{4\pi^2}{\hbar}\right)
\label{eq:Jmagic}
\end{eqnarray}
Note, the duality transformation here is implemented by a discrete Fourier transform with the kernel (``$S$-matrix''):
\begin{eqnarray}
    S_{ab}=\frac{2}{\sqrt{p}} \sin\left(\frac{a b \pi}{p}\right)
    \label{eq:mixing}
\end{eqnarray}
which has eigenvalues $\pm 1$. It is equal to the modular $S$-matrix in a rational CFT (WZW model) and, as argued in \cite{GPV}, determines the transformation of flat connections under the action of $SL(2,\mathbb{Z})$ modular group, {\it cf.} \eqref{Langl}. It also has interpretation in 3d-3d correspondence, namely in terms of the category of line operators MTC$[M_3]$ in 3d $\mathcal{N}=2$ theory \cite{CCFGH}.

The Borel integrals (\ref{eq:js})-(\ref{eq:jsd}) are initially defined for $\hbar<0$, but can be analytically continued 
to positive $\hbar$ via the rotated Borel integrals, valid for $\hbar>0$:
\begin{eqnarray}
J_{(p,a)}^{(\pm)}(\hbar)&:=&\mp i\, e^{\pm i\epsilon} \int_0^\infty du\, e^{-p (e^{\pm i\epsilon} u)^2 \hbar} \, \frac{\sin[(p - a) (e^{\pm i\epsilon} u)\, \hbar]}{\sin[p \, (e^{\pm i\epsilon} u)\, \hbar]} \qquad , \quad \epsilon\to 0^+
\label{eq:jspm}
\end{eqnarray}
Note that the sinh functions are replaced by sin functions. Therefore, we find that under analytic continuation from $\hbar<0$ to  $\hbar>0$ the Mordell-Borel  integral (\ref{eq:js}) acquires a real and imaginary part.
By straightforward contour deformation these $\hbar>0$ expressions can be written as:\footnote{For notational clarity we specialize here to expressions for $p$ odd. And it is convenient to extract a normalization factor of $\sqrt{\frac{-4\, p\, \hbar}{\pi}}$. }
\begin{eqnarray}
{\rm Re}\left[\sqrt{\frac{-4\, p\, \hbar}{\pi}} J_{(p,a)}\left(\hbar\right)\right]&=& 
\sqrt{\frac{4p \hbar}{\pi}}\frac{i}{2}\left(J_{(p,a)}^{(+)}\left(\hbar\right)-J_{(p,a)}^{(-)}\left(\hbar\right)\right) 
\nonumber\\
&&  = e^{\frac{-(p-a)^2 \hbar}{4\, p}}
\sum_{k=0}^\infty e^{-p\, \hbar\left(k+\frac{1}{2}\right)^2}\left(e^{(p-a) \hbar\left(k+\frac{1}{2}\right)}-e^{-(p-a) \hbar \left(k+\frac{1}{2}\right)}\right)
\label{eq:pareal}
\end{eqnarray}
\begin{eqnarray}
{\rm Im}\left[\sqrt{\frac{-4\, p\, \hbar}{\pi}} J_{(p,a)}\left(\hbar\right)\right]&=& 
\sqrt{\frac{4\, p\, \hbar}{\pi}}\frac{1}{2}
\left(J_{(p, a)}^{(+)}\left(\hbar\right) 
+ J_{(p, a)}^{(-)}\left(\hbar\right)\right)
\nonumber
\\
&& \hskip -4cm =\sqrt{\frac{4\pi}{p\hbar}}\sum_{b=1}^{{\rm Floor}\left[\frac{p}{2}\right]}
 \sin\left(\frac{a b \pi}{p}\right) e^{-\frac{b^2\pi^2}{p \hbar}}\left[1+\sum_{m=1}^\infty (-1)^{a\, m}\left(e^{-\frac{ \pi ^2 \left(m^2 p+2 m b\right)}{\hbar}}
 -e^{-\frac{ \pi ^2 \left(m^2 p-2 m b\right)}{\hbar}}\right)\right]
 \nonumber\\
\label{eq:paimag}
\end{eqnarray}
The real part yields a unary $q$-series with an overall factor of a rational power of $q$, while the imaginary part yields a linear combination of different unary $\tilde{q}$-series, each with an overall factor of a rational power of $\tilde{q}$ and a simple trigonometric coefficient.
The real part follows from a residue analysis of the dual Borel integral (\ref{eq:jsd}), while the imaginary part follows from a residue analysis of the original Borel integral (\ref{eq:js}). 
\medskip 

{\bf Comment:} Recalling that $q=e^\hbar$, we recognize that the real part (\ref{eq:pareal}) can be expressed in terms of the False Theta function in (\ref{falsetheta}):
\begin{eqnarray}
 e^{\frac{-(p-a)^2 \hbar}{4\, p}}
\sum_{k=0}^\infty e^{-p\, \hbar\left(k+\frac{1}{2}\right)^2}\left(e^{(p-a) \hbar\left(k+\frac{1}{2}\right)}-e^{-(p-a) \hbar \left(k+\frac{1}{2}\right)}\right)=\tilde{\Psi}_p^{(a)}\left(\frac{1}{q}\right)
\label{eq:psi-real}
\end{eqnarray}

\subsection{Numerical resurgence on the unary side}

In this Section we confirm that the decomposition into unary $q$-series and $\tilde{q}$-series when $\hbar>0$ can be achieved {\bf numerically}, directly from the analytically continued Borel integral, without invoking the analytic expressions (\ref{eq:pareal})-(\ref{eq:paimag}). We demonstrate this with several examples. 
\medskip

\subsubsection{Numerical Resurgence Example: Order 3 Mock Thetas}

Consider  $(p, a)=(3,1)$, corresponding to one of the original examples of Ramanujan \cite{Wat,GM12}, for which the Borel and dual Borel integrals are (for $\hbar<0$):
\begin{eqnarray}
J_{(3,1)}(\hbar)&=&-\frac{1}{\hbar} \int_0^\infty du\, e^{3 u^2/\hbar} \, \frac{\sinh[2u]}{\sinh[3 u]} 
\label{eq:js31}\\
&=& \frac{\sin\left(\frac{\pi}{3}\right)}{\sqrt{-3\pi \hbar }} 
   \int_0^\infty dv  \, e^{3 v^2 \hbar /(\pi^2)} \frac{1}{\cosh[2v] -\cos\left(\frac{\pi}{3}\right)}
   \label{eq:jsd31}
\end{eqnarray}
We analytically continue these integrals, as in (\ref{eq:jspm}), and form the $\pm$ linear combinations appropriate for the real and imaginary parts, as in (\ref{eq:pareal}) and (\ref{eq:paimag}). Plotting the real part as $\hbar\to +\infty$, it is straightforward to identify the leading $e^{-\frac{\hbar}{12}}$ behavior. After factoring out this leading exponential, we plot as a function of $q=e^\hbar$, and study the large $q$ behavior. The Mathematica command InterpolatingPolynomial identifies this as a unary series in $e^{-2\hbar}$:
\begin{eqnarray}
{\rm Re}\left[\sqrt{\frac{-12 \hbar}{\pi}} J_{(3,1)}\left(\hbar\right)\right] 
& =&\sqrt{\frac{12\hbar}{\pi}}\frac{i}{2}
\left(J_{(3,1)}^{(+)}\left(\hbar\right)-J_{(3,1)}^{(-)}\left(\hbar\right)\right) 
 \nonumber \\
&\sim&
e^{-\frac{\hbar}{12}}  \left(
1-e^{-2\hbar}+e^{-4\hbar}-e^{-10\hbar}+e^{-14\hbar}-e^{-24 \hbar}+\dots
\right) 
\label{eq:j31pm_real}
\end{eqnarray}
See Figure \ref{fig:31real}. With our identification $q=e^\hbar$, we recognize the large $\hbar$ limit in (\ref{eq:j31pm_real}) as the large $q$ expansion of the false theta function $\tilde \Psi^{(1)}_3 \left(\frac{1}{q}\right)$ [recall (\ref{falsetheta}]):
\begin{eqnarray}
q^{\frac{1}{12}}\, \tilde \Psi^{(1)}_3 \left(\frac{1}{q}\right) =  1 - q^{-2} + q^{-4} - q^{-10} + q^{-14} + \ldots 
\label{eq:31false}
\end{eqnarray}

For the imaginary part, plotting for $\hbar\to 0^+$, we can similarly identify the leading behavior $\sqrt{\frac{\pi}{\hbar}} e^{-\frac{\pi^2}{3\hbar}}$, and then use interpolation in $\tilde{q}$ to find
\begin{eqnarray}
    {\rm Im}\left[\sqrt{\frac{-12\hbar}{\pi}} J_{(3,1)}\left(\hbar\right)\right] &=&
\sqrt{\frac{12\hbar}{\pi}}\frac{1}{2}\left(J_{(3,1)}^{(+)}\left(\hbar\right)+J_{(3,1)}^{(-)}\left(\hbar\right)\right)
\nonumber
\\
&& \hskip-4cm \sim\sqrt{\frac{\pi}{\hbar}} \, 
e^{-\frac{\pi^2}{3\hbar}}  \left(1+e^{-\frac{\pi^2}{\hbar}}-e^{-\frac{5\pi^2}{\hbar}}
-e^{-\frac{8\pi ^2}{\hbar}}
+e^{-\frac{16 \pi^2}{\hbar}}
+e^{-\frac{21 \pi ^2}{\hbar}}
+\dots\right) 
\label{eq:j31pm_imag}
\end{eqnarray}
Recalling that $\tilde{q}=e^{\frac{\pi^2}{\hbar}}$, we recognize the $\tilde{q}$-series part of this expression as 
\begin{eqnarray}
    \left[q^{\frac{1}{3}}\, \tilde{\Psi}_3^{(2)}\left(\frac{1}{q}\right)\right]_{q\to -\tilde{q}}=1+\tilde{q}^{-1}-\tilde{q}^{-5}-\tilde{q}^{-8}+\tilde{q}^{-16}+\tilde{q}^{-21}-\dots 
\end{eqnarray}
These results for the real and imaginary part are shown in Figures \ref{fig:31real} and \ref{fig:31imag}, respectively.
\begin{figure}[htb]
\centerline{\includegraphics[scale=.5]{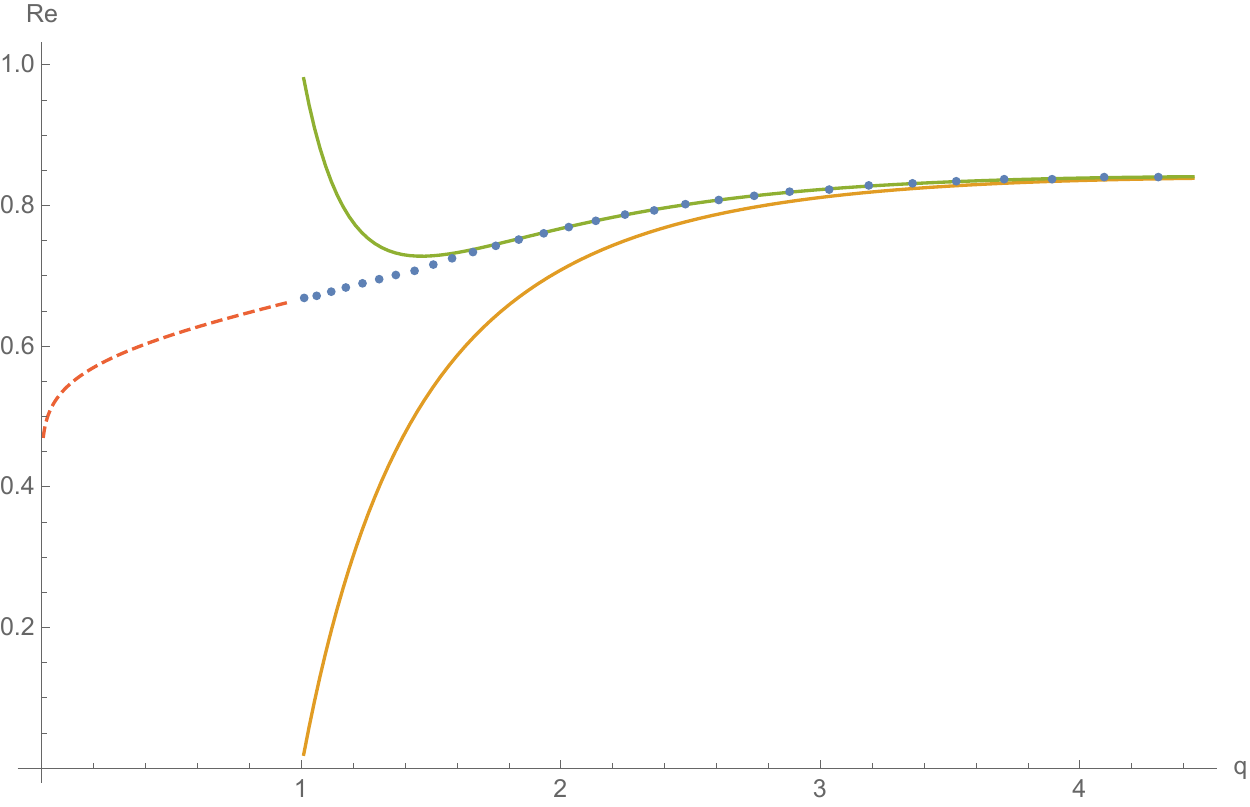}}
\caption{The blue dots denote the real part of $\sqrt{-12\hbar/\pi}J_{(3,1)}(\hbar)$, computed from the analytically continued Borel integral as in the first line of (\ref{eq:j31pm_real}). The orange and green curves show the first and second corrections to the large $q$ behavior ($\hbar\to+\infty$) in (\ref{eq:j31pm_real}). The red dashed line shows the corresponding real quantity for $q<1$ ($\hbar<0$), from the original Borel integral (\ref{eq:js31}).
 }
\label{fig:31real}
\end{figure}
\begin{figure}[htb]
\centerline{\includegraphics[scale=.5]{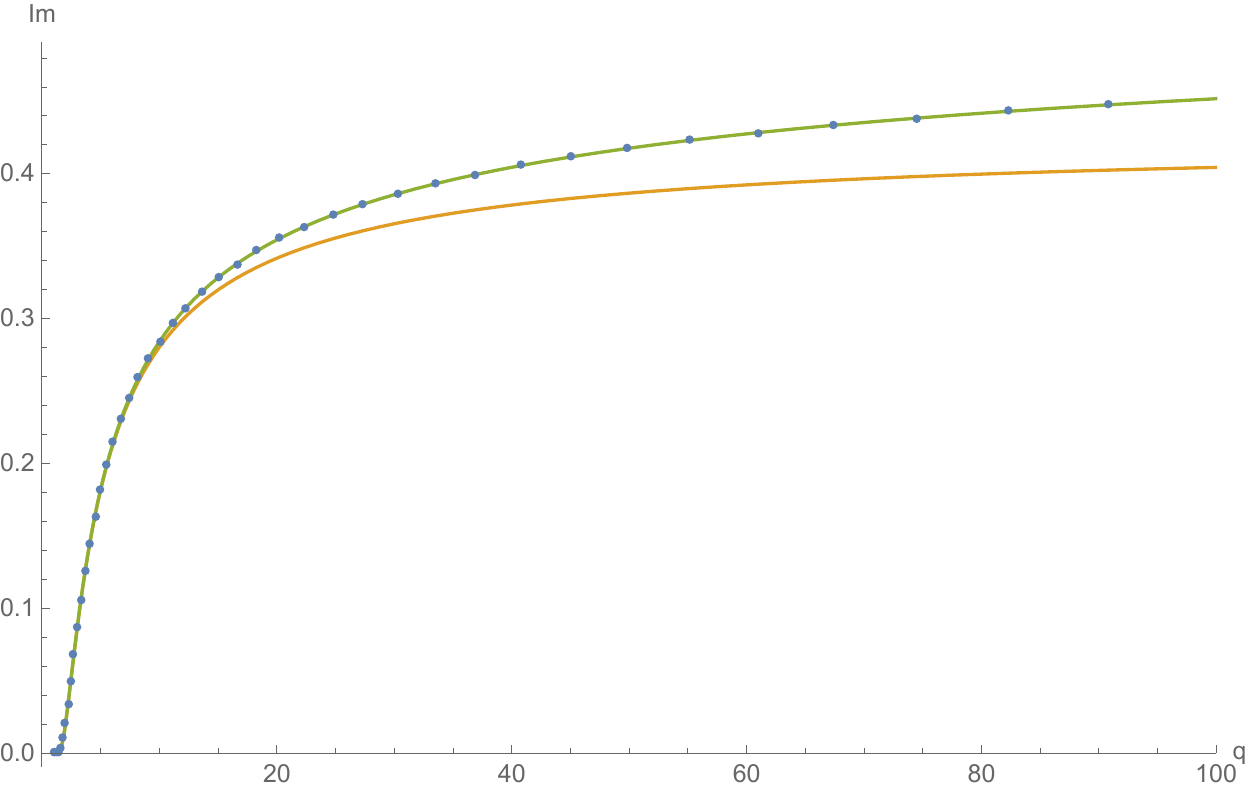}}
\caption{
The blue dots denote the imaginary part of $\sqrt{-12\hbar/\pi}J_{(3,1)}(\hbar)$, evaluated numerically from the analytically continued Borel integral as in the first line of (\ref{eq:j31pm_imag}). This vanishes for $q<1$ ($\hbar<0$) but is non-zero for $q>1$ ($\hbar>0$). The orange curve shows the {\it leading} term of the imaginary part in (\ref{eq:j31pm_imag}), namely $\sqrt{\frac{\pi}{\hbar}}\, e^{-\pi^2/(3\hbar)}$. The green curve shows the effect of including also the next exponentially suppressed term: $\sqrt{\frac{\pi}{\hbar}}\, e^{-\pi^2/(3\hbar)}(1+e^{-\pi^2/\hbar})$. }
\label{fig:31imag}
\end{figure}

{\bf Comments:}
\begin{enumerate}
    \item 
    The numerically derived unary $q$-series for the real part and the unary $\tilde{q}$-series for the imaginary part agree precisely with the analytic residue expressions in (\ref{eq:pareal}) and (\ref{eq:paimag}). 
    \item 
    The expansions (\ref{eq:j31pm_real})-(\ref{eq:j31pm_imag}) agree with the analytic continuation of the known mock-modular relation connecting the order 3 Mock theta functions $f$ and $\omega$ via a Borel integral \cite{Wat,GM12}:
    \begin{eqnarray}
\sqrt{\frac{-12\hbar}{\pi}} J_{(3,1)}\left(\hbar\right)&=& e^{-\frac{\hbar}{12}}\, \frac{1}{2}\, f(e^{2\hbar}) -\sqrt{\frac{4\pi}{-3\hbar}} \sin\left(\frac{\pi}{3}\right) e^{\frac{2\pi^2}{3\hbar}} \omega\left( e^{ \frac{\pi^2}{\hbar}}\right)
\label{eq:31}
\end{eqnarray}
To see this, recall the expansions of the order 3 mock theta functions $f$ and $\omega$:
\begin{eqnarray}
f(q)&:=& 1-\sum_{n=1}^\infty \frac{(-1)^n q^{n}}{(-q; q)_{n-1}}
\label{eq:fq}
\\
&=&1+q-2q^2+3q^3-3q^4+3q^5-5q^6+\dots\quad,\quad q\to 0
\label{eq:fq-small}
\\
&=&2\left(1-\frac{1}{q}+ \frac{1}{q^2}-\frac{1}{q^5}+\frac{1}{q^7}-\frac{1}{q^{12}}+\frac{1}{q^{15}}+
\dots \right) \quad,\quad q\to\infty 
\label{eq:fq-large}
\end{eqnarray}
\begin{eqnarray}
\omega(q)&:=& \sum_{n=0}^\infty \frac{q^{n}}{(q; q^2)_{n+1}}
\label{eq:wq}
\\
&=&1+2 q+3 q^2+4 q^3+6 q^4+8 q^5+10 q^6+\dots\quad,\quad q\to 0
\label{eq:wq-small}
\\
&=&-\frac{1}{q}-\frac{1}{q^2}+\frac{1}{q^6}+
\frac{1}{q^9}-\frac{1}{q^{17}}-\frac{1}{q^{22}}+
\dots \quad,\quad q\to\infty 
\label{eq:wq-large}
\end{eqnarray}
The expansions (\ref{eq:j31pm_real})-(\ref{eq:j31pm_imag}) can be identified with the unary large $q$ expansions of $\frac{1}{2} f(q^2)$ and $-q\, \omega(q)$.

\end{enumerate}

 These results confirm that the resurgent medianization and Stokes phenomenon for the rotated integrals generates the correct unary $q$-series and $\tilde{q}$-series after rotation to the Stokes line ($\hbar>0$). 

\subsubsection{Numerical Resurgence Example: Order 10 Mock Thetas}
\label{sec:mock10}

A similar analysis applies for order 10 mock theta functions, with an interesting new feature. Consider the choice $(p, a)=(5, 1)$:
\begin{eqnarray}
J_{(5,1)}(\hbar)&=&-\frac{1}{\hbar} \int_0^\infty du\, e^{5 u^2/\hbar} \, \frac{\sinh[4u]}{\sinh[5 u]} 
\label{eq:js51}\\
&=& \frac{\sin\left(\frac{\pi}{5}\right)}{\sqrt{-5\pi \hbar }} 
   \int_0^\infty dv  \, e^{5 v^2 \hbar /\pi^2} \frac{1}{\cosh(2v) -\cos\left(\frac{ \pi}{5}\right)}
   \label{eq:jsd51}
\end{eqnarray}
\begin{figure}[htb]
\centerline{\includegraphics[scale=.5]{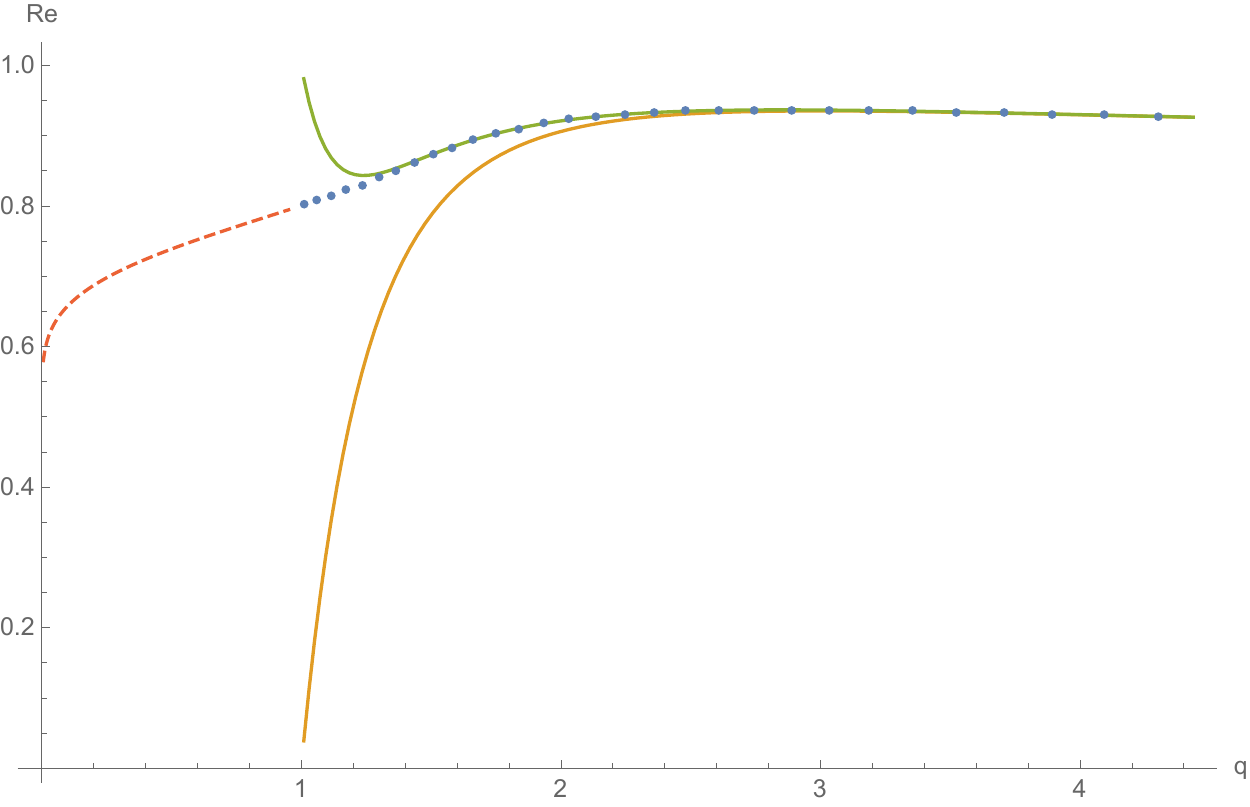}}
\caption{The blue dots denote the real part of $\sqrt{-20\hbar/\pi}J_{(5,1)}(\hbar)$, computed from the analytically continued Borel integral as in  (\ref{eq:51real1}). The orange and green curves show the first and second corrections to the large $q$ behavior ($\hbar\to+\infty$) in (\ref{eq:51real}). The red dashed line shows the corresponding real quantity for $q<1$ ($\hbar<0$).
 }
\label{fig:51real}
\end{figure}
\begin{figure}[htb]
\centerline{\includegraphics[scale=.5]{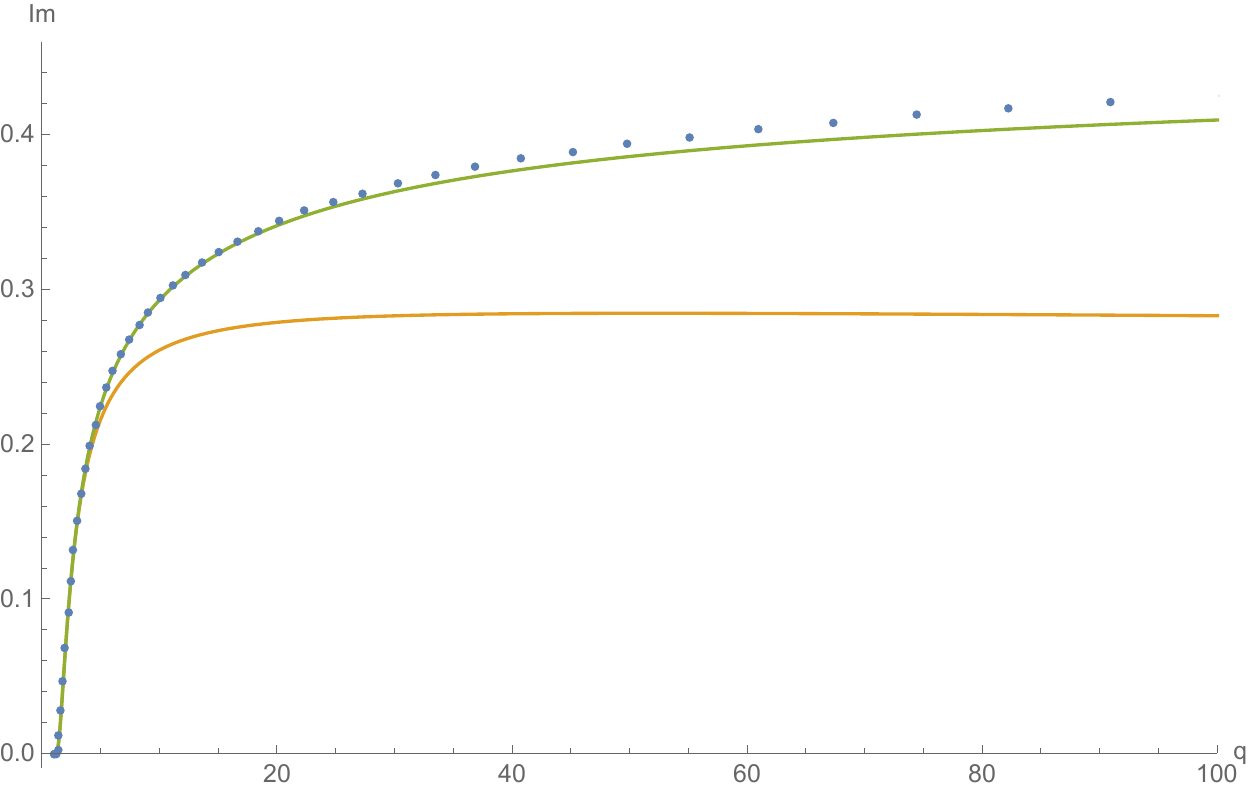}}
\caption{
The blue dots denote the imaginary part of $\sqrt{-20\hbar/\pi}J_{(5,1)}(\hbar)$, evaluated numerically from the analytically continued Borel integral as in (\ref{eq:51imag1}). This vanishes for $q<1$ ($\hbar<0$) but is non-zero for $q>1$ ($\hbar>0$). The orange curve shows the {\it leading} term of the imaginary part in (\ref{eq:51imag}), namely $\sqrt{\frac{4\pi}{5\hbar}}\,\sin\left(\frac{\pi}{5}\right)\, e^{-\pi^2/(5\hbar)}$. The green curve shows the effect of including the next exponentially suppressed term: $\sqrt{\frac{4\pi}{5\hbar}}\left[\sin\left(\frac{\pi}{5}\right) e^{-\pi^2/(5\hbar)}+\sin\left(\frac{2\pi}{5}\right) e^{-4\pi^2/(5\hbar)}\right]$. This green curve is indistinguishable on this scale from the full $\tilde{q}$ dependence in (\ref{eq:511}).}
\label{fig:51imag}
\end{figure}

Analytically continuing from $\hbar<0$ to $\hbar > 0$, from the rotated integrals (\ref{eq:jspm}) and the general expression for the real part (\ref{eq:pareal}) we deduce that for $\hbar>0$:
\begin{eqnarray}
{\rm Re}\left[\sqrt{\frac{-20\hbar}{\pi}} J_{(5,1)}\left(\hbar\right) \right] 
&=& 
\sqrt{\frac{20\hbar}{\pi}}\frac{i}{2}\left(J_{(5,1)}^{(+)}\left(\hbar\right) -J_{(5,1)}^{(-)}\left(\hbar\right) \right) 
\label{eq:51real1}
\\
&&\hskip -5cm \sim
e^{-\frac{\hbar}{20}} \left(1-e^{-4\hbar}+e^{-6\hbar}-e^{-18\hbar}+e^{-22\hbar}-e^{-42\hbar}+\dots \right)\quad, \quad \hbar\to +\infty 
\label{eq:51real}
\end{eqnarray}
See Figure \ref{fig:51real}. With our identification $q=e^\hbar$, we recognize the large $\hbar$ limit in (\ref{eq:51real}) as the large $q$ expansion of the false theta function $\tilde \Psi^{(1)}_5 \left(\frac{1}{q}\right)$ [recall (\ref{falsetheta}]):
\begin{eqnarray}
q^{\frac{1}{20}}\, \tilde \Psi^{(1)}_5 \left(\frac{1}{q}\right) =  1 - q^{-4} + q^{-6} - q^{-18} + q^{-22} + \ldots 
\label{eq:51false}
\end{eqnarray}
For the imaginary part in the $\hbar>0$ region there are two distinct $\tilde{q}$-series. These $\tilde{q}$-series are unary and follow directly from the residue analysis of the analytically continued integrals (\ref{eq:jspm}). See the general expression (\ref{eq:paimag}) for the imaginary part when $\hbar\to 0^+$:
\begin{eqnarray}
{\rm Im}\left[\sqrt{\frac{-20 \hbar}{\pi}} J_{(5,1)}\left(\hbar\right) \right] 
&=& 
 \sqrt{\frac{20 \hbar}{\pi}}\frac{1}{2}\left(J_{(5,1)}^{(+)}\left(\hbar\right) +J_{(5,1)}^{(-)}\left(\hbar\right) \right) 
 \label{eq:51imag1}
\\
&& \hskip -3cm \sim \sqrt{\frac{4\pi}{5\hbar}}\left[\sin\left(\frac{\pi}{5}\right)   e^{-\frac{\pi^2}{5\hbar}}\left(1 +e^{-\frac{3\pi^2}{\hbar}} -e^{-\frac{7\pi^2}{\hbar}} - e^{-\frac{16\pi^2}{\hbar}}+\dots\right) \right. \nonumber\\
&& \hskip -3cm
\left. +\sin\left(\frac{2\pi}{5}\right)  e^{-\frac{4\pi^2}{5\hbar}} \left(1 +e^{-\frac{\pi^2}{\hbar}} - e^{-\frac{9\pi^2}{\hbar}}-e^{-\frac{12\pi^2}{\hbar}}+\dots\right) \right]
\label{eq:51imag}
\end{eqnarray}
See Figure \ref{fig:51imag}. With our identification $\tilde{q}=e^{\frac{\pi^2}{\hbar}}$, we recognize the $\hbar\to 0^+$ limit in (\ref{eq:51imag}) as a linear combination of the false theta function $\tilde \Psi^{(2)}_5$ and $\tilde \Psi^{(4)}_5$ [recall (\ref{falsetheta}]):
\begin{eqnarray}
\left[q^{\frac{1}{5}}\, \tilde \Psi^{(2)}_5 \left(\frac{1}{q}\right)\right]_{q\to -\tilde{q}} &=&  1 + \tilde{q}^{-3} - \tilde{q}^{-7} - \tilde{q}^{-16} + \tilde{q}^{-24} + \tilde{q}^{-39}-\ldots 
\label{eq:51falsetilde1}
\\
\left[q^{\frac{4}{5}}\, \tilde \Psi^{(2)}_5 \left(\frac{1}{q}\right)\right]_{q\to -\tilde{q}} &=&  1 + \tilde{q}^{-1} - \tilde{q}^{-9} - \tilde{q}^{-12} + \tilde{q}^{-28} + \tilde{q}^{-33}-\ldots 
\label{eq:51falsetilde2}
\end{eqnarray}

As in the mock 3 case, these large $q$ and large $\tilde{q}$ expansions can be deduced numerically, for example using the InterpolatingPolynomial command in Mathematica. The resulting unary series match the general analytic forms in (\ref{eq:pareal})-(\ref{eq:paimag}).
These results for the real and imaginary parts are shown in Figures \ref{fig:51real} and \ref{fig:51imag}.

These results also demonstrate that the resurgent expansions and the analytic continuation from $\hbar<0$ to $\hbar>0$ preserve the form of the known mock modular relation: see (\ref{eq:511}). For $q<1$, (i.e. $\hbar<0$) this mock-modular relation \cite{GM12} can be written in the form \eqref{eq:Stokes4}--\eqref{longeqZSmS}:
\begin{eqnarray}
\sqrt{\frac{-20\hbar}{\pi}} J_{(5,1)}\left(\hbar\right) &=& e^{-\frac{\hbar}{20}}\,X(e^{2\hbar})
\nonumber
\\
&& \hskip -2cm -\sqrt{-\frac{4\pi}{5\hbar}} \left[\sin\left(\frac{\pi}{5}\right) e^{-\frac{\pi^2}{5\hbar}} \psi\left( e^{\frac{\pi^2}{\hbar}}\right)+
\sin\left(\frac{2\pi}{5}\right) e^{\frac{\pi^2}{5\hbar}} \phi\left(e^{\frac{\pi^2}{\hbar}}\right)\right]
\label{eq:511}
\end{eqnarray}
where $X(q^2)$, $\psi(q)$ and $\phi(q)$ are  order 10 Mock theta functions. The structure is similar  to the $(p,a)=(3,1)$ case (\ref{eq:31}) of the previous section, but with the new feature that the RHS of (\ref{eq:511}) involves one $q$-series, the order 10 mock theta function $X(q^2)$, but \underline{two different} $\tilde{q}$-series, the order 10 mock theta functions $\psi(\tilde{q})$ and $\phi(\tilde{q})$. Compare (\ref{eq:51real}) and (\ref{eq:51imag}) with the small and large $q$ expansions of the relevant order 10 mock theta functions. All these expansions have integer-valued coefficients, and furthermore the large $q$ expansions are all unary $q$-series.
\begin{eqnarray}
X(q)&=& \sum_{n=0}^\infty \frac{(-1)^n q^{n^2}}{(-q; q)_{2n}}
\label{eq:x10}
\\
&\sim & 1-q+q^2+q^4-2q^5+q^6-q^7 +\dots 
\label{eq:x10_small}
\\
&\sim & 1 -\frac{1}{q^2}+\frac{1}{q^3}-\frac{1}{q^9}+\frac{1}{q^{11}}-\frac{1}{q^{21}}+\dots
\label{eq:x10_large}
\end{eqnarray}
Note that the large $q$ expansion of $X$ corresponds to the false theta function in (\ref{eq:51false}):
\begin{eqnarray}
X(q^2) =  1 - q^{-4} + q^{-6} - q^{-18} + q^{-22} + \ldots = q^{\frac{1}{20}}\, \tilde \Psi^{(1)}_5 \left(\frac{1}{q}\right)
\label{eq:51false2}
\end{eqnarray}
We will see below in Section \ref{sec:numerical-other-side} that the small $q$ expansion of $X$ is related to the dual of a false theta function:
\begin{eqnarray}
     X(q^2) =  1 - q^2 + q^4 + q^8 - 2 q^{10} + \ldots 
    =q^{\frac{1}{20}}\,\tilde \Psi^{(1)}_5 (q)^{\vee}
\end{eqnarray}
The $\tilde{q}$-series in (\ref{eq:511}) are related to two other order 10 mock theta functions. We have
\begin{eqnarray}
\psi(q)&=& \sum_{n=0}^\infty \frac{q^{\frac{1}{2}(n+1)(n+2)}}{(q; q^2)_{n+1}}
\label{eq:psi10}
\\
&\sim & q+q^2+2 q^3+2 q^4+2 q^5+\dots 
\label{eq:psi10_small}
\\
&\sim &-1-\frac{1}{q^3}+\frac{1}{q^7}+\frac{1}{q^{16}}-\frac{1}{q^{24}
   }+\dots 
   \label{eq:psi10_large}
\end{eqnarray}
\begin{eqnarray}
\phi(q)&=& \sum_{n=0}^\infty \frac{q^{\frac{1}{2}n(n+1)}}{(q; q^2)_{n+1}}
\label{eq:phi10}
\\
&\sim & 1+2 q+2 q^2+3 q^3+4 q^4+4 q^5+\dots 
\label{eq:phi10_small}
\\
&\sim &-\frac{1}{q}-\frac{1}{q^2}+\frac{1
   }{q^{10}}+\frac{1}{q^{13}}-\frac{1}{q^{29}}-\dots
   \\
&\sim &-\frac{1}{q}\left(1+\frac{1}{q}-\frac{1
   }{q^{9}}-\frac{1}{q^{12}}+\frac{1}{q^{28}}+\dots\right)
   \label{eq:phi10_large}
\end{eqnarray}
We can identify the large $q$ expansions of $\psi$ and $\phi$ with false theta functions, as in (\ref{eq:51falsetilde1})-(\ref{eq:51falsetilde2}):
\begin{eqnarray}
  -\psi\left(-\frac{1}{q}\right) &=&1-q^3+q^7-q^{16}+q^{24}-q^{39}+\dots
    =q^{\frac{1}{5}} \tilde{\Psi}_5^{(2)}(q)
    \\
\frac{1}{q}\, \phi\left(-\frac{1}{q}\right)&=&    1-q+q^9-q^{12}+q^{28}-q^{33}+\dots
=
q^{\frac{4}{5}} \tilde{\Psi}_5^{(4)}(q)
\end{eqnarray}
We will also see  below in Section \ref{sec:numerical-other-side} that the small $\tilde{q}$ expansions of $\psi$ and $\phi$ are related to duals of false theta functions.
\medskip

\subsubsection{Numerical Resurgence Example: Beyond Mock}
In this Section we show that this behavior extends also to higher values of $p$, for which there is no known mock-modular relation. However, the resurgence argument shows that the Borel integral encodes all the necessary information to go to the other side.
We choose $(p, a)=(7, 1)$ and take as our input the Borel integral, and its dual Borel integral representation, initially defined for $\hbar<0$:
\begin{eqnarray}
J_{(7,1)}(\hbar)&=&-\frac{1}{\hbar} \int_0^\infty du\, e^{7 u^2/\hbar} \, \frac{\sinh[6 u]}{\sinh[7 u]} 
\label{eq:js71}\\
&=& \frac{\sin\left(\frac{\pi}{7}\right)}{\sqrt{-7\pi \hbar }} 
   \int_0^\infty du  \, e^{7 u^2 \hbar /(\pi^2)} \frac{1}{\cosh(2u) -\cos\left(\frac{ \pi}{7}\right)}
   \label{eq:jsd71}
\end{eqnarray}

\begin{figure}[htb]
\centerline{\includegraphics[scale=.5]{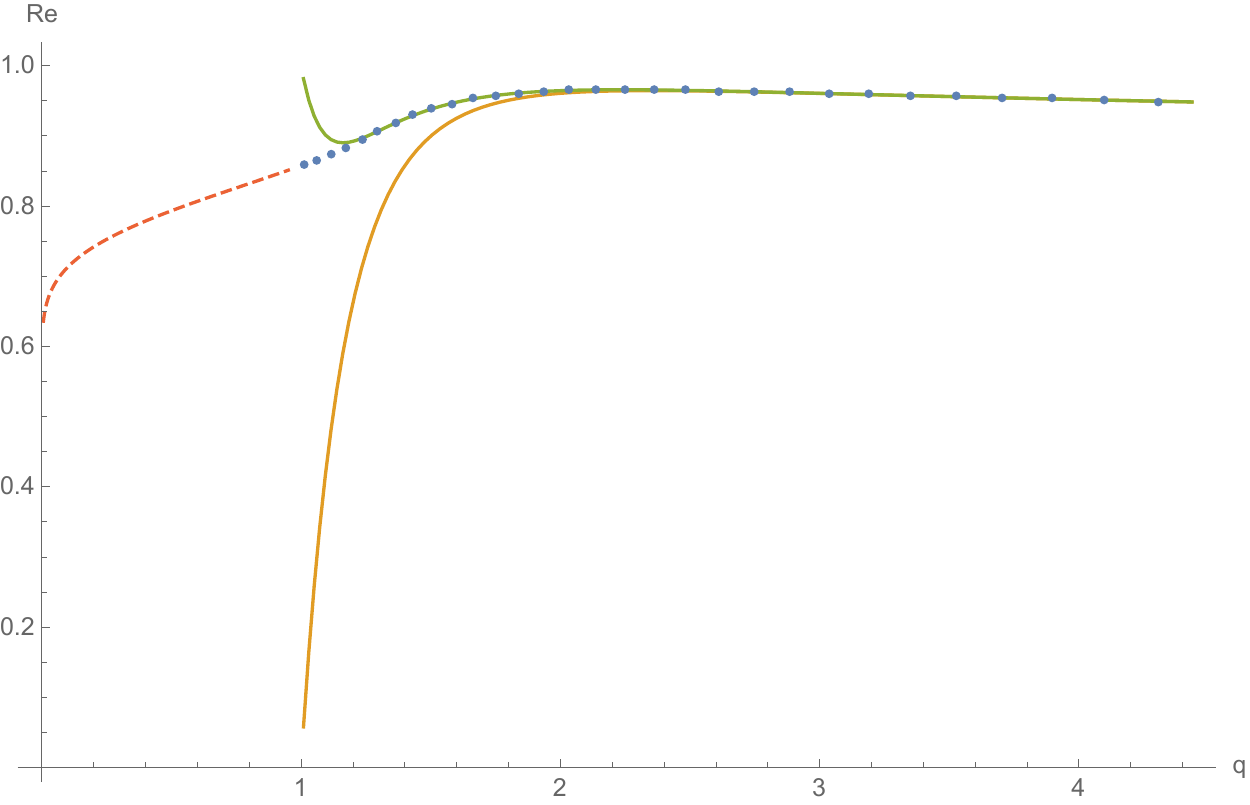}}
\caption{The blue dots show the numerical analytic continuation values of the real part of $\sqrt{\frac{-28 \hbar}{\pi}} J_{(7,1)}(\hbar)$ for $q>1$ (i.e. $\hbar>0$). The orange and green curves show the first and second corrections to the large $q$ behavior ($\hbar\to+\infty$) in (\ref{eq:71real}). The red dashed line shows the corresponding real quantity for $q<1$ ($\hbar<0$). Compare with Figures \ref{fig:31real} and \ref{fig:51real}.
 }
\label{fig:71real}
\end{figure}

This is a clear generalization of (\ref{eq:js31})-(\ref{eq:jsd31}) and (\ref{eq:js51})-(\ref{eq:jsd51}).
However, for $(p,a)=(7,1)$ there is no known mock-modular relation, such as (\ref{eq:31}) or (\ref{eq:511}), with which to compare. Nevertheless, from the Borel integrals (\ref{eq:js71})-(\ref{eq:jsd71}) it is straightforward to analytically continue to  $\hbar>0$ (i.e., $q>1$), and obtain unique decompositions into unary $q$-series and $\tilde{q}$-series. This analytic continuation produces one unary $q$-series for the real part, and a linear combination of three different $\tilde{q}$-series for the imaginary part, as in the general $(p, a)$ expressions (\ref{eq:pareal}) and (\ref{eq:paimag}). 
For $(p,a)=(7,1)$, numerical fitting for $\hbar\to +\infty$ leads to the following large $q$ expansion for the real part:
\begin{eqnarray}
{\rm Re}\left[\sqrt{\frac{-28\hbar}{\pi}}J_{(7,1)}\left(\hbar\right)\right]
&=&  \sqrt{\frac{-28\hbar}{\pi}}\frac{i}{2}\left(J_{(7,1)}^{(+)}\left(\hbar\right)-J_{(7,1)}^{(-)}\left(\hbar\right)\right) \\
  && \hskip -2cm =
e^{-\frac{\hbar}{28}} \left(1-e^{-6\hbar} +e^{-8\hbar}-e^{-26\hbar}+e^{-30\hbar} -e^{-60\hbar}+\dots\right)
\label{eq:71real}
\end{eqnarray}
See Figure \ref{fig:71real}. With our identification $q=e^\hbar$, we recognize the large $\hbar$ limit in (\ref{eq:71real}) as the large $q$ expansion of the false theta function $\tilde \Psi^{(1)}_7 \left(\frac{1}{q}\right)$ [recall (\ref{falsetheta}]):
\begin{eqnarray}
q^{\frac{1}{28}}\, \tilde \Psi^{(1)}_7 \left(\frac{1}{q}\right) =  1 - q^{-6} + q^{-8} - q^{-26} + q^{-30} + \ldots 
\label{eq:71false}
\end{eqnarray}
Similarly, numerical fitting for $\hbar\to 0^+$ leads to the following large $\tilde{q}$ expansion for the imaginary part:
\begin{eqnarray}
{\rm Im}\left[\sqrt{\frac{-28\hbar}{\pi}} J_{(7,1)}\left(\hbar\right)\right]
&=& \sqrt{\frac{-28\hbar}{\pi}}
\frac{1}{2}\left(J_{(7,1)}^{(+)}\left(\hbar\right) + J_{(7,1)}^{(-)}\left(\hbar\right) \right)
\label{eq:71imag}\\
&&
\hskip -3cm \sim \sqrt{\frac{4\pi}{7\hbar}}\left\{\sin\left(\frac{\pi}{7}\right) e^{-\frac{\pi^2}{7\hbar}}+\sin\left(\frac{2\pi}{7}\right) e^{-\frac{4\pi^2}{7\hbar}}+\sin\left(\frac{3\pi}{7}\right) e^{-\frac{9\pi^2}{7\hbar}}\right\}
\label{eq:71imag_leading}
\end{eqnarray}
See Figure \ref{fig:71imag}. In the last line (\ref{eq:71imag_leading}), for each trigonometric coefficient  we have kept just the leading term in the large $\tilde{q}$ limit (i.e. $\hbar\to 0^+$ and $q\to 1^+$). In Figure \ref{fig:71imag} we plot
the successive inclusions of these first three contributions, $\sin\left(\frac{\pi}{7}\right) e^{-\frac{\pi^2}{7\hbar}}$, $\sin\left(\frac{2\pi}{7}\right) e^{-\frac{4\pi^2}{7\hbar}}$ and $\sin\left(\frac{3\pi}{7}\right) e^{-\frac{9\pi^2}{7\hbar}}$, shown as the orange, green and red curves. These are expansions generated about $q=1^+$, and yet they are remarkably accurate at large values of $q$. At extremely large values of $q$ we require the further large $\tilde{q}$ exponential terms in (\ref{eq:paimag}).
\begin{figure}[htb]
\centerline{\includegraphics[scale=.6]{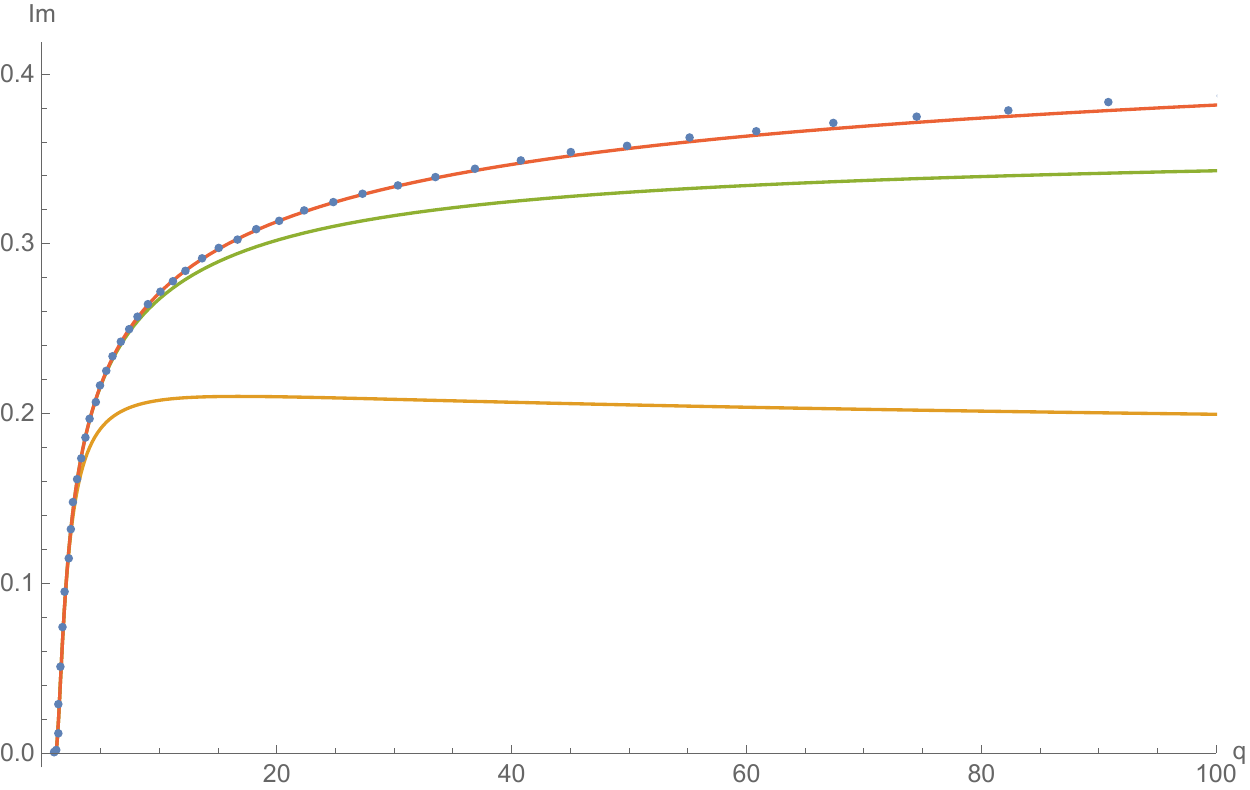}}
\caption{
The blue dots denote the imaginary part of $\sqrt{\frac{-28\hbar}{\pi}}\, J_{(7,1)}(\hbar)$, evaluated numerically from the analytically continued Borel integral as in (\ref{eq:71imag}). This vanishes for $q<1$ ($\hbar<0$) but is non-zero for $q>1$ ($\hbar>0$). The orange curve shows the {\it leading} term of the imaginary part in (\ref{eq:71imag_leading}), namely $\sqrt{\frac{4\pi}{7\hbar}}\,\sin\left(\frac{\pi}{7}\right)\, e^{-\pi^2/(7\hbar)}$. The green and red curves show the effects of including successively the next exponentially suppressed terms in (\ref{eq:71imag_leading}). Compare with Figures \ref{fig:31imag} and \ref{fig:51imag}. }
\label{fig:71imag}
\end{figure}

\subsection{Numerical $q$-series decomposition on the other side}
\label{sec:numerical-other-side}

\subsubsection{Resurgence and Preservation of Relations}
\label{sec:preservation}

Having verified that the {\it unary} $q$-series and $\tilde{q}$-series may be decoded numerically from the original Borel integral (\ref{eq:borel}), analytically continued from $\hbar<0$ to $\hbar>0$, we now turn to the more difficult problem of numerically extracting $q$-series and $\tilde{q}$-series on the original side, where $\hbar<0$. Here the Borel integral is real, and the fact that the small and large $\hbar$ expansions are both divergent implies that it is not possible to expand the real part purely as a $q$-series or a $\tilde{q}$-series; {\it it must contain both}. The challenge is to uniquely disentangle these $q$-series and $\tilde{q}$-series parts. Naive attempts to do this immediately face ambiguities. However, by invoking resurgence, the analytic continuation of the Borel-\'Ecalle resummed expressions preserves the structure of the non-trivial decomposition of the integral into $q$-series and $\tilde{q}$-series parts. {\it We can therefore use this fact to write down the form of the decomposition and then use numerical fitting to deduce the actual coefficients of the $q$-series and $\tilde{q}$-series.}

Given the fact that Borel-\'Ecalle summation preserves relations under continuation, the basic strategy is simple: we fit the form of the expansion of the Borel integrals (and dual Borel integrals) in various parametric regimes. In practice, it is important that we can use Fourier-Poisson transformations to go back and forth between the large and small $|\hbar|$ regimes, on both sides of the $|q|=1$ boundary. 

To be very explicit, we consider the Mordell-Borel class of Borel integrals (\ref{eq:js}). Based on the information deduced near the Stokes line, $\hbar\in \mathbb R^+$, we can write symbolically the unique decomposition of the Borel integral into $q$-series and $\tilde{q}$-series as informed by \eqref{eq:Stokes4}--\eqref{longeqZSmS}:
\begin{eqnarray}
\sqrt{\frac{4\, p\, (-\hbar)}{\pi}} J_{(p,a)}\left(\hbar\right)
&=& 
q^{\Delta_a}
F_a(q)
+i \sqrt{\frac{\pi}{\hbar}}\sum_{b=1}^{{\rm Floor}\left[\frac{p}{2}\right]}
S_{ab}\, \tilde{q}^{\tilde{\Delta}_b}\, W_{b}(\tilde{q})
\label{eq:paq-sym1}
\end{eqnarray}
Here $S_{ab}$ is the S-matrix (\ref{eq:mixing}), and the exponents ${\Delta_a}$ and $\tilde{\Delta}_b$ are known from the decompositions on the Stokes line. Any integer part of the exponents can be absorbed into $F_a(q)$ and $W_b(\tilde{q})$, which are $q$-series and $\tilde{q}$-series that need to be determined. Note that (\ref{eq:paq-sym1}) has both a real and imaginary part on the unary side ($\hbar>0$), but is purely real when $\hbar <0$.

Continuation to the other side relies on: (i) resurgence; and (ii) unique continuation, as follows. First note that $J(\hbar)$ is analytic in the upper half plane, and that \eqref{eq:paq-sym1} shows that it decomposes uniquely into a pair of functions of  $q$ and $\tilde q$, each of them ramified-analytic for small argument. Eq. \eqref{eq:paq-sym1} also shows that $q^{\Delta_a} F_a(q)$ is resurgent at $q=1$, with $J$ the Borel sum of its asymptotic series. A similar statement holds for $W_b$ with Borel kernel equal to the  Fourier transform of the Borel kernel of $J$. By preservation of relations -- a property characteristic of analytic functions, and more generally of resurgent functions \cite{Ec81}, we expect the same type of unique decomposition to hold for $\hbar<0$ as well, thereby providing a (similarly unique) continuation for $F$ and $W$.  We mathematically prove that this is indeed the case for order 3 mock theta functions (see Section \ref{S-uniq-omega}), and provide strong numerical evidence for a wider class of Mordell integrals.  Ultimately this approach relies on the fundamental fact that all the information about the transseries is encoded in the behavior at the Stokes line, where we know the unary series decomposition. We furthermore demonstrate by a number of examples that the unique continuation can be achieved {\it effectively} by numerical methods. We defer further mathematical details for more general cases to another paper. 
 
Given this perspective, we can isolate different parts of the relation (\ref{eq:paq-sym1}) by taking different limits, and this can be used to determine information about the $q$-series and $\tilde{q}$-series expansions.
This is numerically possible because the Borel integral $J_{(p,a)}\left(\hbar\right)$ can be evaluated to high precision as $\hbar\to -\infty$ (i.e. $q\to 0^+$ and $\tilde{q}\to 1^-$) and also as $\hbar\to 0^-$ (i.e., $q\to 1^-$ and $\tilde{q}\to 0^+$). This is because the dual Borel integral can be evaluated using the Fourier transform of the Borel transform. 

We first describe a very simple numerical method which is able to determine the first few coefficients of these $q$-series and $\tilde{q}$-series expansions, under the assumption that they are integers. We introduce this numerical procedure by some examples.
In Section \ref{sec:precision} we describe more sophisticated numerical methods which demonstrate to very high precision that these coefficients are indeed integers, and which can generate many more coefficients. These numerical methods also provide an alternative independent approach to a proof of uniqueness and of the fact that the coefficients are integer-valued.

\subsubsection{Numerical Example: Mock Order 3 on the Other Side}

Consider our first example, $(p, a)=(3,1)$, corresponding to one of the original examples of Ramanujan \cite{Wat,GM12}, for which the Borel and dual Borel integrals are in (\ref{eq:js31}) and (\ref{eq:jsd31}). By the principle of preservation of relations and the uniqueness of continuation of resurgent transseries, this Borel integral must have a unique decomposition of the form \eqref{eq:Stokes4}--\eqref{longeqZSmS}:
\begin{eqnarray}
    \sqrt{\frac{-12\hbar}{\pi}}\, J_{(3,1)}(\hbar)=q^\Delta\, F(q) -\sqrt{\frac{\pi}{-\hbar}}\, \tilde{q}^{\tilde{\Delta}}\, W(\tilde{q})
    \label{eq:3decomp}
\end{eqnarray}
where the exponents are uniquely determined (from the Borel integral and its dual) by information near the Stokes line $\hbar\in \mathbb R^+$:
\begin{eqnarray}
      \Delta=-\frac{1}{12} \qquad, \qquad \tilde\Delta=\frac{2}{3}
      \label{eq:3deltas}
  \end{eqnarray}
To probe the expansions of $F(q)$ and $W(\tilde{q})$ for small $q$ and small $\tilde{q}$, we re-express the fundamental relation (\ref{eq:3decomp}) in two different ways: 
\begin{eqnarray}
    F(q)&=&q^{\frac{1}{12}}\left(\sqrt{\frac{-12\hbar}{\pi}}\, J_{(3,1)}(\hbar)+\sqrt{\frac{\pi}{-\hbar}}\, \tilde{q}^{\frac{2}{3}}\, W(\tilde{q})\right)
    \label{eq:3decompf}
    \\
    W(\tilde{q})&=& \tilde{q}^{-\frac{2}{3}}\left(\frac{\hbar}{\pi}\sqrt{12}\, J_{(3,1)}(\hbar)+\sqrt{\frac{-\hbar}{\pi}}\, q^{-\frac{1}{12}}\, F(q)\right)
    \label{eq:3decompw}
\end{eqnarray}
In the limit $\hbar\to -\infty$, $q\to 0^+$, it is easy to see numerically from (\ref{eq:3decompf}) that with $W(\tilde{q})\sim 1$, we find $F(q)\sim \frac{1}{2}$. Furthermore, subtracting this leading behavior, one observes that the further corrections to $F(q)$ are exponentially small, in powers of $q^2$. With this simple procedure, using just the leading behavior $W(\tilde{q})\sim 1$, the first 4 coefficients of the small $q^2$ expansion of $2F(q)$ appear to be $1$, $1$, $-2$ and $+3$, assuming they are integers:
\begin{eqnarray}
 2 F(q)\approx 1+q^2-2q^4+3q^6+\dots 
 \label{eq:3decompf2}
\end{eqnarray}
See the dashed curves in Figure \ref{fig:31sub1}. Indeed, these coincide with the first 4 terms of the small $q$ expansion of the order 3 mock theta function $f(q^2)$ in (\ref{eq:fq-small}).

This fitting procedure can be improved by considering also the small $\tilde{q}$ expansion of $W(\tilde{q})$ in (\ref{eq:3decompw}). This is the $\hbar\to 0^-$ limit. In this limit we find (as a consistency check) that the leading behavior of $F(q)$ is indeed $\frac{1}{2}$. Furthermore, this leading term is sufficient to deduce that the first 4 coefficients of $W(\tilde{q})$ appear to be $1$, $2$, $3$ and $4$, assuming they are integers:
\begin{eqnarray}
 W(\tilde{q})\approx 1+2\tilde{q}+3\tilde{q}^2+4\tilde{q}^3 +\dots 
 \label{eq:3decompw2}
\end{eqnarray}See the dashed curves in Figure \ref{fig:31sub2}. Indeed, these coincide with the first 4 terms of the small $q$ expansion of the order 3 mock theta function $\omega(q)$ in (\ref{eq:wq-small}).
  \begin{figure}[htb]
\centerline{\includegraphics[scale=.7]{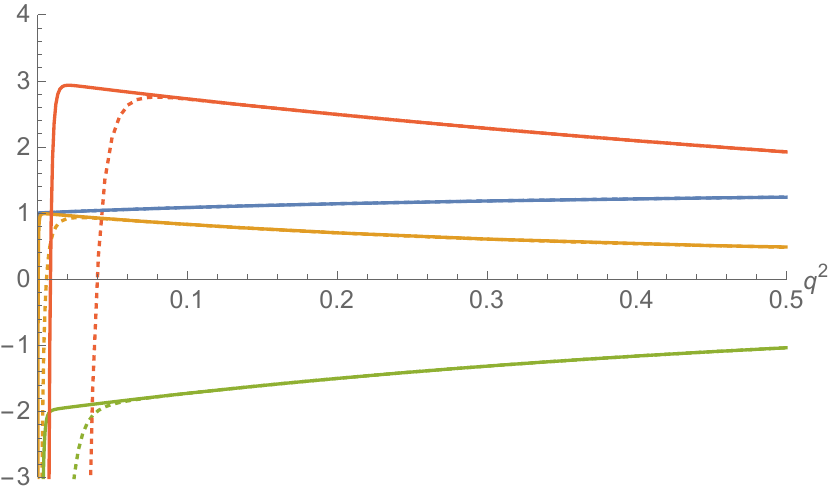}}
\caption{For $(p, a)=(3, 1)$, the coefficients of the small $q^2$ (i.e. $\hbar\to -\infty$) expansion of $2F(q)$ in (\ref{eq:3decompf}) [as shown in (\ref{eq:3decompf2})] are indicated by the blue, orange, green and red  curves, respectively. The  dashed curves use just the leading  $\tilde{q}$ behavior of $W(\tilde{q})$ in the RHS of  (\ref{eq:3decompf}), while the solid curves include two further correction terms deduced from analysis of (\ref{eq:3decompw}), as shown in the accompanying Figure \ref{fig:31sub2}. This simple fitting procedure can be iterated to obtain more coefficients and with higher precision.}
\label{fig:31sub1}
\end{figure}
\begin{figure}[htb]
\centerline{\includegraphics[scale=.7]{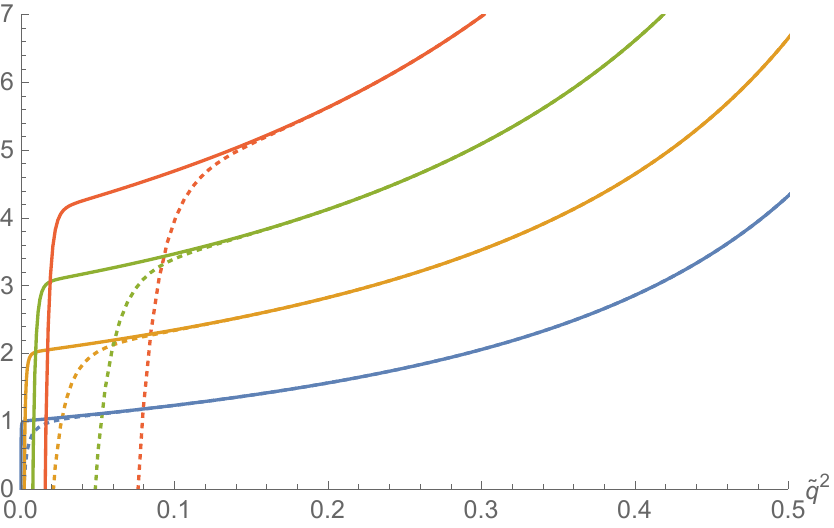}}
\caption{For $(p, a)=(3, 1)$, the coefficients of the small $\tilde{q}$ (i.e. $\hbar\to 0^-$) expansion of $W(\tilde{q})$ in (\ref{eq:3decompw}) are indicated by the blue, orange, green and red  curves, respectively. The  dashed curves use just the leading $q$ behavior of $F(q)$ in the RHS of  (\ref{eq:3decompw}), while the solid curves include two further correction terms deduced from analysis of (\ref{eq:3decompf}), as shown in the accompanying Figure \ref{fig:31sub1}. This simple fitting procedure can be iterated to obtain more coefficients and with higher precision.}
\label{fig:31sub2}
\end{figure}

Now we can {\it combine} these numerical results, inserting these subleading corrections for $W(\tilde{q})$ back into the RHS of  (\ref{eq:3decompf}), thereby obtaining significantly higher precision for the coefficients of $2F(q)$ -- see the solid curves in Figure \ref{fig:31sub1}. Similarly, inserting the subleading corrections for $F(q)$ back into the RHS of (\ref{eq:3decompw}), we obtain greater precision for the expansion coefficients of $W(\tilde{q})$ -- see the solid curves in Figure \ref{fig:31sub2}. We note that if an incorrect integer fit is made at a given order, then the next order fitting fails dramatically. 

Using these expansions, valid in the $0<q<1$ region, we can combine them with the $q>1$ results on the unary side, to make a combined plot for the $q$ dependence of $F(q)$ in (\ref{eq:3decompf}). See Figure \ref{fig:31f}.
\begin{figure}[htb]
\centerline{\includegraphics[scale=.7]{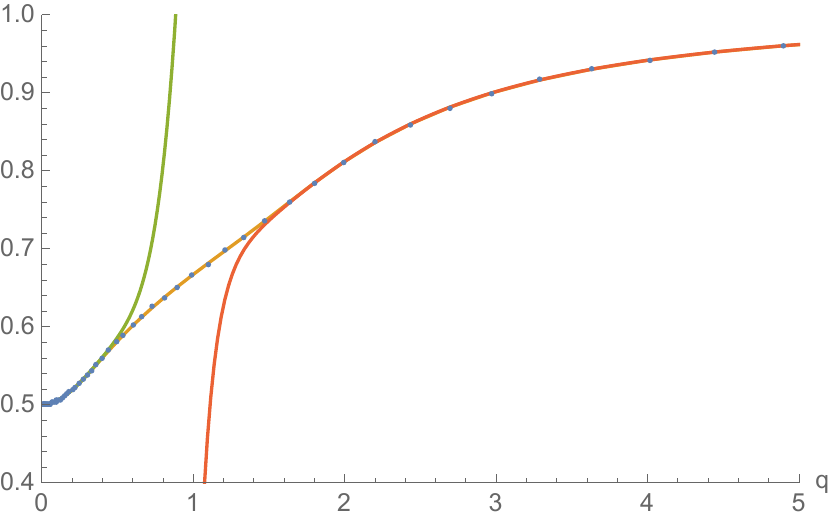}}
\caption{For $(p, a)=(3, 1)$, the full $q$ dependence of the RHS of (\ref{eq:3decompf}). The blue dots give the numerical values, and the orange curve gives the exact result using the known order 3 mock-modular relation \cite{GM12}. The green curve shows the first 4 terms of the small $q$ expansion of $F(q)$ derived by numerical fitting in the $0<q<1$ region, while the red curve shows the first 4 terms of the large $q$ expansion of $F(q)$ derived by numerical fitting in the $q>1$ region.}
\label{fig:31f}
\end{figure}
This procedure can be bootstrapped into a combined algebraic fitting method, {\it simultaneously} solving for the expansions of both $F(q)$ and $W(\tilde{q})$. This can be shown to be convergent and it reveals that the coefficients are integers to extremely high precision. See also Section \ref{sec:precision} below. This perspective provides an alternative and complementary approach to proving uniqueness, as will be described in detail in a future paper.

In summary, we have shown that the Borel integral $J_{(3,1)}(\hbar)$ contains all the information, and provides a rigorous, unique and numerically accessible continuation across the natural boundary. 
\begin{enumerate}
    \item The Borel integral and its dual, obtained by Fourier-Poisson transformation, give the unique $q$-series and $\tilde{q}$-series decomposition on the unary side, at the Stokes line, using straightforward residue analysis.
    \item Invoking the preservation of relations for Borel-\'Ecalle transseries, we obtain a precise skeleton structure (\ref{eq:paq-sym1}) for the $q$-series and $\tilde{q}$-series decomposition on the non-unary side. 
    \item Numerical asymptotics of the Borel integral and its dual can be used to determine the integer-valued coefficients of the resulting $q$-series and $\tilde{q}$-series.
\end{enumerate}

\subsubsection{Numerical Example: Mock Order 10 on the Other Side}

A similar analysis works for higher order mock theta functions. The main complication is an increase in the algebraic complexity of the fitting procedure. For example, for $(p, a)=(5,1)$, associated with order 10 mock theta functions, the structure on the unary side, derived from the behavior at the Stokes line where $\hbar>0$, suggests an expansion of the form
\begin{eqnarray}
    \sqrt{\frac{-20\hbar}{\pi}}\, J_{(5,1)}(\hbar)=q^{-\frac{1}{20}}\, F(q) -\sqrt{\frac{\pi}{-\hbar}}\frac{2}{\sqrt{5}}\left[
    \sin\left(\frac{\pi}{5}\right)\tilde{q}^{-\frac{1}{5}}\,  W_1(\tilde{q})
    +\sin\left(\frac{2\pi}{5}\right)\tilde{q}^{-\frac{4}{5}}\, W_2(\tilde{q})\right]
    \nonumber\\
    \label{eq:51decomp}
\end{eqnarray}
where the exponents are inherited from the unary side.
There are now 3 series to be determined: one $q$-series, $F(q)$, and two $\tilde{q}$-series, which we denote as $W_1(\tilde{q})$ and $W_2(\tilde{q})$. 
As above, we can probe the small $q$ and small $\tilde{q}$ expansions by considering both the $\hbar\to -\infty$ and $\hbar\to 0^-$ limits. For example, using just the leading behavior of $W_1(\tilde{q})$ and $W_2(\tilde{q})$, we deduce that at small $q$, $F(q)$ has an expansion in powers of $q^2$, with integer coefficients:
\begin{figure}[htb]
\centerline{\includegraphics[scale=.7]{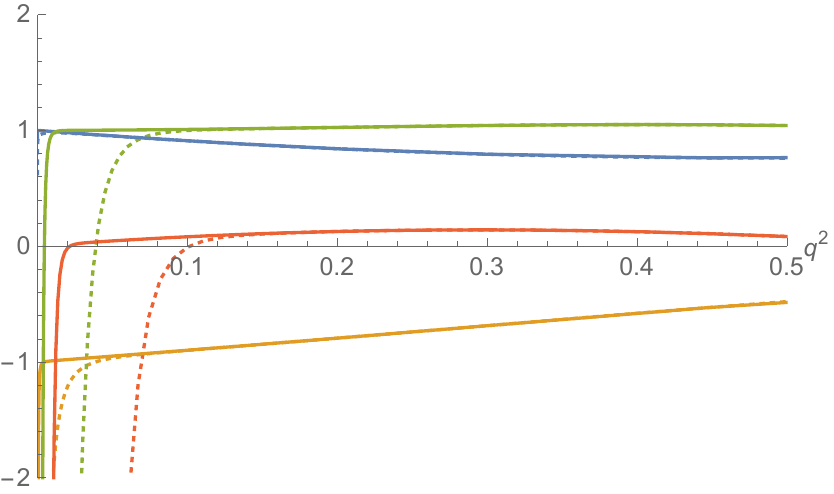}}
\caption{For $(p, a)=(5, 1)$, the coefficients of the small $q^2$ (i.e. $\hbar\to -\infty$) expansion of $F(q)$ in (\ref{eq:51decomp}) are indicated by the blue, orange, green and red curves, respectively. The  dashed curves use just the leading  $\tilde{q}$ behavior of $W_1(\tilde{q})$ and $W_2(\tilde{q})$ in (\ref{eq:w1w21})-(\ref{eq:w1w22}), while the solid curves include two further correction terms deduced from analysis of the $\hbar\to 0^-$ limit. This simple fitting procedure can be iterated to obtain more coefficients and with higher precision.}
\label{fig:51sub}
\end{figure}
\begin{eqnarray}
    F(q)= 1 - q^2 + q^4 +0\cdot q^6 +O(q^8)
\end{eqnarray}
See Figure \ref{fig:51sub}. Note that the coefficient of $q^6$ is zero. 
Similarly, the leading integer-valued coefficient expansions of $W_1(\tilde{q})$ and $W_2(\tilde{q})$ give:
\begin{eqnarray}
    W_1(\tilde{q})&=& \tilde{q}+\tilde{q}^2+2\tilde{q}^3+\dots
    \label{eq:w1w21} \\
    W_2(\tilde{q})&=& \tilde{q}+2\tilde{q}^2+2\tilde{q}^3+\dots
    \label{eq:w1w22}
\end{eqnarray}
These results are in agreement with the expansions (\ref{eq:x10})-(\ref{eq:phi10}) of the order 10 mock theta functions $X(q)$, $\psi(q)$ and $\phi(q)$ entering the known mock modular relation (\ref{eq:511}) for $0<q<1$ \cite{GM12}, identifying $F(q)=X(q^2)$, $W_1(\tilde{q})=\psi(\tilde{q})$, and $W_2(\tilde{q})=\tilde{q}\, \phi(\tilde{q})$:
\begin{equation}
\sqrt{\frac{-20\hbar}{\pi}} J_{(5,1)}\left(\hbar\right) = q^{-\frac{1}{20}}\,X(q^2)
- \sqrt{\frac{\pi}{-\hbar}}\, \frac{2}{\sqrt{5}} \left[\sin\left(\frac{\pi}{5}\right) \tilde{q}^{-\frac{1}{5}} \psi\left( \tilde{q}\right)+
\sin\left(\frac{2\pi}{5}\right) \tilde{q}^{\frac{1}{5}} \phi\left(\tilde{q}\right)\right]
\label{eq:51}
\end{equation}
Furthermore, using these numerically extracted expansions, valid in the $0<q<1$ region, we can combine them with the $q>1$ results on the unary side, to make a combined plot for the $q$ dependence of $F(q)$ in (\ref{eq:51decomp}). See Figure \ref{fig:51f}. 
\begin{figure}[htb]
\centerline{\includegraphics[scale=.6]{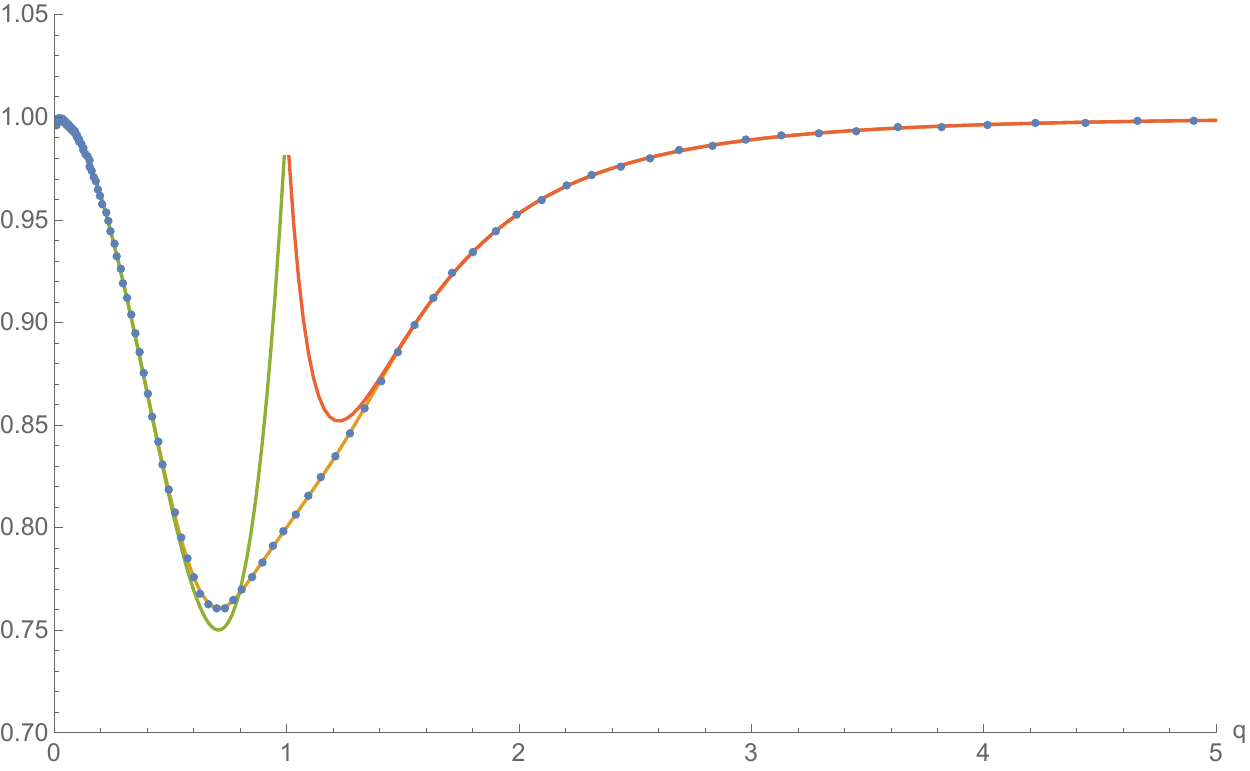}}
\caption{For $(p, a)=(5,1)$, the full $q$ dependence of $F(q)$ extracted from (\ref{eq:51decomp}). The blue dots give the numerical values from the RHS of (\ref{eq:51decomp}), and the orange curve gives the exact result using the known order 10 mock-modular relation \cite{GM12}. The green curve shows the first 3 terms of the small $q$ expansion of $F(q)$ derived by numerical fitting in the $0<q<1$ region, while the red curve shows the first 3 terms of the large $q$ expansion of $F(q)$ derived by numerical fitting in the $q>1$ region.}
\label{fig:51f}
\end{figure}

\subsubsection{Numerical Example: Beyond Mock on the Other Side}

As a first example beyond the mock theta functions, we consider $(p, a)=(7,1)$. Here there is no known mock-modular relation with which to compare. However, the main  difference from the previous cases is simply an increase in algebraic complexity of the fitting procedure.
We choose as our input the Borel integral (\ref{eq:js71}), and its dual Borel integral representation (\ref{eq:jsd71}), initially defined for $\hbar <0$.
The structure on the unary side, derived from the behavior at the Stokes line where $\hbar>0$, suggests an expansion of the form, {\it cf.} \eqref{eq:Stokes4}--\eqref{longeqZSmS}:
\begin{eqnarray}
    \sqrt{\frac{-28\hbar}{\pi}}\, J_{(7,1)}(\hbar)&=&q^{-\frac{1}{28}}\, F(q)
    \nonumber\\
    &&\hskip -4cm 
    -\sqrt{\frac{\pi}{-\hbar}}\frac{2}{\sqrt{7}}\left[
    \sin\left(\frac{\pi}{7}\right)\tilde{q}^{-\frac{1}{7}}\,  W_1(\tilde{q})
    +\sin\left(\frac{2\pi}{7}\right)\tilde{q}^{-\frac{4}{7}}\, W_2(\tilde{q})
    +\sin\left(\frac{3\pi}{7}\right)\tilde{q}^{-\frac{9}{7}}\, W_3(\tilde{q})\right]
    \label{eq:71decomp}
\end{eqnarray}
where the exponents are inherited from the unary side.
There are now 4 series to be determined: one $q$-series, $F(q)$, and three $\tilde{q}$-series, which we denote as $W_1(\tilde{q})$, $W_2(\tilde{q})$, and $W_3(\tilde{q})$. As above, we can probe the small $q$ and small $\tilde{q}$ expansions by considering both the $\hbar\to -\infty$ and $\hbar\to 0^-$ limits of the Borel integral and its dual. 
For example, numerical asymptotics of the small $q$ (i.e. $\hbar\to -\infty$) limit shows that the leading behavior is:
\begin{eqnarray}
&&F(q)=
\nonumber\\
&&e^{\frac{\hbar}{28}}\left[\sqrt{\frac{-28\hbar}{\pi}} J_{(7,1)}\left(\hbar\right)
+ \sqrt{\frac{\pi}{-\hbar}}\,  \frac{2}{\sqrt{7}}
\left(
\sin\left(\frac{\pi}{7}\right) e^{\frac{19\pi^2}{14\hbar}}
+\sin\left(\frac{2\pi}{7}\right) e^{\frac{13\pi^2}{14\hbar}} 
+
\sin\left(\frac{3\pi}{7}\right) e^{\frac{3\pi^2}{14\hbar}} \right)\right] 
\nonumber\\
&\sim & 1-e^{2\hbar}+2e^{4\hbar}-2e^{6\hbar}
+\dots 
 \label{eq:71sub}
\end{eqnarray}
This is the ``dual'' of a false theta, that in our other notation we write as
\begin{equation}
\tilde \Psi^{(1)}_7 (q)^{\vee} = q^{-\frac{1}{28}} \left( 1 - q^2 + 2 q^4 - 2 q^6 + \ldots \right)
\end{equation}
This provides an explicit answer to the Question~\ref{quest:dualtheta}, beyond the previously known cases.
The first four coefficients, +1, -1, +2, -2, are shown as the dotted curves in Figure \ref{fig:71sub}. Note that the $\tilde{q}$ exponents in (\ref{eq:71sub}) differ from those on the unary side by a common shift: $e^{\frac{19\pi^2}{14\hbar}}=e^{\frac{3\pi^2}{2\hbar}}\, e^{-\frac{\pi^2}{7\hbar}}$, $e^{\frac{13\pi^2}{14\hbar}}=e^{\frac{3\pi^2}{2\hbar}}\, e^{-\frac{4\pi^2}{7\hbar}}$, and $e^{\frac{3\pi^2}{14\hbar}}=e^{\frac{3\pi^2}{2\hbar}}\, e^{-\frac{9\pi^2}{7\hbar}}$.
\begin{figure}[htb]
\centerline{\includegraphics[scale=.7]{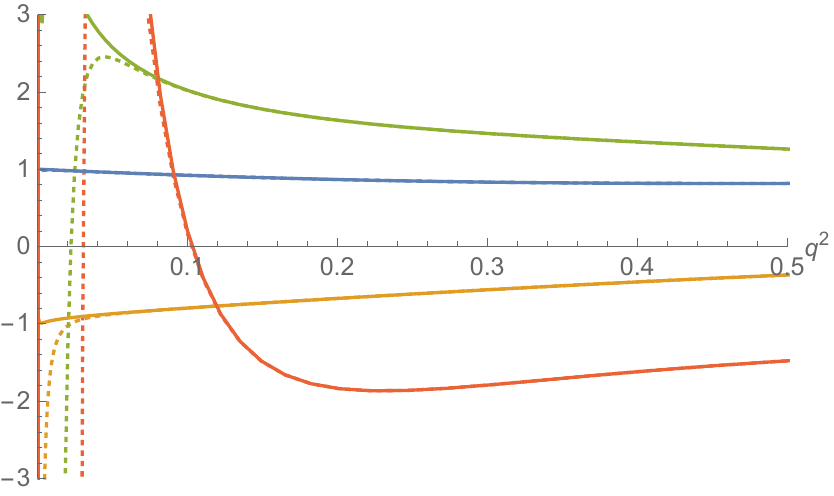}}
\caption{For $(p, a)=(7, 1)$, the coefficients of the small $q$ expansion of $F(q)$, as an expansion in $q^2$, obtained by fitting to the $\hbar\to -\infty$ limit in the decomposition (\ref{eq:71decomp}), are indicated by the blue, orange, green and red  curves, respectively. The  dashed curves use just the leading  $\tilde{q}$ behavior of $W_j(\tilde{q})$ as shown in  (\ref{eq:71sub}), while the solid curves include one further correction term deduced from analysis of the $\hbar\to 0^-$ limit. This simple fitting procedure can be iterated to obtain more coefficients and with higher precision.}
\label{fig:71sub}
\end{figure}
The small $q$ expansion of $F(q)$ suggests that it is an expansion in $q^2$, with the first four integer-valued coefficients being +1, -1, +2, -2. Further fitting of the decomposition (\ref{eq:71decomp})-(\ref{eq:71sub}) suggests that the next corrections to $W_j(\tilde{q})$ are:
\begin{eqnarray}
    W_1(\tilde{q})&\sim & \tilde{q}^{\frac{3}{2}}\left(1-\tilde{q}+\dots\right)\\
    W_2(\tilde{q})&\sim & \tilde{q}^{\frac{3}{2}}\left(1+3\tilde{q}+\dots\right)\\
    W_3(\tilde{q})&\sim & \tilde{q}^{\frac{3}{2}}\left(1+\tilde{q}+\dots\right)
    \label{eq:w123}
\end{eqnarray}
Including these higher corrections $W_j(\tilde{q})$ leads to more precise fitting of the first coefficients of $F(q)$, as shown as solid curves in Figure \ref{fig:71sub}.
Further correction terms can be obtained by a more systematic unified algebraic fitting procedure, along the lines of the discussion in Section \ref{sec:precision} below, as will be described in detail in a future paper. 

Furthermore, using these numerically extracted expansions, valid in the $0<q<1$ region, we can combine them with the $q>1$ results on the unary side, to make a combined plot for the $q$ dependence of $F(q)$ in (\ref{eq:71decomp}). See Figure \ref{fig:71f}. Note that unlike the analogous examples plotted in Figures \ref{fig:31f} and \ref{fig:51f}, in the case of $(p, a)=(7, 1)$ we do not have an exact mock-modular relation to plot, so there is no orange curve in Figure \ref{fig:71f}. However, the numerical values (blue dots) match smoothly across the boundary, and also match the asymptotic small and large $q$ behavior inidcated by the green and red curves.
\begin{figure}[htb]
\centerline{\includegraphics[scale=.7]{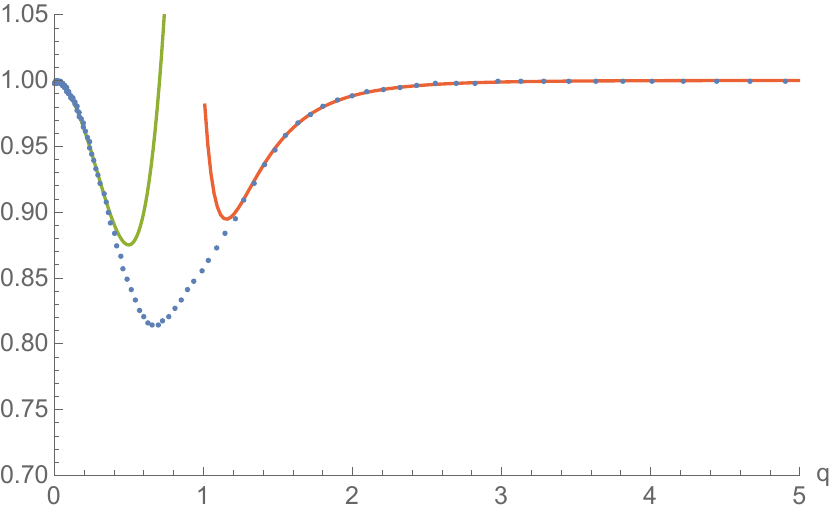}}
\caption{For $(p, a)=(7, 1)$, the full $q$ dependence of $F(q)$ extracted numerically from (\ref{eq:71decomp}). The blue dots give the numerical values. The green curve shows the first 3 terms of the small $q$ expansion of $F(q)$ derived by numerical fitting in the $0<q<1$ region, while the red curve shows the first 3 terms of the large $q$ expansion of $F(q)$ derived by numerical fitting in the $q>1$ region.}
\label{fig:71f}
\end{figure}

{\bf Comment:} we re-iterate that this numerical procedure is extremely simple:
\begin{enumerate}
    \item The Borel integral, when taken to its Stokes line, has a unique splitting into a $q$-series and a linear combination of $\tilde{q}$-series. These are easily obtained from a residue analysis of the Borel integral and its dual (which in turn is obtained from the Fourier transform of the original Borel transform).
    \item We then rotate back again and seek a decomposition of the Borel integral into a $q$-series and a linear combination of $\tilde{q}$-series which has the same overall structure, and for which the actual series coefficients can be deduced by numerical fitting under the assumption that they are integers.
\end{enumerate}
It is somewhat surprising that such a simple and crude numerical fitting procedure is able to `discover' known mock modular relations, and also to go beyond the known cases. In Section \ref{sec:precision} we introduce a more sophisticated numerical method which can achieve much higher precision and also find many more coefficients. A full description of this extended method is deferred to a future publication. But first we give an explicit proof of uniqueness of the procedure, as this also helps to motivate the more precise numerical methods.

\subsection{Uniqueness proof}
\label{S-uniq-omega}

In this Section we present a proof that the resurgent analytic continuation procedure is unique. The proof is given here for the order 3 mock theta function case, but we conjecture that it can be extended to more general Borel integrals, even beyond the mock theta function class of problems.

By Section \ref{sec:sl2z}, any one of the Borel $J$ functions in a Mock Theta class generates all the others by repeated applications of the $SL(2, \mathbb Z)$ maps $q\mapsto \tilde{q}$ and $q\to -q$ (see also \cite{GM12}). In other words, these functions are merely values of one of them in various limits. Hence, to show uniqueness of the $q,\tilde{q}$ series decompositions it suffices to do so for one integral in each given class. 

For example, the  transformation $q\mapsto -q$ maps the integral $J_{3,1}=J(B^{(s)}_{3,1})$ into the function $2 J^{(c)}_{3,1}$; it effectively turns the sinh-like Borel functions into cosh-like Borel functions.\footnote{It is simpler to consider the action of $SL(2, \mathbb Z)$ on the Borel transforms than on the $q$-series themselves.}
For the sake of comparison with the mock-modular literature, for example as in the review paper \cite{GM12}, in this subsection we adopt the common notation:
\begin{eqnarray}
    q=e^{-\alpha} \qquad; \qquad \tilde{q}=e^{-\frac{\pi^2}{\alpha}}
    \label{eq:alpha}
\end{eqnarray}
Thus, $\hbar$ is replaced by $-\alpha$. Then the unary relations (\ref{eq:pareal}) and (\ref{eq:paimag}) map to the known \cite{Wat,GM12} mock modular relations for $\alpha>0$:
 \begin{equation}
    \label{eq:eqt1}
    q^{2/3}\omega(-q)=-\sqrt{\frac{\pi}{\alpha}}\tilde{q}^{2/3}\omega(-\tilde{q})+\sqrt{\frac{12\alpha}{\pi}} W_3(\alpha)
  \end{equation}
  and
  \begin{equation}
    \label{eq:eqt2}
    q^{2/3}\omega(q)=\sqrt{\frac{\pi}{4\alpha}}\tilde{q}^{-1/12}f(\tilde{q}^2)-\sqrt{\frac{3\alpha}{\pi}} W_2(\alpha/2)
  \end{equation}
  Here $f$ and $\omega$ are the standard order 3 mock theta functions, and $W_3(\alpha):=J_{3,1}(-\alpha)$, and $W_2(\alpha/2):=2 J^{(c)}_{3,1}(-\alpha)$ \cite{Wat,GM12}.
 \begin{mythm}[Uniqueness of $\omega$]\label{T:uniq-omega}
 There is a unique $\omega$ satisfying \eqref{eq:eqt1}  such that $q^{-2/3}\omega(q)$ is analytic in the unit disk $\mathbb{D}$ and $\omega(q)\tilde{q}^{1/12}$ is bounded as $q\to 1$.
\end{mythm}
The proof is given below. 
  
{\bf Comments:}
\begin{enumerate}
    \item 
    The bounds coming from the behavior of $W_3$ under $q\mapsto-q$ are necessary to ensure uniqueness. By itself, the identity  \eqref{eq:eqt1}  is insufficient. Indeed, the function $q^{2/3}{\tilde{\omega}}= q^{2/3}\omega-(2^{-5/3})[\lambda(1-\lambda)]^{2/3} \vartheta _3(q) \left(\frac{1}{2}-\lambda (q)\right)$, where $\lambda=\lambda(\tau)$ is the elliptic modular lambda functions, also satisfies \eqref{eq:eqt1} as can be shown by a straightforward calculation.

    \item 
  By this uniqueness theorem, there is one and only one choice of $\omega$ for $|q|>1$ ensuring preservation of properties with respect to the integral $W_2$ and the action of $SL(2,\mathbb Z)$. 
    
    \item 
    This proof of uniqueness for the function $\omega$ can be straightforwardly adapted to give a proof of uniqueness of the function $f$, which satisfies 
    \begin{eqnarray}
        q^{-1/24}f(-q)=-\sqrt{\frac{\pi}{\alpha}}\, \tilde{q}^{-1/24}\, f(-\tilde{q})+\sqrt{\frac{24\alpha}{\pi}}\, W(\alpha)
    \end{eqnarray}
    The function $W(\alpha)$ is a linear combination of two cosh-like Borel integrals, defined in \cite{GM12}. This is another way to see that it is sufficient to prove uniqueness for just one of the Borel integrals.

\end{enumerate}

  \begin{proof}
    
Let $\omega_1$ be another function that satisfies the assumptions of the theorem and $\omega_f=\omega_1-\omega$. Then $\omega_f$ satisfies the same assumptions, except that it solves
\begin{equation}
    \label{eq:eqtf}
     q^{2/3}\omega_f(-q)=-\sqrt{\frac{\pi}{\alpha}}\tilde{q}^{2/3}\omega_f(-\tilde{q})
   \end{equation}
We let $w(q)=\omega_f(-q)$ and define $h= (q^{2/3}w(q))^6/\vartheta_3^6(q)$. This has the effect of removing the rational exponent. Then $h(\tau)$, where $q=e^{i\pi \tau}$, satisfies 
  \begin{equation}
  \label{eq:rget3}
  h(\tau +2)=h(\tau );\ h(-1/\tau )=h(\tau )
\end{equation}
Let  $\lambda=\lambda(\tau)$ be the elliptic modular function defined  in the upper half plane, and  $\lambda_q$ the inverse nome function; we have $\lambda(\tau)=\lambda_q(e^{i\pi \tau})$. To avoid complicating the notation we will omit the subscript $q$  whenever it's clear what we mean. Since $h(-1/\tau )=h(\tau )$,  $\lambda(\tau)=1-\lambda(-1/\tau)$  and  $\lambda$ is a Hauptmodul for $\Gamma(2)$ we have\footnote{ A direct argument using the conformal properties of $\lambda$ [Ahlfors, pp, 279--282] is the following. In the fundamental domain $\Omega$ of $\lambda$, [Ahlfors, Fig 7-3]   $\lambda$ takes all values in the UHP exactly once, while in $\Omega'$, the reflection of $\Omega$ across $i\mathbb{R}$, it takes all values in the LHP exactly once.  Using these facts and analyticity of $h$ in $\Omega\cup\Omega'=\Omega_t$,  it follows that  $h=G(\lambda)$ in $\Omega_t$. On the left boundary $i\mathbb{R}$,  $\lambda$ is purely real and one-to-one onto $(0,1)$.   Furthermore,  on $i\mathbb{R}$ we have, by the action of the inversion $-1/\tau$ on $\lambda$,  $\lambda(i t)=1-\lambda (i/t)$.  Since  $h(it)=h(i/t)$, we have $G(\lambda)=G(1-\lambda)$ for  $\lambda\in (0,1$.  If $\mu=1/2-\lambda$ we have $F(\mu)=F(-\mu)$ for   $\mu\in (-1/2,1/2)$ hence $F(\mu)=F(-\mu)$ for all $\mu$ where analyticity is preserved. We now continue from $\tau=i$ (where $\mu=0$), through, say, the lower circular boundary of $\Omega_t$ into $-1/\Omega_t=:\Omega_2$, and we have to show that the construction  $F(\mu(s)):=h(s)$ is consistent.   Since $h(s)=h(-1/s)$,  and the biholomorphism $\mu$ satisfies $\mu(s)=-\mu(-1/s)$, analytically continuing along a path $\gamma$ through the right arccircle is equivalent  to performing the continuation along the inverted path $-1/\gamma$ through $\Omega'$, where we have already shown consistency.}
\begin{equation}
  \label{eq:iden0}
   h(q)= q^2w^6(q)/\vartheta_3^6(q)=F(\lambda(q)),
\end{equation}
and $F\circ \lambda$ is analytic in $\mathbb{D}$ since $h$ is. We now examine the function $F$ itself. By the conformal map properties of $\lambda$, $e^{\pi i \tau}\mapsto \lambda$ is a biholomorphism at any point $\tau$ in the upper half plane. Since the only avoided values of $\lambda$ are $0,1$, $F(\zeta)$ is analytic except possibly at zero and $1$.

{\bf Properties of $F$ for $\zeta\in\{0,1\}$.} At $q= 0$, $\lambda_q (q)$ is biholomorphic. It only remains to examine $q\to 1$. We have,  by the properties of $h$,
$$h(q)=F(\lambda(q))=F(1-\lambda(\tilde{q}))=h(\tilde{q})=F(\lambda(\tilde{q}));\ \text{hence}\  F(1-s)=F(s)$$ Thus $F$ is analytic at 1 as well. Since $\lambda(\mathbb{D})=\mathbb{C}\setminus\{1\}$, and $F$ is analytic in $\mathbb{D}$ and at $1$, $F$ is entire.

{\bf End of the proof:} Now $q\mapsto -q$ is equivalent to $\tau\mapsto \tau+1$.  Hence
\begin{equation}
  \label{eq:near-1}
  \lambda(-q)=\frac{\lambda(q)}{\lambda(q)-1}=1-\frac{1}{\lambda(\tilde{q})} 
\end{equation}
Hence, using the series of $\lambda$ at $q=0$ we get
\begin{equation}
  \label{eq:at-1}
  F\left(-\frac{1}{16 \tilde{q}}+\frac{1}{2}+\cdots\right)=(q ^{2/3}w(-q)/\theta_3(-q))^6=O[(\tilde{q}^{-1/12-1/4})^6]=O(\tilde{q}^{-2})
\end{equation}
which, combined with the easily checked fact that $\tilde{q}^{-1}$ covers a neighborhood of the $\infty$ on the Riemann sphere as $q\to -1$,  means that $F(x)=O(x^2)$ for large $x$, implying $F$ is a quadratic polynomial. Since $P=\sqrt{F}$ is analytic (as being $q^2 w(q)/\theta_3(q)^3$), $P$  is polynomial of degree 1, $P(x)=Ax+B$. Now, $P(q=0)=0$ implies $B=0$ and thus $P( x)=A x$. But, for small $q$,  $O(q^2)=q^2 w(q)/\theta_3(q)^3=P(q)=A\lambda(q)=16 Aq+O(q^2)$ implies $A=0$. Hence $F=w=0$.
\end{proof}

\subsection{Precise numerical fitting}
\label{sec:precision}

The uniqueness theorem in Section \ref{S-uniq-omega} suggests that the $(q,\tilde{q})$ decomposition is generated by the Borel kernel and should be fully recoverable from it, both theoretically and numerically. Here we discuss one such procedure. We illustrate with the example in Section \ref{S-uniq-omega}, which has a symmetry under $q\leftrightarrow \tilde{q}$, but we stress that this just simplifies the algebraic procedure and is not a fundamental restriction. We can make the $q\leftrightarrow \tilde{q}$ symmetry of the problem explicit by dividing equations (\ref{eq:eqt1}) and (\ref{eq:eqt2}) through by $\theta_3(q)$, using the familiar property of $\theta_3(q)$ under $q\to \tilde{q}$. Then the best numerical matching point between Mordell-Borel integrals and their $(q,\tilde{q})$ decompositions is near $q_0=e^{-\pi}$ where $q=\tilde{q}$ and $q_0$ is small.\footnote{Even without an explicit $q\leftrightarrow \tilde{q}$ symmetry it is still true that a ``central point'' is an effective location to self-consistently match both small and large $q$ expansions.}

{\bf Procedure.} We illustrate the approach  on the function $\omega(-q)$
which satisfies (\ref{eq:eqt1}). The Mordell-Borel $(q, \tilde{q})$  decomposition is calculated by expanding the Mordell-Borel integral and a sum $\sum_{k=0}^nc_k(q^{k+2/3}+\tilde{q}^{k+2/3})$ in a $2n$-order Taylor series at $q_0$. Because of the underlying symmetry of the problem, we only need the even coefficients. The condition that  the two expansions be equal translates into an algebraic equation $\hat{A}_n {\bf c}=\bf m$, where $\mathbf c=(c_1,...,c_n)$, and $\mathbf m=(m_1,...,m_n)$ are the Taylor coefficients of the integral. For  $n=100$, we get the following values for the coefficients $c_k$:

{\scriptsize\begin{equation*}
\begin{matrix}
\text{coeff:}\  c_0 & c_1 & c_2 & c_3 & c_4 & c_5 & c_6 & c_7 & c_8 & c_9 \\
  \hline\\
\text{value:}\ \  1 & -2 & 3 & -4 & 6 & -8 & 10 & -14 & 18 & -22 \\
  \hline\\
\text{error:}\,\,\, 10^{-22} & 10^{-21} & 10^{-19} & 10^{-18} & 10^{-17} & 10^{-17} & 10^{-16} & 10^{-15} & 10^{-14} & 10^{-14} \\
\end{matrix}
\,\,\,\,\begin{matrix}  c_{10} & c_{11} & c_{12}  \\ \hline\\
 29 & -36 & 44 \\ \hline\\
 10^{-13} & 10^{-13} & 10^{-12}  \\\end{matrix}
\end{equation*}}
{\scriptsize\begin{equation*}
\begin{matrix}  c_{13} & c_{14}  \\ \hline\\
-56 & 68 \\ \hline\\
 10^{-12} & 10^{-11}  \\\end{matrix}\,\,\, \begin{matrix}  c_{15} & c_{16} & c_{17} & c_{18} & c_{19} & c_{20} & c_{21} & c_{22} & c_{23} \\\hline\\
 -82 & 101 & -122 & 146 & -176 & 210 & -248 & 296 & -350 \\\hline\\
 10^{-11} & 10^{-10} & 10^{-10} & 10^{-9 }& 10^{-9 }& 10^{-8 }& 10^{-8 }& 10^{-7 }& 10^{-7 }\\
\end{matrix}
\,\,\,\begin{matrix}
   c_{24} & c_{25}  \\\hline\\
 410 & -484  \\\hline\\
 10^{-7 }& 10^{-6 }\\
\end{matrix}\,\,\,\,\begin{matrix}
   c_{26} & c_{27}   \\\hline\\
 566 & -660    \\\hline\\
  10^{-6 }& 10^{-6 }\\
\end{matrix}
\end{equation*}
}
Note that these coefficients are integers to very high precision, and furthermore these integer values coincide with the expansion of $\omega(-q)$: see (\ref{eq:wq-small}).

\begin{figure}
      \centerline{\includegraphics[scale=0.3]{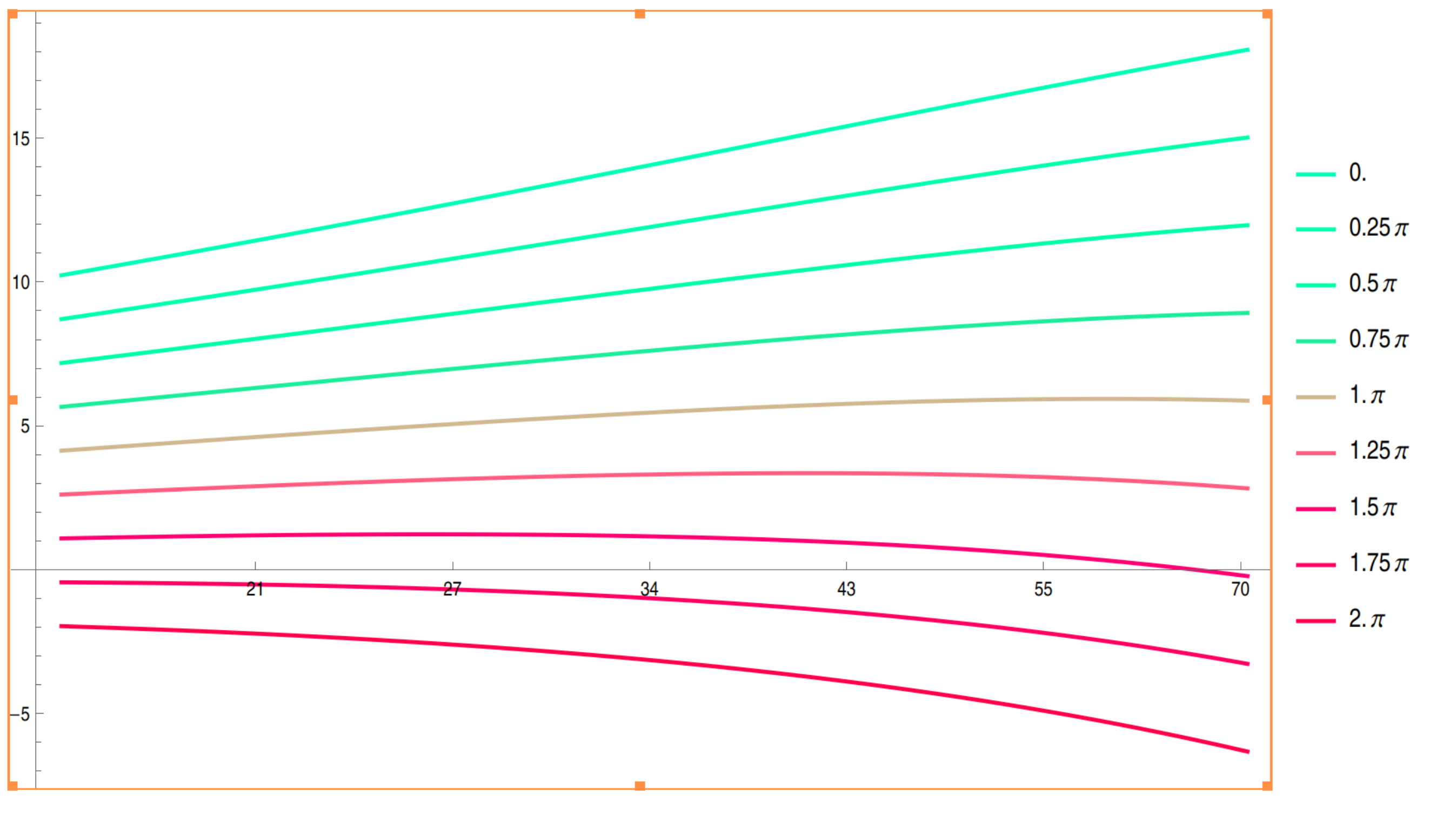}}
      \caption{Accuracy (number of exact digits) of $c_0$ as a function of the order of the system; see discussion in the text about the way the accuracy was determined. The accuracy curves of $c_k,k>0$  are nearly perfectly parallel to that of $c_0$, with an offset which gets worse with the order.}
      \label{Fig-conv}
    \end{figure}

{\bf Checking convergence.} Here we give an empirical argument that  $\hat{A}^{-1}_n \mathbf m$ converges, leaving a theoretical analysis (based on relatively delicate norm estimates) for another paper. We replace $\mathbf m$ by the Taylor coefficients of a general expansion of $\sum_{k=0}^{N}d_k(q^{k+2/3}+\tilde{q}^{k+2/3})$, with $N$ significantly greater than $n$ (we took $N=2n$,  as higher coefficients $d_k,k>2 n$ contribute negligibly to the outcome). We bounded the size of the error, which comes from the ``residual coefficients'', $d_{n+1},...d_{2n},...$, by placing absolute values on all terms in the error expression. This bound should be realistic,  as the matrix controlling the error has non-alternating coefficients. In Figure \ref{Fig-conv} we took for $d_k$ asymptotic approximations that represent the generic behavior of modular forms and modular theta functions: $c_n\propto e^{{\rm const}\, \pi \sqrt{n}}$. For $\omega$, this constant is given by ${\rm  const}=3^{-1/2}\approx 0.577$, while the possible freedom in (\ref{eq:eqt1}) 
    has  ${\rm const} \ge 2$. 
    All these calculations were done numerically, see Figure \ref{Fig-conv}, but the procedure is amenable to rigorous estimates.
    This means that the procedure converges for our Mordell integral but, in general, convergence depends on the rate of growth of the coefficients of the right hand side. 
This can be estimated in terms of exponential growth, for the following reason:

\noindent{\bf Theorem:}
  Assume $f$ is analytic in $\mathbb D$ and bounded by $e^{a^2/d},d=|1-z|$ for $z\in \mathbb D$. Then,
 \begin{equation}
  \label{eq:coefcn}
  |c_n|\le \frac{a\, C}{\sqrt{n}}\,  e^{a^2/2}\, e^{2 a\sqrt{n}}\, {\rm const}
\end{equation}

Conversely, assume that the Maclaurin coefficients of $f$ are bounded by $n^a e^{b\sqrt{n}}$, and for simplicity take $f(0)=0$. Then $f$ is analytic in $\mathbb D$ and bounded by
\begin{equation}
  \label{eq:boundf}
  |f(z)|\le 2^{-2 a} b^{2 a+1}  \sqrt{\pi }|z|^{-2 a-\frac{3}{2}}\,e^{\frac{b^2}{4 |z|}} 
\end{equation}
\begin{proof}
 Using Cauchy's formula to calculate the Maclaurin coefficients of $f$ we get
\begin{equation}
  |c_n|=\left|\frac{1}{2\pi i}\oint_{|z|=1-d} \frac{f(s)ds}{s^{n+1}}\right|\le e^{a^2/d} (1-d)^{-n}\le \inf_{d\in(0,1)}e^{a^2/d} (1-d)^{-n}=e^{a^2/2}e^{2 a\sqrt{n}}{\rm const}
\end{equation}

For the second part we use the integral test to see that
\begin{multline}
  \label{eq:boundf2}
  |f(z)|\le \int_0^{\infty } n^a  e^{b \sqrt{n}-n t} \, dn\\
  =t^{-a-\frac{3}{2}} \left(\sqrt{t} \Gamma (a+1) \,
   _1F_1\left(a+1;\frac{1}{2};\frac{b^2}{4 t}\right)+b \Gamma
   \left(a+\frac{3}{2}\right) \,
   _1F_1\left(a+\frac{3}{2};\frac{3}{2};\frac{b^2}{4
   t}\right)\right)\end{multline}
The rest of the argument follows directly from the asymptotics of the special functions involved. Alternatively, and rather straightforwardly, one can obtain the asymptotics from the integral representation and Laplace's method. 
\end{proof}

If $a\lessapprox 2$ the bound is valid down to $n=1$ with ${\rm const}\sim 1.025$, and ${\rm const}\to 1$ for large $n$.  Similarly, if the bound is $C d^{1/2}e^{a^2/d}$, as it is in \eqref{eq:eqt1} (since $1/\alpha=1/\ln(q)\sim d$ if $q=1-d$) the bound becomes \eqref{eq:coefcn}. The asymptotics of the matrix $\hat{A}_n$ is extremely rich and contains detailed information about modular forms and mock theta functions, which will be described in a future paper.

\subsection{Self-dual decompositions}
\label{sec:selfdual}

The numerical fitting procedure in the previous section was optimized for situations where the Borel integral has a precise duality symmetry under $q\leftrightarrow \tilde{q}$, in which case the numerical fitting is extremely sensitive to an expansion about the self-dual point $q_0=e^{-\pi}\approx 0.043$, where $q=\tilde{q}$ and $q_0$ is small. This duality symmetry essentially halves the algebraic complexity of the fitting problem, which is not essential but can be a numerical advantage. In this Section we point out that this kind of duality symmetry is a feature of the Mordell-Borel integrals $J_{(p,a)}(\hbar)$, not just for the order 3 mock theta case discussed in Section \ref{sec:precision}.

First, recall this order 3 example. For $(p,a)=(3,2)$, the Borel (\ref{eq:js}) and dual Borel (\ref{eq:jsd}) integrals are (recall that $q=e^\hbar$, and $0<q<1$ for $\hbar<0$)
\begin{eqnarray}
J_{(3,2)}(\hbar)&=&-\frac{1}{\hbar} \int_0^\infty du\, e^{3 u^2/\hbar} \, \frac{\sinh(u)}{\sinh(3 u)}
\label{eq:j32}
\\
&=&
\frac{\sin\left(\frac{2\pi}{3}\right) }{\sqrt{3\pi (- \hbar) }} 
   \int_0^\infty dv  \, e^{3 v^2 \hbar /(\pi^2)} \frac{1}{\cosh[2v] -\cos\left(\frac{2 \pi}{3}\right)}
   \label{eq:jd32}
\end{eqnarray}
These have an additional symmetry because
from the identities (\ref{eq:identities}) we see that the Borel transform function is proportional to its dual
\begin{eqnarray}
\frac{\sinh(u)}{\sinh(3 u)}=\frac{1}{2} \left(\frac{1}{\cosh(2u)-\cos\left(\frac{2\pi}{3}\right)}\right)
\label{eq:W3borel}
\end{eqnarray}
Therefore, $J_{(3,2)}(\hbar)$ transforms in a simple way under $\hbar\to \frac{\pi^2}{\hbar}$:
\begin{eqnarray}
J_{(3,2)}\left(\frac{\pi^2}{\hbar}\right)=\left(\frac{-\hbar}{\pi}\right)^{3/2} J_{(3,2)}(\hbar)
\label{eq:W3duality}
\end{eqnarray}
The corresponding mock-modular relation \cite{Wat,GM12} can be written as, {\it cf.} \eqref{eq:Stokes4}--\eqref{longeqZSmS}:
\begin{eqnarray}
\sqrt{\frac{-12\hbar}{\pi}} J_{(3,2)}\left(\hbar\right)&=& e^{\frac{2\hbar}{3}}\omega\left(-e^{\hbar}\right)+
 \sqrt{\frac{\pi}{-\hbar}}\, e^{\frac{2\pi^2}{3\hbar}}\omega\left(-e^{\frac{\pi^2}{\hbar}}\right)
\label{eq:32}
\end{eqnarray}
This has the consequence that when fitting the $q$-series and $\tilde{q}$-series decompositions of the Mordell-Borel integral $J_{(3,2)}\left(\hbar\right)$, there is only one $q$-series to be determined: $\omega(-q)$. This is the example proved rigorously in Section \ref{S-uniq-omega}, and analyzed numerically in Section \ref{sec:precision}.

A similar decomposition exists for higher $p$ values. To see this, 
recall the duality expression (\ref{eq:Jmagic}). This can be further reduced to an identity for the duality transformation $\hbar\to \pi^2/\hbar$, with details depending on the parity of $p$ and $a$. For example, if $a=2r$ is even:
\begin{eqnarray}
J_{(p,2r)}(\hbar)=\left(\frac{\pi}{-\hbar}\right)^{3/2}\, \frac{2}{\sqrt{p}} \sum_{b=1}^{{\rm Floor}\left[\frac{p}{2}\right]} \sin\left(\frac{2 r b \pi}{p}\right) J_{(p,2b)}\left(\frac{\pi^2}{\hbar}\right)
\label{eq:Jmagiceven}
\end{eqnarray}
This has a direct consequence on the $q$-series and $\tilde{q}$-series appearing on the unary side. For example, when the parameter $a$ is even, the $\tilde{q}$-series appearing in the imaginary part (\ref{eq:paimag}) are {\it exactly the same as} the $q$-series appearing in the real part (\ref{eq:pareal}). This is because for $a=2r$, with $r=1,2, \dots, {\rm Floor}\left[\frac{p}{2}\right]$, we have:
\begin{eqnarray}
&&e^{-\frac{r^2\pi^2}{p \hbar}}\left[1+\sum_{m=1}^\infty (-1)^{2r\, m}\left(e^{-\frac{ \pi ^2 \left(m^2 p+2 m r\right)}{\hbar}}
 -e^{-\frac{ \pi ^2 \left(m^2 p-2 m r\right)}{\hbar}}\right)\right]  \\  
 &&\hskip 2cm =\left\{ 
 e^{-\frac{(p-2r)^2\hbar}{4\, p}}
\sum_{k=0}^\infty e^{-p\, \hbar\left(k+\frac{1}{2}\right)^2}\left(e^{(p-2r) \hbar\left(k+\frac{1}{2}\right)}-e^{-(p-2r) \hbar \left(k+\frac{1}{2}\right)}\right)
\right\}_{\hbar\to \frac{\pi^2}{\hbar}}
\end{eqnarray}
This remarkable relation between the $q$-series and $\tilde{q}$-series (when $\hbar>0$ and the $a$ parameter is even) is a manifestation of the $\hbar<0$ duality identity (\ref{eq:Jmagiceven}), which is preserved under analytic continuation of $\hbar$ to $\hbar>0$ in the Borel integrals (\ref{eq:js}):
\begin{eqnarray}
{\rm Im}\left[\sqrt{\frac{-4\, p\, \hbar}{\pi}} J_{(p,2r)}\left(\hbar\right)\right]=
\sqrt{\frac{\pi}{\hbar}} \sum_{b=1}^{{\rm Floor}\left[\frac{p}{2}\right]} \frac{2}{\sqrt{p}} \sin\left(\frac{2 r b \pi}{p}\right) \, {\rm Re}\left[\sqrt{\frac{-4\, p\, \pi}{\hbar}} J_{(p,2b)}\left(\frac{\pi^2}{\hbar}\right)\right]
\label{eq:real-imag-even}
\end{eqnarray}
The real and imaginary parts are coupled because of the $\sqrt{\frac{\pi}{\hbar}}$ factor. 

This means that for larger $p$ values, we have a closed matrix system of duality relations. For example, for $p=5$, we have the $2\times 2$ matrix structure:
\begin{eqnarray}
\sqrt{\frac{-20 \hbar}{\pi}}
\begin{pmatrix}
J_{(5,2)}(\hbar)
\cr
J_{(5,4)}(\hbar)
\end{pmatrix} 
= \sqrt{\frac{\pi}{-\hbar}} \frac{2}{\sqrt{5}}
\begin{pmatrix}
\sin\left(\frac{2\pi}{5}\right) & \sin\left(\frac{4\pi}{5}\right)\cr
\sin\left(\frac{4\pi}{5}\right) & \sin\left(\frac{8\pi}{5}\right)
\end{pmatrix}
\sqrt{\frac{20\pi}{-\hbar}}
\begin{pmatrix}
J_{(5,2)}\left(\frac{\pi^2}{\hbar}\right)
\cr
J_{(5,4)}\left(\frac{\pi^2}{\hbar}\right)
\end{pmatrix} 
\end{eqnarray}
Correspondingly, we have  a closed 2-component $q$-series relation of the following form:
\begin{eqnarray}
\sqrt{\frac{-20\hbar}{\pi}}
\begin{pmatrix}
J_{(5,2)}(\hbar)
\cr
J_{(5,4)}(\hbar)
\end{pmatrix} 
&=& \begin{pmatrix}
e^{-\frac{4\hbar}{20}}\, W_1\left(e^{\hbar}\right)
\cr
e^{\frac{-16\hbar}{20}}\,W_2\left(e^{\hbar}\right)
\end{pmatrix} \nonumber\\
&&  +\sqrt{\frac{\pi}{-\hbar}} \frac{2}{\sqrt{5}}
\begin{pmatrix}
\sin\left(\frac{2\pi}{5}\right) & \sin\left(\frac{4\pi}{5}\right)\cr
\sin\left(\frac{4\pi}{5}\right) & \sin\left(\frac{8\pi}{5}\right)
\end{pmatrix}
\begin{pmatrix}
e^{-\frac{4\pi^2}{20\hbar}}\, W_1\left(e^{\frac{\pi^2}{\hbar}}\right)
\cr
e^{\frac{-16\pi^2}{20\hbar}} \, W_2\left(e^{\frac{\pi^2}{\hbar}}\right)
\end{pmatrix}
\label{eq:5even}
\end{eqnarray}
Where $W_1(q)$ and $W_2(\tilde{q})$ have the unary expansions in (\ref{eq:pareal}) and (\ref{eq:paimag}) when $\hbar>0$, and which can be determined numerically when $\hbar<0$. 
Note that there are two equations in (\ref{eq:5even}) and there are two series to determine.
The eigenvalues of the S-matrix $\frac{2}{\sqrt{5}}
\begin{pmatrix}
\sin\left(\frac{2\pi}{5}\right) & \sin\left(\frac{4\pi}{5}\right)\cr
\sin\left(\frac{4\pi}{5}\right) & \sin\left(\frac{8\pi}{5}\right)
\end{pmatrix}
$ are $\pm 1$, and by diagonalizing this matrix we obtain two different linear combinations of $J_{(5, 2)}(\hbar)$ and $J_{(5, 4)}(\hbar)$, which are self-dual, and anti-self-dual, under $\hbar\to \frac{\pi^2}{\hbar}$, respectively.  
We can now apply the high-precision numerical procedure described in Section \ref{sec:precision} for the order 3 mock self-dual case. This reveals that
\begin{eqnarray}
    W_1(q)&\sim & -q+q^2-2q^3+2q^4-2 q^5+\dots \\
    W_2(q)&\sim &  -q +2q^2-2q^3+3 q^4 -4q^5+\dots
\end{eqnarray}
We recognize these expansions as those of the order 10 mock theta functions $\psi$ and $\phi$ in (\ref{eq:psi10}) and (\ref{eq:phi10}), respectively:
\begin{eqnarray}
    W_1(q)= -\psi(-q)
    \qquad
    \text{and}
    \qquad
    W_2(q)= q\, \phi(-q)
\end{eqnarray}
This is in complete agreement with the known order 10 mock modular relations \cite{GM12}, which can be written as
\begin{eqnarray}
\sqrt{\frac{-20\hbar}{\pi}}
\begin{pmatrix}
J_{(5,2)}(\hbar)
\cr
J_{(5,4)}(\hbar)
\end{pmatrix} 
&=& \begin{pmatrix}
e^{-\frac{4\hbar}{20}}\left[-\psi\left(-e^{\hbar}\right)\right]
\cr
e^{-\frac{16\hbar}{20}}\left[e^{\hbar} \phi\left(-e^{\hbar}\right)\right]
\end{pmatrix} \nonumber\\
&& \hskip -2cm +\sqrt{\frac{\pi}{-\hbar}} 
\frac{2}{\sqrt{5}}
\begin{pmatrix}
\sin\left(\frac{2\pi}{5}\right) & \sin\left(\frac{4\pi}{5}\right)\cr
\sin\left(\frac{4\pi}{5}\right) & \sin\left(\frac{8\pi}{5}\right)
\end{pmatrix}
\begin{pmatrix}
e^{-\frac{4\pi^2}{20\hbar}}\left[-\psi\left(-e^{\frac{\pi^2}{\hbar}}\right)\right]
\cr
e^{-\frac{16\pi^2}{20\hbar}}\left[e^{\frac{\pi^2}{\hbar}} \phi\left(-e^{\frac{\pi^2}{\hbar}}\right)\right]
\end{pmatrix}
\label{eq:5even2}
\end{eqnarray}
Once  these are determined, $SL(2, \mathbb Z)$  transformations can be used to generate all the other mock 10 $q$-series expressions. Further details of this numerical procedure will be described in a future paper.

This approach can be extended to higher $p$ values, beyond the mock theta functions. For example,  for $(p, a)=(7, 2)$, $(p, a)=(7,4)$ and $(p, a)=(7, 6)$, we postulate a closed 3-component "mock-modular relation":
\begin{eqnarray}
\sqrt{-\frac{28\hbar}{\pi}}
\begin{pmatrix}
J_{(7,2)}(\hbar)
\cr
J_{(7,4)}(\hbar)
\cr
J_{(7,6)}(\hbar)
\end{pmatrix} &=&
\begin{pmatrix}
e^{-\frac{4\hbar}{28}}W_1\left(e^{\hbar}\right)
\cr
e^{-\frac{16\hbar}{28}}W_2\left(e^{\hbar}\right)
\cr
e^{-\frac{36\hbar}{28}}W_3\left(e^{\hbar}\right)
\end{pmatrix} \nonumber\\
&& \hskip -3cm 
+\sqrt{\frac{\pi}{-\hbar}} \frac{2}{\sqrt{7}}
\begin{pmatrix}
\sin\left(\frac{2\pi}{7}\right) & \sin\left(\frac{4\pi}{7}\right) & \sin\left(\frac{6\pi}{7}\right)\cr
\sin\left(\frac{4\pi}{7}\right) & \sin\left(\frac{8\pi}{7}\right) & \sin\left(\frac{12\pi}{7}\right)\cr
\sin\left(\frac{6\pi}{7}\right) & \sin\left(\frac{12\pi}{7}\right) & \sin\left(\frac{18\pi}{7}\right)
\end{pmatrix}
\begin{pmatrix}
e^{-\frac{4\pi^2}{28\hbar}}
W_1\left(e^{\frac{\pi^2}{\hbar}}\right)
\cr
e^{-\frac{16\pi^2}{28\hbar}}
W_2\left(e^{\frac{\pi^2}{\hbar}}\right)
\cr
e^{-\frac{36\pi^2}{28\hbar}}
W_3\left(e^{\frac{\pi^2}{\hbar}}\right)
\end{pmatrix}
\label{eq:7246}
\end{eqnarray}
Here $W_j$, for $j=1, 2, 3$, are 3 different functions which appear both as $q$-series and as $\tilde{q}$-series in this expression.  This symmetric form of the expansion is consistent with the duality (\ref{eq:Jmagiceven}) under $\hbar\to \frac{\pi^2}{\hbar}$, and matches the structure of the $\hbar\to -\hbar$ analytic continuation of the Borel integrals $J_{(7,2)}(\hbar)$, $J_{(7,4)}(\hbar)$ and $J_{(7,6)}(\hbar)$ on the unary side. We can diagonalize the coupling matrix of sine factors, leading to 2 linear combinations of $J_{(7,2)}(\hbar)$, $J_{(7,4)}(\hbar)$ and $J_{(7,6)}(\hbar)$ which are self-dual under $\hbar\to \frac{\pi^2}{\hbar}$, and 1 linear combination of $J_{(7,2)}(\hbar)$, $J_{(7,4)}(\hbar)$ and $J_{(7,6)}(\hbar)$ which is anti-self-dual under $\hbar\to \frac{\pi^2}{\hbar}$. Note that we have 3 different relations and we need to determine 3 different $q$-series. The same series appear on both sides, so these combinations can be precisely probed near $q_0=e^{-\pi}$ in order to determine the integer-coefficient expansions of $W_j(q)$.  Finally, given these expansions, $SL(2, \mathbb Z)$ transformations generate all the further $q$-series with $p=7$.
These numerical methods will be described in more detail in a future paper.

\subsection{$SL(2, \mathbb Z)$ action on Borel transforms}
\label{sec:sl2z}

The Borel functions are the primary building blocks, and in each group of Mock Theta functions, one Borel kernel generates all the others via the action of $SL(2,\mathbb Z)$. Furthermore, Laplace transforming these generated kernels, one obtains all special $q$-series in the group. The Laplace transform of the Borel function is analytic across $|q|=1$.
On the other hand, there is a unique $q$-series representation on each side of the boundary, and these series are a continuation of each other by virtue of the analyticity of the Borel integral and the aforementioned uniqueness.

The analysis of the transformations of the Mordell-Borel integrals of the mock theta functions is more conveniently  done for $\hbar>0$ where the real part of the Mordell integral simply equals one q-series rather than a combination of $q$ and $\tilde{q}$ series. We note that none of the sinh or cosh integrals can equal a pure $q$ series (for $\hbar<0$), because of the divergence of their transseries at both limits $\hbar\to 0^-$ and $\hbar\to -\infty$. 

The action of $\hbar$ to $1/\hbar$ (equivalently, $q$ to $\tilde{q}$) is given by a Fourier transform. We note the following identities (for $\hbar<0$):
\begin{align}\label{eq:sum1}
  \int_0^\infty \frac{\sinh(a x)}{\sinh(p x)}e^{p x^2/\hbar}= \sqrt{\frac{\pi(- \hbar)}{p}}\int_0^\infty \frac{\sin(\pi a/p)}{\cos(a\pi/p)+\cosh(2\pi x)}e^{p x^2 \hbar}sdx\\
  \int_0^\infty \frac{\cosh(a x)}{\cosh(p x)}e^{p x^2/\hbar}= 2\sqrt{\frac{\pi (-\hbar)}{p}}\int_0^\infty \frac{\cos(\pi a/(2p))\cosh(\pi x)}{\cos(a\pi/p)+\cosh(2\pi x)}e^{p x^2 \hbar}dx \label{eq:sum2}
\end{align}
The action of $\hbar$ to $\hbar+\pi i$ (equivalently, $q$ to $-q$) is given by a Weierstrass transform. The effect of these transformations is calculated conveniently for $\hbar>0$ using the identities (here we write $q=e^\hbar=e^{-s}$)
\begin{equation}
  \label{eq:sum2b}
 \sum_{k=0}^\infty\Big( e^{ -\frac{(a+(-2 k-1) p)^2}{4 p s}}-e^{ -\frac{(a+2 k p+p)^2}{4 p
   s}}\Big)=2\sqrt{\frac{p s}{\pi}}PV \int_0^\infty e^{-p v^2 s}\frac{\sin (a v)}{\sin (p v)}dv
\end{equation}
 \begin{equation}
  \label{eq:sum3}
 \sum_{k=0}^\infty\Big((-1)^k e^{-\frac{(a+(-2 k-1) p)^2}{4 p s}}+(-1)^k e^{-\frac{(a+2
   k p+p)^2}{4 p s}}\Big)=2\sqrt{\frac{p s}{\pi}}PV \int_0^\infty e^{-p v^2 s}\frac{\cos (a v)}{\cos (p v)}dv
\end{equation}
For example, consider the action of $SL(2,\mathbb Z)$ on the order 3 mock theta functions. We denote as usual the elements of the group $SL(2,\mathbb Z)$ as $S$ corresponding to $\tau\to -1/\tau$, or in our normalization, $\hbar\to \pi^2/\hbar$, and $T$ corresponding to $\tau\to \tau+1$, for us $q\to -q$. Using the results above, and denoting by $v$ the column vector with components $(\mathcal W, \mathcal W_1, \mathcal W_2, \mathcal W_3)$ in the notations of \cite{GM12}, we get the action of the group (the matrices are denoted by $\hat{S}$ and $\hat{T}$) as:
\begin{equation}
  \label{eq:matrices}
 \hat{S}= \left(
\begin{array}{cccc}
 1 & 0 & 0 & 0 \\
 0 & 0 & 1 & 0 \\
 0 & 1 & 0 & 0 \\
 0 & 0 & 0 & 1 \\
\end{array}
\right);\ \ \hat{T}=\left(
\begin{array}{cccc}
 0 & 1 & 0 & 0 \\
 1 & 0 & 0 & 0 \\
 0 & 0 & 0 & -1 \\
 0 & 0 & -1 & 0 \\
\end{array}
\right);\ \ \hat{T}\hat{S}=\left(
\begin{array}{cccc}
 0 & 0 & 1 & 0 \\
 1 & 0 & 0 & 0 \\
 0 & 0 & 0 & -1 \\
 0 & -1 & 0 & 0 \\
\end{array}
\right)
\end{equation}
\noindent{Proposition}
  $\hat{T}^2=\hat{S}^2=(\hat{T}\hat{S})^4=1$. The eigenvalues of $\hat{T}\hat{S}$ are $\pm 1,\pm i$.  The group generated by $\hat{T},\hat{S}$ is isomorphic to the dihedral group $D_4$.The representation induced by $\hat{S}$ and $\hat{T}$ is reducible, and the space generated by the vectors $(-1,0,0,1),(0,1,1,0)$ is invariant under $\hat{S},\hat{T}$. However, each basis vector $(1,0,0,0),...,(0,0,0,1)$ is cyclic for $\hat{T}\hat{S}$, hence any of the $\mathcal W's$ generates the whole set $\mathcal W, \mathcal W_1, \mathcal W_2, \mathcal W_3$. 

As another example, consider the action of $SL(2,\mathbb Z)$ on the order 10 mock theta functions. The analysis is similar, but the results are different. In the notation of \cite{GM12} there are 8 fundamental functions, denoted $K_4,...,K_7,J_4,...,J_7$, and the matrices are 
\begin{equation}
  \label{eq:matrices10}
 \hat{S}= \left(
\begin{array}{cccccccc}
 0 & 0 & 0 & 0 & 0 & -a & 0 & b \\
 0 & 0 & 0 & 0 & 0 & b & 0 & a \\
 0 & 0 & -a & b & 0 & 0 & 0 & 0 \\
 0 & 0 & b & a & 0 & 0 & 0 & 0 \\
 0 & 0 & 0 & 0 & -b & 0 & a & 0 \\
 -a & b & 0 & 0 & 0 & 0 & 0 & 0 \\
 0 & 0 & 0 & 0 & a & 0 & b & 0 \\
 b & a & 0 & 0 & 0 & 0 & 0 & 0 \\
\end{array}
\right);\ \ \hat{T}=\left(
\begin{array}{cccccccc}
 0 & 0 & 0 & 0 & -1 & 0 & 0 & 0 \\
 0 & 0 & 0 & 0 & 0 & 0 & 1 & 0 \\
 0 & 0 & 0 & 0 & 0 & 0 & 0 & 1 \\
 0 & 0 & 0 & 0 & 0 & 1 & 0 & 0 \\
 -1 & 0 & 0 & 0 & 0 & 0 & 0 & 0 \\
 0 & 0 & 0 & 1 & 0 & 0 & 0 & 0 \\
 0 & 1 & 0 & 0 & 0 & 0 & 0 & 0 \\
 0 & 0 & 1 & 0 & 0 & 0 & 0 & 0 \\
\end{array}
\right)
\end{equation}
Here $b/a=\phi=\frac12(1+\sqrt{5})$ and $a^2+b^2=1$. Again each of the basis vectors is cyclic w.r.t. $\hat{T}\hat{S}$, as follows from the analysis below. The eigenvalues of are $\hat{T}\hat{S}$ are $(e^{\frac{i \theta }{2}+\frac{3 i \pi }{8}},e^{\frac{i \pi
   }{8}-\frac{i \theta }{2}},$ $ -e^{\frac{i \theta }{2}+\frac{7 i
   \pi }{8}},$ $ -e^{\frac{5 i \pi }{8}-\frac{i \theta
   }{2}},$ $ e^{\frac{i \theta }{2}+\frac{7 i \pi }{8}},$ $ e^{\frac{5 i
   \pi }{8}-\frac{i \theta }{2}},$ $ -e^{\frac{i \theta }{2}+\frac{3
   i \pi }{8}},$ $ -e^{\frac{i \pi }{8}-\frac{i \theta }{2}})$ where $e^{i\theta}=a+ib$, all distinct. The eigenvalues of  $(\hat{T}\hat{S})^8$ are $(\frac{3}{5}+\frac{4 \,\mathrm{i}}{5}, $ $ 
\frac{3}{5}-\frac{4 \,\mathrm{i}}{5}, $ $ 
\frac{3}{5}+\frac{4 \,\mathrm{i}}{5},$ $  
\frac{3}{5}-\frac{4 \,\mathrm{i}}{5},$ $ 
\frac{3}{5}+\frac{4 \,\mathrm{i}}{5}, $ $ 
\frac{3}{5}-\frac{4 \,\mathrm{i}}{5}, $ $ 
\frac{3}{5}+\frac{4 \,\mathrm{i}}{5},$ $  
\frac{3}{5}-\frac{4 \,\mathrm{i}}{5}
)$, from which it follows that $\hat{T}\hat{S}$ is of infinite order. Passing to the basis of  $\hat{T}\hat{S}$ it is easy to see that there is no nontrivial invariant subspace for the group, hence the representation is irreducible, and that the original basis vectors are cyclic w.r.t. $\hat{T}\hat{S}$.

This approach can in principle be generalized to higher orders.

\section{Exploring new limits: small surgeries}
\label{sec:thooft-small}

As briefly mentioned in the Introduction, in low-dimensional topology 0-surgeries \eqref{0surgeryonK} play a rather special role among general surgeries, {\it cf.} \eqref{Msurgery}. They appear at the core of some of the most challenging questions. For example, ``Property R'' conjecture (proved by D.~Gabai in 1983 \cite{Gab83}) states that if the 0-surgery on $K \subset S^3$ is homeomorphic to $S^1 \times S^2$, then $K$ is the unknot. In that same series of works Gabai showed that the trefoil and figure-8 knots are likewise characterized by their 0-surgeries \cite{Gab87}. Since then, many generalizations of the Property R conjecture have been studied. One such generalization, closely related to the slice-ribbon conjecture --- another major open problem in low-dimensional topology --- asks to identify links (up to handle slides) that produce connected sums of $S^1 \times S^2$ via surgery \cite{GST}.

Furthermore, 0-surgeries offer a strategy to disprove the famous smooth Poincar\'e conjecture in dimension four (SPC4), which is false as soon as anyone finds a pair of knots with the same 0-surgery, such that one is slice and the other is not. See \cite{MP21,Nak22,GHMR} for some recent work. Another manifestation of the connection between knot sliceness, 0-surgeries, and smooth 4-manifold topology is the recent result \cite{Tru} that establishes a bound on the topology of a two-handlebody $M_4$ (i.e. two-handles attached to a 4-ball) bounded by a 0-surgery on a slice knot $K$, $\partial M_4 = S^3_0 (K)$,
\begin{equation}
b_2 (M_4) \ge \frac{10}{8} |\sigma (M_4)| + 5
\end{equation}
such that $b_2 (M_4) \ne 1$, 3, or 23.

While all 0-surgeries on knots in the 3-sphere have homology of $S^1 \times S^2$, the converse is not true and there are many 3-manifolds with the homology of $S^1 \times S^2$ that are not surgeries on any knot in the 3-sphere. Thus, for $M_3 = S^3_0 (K)$ at least one of Rokhlin invariants vanishes.\footnote{This follows from the general surgery formulae $\mu (Y_0 (K),s_0) = \mu (Y)$ and $\mu (Y_0 (K), s_1) = \mu (Y) + \text{Arf} (K)$ for an integral homology sphere $Y$.} Therefore, if $M_3$ is integral homology $S^1 \times S^2$ with two non-trivial Rokhlin invariants, then $M_3 \ne S^3_0 (K)$. See \cite{HKMP} for an infinite family of such 3-manifolds.

Curiously, despite their prominent role in low-dimensional topology, 0-surgeries on knots are not as prominently used as examples in quantum topology and, to the best of our knowledge, even the perturbative expansion $Z_{\text{pert}} (\hbar)$ appears to be not so well studied. One difficulty in studying the $0 = \frac{0}{1}$ surgery is that some of the simpler techniques do not apply, e.g. those which require $M_3$ to be a homology sphere. In particular, the BPS $q$-series invariants $\widehat Z_b (S^3_0 (K), q)$ are not known for general $K$. In part, these difficulties are related to the fact that the space of flat connections has a component of dimension $1$ (the abelian branch), which makes the resurgent analysis and the surgery formulae more delicate.

In this section, we make a step toward addressing this problem via a sequence of better-understood $\frac{p}{r} \ne 0$ surgeries with small surgery coefficients gradually approaching the limit $\frac{p}{r} \to 0$. Although we find that some topological properties and quantum invariants behave discontinuously in this limit, there are certain aspects for which this is a smooth limit. In particular, we find that the following double-scaling limit is especially effective in capturing such aspects and the regularities as $r\to\infty$ and $\hbar \to 0$, with the double-scaling parameter $t$ kept fixed,
\begin{eqnarray}
\text{double-scaling parameter:} \quad t:=\frac{\hbar\, r}{p}
\label{eq:thooft}
\end{eqnarray}
We mostly focus on $\pm \frac{1}{r}$ surgeries (called ``small surgeries'') here but the analysis extends to the general case.

\subsection{Physics of the double-scaling limit from 3d-3d correspondence}
\label{sec:3d3dlimit}

The behavior of the Chern-Simons invariants for $\frac{1}{r}$ surgeries can be understood from the perspective of the 3d-3d correspondence and the theory $T [M_3]$. Let us recall a few basic elements of the 3d-3d correspondence that will be useful to us below:

\begin{itemize}

\item $G_{\mathbb{\C}}$ flat connections on $M_3$ are the supersymmetric vacua of $T[G,M_3]$, a $3$-dimensional ${\mathcal N}=2$ supersymmetric quantum field theory defined on the manifold $M_3$, and the corresponding values of the twisted superpotential are the Chern-Simons invariants:
    \begin{equation}
    \text{CS} (\alpha) \; = \; \mathcal{W} \big\vert_{\alpha}
    \label{CSviaW}
    \end{equation}
    In particular, $SL(2,\mathbb{C})$ Chern-Simons values on $M_3$ are precisely the values of the twisted superpotential in 3d $\mathcal{N}=2$ theory $T[SU(2), M_3]$, which we denote simply by $T[M_3]$ to avoid clutter (and because we do not discuss other choices of the ``gauge group'' $G$).
    
    \item Next we  recall how the theory $T[M_3]$ and its superpotential $\mathcal{W}$ behave under cutting and gluing (see {\it e.g.} \cite{GGP16}):
    \begin{equation}
    \mathcal{W}_{T[S^3_{p/r} (K) ]} \; = \;
    \left( \mathcal{W}_{T[S^3 \setminus K ]} + \frac{p}{2r} (\log x)^2 \right)_{{\text{extremize} \atop \text{w.r.t. } x}}
    \end{equation}
    where $\frac{p}{r}$ is the surgery coefficient, and the extremization is understood in the standard ``K-theory'' sense of the 3d Bethe ansatz equations that involves the exponential of the log-derivative:
    \begin{equation}
    \exp \left[\frac{\partial}{\partial \log x} \left( \mathcal{W}_{T[S^3 \setminus K ]} + \frac{p}{2r} (\log x)^2 \right) \right]\; = \; 1
    \label{Wvacua}
    \end{equation}

\item In general, for a knot $K$ the superpotential $\mathcal{W}_{T[S^3 \setminus K ]}$ as function of $x$ is given by the integral of a 1-form differential on the $A$-polynomial curve\footnote{This is the Liouville 1-form in quantization of complex Chern-Simons theory \cite{Guk05} and can be interpreted as an analogue of the Seiberg-Witten differential in 3d $\mathcal{N}=2$ theory. In turn, the $A$-polynomial curve plays the role of the Seiberg-Witten curve in 3d $\mathcal{N}=2$ theory $T[M_3]$.}
    \begin{equation}
    \mathcal{W}_{T[S^3 \setminus K ]} \; = \; \int \log y \; d \log x
\end{equation}

\end{itemize}

These (and other) rules of the 3d-3d correspondence can be implemented explicitly for many simple knots and the ``knot complement theory'' $T[S^3 \setminus K ]$ is also known for many knots. In particular, this makes manifest several aspects of the $r$-dependence for $\frac{1}{r}$-surgeries.

In tracking particular flat connections (= vacua of $T[M_3]$) we should keep in mind that the total number of such flat connections / vacua is actually increasing with $r$. So, what we really mean is a smooth behavior of a {\it subset} of vacua / flat connections that can be traced all the way to $r=2$ or $r=1$.
Another important aspect --- that is more subtle than the existence of a limit $r \to \infty$ --- is whether this limit agrees with 0-surgery.
In other words, matching (or, even a relation) between the limit $\frac{1}{r} \to 0$ and the 0-surgery is not as obvious from the above 3d-3d perspective as the existence of a limit itself. There could be ``jumps'' between different branches.
    
The ``smoothness'' of the behavior in $r$ in this 3d-3d setup is ultimately a consequence of the fact that $\mathcal{W}$ is a nice function in both $x$ and $r$, if we treat $r$ as a continuous variable. In particular, the extremization with respect to $x$ returns a nice continuous function in $r$; jumps between branches can only occur at particular values of $r$. This explains why away from these points we observe continuous behavior.
    
Also, note that, since $\mathcal{W}_{T[S^3 \setminus K ]}$ is a linear combination of at most dilogarithms, its derivative involves only logs and so the exponentiated Bethe ansatz equation \eqref{Wvacua} is an algebraic equation. This is the basic reason why, for integer $r$, the solutions are algebraic integers and, therefore, the Chern-Simons invariants \eqref{CSviaW} are sums of logs and dilogs evaluated at these algebraic integers.
    
Another interesting feature of this physical perspective is that $\frac{p}{r}$ (or, more precisely, integers in its continued fraction expansion) have a clear physical meaning as coupling constants: they are supersymmetric Chern-Simons levels in the 3d $\mathcal{N}=2$ theory $T[M_3]$. In particular, this provides a well-defined physical framework for interpreting $r \hbar$ as a double-scaling parameter in the double expansion of $Z_{\text{pert}} (S^3_{-1/r} (K),\hbar)$ as $\hbar \to 0$ and $r \to \infty$.

According to the rules of 3d-3d correspondence, the limit $\hbar \to 0$ means that the $Q$-cohomology that implements the holomorphic twist (a.k.a. Omega-background) turns into the usual BRST cohomology of a partial topological twist in 3d $\mathcal{N}=2$ theory \cite{GHNPPS}. In other words, in this limit the supercharge $Q$ becomes the BRST differential of a partial topological twist. Therefore, from the perspective of 3d-3d correspondence, one might expect that the double-scaling limit \eqref{eq:thooft} could be achieved by considering a partial topological twist of 3d $\mathcal{N}=2$ theory $T[M_3]$ with $M_3 = S^3_0 (K)$. In the rest of Section~\ref{sec:thooft-small}, we provide further support to this preliminary conclusion. Note that topologically twisted observables in 3d $\mathcal{N}=2$ theories are usually much simpler in structure than the observables in the backgrounds $S^1 \times_q D^2$ and $S^1 \times_q S^2$, in perfect agreement with the fact that we find drastic simplifications in the double-scaling limit \eqref{eq:thooft} reducing quantum invariants to much simpler invariants reminiscent of the classical torsion. This also agrees with the conclusion that the double-scaling limit \eqref{eq:thooft} is directly related to the partial topological twist of $T[S^3_0 (K)]$ since the latter indeed computes \cite{GHNPPS} the Alexander polynomial of $K$ and its various generalizations.

\subsection{Small surgery limit via A-polynomial}
\label{sec:small-cs-values}

We start with considering this problem algebraically by looking at the $A$ polynomial curve. Recall that we can associate flat connections on the $\frac{p}{r}$ surgery to intersections between the affine varieties $\B{A}$ and $\{y = x^{-\frac{p}{r}}\}$.

The overarching idea is that for surgeries that are ``close'', in the sense that $|\frac{p}{r} - \frac{p'}{r'}| << 1$, we are able to estimate some CS invariants for the $\frac{p}{r}$ surgery from corresponding the CS invariants on the $\frac{p'}{r'}$ surgery. To illustrate this technique we focus on the case where $\frac{p'}{r'} = 0$. In this case we find a  collection of flat connections with CS invariant $0$ and these naturally link to flat connections with minimal CS invariants.

\subsubsection{CS values from roots of the Alexander Polynomial}
\label{sec:cs-alexander}

Fix a square root $x^*$ of the Alexander polynomial $\Delta_K(x^2)$ and note that this additionally forces the A-polynomial to vanish: $A(1, x^*) = 0$. Then for $\theta = \frac{p}{r} \in \m{Q}$ close to $0$ we expand
\begin{equation} \label{eq:x ansatz}
					x(\theta) = x^{*} + c_1 \theta +  \cdots + c_n \theta^n + \cdots
				\end{equation}
				Enforcing the surgery condition, $y(\theta)) = (x(\theta))^{-\theta}$, we find a collection of choices
				\begin{equation} \label{eq:y ansatz}
					y(\theta; n) = x(\theta)^{-\theta} = 1 - (\log(x^*) + 2n\pi i) \theta + \left(\frac{2c_1}{x} + (\log(x^*) + 2n\pi i)^2\right) \frac{\theta^2}{2} + \cdots
\end{equation}
labelled by choice of log branch. As $\theta \to 0$ all branch choices appear as intersection points. 
Imposing that the A-polynomial vanishes, $A(x(\theta), y(\theta)) = 0$, we can solve for $c_i$ by looking at the coefficient of $\theta^i$. Note that there can be multiple solutions along different branches if multiple irreducible branches of $\B{A}$ meet at $(x^*, 1)$.

Next, recall the formula for the CS invariant given in Equation \eqref{eq: Computing CS Values, explicit}. Provided $\theta$ is small, we can ignore the branching of $\log$ and simplify this to
				\begin{align} \label{eq: Computing CS Values}
	    		   CS(\theta, n, x^*) = \frac{1}{2\pi^2}\left(\int_{\gamma} \frac{\log(y)}{x} \ dx + \frac{1}{2\theta}\log(y(\theta))^2\right) = a_0 + a_1 \theta + a_2 \theta^2 + \cdots
\end{align}
The higher terms $a_2, \cdots$ depend on the specific knot but the first $2$ terms admit a more general description. For $|\theta| << 1$, the integrand is analytic in a small region around $\gamma \subset \B{A} \subset \m{C}^2$ and so
	    		\begin{equation} \label{eq: Small Theta Simplification}
	    		    \int_{\gamma} \frac{\log(y)}{x} \ dx = \int_0^{\theta} \theta \frac{\log(x(\theta))}{x(\theta)}x'(\theta) d\theta = O(\theta^2).
	    		\end{equation}
	    	    Hence the integral does not contribute to the first $2$ terms and so
    	    	\begin{align*}
    	    		  CS(\theta, n, x^*) & = \frac{1}{4\pi^2 \theta}\Log(y(\theta)) + O(\theta^2)
    	    		  \\ & = \frac{(\log(x^*) + 2n\pi i)^2}{4\pi^2}\theta + O(\theta^2).
    	    	\end{align*}
	    	This is a universal prediction for all knots and all roots $x^*$ of $\Delta_K(x^2)$. Due to symmetry, if $x^*$ is a root, so are $\{-x^*, (x^*)^{-1}, -(x^*)^{-1}\}$. Hence choosing a unique element from this set, we have a collection of CS values given approximately by
	    	\begin{equation} \label{eq: Universal Approximation}
	    	    \frac{(\log(x^*) + n\pi i)^2}{4\pi^2}\theta \quad \quad n \in \m{Z}.
	    	\end{equation}
	    	Observe that all of these CS invariants go to $0$ as $\theta \to 0$. Hence for small $\theta$, if we choose $x^*$ to minimise $|\log(x^*)|$, then the minimal CS invariant is $CS(\theta, 0, x^*)$.
	    	
	    	Let us apply this theory to our two examples, the $4_1$ and the $5_2$ knots.
	    \subsubsection{$4_1$ Knot}
	    \label{sec:41-large-r}
	        For the $4_1$ knot the universal small $\theta$ estimate (\ref{eq: Universal Approximation}) is already an excellent approximation as, due to the amphichirality of the $4_1$ knot, the $\theta^{2n}$ terms in the expansion vanish. Setting $x^* =  \frac{1 + \sqrt{5}}{2}$, we find that for the $\theta=-\frac{1}{2}$ surgery we get three predicted CS invariants,
	    \begin{eqnarray}
	    	    \left(-\frac{1}{2}\right)\frac{(\log(x^*) + n\pi i)^2}{4\pi^2} = \begin{cases}
	    	        -0.0029328 & \quad, \quad  n = 0
	    	        \\ 0.1220672 \mp 0.0382936 i & \quad, \quad n = \pm 1
	    	    \end{cases}
	    	    \label{eq:41_leading_theta}
	    	\end{eqnarray}
	    	These Chern-Simons invariants compare well to $\alpha = 1$, $4$, $5$ in Table \ref{tab: 4_1 knot}. We can further improve the approximation by computing higher corrections. Let $x_n = e^{n \pi i}x^*$ denote the solution on the branch $\log(x_n) = \log(x^*) + n \pi i$. Then including the next terms in equations \eqref{eq:x ansatz} and \eqref{eq:y ansatz} in $A_{4_1}(x(\theta), y(\theta)) = 0$ we find that
        	\begin{align*}
        		x_n(\theta) & = x_n + (-1)^{n}\frac{5 + \sqrt{5}}{100}\log(x_n)^2\theta^2 + O(\theta)^4\\
        		y_n(\theta) & = x_n(\theta)^{-\theta} = 1 - \log(x_n) \theta + \frac{1}{2}\log(x_n)^2 \theta^2 - \left(\frac{\log(x_n)^2}{10\sqrt{5}} + \frac{\log(x_n)^3}{6}\right)\theta^3 + O(\theta^4)
        	\end{align*}
        	Then, assuming $\theta$ is small, we can compute the CS invariant integral in \eqref{eq: Computing CS Values} to $O(\theta^3)$:
        	\begin{align} \label{eq: Shifted CS value 41}
        		CS_{4_1}(n; \theta) & = \frac{\log(x_n)^2}{4\pi^2}\theta + \frac{\sqrt{5} \log(x_n)^3}{300 \pi^2}\theta^3 - O(\theta)^5.
        	\end{align}
        	Again specialising to the $\theta = -\frac{1}{2}$ case, this improves our earlier predictions to
        	\[
	    	    CS_{4_1}\left(n; -\frac{1}{2}\right) = \begin{cases}
	    	        -0.0029433 & n = 0
	    	        \\ 0.1234017 \mp 0.0355726 i & n = \mp 1
	    	    \end{cases}
	    	\]
	    	Comparing again to Table \ref{tab: 4_1 knot} we see  improved accuracy.
	    	
	    	Another useful observation from these calculations is that if we normalize the smallest magnitude Borel singularity to be at $\pm 1$, then for small $\theta$ we find a family of poles which do not move much as $\theta$ changes. Explicitly these  occur at
        	\[
        		\frac{(\log(x^*) + n \pi i)^2}{\log(x^*)^2}
        	\]
        	and become a dominant family of subleading poles as $\theta \to 0$. The closest and most visible poles correspond to $n = \pm 1$, occurring at $(-41.62 \pm 13.06i)$ in this normalization. We analyze this phenomenon further in Section \ref{sec:thooft-exact}.
           	
        \subsubsection{$5_2$ Knot}
        \label{sec:52-large-r}
            
            The same analysis can be applied to the roots of the Alexander polynomial for the $5_2$ knot. In this case, since the $5_2$ knot is not amphichiral there will also be $\theta^2$ corrections to the universal approximation in Equations \eqref{eq: Universal Approximation} and \eqref{eq: Shifted CS value 41}. Following the same procedure we find
            \[
                x_n(\theta) = x_n + \frac{x_n}{8}\log(x_n)\theta + \frac{x_n}{896}\log(x_n)(14 + (7 - 2\sqrt{7} i)\log(x_n))\theta^2 + O(\theta)^3\\
            \]
            Here $x_n = e^{n \pi i}x^* = e^{n \pi i}\frac{1}{2\sqrt{2}} (\sqrt{7}+i)$ is defined similarly as before. From this we compute a family of predicted Chern-Simons invariants:
            \begin{align} \label{eq: Shifted CS value 52}
        		CS_{5_2}(n; \theta) & = \frac{\log(x_n)^2}{4\pi^2}\theta - \frac{\log(x_n)^2}{32 \pi^2}\theta^2 + \frac{\log(x_n)^2(21 - 2 \sqrt{7} i \log(x_n))}{5376\pi^2}\theta^3 + O(\theta)^4.
        	\end{align}
        	Despite appearances, these are strictly real as $\log(x_n)$ is imaginary. Setting $\theta = \pm \frac{1}{2}$ we find the following approximate values for $\mp\frac{1}{2}$ surgery for $5_2$, compared here with their exact numerical values from Tables \ref{tab: 1/2 surg 52} and \ref{tab: -1/2 surg 52}:
        	\begin{align}
        	 CS_{5_2}\left(n; -\frac{1}{2}\right) & = \begin{cases}
	    	        0.0017643 & n = 0
	    	        \\ 0.1662666 & n = 1
	    	        \\ 0.1041303 & n = - 1
	    	    \end{cases}
          \qquad  CS_{5_2}^{(-\frac{1}{2}) exact}  = \begin{cases}
	    	        0.001764890...  
	    	        \\ \frac{1}{6}=0.16666666...  
	    	        \\ \frac{5}{48}=0.1041666...  
	    	    \end{cases}\nonumber\\
	    	    CS_{5_2}\left(n; \frac{1}{2}\right)
                    & = \begin{cases}
	    	         -0.0015575 & n = 0
	    	        \\  -0.1468403 & n = 1
	    	        \\  -0.0918932 & n = - 1
	    	    \end{cases}
          \qquad  CS_{5_2}^{(+\frac{1}{2}) exact}  = \begin{cases}
	    	        -0.00155708...  
	    	        \\ -0.14661662...  
	    	        \\ -0.09186365...  
	    	    \end{cases}
          \label{eq:cs-comp}
	    	\end{align}
	    	We observe that for $\theta=\mp \frac{1}{2}$ surgery these simple estimates closely reproduce the three smallest magnitude exact Chern-Simons invariants in Tables \ref{tab: 1/2 surg 52} and \ref{tab: -1/2 surg 52}. 
	    		    
\subsection{Perturbative series expansions in the double-scaling limit}
    
Consider applying the procedure in Section \ref{sec:borel} to produce perturbative expansions around the trivial flat connection with $\frac{p}{r}$ kept generic. The key observation is that even with generic $\frac{p}{r}$
	\begin{equation}
		\B{L}_{\frac{p}{r}}\left((x^{\frac{1}{2}} - x^{\frac{-1}{2}})(x^{\frac{1}{2r}} - x^{\frac{-1}{2r}})C_{K; m}(q)(qx; q)_m(qx^{-1}, q)_m\right) = O(\hbar^{m+1})
	\end{equation}
and so we will end up with a series in $\hbar, r$ and $\frac{1}{p}$. We first illustrate this with examples.
    	
\subsubsection{Perturbative Series Expansions in the Double-Scaling Limit for  $4_1$ Knot}
\label{sec:41-perturbative-large-r}

Specializing to the $4_1$ knot we find the formal perturbative series:
\begin{align}
    		\B{Z}_{0}^{pert}(S^3_{-\frac{p}{r}}(4_1)) & = \frac{1}{p} + \left(-\frac{1}{2p} + \frac{r^{-1} + 25r}{4p^2}\right)\hbar
    		\nonumber\\ & \quad + \left(-\frac{11}{12p} - \frac{r^{-1} - 25r}{8p^2} + \frac{3r^{-2} + 250 + 6483r^2}{96p^3}\right)\hbar^2 + O(\hbar^3)
    		\label{eq:Zthooft-41}
\end{align}
Already, even for these low orders of the expansion we can see one of the general features of this series. Writing the coefficient of $\hbar^i$ as a Laurent series in $r$ and a polynomial in $\frac{1}{p}$, there is a unique monomial with largest $r$ power, $a_{i,i+1}\frac{r^i}{p^{i+1}}$, and $|a_{i,i+1}|$ is the largest absolute coefficient.

For small surgeries where $|r| \gg |p|$, the polynomials are dominated by this term:
\begin{align}
    		\B{Z}_{0}^{pert}(S^3_{-\frac{p}{r}}(4_1)) & \sim \frac{1}{p} + \frac{25 r \hbar}{4p^2} + \frac{2161 r^2 \hbar^2}{32 p^{n+1}} + \frac{391945 r^3 \hbar^3}{384 p^4} + \frac{121866721 r^4 \hbar^4}{6144 p^5} + O\left(\frac{1}{p}\left(\frac{r \hbar}{p}\right)^5\right)
      \label{eq:z41rp}
\end{align}
    	
This fact makes it natural to define a double-scaling limit, $\frac{r}{p}\to\infty$ and $\hbar\to 0$ such that the parameter $t=\frac{\hbar\,r}{p}$ defined in (\ref{eq:thooft}) is fixed, by selecting the $k^{th}$ power of $r/p$ in $a_k(r)$. We  choose $p=1$ to simplify the notation, concentrating on $-\frac{1}{r}$ surgery, with $r\to\infty$.
	   This defines the following sequence of coefficients:
     \begin{eqnarray}
b_n:=   \left\{1,\frac{25}{4},\frac{2161}{32},\frac{391945}{384},\frac{121866721}{6144},\frac{11578044773}{24576},\frac{38999338931281}{2949120}, ...\right\}
\label{eq:thooft-41-coeffics}
   \end{eqnarray}
   For $+\frac{1}{r}$ surgery, since the $4_1$ knot is amphichiral, the only difference is that the expansion coefficients alternate in sign. So we can define two new perturbative sums, for the double-scaling limit for $\mp\frac{1}{r}$ surgery for the $4_1$ knot:
\begin{eqnarray}
   T^{(\mp)}_{4_1}(t):=\sum_{n=0}^\infty (\pm 1)^n b_n\, t^n\qquad, \quad t:=\hbar\, r
   \label{eq:t41}
   \end{eqnarray}
Studying the large order growth of the coefficients reveals that:
   \begin{eqnarray}
   b_n\sim\frac{1}{\sqrt{5\pi}}\frac{\Gamma\left(n+\frac{3}{2}\right)}{\left[\ln\left(\frac{1+\sqrt{5}}{2}\right)\right]^{2n+2}}\left( 1+{\rm exponentially\,\,small\,\, corrections}\right)\qquad, \quad n\to\infty
   \label{eq:41thooft-leading}
   \end{eqnarray}
Therefore, the weak coupling double-scaling  expansion in (\ref{eq:t41}) is an asymptotic expansion.
Note that with about 100 coefficients in the double-scaling expansion we have enough data to extract numerically the exact large order growth of the coefficients, which tells us the exact Borel radius of convergence and also the corresponding exact Stokes constant shown in (\ref{eq:41thooft-leading}). 
   The fact that the first sub-leading corrections to the leading large-order growth in (\ref{eq:41thooft-leading}) are {\it exponential} rather than power-law tells us that the associated leading Borel singularity is a pole. These numerical observations are confirmed analytically below (see Section \ref{sec:thooft-exact}).

The radius of convergence of the Borel transform of the double-scaling expansion in (\ref{eq:t41}) is  $\left[\ln\left(\frac{1+\sqrt{5}}{2}\right)\right]^2\approx 0.231564820577$. Surprisingly, this is remarkably close to the radius of convergence for the $\pm \frac{1}{r}=\pm \frac{1}{2}$ surgery, even though $r=2$ is very far from the $r\to\infty$ limit. If we re-instate the factor of $1/r$ (in order to compare the double-scaling expansion in powers of $(r\, \hbar)$ with an expansion in powers of $\hbar$), as well as the $\frac{1}{4\pi^2}$ normalization,  then with $r=2$, we can compare with the exact leading Chern-Simons invariant (see the first row of Table \ref{tab: 4_1 knot}) and also with the numerical Borel evaluation in (\ref{eq:cs41-leading}):
   \begin{eqnarray}
 \frac{CS^{leading}_{4_1, \mp \frac{1}{2}}}{4\pi^2}&=& \mp 0.0029434 ...
   \\
 \left[\mp \frac{1}{r}\frac{\left[\ln\left(\frac{1+\sqrt{5}}{2}\right)\right]^{2}}{4\pi^2 } \right]_{r=2}&=& \mp 0.0029328 ...
  \end{eqnarray}
In fact, this identification can be made even more precise using the large $r$ expansion introduced in Sections \ref{sec:small-cs-values} and \ref{sec:41-large-r} to probe the approach to the small surgery limit.

Furthermore, from the leading large-order growth in (\ref{eq:41thooft-leading}) we extract the leading Stokes constant as
        \begin{eqnarray}
  \frac{1}{\sqrt{5\pi}} \frac{1}{\left[\ln\left(\frac{1+\sqrt{5}}{2}\right)\right]^{2}} &=& 1.089600966040975 ...
  \end{eqnarray}
Once again, this is remarkably close to the exact value for the associated Adjoint Reidemeister torsion in Table \ref{tab: 4_1 knot}, and the numerical Borel Stokes constant (\ref{eq:41-stokes}) computed in Section \ref{sec:41borel}, which were derived for the $\pm \frac{1}{r}=\pm \frac{1}{2}$ surgery:
     \begin{eqnarray}
   {\mathcal S}_{4_1}&=&1.10366976209388967 ...
  \end{eqnarray}
  
\subsubsection{Perturbative Series Expansions in the Double-scaling Limit for  $5_2$ Knot}
\label{sec:52-perturbative-large-r}

  A similar analysis for $-\frac{1}{r}$ surgery for the $5_2$ knot yields the formal expansion:
\begin{eqnarray}
    \B{Z}_{0}^{pert}(S^3_{-\frac{1}{r}}(5_2)) &=& 1+\hbar \left(-\frac{47 r}{4}+\frac{1}{4 r}-\frac{1}{2}\right)+\hbar^2 \left(\frac{7201 r^2}{32}+\frac{1}{32 r^2}-\frac{241 r}{8}-\frac{1}{8
   r}-\frac{45}{16}\right)+O(\hbar^3) \nonumber\\
   &\sim &
   1-\frac{47(\hbar r)}{4}+\frac{7201 (\hbar r)^2}{32}-O((\hbar r)^3) 
      \end{eqnarray}
 The double-scaling limit ($\hbar\to 0$ and $r\to\infty$, with $\hbar r$ fixed) leads to:  
 \begin{eqnarray}
   T^{(\mp)}_{5_2}(t):=\sum_{n=0}^\infty (\mp 1)^{n} b_n\, t^n\qquad, \quad t:=\hbar\, r
   \label{eq:t52}
   \end{eqnarray}
   where the first coefficients are:
   \begin{eqnarray}
   b_n=\left\{1,\frac{47}{4},\frac{7201}{32},\frac{2316047}{384},\frac{1276975681}{6144},\frac{1075667467247}{122880}, \dots \right
   \}
   \label{eq:thooft52-coeffics}
   \end{eqnarray}
   In the double-scaling limit the only difference between the $\mp \frac{1}{r}$ surgeries is in the sign alternation pattern of the coefficients. However, note that the sign alternation pattern is the opposite in \eqref{eq:t41} and \eqref{eq:t52}. This matches the corresponding difference between the sign patterns for $\mp\frac{1}{2}$ surgery in the $4_1$ and $5_2$ cases: compare \eqref{eq:z41}-\eqref{eq:z41plus} and \eqref{eq:z52minus}-\eqref{eq:z52plus}.
    
    The leading large order growth of these coefficients is:
   \begin{eqnarray}
   b_n\sim \frac{1}{\sqrt{14\pi}}\frac{\Gamma\left(n+\frac{3}{2}\right)}{\left[\ln\left(\frac{\sqrt{7}+i}{2\sqrt{2}}\right)\right]^{2n+2}}\left( 1+{\rm exponentially\,\,small\,\, corrections}\right), \quad n\to\infty
   \label{eq:52thooft-leading}
   \end{eqnarray}
So the weak coupling double-scaling expansion in (\ref{eq:t52}) is also asymptotic.
As in the $4_1$ case, the exact Borel radius of convergence and the exact Stokes constant can each be deduced numerically.
Also, the fact that the first sub-leading corrections to the leading large-order growth are exponential rather than power-law tells us that the associated leading Borel singularity is a pole. These numerical observations are confirmed analytically below (see Section \ref{sec:thooft-exact}).

From the large order growth (\ref{eq:52thooft-leading}) we estimate
the leading Chern-Simons invariant to be
\begin{eqnarray}
\left(\mp \frac{1}{2}\right)\frac{1}{4 \pi ^2}{\log ^2\left(\frac{\sqrt{7}+i}{2 \sqrt{2}}\right)} = \pm 0.00165389...
\label{eq:52thooft-leading-radius}
\end{eqnarray}
However, recall that unlike the $4_1$ knot, the $5_2$ knot is {\it not} amphichiral, so the radii of convergence for the $\mp \frac{1}{2}$ surgeries are different for the $5_2$ case. As discussed in Section \ref{sec:cs-alexander}, after including the further large $r$ corrections (\ref{eq: Shifted CS value 52}), the symmetry in  \eqref{eq:52thooft-leading-radius} is broken, yielding values in \eqref{eq:cs-comp}, which are in close agreement with the exact numerical values in Tables \ref{tab: 1/2 surg 52} and \ref{tab: -1/2 surg 52}.
    
\subsection{Exact Borel analysis of the small-$t$  double-scaling limit}
\label{sec:thooft-exact}
 
 The fact that the first corrections to the leading growth in (\ref{eq:41thooft-leading}) and (\ref{eq:52thooft-leading}) are exponential, rather than power-law, suggests defining a modified Borel transform that divides out the exact leading large-order factorial growth: 
 \begin{eqnarray}
   B(\zeta):=\sum_{n=0}^\infty \frac{b_n\, \zeta^n}{\Gamma\left(n+\frac{3}{2}\right)}
  \end{eqnarray}
  Here we divide by $\Gamma\left(n+\frac{3}{2}\right)$ instead of by $\Gamma\left(n+1\right)$.
   The formal expansion for $T(t)$ is recovered by the formal Laplace integral (note the extra $\sqrt{\zeta}$ factor)
  \begin{eqnarray}
   T(t)=\frac{1}{t^{3/2}} \int_0^\infty d\zeta \, e^{-\zeta/t}\, \sqrt{\zeta}\, B(\zeta)
   \end{eqnarray}
   We will be more precise about subtleties concerning the contour of the Borel $\zeta$ integration in Sections \ref{sec:41-thooft-exact} and \ref{sec:52-thooft-exact} below.
  
In particular, we have the following integral
    \begin{equation}
    \widehat Z \big( S^3_{-1/r} (K) \big)
    = \int_{|x|=1} \frac{dx}{2\pi i x} (x^{\frac{1}{2r}} - x^{-\frac{1}{2r}}) F_K (q,x)
    \sum_{n \in \mathbb{Z}} q^{r n^2} x^n
    \label{smallsurgery}
    \end{equation}
    The semiclassical ($\hbar \to 0$) limit of the integrand is $e^{\frac{1}{\hbar} \mathcal{W} (x) + \ldots}$, where $\mathcal{W} (x)$ is the twisted superpotential of the 3d $\mathcal{N}=2$ theory. In this limit, we also have
    \begin{equation}
    F_K (x,q) \; = \; \frac{1}{2} \sum_{{m \ge 1 \atop \text{odd}}} f_m (q) \cdot (x^{\frac{m}{2}} - x^{- \frac{m}{2}}) ~~\longrightarrow~~ \frac{x^{1/2} - x^{-1/2}}{\Delta_K (x)}
    \end{equation}
    Here $\Delta_K(x)$ is the Alexander polynomial for the associated knot. This is the first term in the MRR expansion (see Appendix \ref{sec:MMR}). Let us illustrate this expansion first by specializing to the $4_1$ and $5_2$ knots.
    
\subsubsection{Exact Borel Analysis of the double-scaling Limit: $4_1$ Knot}
\label{sec:41-thooft-exact}

Evaluating the integral \eqref{smallsurgery} for the figure-8 knot $K = {\bf 4_1}$, we get
    \begin{equation}
    \widehat Z \big( S^3_{-1/r} (K) \big)
    \; = \; \sum_{m=1}^{\infty}  q^{\frac{(mr-1)^2}{4r}} (q^m-1) f_m (q)
    \end{equation}
    In  the conventional normalization we divide by $(q-1)$ and then take the double-scaling limit:
    \begin{equation}
    \hbar \to 0
    \,, \qquad
    r \to \infty
    \,, \qquad
    t := r \hbar = \text{fixed}
    \end{equation}
This leads to the following formal expansion:
    \begin{equation}
    G(t) \; = \; \sum_{m=1}^{\infty}  m e^{t \frac{m^2}{4}} f_m \; = \; \sum_{n=0}^{\infty} b_n \, t^n
    \label{bviaf}
    \end{equation}
    where $f_m$ refers to $f_m (1)$, which for the $4_1$ knot are all finite. This is basically an Eichler integral (``half-derivative'') of the Laplace transform applied to $\frac{x^{1/2} - x^{-1/2}}{\Delta_K (x)}$, where the Alexander polynomial of the $4_1$ knot is
    \begin{eqnarray}
    \Delta_{4_1}(x)=-x-x^{-1}+3
    \label{eq:D41}
    \end{eqnarray}
Therefore:    
    \begin{equation}
    \sum_{m=1}^{\infty}  e^{t \frac{m^2}{4}} f_m \; = \; \int_{|x|=1} \frac{dx}{2\pi i x} \frac{x^{1/2} - x^{-1/2}}{(-x- x^{-1} + 3)} \sum_{n \in \mathbb{Z}} e^{t n^2} x^n
    \end{equation}

Using \eqref{bviaf}, we can express the new expansion coefficients $b_n$ via $f_m$. For example, the limit $t \to 0$ is given by the logarithmic derivative of $\frac{x^{1/2} - x^{-1/2}}{- x^{-1} + 3 - x}$ at $x=1$:
    $$
    b_0 = 1
    $$
    The next term $b_1$ is equal to the third logarithmic derivative of $\frac{x^{1/2} - x^{-1/2}}{- x^{-1} + 3 - x}$ at $x=1$:
    $$
    b_1 = \frac{25}{4}
    $$
and so on. Compare with \eqref{eq:thooft-41-coeffics}. For general $n$, we have
    \begin{equation}
    b_n \; = \; \frac{1}{n!} \frac{d^{2n+1}}{d (\log x)^{2n+1}} \left( \frac{x^{1/2} - x^{-1/2}}{- x^{-1} + 3 - x} \right) \Big|_{x=1}
    \end{equation}
    \begin{figure}[h]
   \centerline{\includegraphics[scale=.75]{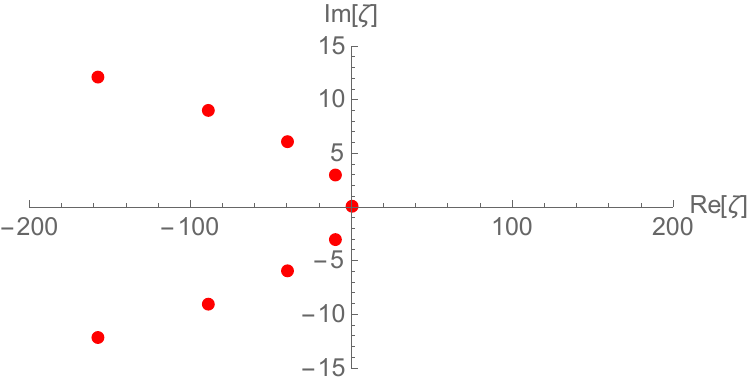}}
   \centerline{\includegraphics[scale=.75]{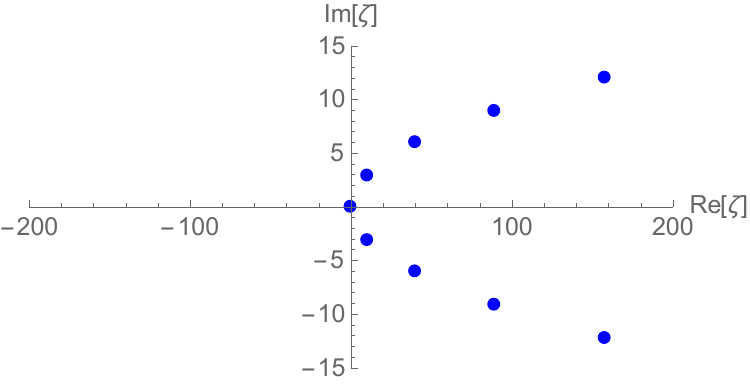}}
   \caption{Poles of the small $t$ double-scaling limit Borel transform for the $(-/+)\frac{1}{r}$ surgery (upper/lower plots) for $r\to\infty$ in the $4_1$ knot case. See Eq. (\ref{eq:41thooft-poles}). Note that the dominant pole lies on the real axis, and is positive for $-$ surgery and negative for $+$ surgery: See Eq. (\ref{eq:thooft41-leading-pole}). Contrast with the $5_2$ case in Figure \ref{fig:52-modified-borel-poles}.}
   \label{fig:41-modified-borel-poles}
   \end{figure}

Thus, in order to obtain the coefficients $b_n$ we need to evaluate at $t=0$ the $(2n+1)$-th derivative of
    \begin{equation}
    \frac{e^{s/2} - e^{-s/2}}{- e^{s} + 3 - e^{-s}} = s+ \frac{25}{24} s^3 + \frac{2161}{1920} s^5 + \ldots
    \end{equation}
    Therefore, 
    the closed-form Borel transforms are (compare with (\ref{eq:thooft-41-coeffics})):
    \begin{eqnarray}
    \frac{\sqrt{\pi}}{2}\, B^{(-)}_{4_1}(\zeta)&=&\frac{ \sinh(\sqrt{\zeta})}{\sqrt{\zeta} \left( 3-2\cosh (2\sqrt{\zeta})\right)}  \\
    &=&1+\frac{25 \zeta }{6}+\frac{2161 \zeta ^2}{120}+\frac{78389 \zeta^3}{1008}+\frac{121866721 \zeta ^4}{362880} +\dots
   \label{eq:thooft-41m-borel}\\
    \frac{\sqrt{\pi}}{2}\, B^{(+)}_{4_1}(\zeta)&=&\frac{ \sin(\sqrt{\zeta})}{\sqrt{\zeta} \left( 3-2\cos (2\sqrt{\zeta})\right)}\\
    &=&1-\frac{25 \zeta }{6}+\frac{2161 \zeta ^2}{120}-\frac{78389 \zeta^3}{1008}+\frac{121866721 \zeta ^4}{362880} -\dots
    \label{eq:thooft-41p-borel}
    \end{eqnarray}
    The leading Borel singularities are therefore at
    \begin{eqnarray}
    \zeta^{(\mp)} = \begin{cases}\left[\frac{1}{2} {\rm arccosh}\left(\frac{3}{2}\right)\right]^2 = \left[\log\left(\frac{1+\sqrt{5}}{2}\right)\right]^2 = 0.231565...
     \cr\cr
    \left[\frac{1}{2} {\rm arccos}\left(\frac{3}{2}\right)\right]^2 = -\left[\log\left(\frac{1+\sqrt{5}}{2}\right)\right]^2 = - 0.231565...
    \end{cases}
    \label{eq:thooft41-leading-pole}
    \end{eqnarray}
    Dividing by the normalization factor $\frac{1}{4\pi^2}$ and multiplying by the surgery factor $\frac{1}{2}$ we obtain $\pm 0.0029328$, in agreement with the numerical result in Section \ref{sec:41-perturbative-large-r}. This value is also very close to the leading Chern-Simons values for the $\pm \frac{1}{2}$ surgeries studied in Sections \ref{sec:41-Apoly} and \ref{sec:41borel}, in agreement with the small surgery expansion result (\ref{eq: Universal Approximation}).

The Borel transform functions in (\ref{eq:thooft-41m-borel}) and (\ref{eq:thooft-41p-borel}) are meromorphic functions of the Borel variable $\zeta$, with simple poles:
     \begin{eqnarray}
 \zeta_k^{(\mp)}&=&\pm \left[\ln\left(\frac{1+\sqrt{5}}{2}\right)+ i\, \pi\, k\right]^2
 \quad, \quad k = \dots, -3, -2, -1, 0, 1, 2, 3, \dots 
 \nonumber\\
 &&= 
 \mp \left\{ \dots, 88.5949\, +9.07063 i,39.2469\, +6.04709 i,9.63804\, +3.02354 i,-0.231565, \right.  \nonumber\\
 && \left. 9.63804\,
   -3.02354 i,39.2469\, -6.04709 i,88.5949\, -9.07063 i\, \dots \right\} 
 \label{eq:41thooft-poles}
 \end{eqnarray}
See Figure \ref{fig:41-modified-borel-poles}. 
Note the parabolic shape of the complex poles. This is analogous to the complex Borel poles, associated with complex geodesics, for the short-time asymptotic expansion of the Laplacian heat kernel on the two dimensional hyperbolic plane $\mathbb H^2$ \cite{McK72,Grig98,Dun21}. We comment further on this analogy below.
 
 Normalizing by the leading singularity we find the first subleading Borel singularities at:  
   \begin{eqnarray}
 \pm  \frac{\left[\ln\left(\frac{1+\sqrt{5}}{2}\right)\pm i\pi \right]^{2}}{\left[\ln\left(\frac{1+\sqrt{5}}{2}\right)\right]^{2}}=\mp\left(41.621346267 \pm  13.05700521 \, i\right)
   \end{eqnarray}
As noted previously, these are very close to the first subleading normalized Chern-Simons invariants for $\mp \frac{1}{2}$ surgery, as for $\alpha = 4$ and $5$ in Table~\ref{tab: 4_1 knot}.

\subsubsection{Exact Borel Analysis of the double-scaling Limit: $5_2$ Knot}
   \label{sec:52-thooft-exact}     
        A similar analysis for the double-scaling limit of the $5_2$ knot case relies on the Alexander polynomial of the $5_2$ knot:
 \begin{eqnarray}
 \Delta_{5_2}(x)= 2\left(x+x^{-1}\right)-3
 \label{eq:alexander-52}
 \end{eqnarray}
 Then an analogous argument leads to the generating function
for the expansion coefficients in (\ref{eq:thooft52-coeffics}): 
    \begin{equation}
    \frac{e^{s/2} - e^{-s/2}}{2e^{s} +2 e^{-s}-3} = s- \frac{47}{24} s^3 + \frac{7201}{1920} s^5  -\frac{2316047}{322560} s^7 + 
        \ldots
    \end{equation}
We therefore deduce the double-scaling limit  Borel transform functions (compare with (\ref{eq:thooft52-coeffics})):
      \begin{eqnarray}
   \frac{\sqrt{\pi}}{2}\, B^{(-)}_{5_2}(\zeta)&=&\frac{ \sinh(\sqrt{\zeta})}{\sqrt{\zeta} \left( -3+4\cosh (2\sqrt{\zeta})\right)}
     \label{eq:thooft-52m-borel} \\
    &=&
    1-\frac{47 \zeta}{6}+\frac{7201 \zeta ^2}{120}
    -\frac{2316047 \zeta ^3}{5040}+
    \frac{1276975681 \zeta ^4}{362880}-\dots
   \nonumber \\
   \frac{\sqrt{\pi}}{2}\,  B^{(+)}_{5_2}(\zeta)&=&\frac{ \sin(\sqrt{\zeta})}{\sqrt{\zeta} \left( -3+4\cos (2\sqrt{\zeta})\right)} \label{eq:thooft-52p-borel}\\
    &=&
    1+\frac{47 \zeta}{6}+\frac{7201 \zeta ^2}{120}
    +\frac{2316047 \zeta ^3}{5040}+
    \frac{1276975681 \zeta ^4}{362880}+\dots
  \nonumber
    \end{eqnarray}
    Note once again the opposite sign-alternation pattern for $\mp$ surgery compared to the $4_1$ knot case in \eqref{eq:thooft-41m-borel}-\eqref{eq:thooft-41p-borel}. 
        The leading Borel singularities are therefore at
    \begin{eqnarray}
    \zeta^{(\mp)} = \begin{cases}\left[\frac{1}{2} {\rm arccosh}\left(\frac{3}{4}\right)\right]^2 = \left[\log\left(\frac{\sqrt{7}+i}{2\sqrt{2}}\right)\right]^2 = -0.130586 ...
     \cr\cr
    \left[\frac{1}{2} {\rm arccos}\left(\frac{3}{4}\right)\right]^2 = - \left[\log\left(\frac{\sqrt{7}+i}{2\sqrt{2}}\right)\right]^2 = 0.130586 ...
    \end{cases}
    \label{eq:thooft-poles-52-leading}
    \end{eqnarray}
 \begin{figure}[h]
   \centerline{\includegraphics[scale=.75]{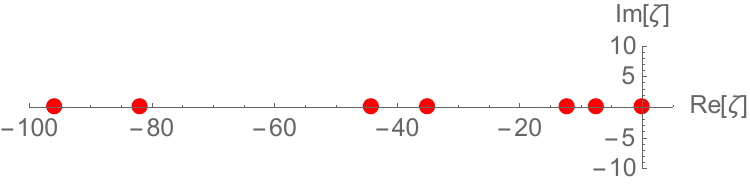}}
   \centerline{\includegraphics[scale=.75]{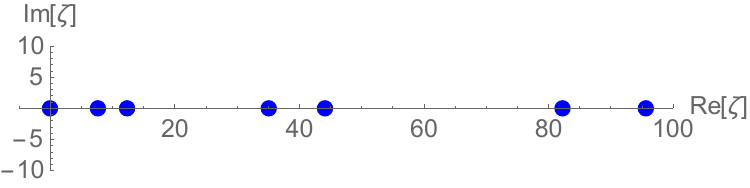}}
   \caption{Poles of the small $t$ double-scaling limit Borel transform for the $(-/+)\frac{1}{r}$ surgery (upper/lower plots) for $r\to\infty$ in the $5_2$ knot case. See Eq. (\ref{eq:thooft-52-poles}). Note that all poles, including the dominant pole, are negative/positive for the $-/+$ surgery. Contrast with the $4_1$ case in Figure \ref{fig:41-modified-borel-poles}. }
   \label{fig:52-modified-borel-poles}
   \end{figure}

The double-scaling limit  Borel transform functions in (\ref{eq:thooft-52m-borel})-(\ref{eq:thooft-52p-borel}) are meromorphic functions of the Borel variable $\zeta$, with simple poles at: 
     \begin{eqnarray}
 \zeta_k^{(\mp)}& =&\pm\left[\ln\left(\frac{\sqrt{7}+i}{2\sqrt{2}}\right)+ i\, \pi\, k\right]^2
 \qquad, \quad k = \dots, -3, -2, -1, 0, 1, 2, 3, \dots 
\\
  & = &\pm \{-0.130586198,-7.72965400,-12.27072720,-35.0679306,-44.1500770,\dots\}
 \label{eq:thooft-52-poles}
 \end{eqnarray}
 In contrast to the $4_1$ knot case in the previous section, for the $5_2$ knot case these double-scaling limit Borel poles are all real. Furthermore, apart from the leading pole, each pole has degeneracy 2, corresponding to the two possible signs for $k$ in (\ref{eq:thooft-52-poles}). For the negative surgery double-scaling limit they all lie on the negative Borel $\zeta$ axis, while for the positive surgery double-scaling limit they all lie on the positive Borel $\zeta$ axis.  See Figure \ref{fig:52-modified-borel-poles}. Contrast with the $4_1$ case shown in Figure \ref{fig:41-modified-borel-poles}.
 
 If we normalize these poles by dividing by the magnitude of the leading singularity at $\zeta_0=\pm\left[\ln\left(\frac{\sqrt{7}+i}{2\sqrt{2}}\right)\right]^2\approx \mp 0.130586198$, then we obtain a picture of the relative distances from the origin of the various poles:
 \begin{eqnarray*}
    \frac{\zeta_k^{(\mp)}}{|\zeta_0|}=\pm \{1.000000000,59.191967,93.966494,268.54240,338.09145,629.05129, ...\}
 \end{eqnarray*}
 This shows that the subleading poles are very far from the origin compared to the distance of the leading pole from the origin. Hence the corresponding subleading effects will be strongly suppressed. This pattern is consistent with what we observed in Section \ref{sec:52borel} for the first two subleading Borel singularities  in the Borel plane for $\mp\frac{1}{2}$ surgery in the $5_2$ knot case.

\subsection{Double-scaling limit for hyperbolic twist knots}
\label{sec:thooft-twist}

Hyperbolic twist knots form a class of knots with Alexander polynomials of a similar form compared to those of the $4_1$ and $5_2$ knots. We denote the twist knots as $K_N$, labelled by an integer $N$, taking positive and negative values. Except for $N=0$ and $N=1$, the twist knots are hyperbolic.
The Alexander polynomials $\Delta_N(x)$ are (see also Table \ref{tab:twist-knots} for $\Delta_N(x^2)$)\footnote{In this Section it is more convenient to use a different labelling convention for the twist knots compared to those in Table \ref{tab:twist knots}. This of course does not affect any of the results.} :
\begin{eqnarray}
\Delta_N(x)=-N\left(x+\frac{1}{x}\right)+(2N+1) 
\label{eq:twist-alexander}
\end{eqnarray}
   	\begin{table}
				\begin{tabular}{|c|c|c|c|c|}
	                \hline
					$N$ & Alexander polynomial $\Delta_N(x^2)$ &  Twist Knot  & hyperbolic & $x_\pm(N$) \\				 \hline
					$0$ & $1$ & $0_1$ & no & -- \\ \hline 
					$1$ & $-1\left(x^2+x^{-2}\right)+3$ & $4_1$ & yes & $\frac{1}{2}(\sqrt{5}\pm 1)$ \\  \hline
					$2$ & $-2\left(x^2+x^{-2}\right)+5$ & $6_1$  & yes & $\frac{1}{2\sqrt{2}}(3\pm 1)$ \\ \hline
					$3$ & $-3\left(x^2+x^{-2}\right)+7$ & $8_1$  & yes & $\frac{1}{2\sqrt{3}}(\sqrt{13}\pm 1)$ \\ \hline
					$4$ & $-4\left(x^2+x^{-2}\right)+9$ & $10_1$  & yes & $\frac{1}{4}(\sqrt{17}\pm 1)$ \\ \hline
					$\vdots$ & $\vdots$ & $\vdots$ & yes & $\vdots$ \\ \hline\hline
					$-1$ & $1\left(x^2+x^{-2}\right)-1$ & $3_1$  & no & $\frac{1}{2}(\sqrt{3}\pm i)$ \\ \hline
					$-2$ & $2\left(x^2+x^{-2}\right)-3$ & $5_2$  & yes & $\frac{1}{2\sqrt{2}}(\sqrt{7}\pm i)$ \\ \hline
					$-3$ & $3\left(x^2+x^{-2}\right)-5$ & $7_2$  & yes & $\frac{1}{2\sqrt{3}}(\sqrt{11}\pm i)$ \\ \hline
					$-4$ & $4\left(x^2+x^{-2}\right)-7$ & $9_2$  & yes &  $\frac{1}{4}(\sqrt{15}\pm i)$\\ \hline
					$-5$ & $5\left(x^2+x^{-2}\right)-9$ & $11_2$  & yes & $\frac{1}{2\sqrt{5}}(\sqrt{19}\pm i)$ \\ \hline
					$\vdots$ & $\vdots$ & $\vdots$  & yes & $\vdots$ \\ \hline
				\end{tabular}
    \vspace{3mm}
				\caption{Alexander polynomials $\Delta_N(x^2)$ for the twist knots, together with the roots $x_\pm(N)$, as in (\ref{eq:alexander-roots}). Except for $N=0$ and $N=-1$, the twist knots are hyperbolic, and they naturally generalize the $4_1$ and $5_2$ knots studied in earlier sections of this paper.}
				\label{tab:twist-knots}
	\end{table}
\noindent
Taking $N=1$ we obtain the Alexander polynomial of the $4_1$ knot, while  $N=-2$ gives that of the $5_2$ knot.
Therefore, since the Alexander polynomial governs the double-scaling limit, we expect this limit for the manifolds $S_{\pm \frac{1}{r}}^3(K_N)$ to be very similar to the cases of $4_1$ (for positive values of $N$) and $5_2$ (for negative values of $N$), as  discussed in the previous sections. The poles of the corresponding Borel transform functions are given by the roots of $\Delta_N(x^2)$:
\begin{eqnarray}
\Delta_N(x^2)=0\qquad \Rightarrow\qquad x_\pm(N) := \frac{\sqrt{4N+1}\pm 1}{2\sqrt{N}}
\label{eq:alexander-roots}
\end{eqnarray}
{\bf Comments:}
\begin{enumerate}
    \item Observe that in all cases $x_+(N) x_-(N)=1$.
    \item For $N\geq 0$ the roots $x_\pm(N)$ are real, while for $N<0$ the roots $x_\pm(N)$ form a complex conjugate pair of pure phases: see Table \ref{tab:twist-knots}.
    \item We note the identities:
\begin{eqnarray} 
\frac{1}{2}\left(x_\pm(N)^2+x_\pm(N)^{-2}\right)&=&\frac{2N+1}{2N} \\
\left(x_\pm(N)+x_\pm(N)^{-1}\right)&=&\sqrt{\frac{4N+1}{N}}
\label{eq:x-identity}
\end{eqnarray}. 
\end{enumerate}
 
\subsubsection{Double-Scaling Limit Expansions for Hyperbolic Twist Knots at Small t}
\label{sec:thooft-twist-weak}

In this Section we discuss the small $t$ double-scaling limit  expansions for hyperbolic twist knots. We distinguish between the negative and positive surgery cases, $\mp \frac{1}{r}$, in the large $r$ double-scaling limit:
\begin{eqnarray}
T_{N}^{(-)}(t)= \frac{2}{\sqrt{\pi}\, t^{3/2}} \int_0^\infty d\zeta \, e^{-\zeta/t}\left(\frac{\sinh\left(\sqrt{\zeta}\right)}{-2N\cosh\left(2\sqrt{\zeta}\right)+(2N+1)}\right)
\label{eq:exact-thooft-minus}
\end{eqnarray}
\begin{eqnarray}
T_{N}^{(+)}(t)=\frac{2}{\sqrt{\pi}\, t^{3/2}} \int_0^\infty d\zeta \, e^{-\zeta/t}\left(\frac{\sin\left(\sqrt{\zeta}\right)}{-2N\cos\left(2\sqrt{\zeta}\right)+(2N+1)}\right)
\label{eq:exact-thooft-plus}
\end{eqnarray}
Note that these integrals are well defined for $-$ surgery for negative knots ($N<0$) and for $+$ surgery for positive knots ($N>0$), but require analytic continuation and contour deformation for $-$ surgery with positive twist knots ($N>0$), and for $+$ surgery with negative twist knots ($N<0$). These Borel integrals lie in the Mordell-Borel class so the analytic continuation methods discussed in Section \ref{sec:other-side} apply.
We also note the curious feature that in the $N\to \infty$ limit  $T_{N}^{(-)}(t)$ reduces (up to an overall factor) to the diagonal heat kernel trace on the 2 dimensional hyperbolic manifold $\mathbb H^2$ \cite{Grig98,Dun21}:
  \begin{eqnarray}
 K(t)= \frac{e^{-t/4}}{2(\pi t)^{3/2}} \int_0^\infty d\zeta
 \, e^{-\zeta/(4t)}\left(\frac{1}{\sinh(\sqrt{\zeta})}\right) 
 \label{eq:heat-short}
    \end{eqnarray}
Correspondingly, in the large $N$ limit, the Borel singularities of $T_{N}^{(-)}(t)$ in (\ref{eq:exact-thooft-minus}) tend towards the common values of $-n^2\pi^2$, for $n=1, 2, 3, ...$ on the negative real axis, as shown  in Figure \ref{fig:borel-general}. Compare with Figures \ref{fig:41-modified-borel-poles} and \ref{fig:52-modified-borel-poles} for $-\frac{1}{2}$ surgery on the $4_1$ and $5_2$ knots, the hyperbolic twist knots for $N=1$ and $N=-2$, respectively. 
 \begin{figure}[htb]
 \centerline{\includegraphics[scale=0.7]{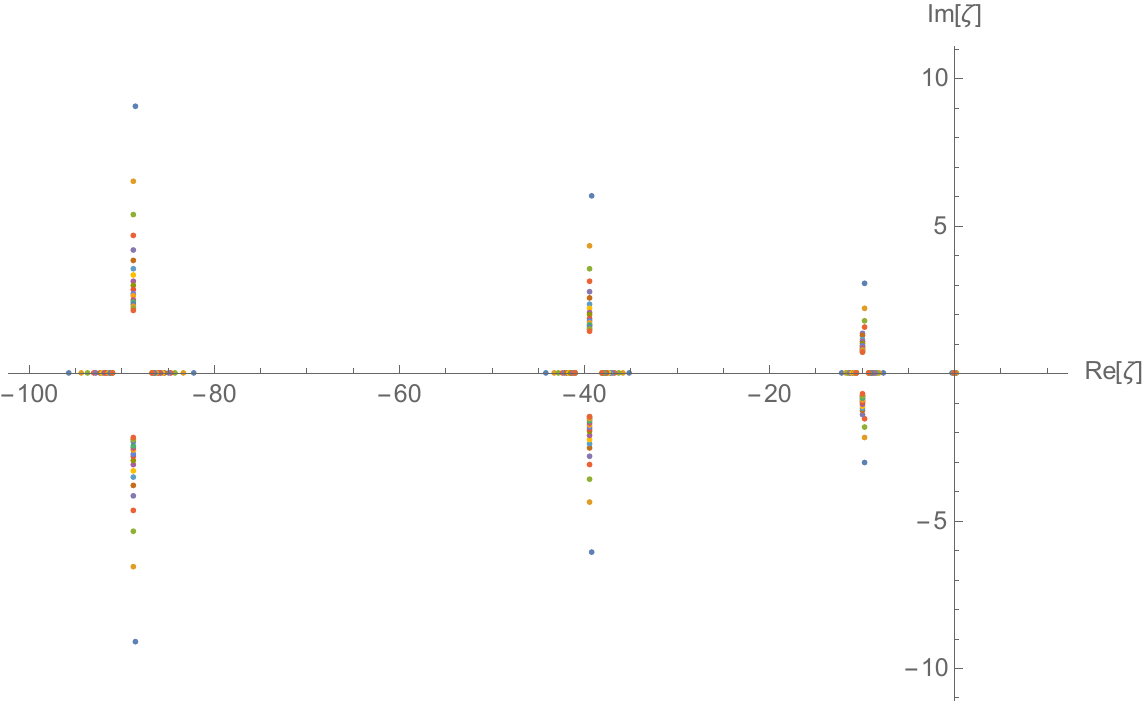}}
 \caption{The Borel poles for the hyperbolic twist knots, for various values of the twist knot label $N$, with different colored points corresponding to different $N$ values. Those for positive $N$ are complex conjugate pairs forming a parabolic pattern, as in Figure \ref{fig:41-modified-borel-poles}, while those for negative $N$ all lie on the real axis, as in Figure \ref{fig:52-modified-borel-poles}.
 As $N\to\infty$ the poles coalesce to real values, $-(n \pi)^2$, the Borel singularities for the heat kernel trace (\ref{eq:heat-short}) for the 2 dimensional hyperbolic manifold $\mathbb H^2$.}
 \label{fig:borel-general}
 \end{figure}

The formal small $t$ expansions in the double-scaling limit are generated from the expansion of the Borel transform function about $\zeta=0$. This expansion follows from the following straightforward Lemmas: 
\begin{mylem}
\label{lem:twist-borel}
\begin{eqnarray}
\frac{\sinh(\sqrt{\zeta})}{\cosh(2\sqrt{\zeta})-\frac{1}{2}(x^2+x^{-2})}
&=&\frac{1}{x+x^{-1}}\left(\frac{x\, e^{-\sqrt{\zeta}}}{1-x^2\, e^{-2 \sqrt{\zeta}}} 
+ \frac{x\, e^{\sqrt{\zeta}}}{x^2\, e^{2\sqrt{\zeta}}-1}\right) \\
&=&
\sum_{n=0}^\infty \frac{\zeta^{n+1/2}}{\Gamma(2n+2)}\left(\frac{{\rm Li}_{-(2n+1)}(x)-{\rm Li}_{-(2n+1)}(-x)}{(x+x^{-1})}\right)
\label{eq:sinh-expansion}
\end{eqnarray}
Similarly,
\begin{eqnarray}
\frac{\sin(\sqrt{\zeta})}{\cos(2\sqrt{\zeta})-\frac{1}{2}(x^2+x^{-2})}
&=&
\sum_{n=0}^\infty \frac{(-1)^n \zeta^{n+1/2}}{\Gamma(2n+2)}\left(\frac{{\rm Li}_{-(2n+1)}(x)-{\rm Li}_{-(2n+1)}(-x)}{(x+x^{-1})}\right)
\label{eq:sine-expansion}
\end{eqnarray}
The only difference between these series expansions is the alternating sign factor $(-1)^n$.
\end{mylem}

This lemma leads directly to the following theorem giving closed-form expressions for the formal small $t$ expansions in the double-scaling limit:
\begin{mythm}
\label{thm:smallttwist}
The small $t$ double-scaling limit expansion of the Chern-Simons partition function for $\pm$ surgery on the knot $K_N$ 
has the following formal perturbative asymptotic expansion as $t\to 0^+$
\begin{eqnarray}
  T_{N}^{(\pm)}(t)\sim
  {\rm sign}(N)\sum_{n=0}^\infty \frac{\left(\mp \frac{t}{4}\right)^n}{n!}\left[\frac{{\rm Li}_{-(2n+1)}(x_+(N)) -{\rm Li}_{-(2n+1)}(-x_+(N))}{2\sqrt{N(4N+1)}}\right]
  \label{eq:weak-thooft-twist}
  \end{eqnarray}
\end{mythm}

{\bf Comments:}
\begin{enumerate}
    \item Note that the RHS of (\ref{eq:weak-thooft-twist}) may also be written in terms of polylogarithms of $\pm x_-(N)$, since ${\rm Li}_{-(2n+1)}(x_+(N)) -{\rm Li}_{-(2n+1)}(-x_+(N))={\rm Li}_{-(2n+1)}(x_-(N)) -{\rm Li}_{-(2n+1)}(-x_-(N))$. This follows from an identity for polylogarithm functions, recalling that $x_-(N)=1/x_+(N)$.
    \item The formal double-scaling limit small $t$ series in (\ref{eq:weak-thooft-twist}) have factorially growing coefficients that alternate in sign for positive surgery and $N>0$, or for negative surgery and $N<0$. When the sign of the surgery and the sign of $N$ differ, the coefficients in (\ref{eq:weak-thooft-twist}) all have the same sign.
    \item We also note the curious fact that the polylogarithmic combinations appearing in these asymptotic expansions are integer valued. This is despite the fact that for $N>0$ the $x_\pm(N)$ are irrational, and for $N<0$ they are pure phases which are not rational roots of unity. It is tempting to interpret these integers as (close cousins of) the BPS state count in 3d-3d correspondence that we briefly reviewed in Section~\ref{sec:3d3dlimit}. As explained there, we expect the relevant theory to be $T[S^3_0 (K)]$ and the count to be with respect to the supercharge $Q$ of a partial topological twist of the 3d $\mathcal{N}=2$ theory. We will return to this interpretation below, after Lemma~\ref{lem:twist-borel-strong} and before Lemma~\ref{lem:twist-borel-strongB}.
\end{enumerate}

These properties follow from the following Lemma:
\begin{mylem}
\label{lem:lerch}
\begin{enumerate}
    \item 
If $x=x_\pm(N)$ is a solution (\ref{eq:alexander-roots}) of $\Delta_N(x^2)=0$, where $\Delta_N(x)$ is the Alexander polynomial (\ref{eq:twist-alexander}) for the twist knot $K_N$, then for all integers $n\geq 0$ the following combination [which appears in the coefficients of the double-scaling limit small $t$ expansion in (\ref{eq:weak-thooft-twist})] is an integer:
\begin{eqnarray}
\left(\frac{{\rm Li}_{-(2n+1)}(x_\pm(N))-{\rm Li}_{-(2n+1)}(-x_\pm(N))}{2\sqrt{N(4N+1)}}\right) &=& \frac{x_\pm(N)\, 4^{n+1} \Phi\left(x_\pm(N)^2,-(2n+1), \frac{1}{2}\right)}{2\sqrt{N(4N+1)}}
\nonumber\\
&\in& \mathbb Z 
\label{eq:weak-coefficients}
\end{eqnarray}
Here $\Phi(z, s, v)$ is the Lerch function. 
\item
The leading large order behavior of the expansion coefficients in (\ref{eq:weak-thooft-twist}) is given by
\begin{eqnarray}
\frac{1}{4^n n!}\left(\frac{{\rm Li}_{-(2n+1)}(x_\pm(N))-{\rm Li}_{-(2n+1)}(-x_\pm(N))}{2\sqrt{N(4N+1)}}\right)&&\nonumber\\
&&\hskip -5cm \sim 
\frac{1}{\sqrt{\pi N(4N+1)}} \frac{\Gamma\left(n+\frac{3}{2}\right)}{\left(\log(x_+(N))\right)^{2n+2}}\quad, \quad n\to\infty
\label{eq:leading-large-weak-thooft}
\end{eqnarray}
\end{enumerate}
\end{mylem}

This Lemma follows from an identity relating the PolyLog function to the Lerch function, combined with the remarkable Lerch duality formula \cite{Erd}:
\begin{eqnarray}
\Phi(z, s, v)&=& i\, z^{-v}\frac{\Gamma(1-s)}{(2\pi)^{1-s}}\left\{
e^{-i\pi \frac{s}{2}}\Phi\left(e^{-2\pi i v}, 1-s, \frac{\log(z)}{2\pi i}\right) \right.
\notag \\
&&\qquad \left. -e^{i\pi\left(\frac{s}{2}+2v\right)}\Phi\left(e^{2\pi i v}, 1-s, 1-\frac{\log(z)}{2\pi i}\right)\right\}
\label{eq:lerch-dual}
\end{eqnarray}
As a consequence of this Lerch duality formula, we can express the coefficients in terms of Hurwitz zeta functions, whose asymptotic behavior is well known. We have
\begin{eqnarray}
&& z\, \Phi\left(z^2,-(2n+1), \frac{1}{2}\right) = \nonumber\\
&&
\frac{(-1)^{n+1} \Gamma(2n+2)}{(2\pi)^{2n+2}}\left[\Phi\left(-1,2n+2, \frac{\log(z)}{\pi i}\right) - \Phi\left(-1,2n+2,1- \frac{\log(z)}{\pi i}\right)\right] \\
&=& \frac{(-1)^{n+1} \Gamma(2n+2)}{(4\pi)^{2n+2}}\left[ \zeta\left(2n+2, \frac{\log(z)}{2\pi i}\right) + \zeta\left(2n+2, 1-\frac{\log(z)}{2\pi i}\right) \right.
\nonumber\\
&&\left. -\zeta\left(2n+2,\frac{1}{2}+ \frac{\log(z)}{2\pi i}\right) -\zeta\left(2n+2,\frac{1}{2}- \frac{\log(z)}{2\pi i}\right)\right]
\end{eqnarray}
Evaluating at $z=x_+(N)$, and using the large order behavior of the Hurwitz zeta function, leads to the asymptotic result in part (2) of the Lemma.

{\bf Comments:}
\begin{enumerate}
    \item Note that when $N=1$ we recover from (\ref{eq:leading-large-weak-thooft}) the leading large-order behavior of the small $t$ double-scaling limit expansion for the $4_1$ knot case, see Eq. (\ref{eq:41thooft-leading}), and when $N=-2$ we recover the leading large-order behavior of the small $t$ double-scaling limit expansion for the $5_2$ knot case, see Eq. (\ref{eq:52thooft-leading}).
    \item The sign-alternation pattern of the factorially growing coefficients in the small $t$ expansion (\ref{eq:weak-thooft-twist}) depends both on the sign of the surgery (in the double-scaling limit this means an expansion in powers of $t$ or of $-t$) as well as on the sign of the twist knot label $N$. This is because when $N>0$, the roots $x_\pm(N)$ are real, whereas when $N<0$ they are pure phases. Specifically: the coefficients  alternate in sign when the sign of the surgery and the sign of $N$ are the same, and all have the same sign when the sign of the surgery and the sign of $N$ are opposite. 
    \item It is interesting to note that this sign pattern matches the behavior found in Section~\ref{sec:borel} for the $\mp \frac{1}{2}$ surgery of the $4_1$ knot ($N=1$) and of the $5_2$ knot ($N=-2$), which were not in the double-scaling limit.
\end{enumerate}

\subsubsection{Double-Scaling Limit Expansions for Hyperbolic Twist Knots at Large t}
\label{sec:thooft-twist-strong}

The large $t$ double-scaling limit behavior reveals extra structure. 
As described in Section \ref{sec:other-side}, within this class of Borel integrals, the transformation $t\to \frac{\pi^2}{t}$, which maps small to large $t$, is implemented by Fourier-Poisson transformation. This leads to a rich structure of behaviors in the double-scaling limit, depending on the sign of the surgery and on the sign of $N$, the twist knot parameter.
The large $t$ behavior is governed by the large $\zeta$ behavior of the Borel transform functions. These are described by the following straightforward Lemmas.

\begin{mylem}
\label{lem:twist-borel-strong}
(1) For all $N$, the  $\zeta\to +\infty$ behavior of the following Borel transform function is:
\begin{eqnarray}
\frac{\sinh(\sqrt{\zeta})}{\cosh(2\sqrt{\zeta})-\frac{2N+1}{2N}}
\sim 
\sum_{k=0}^\infty \frac{\left(x_+(N)^{2k+1}+x_+(N)^{-(2k+1)}\right)}{\left(x_+(N)+x_+(N)^{-1}\right)}\, e^{-(2k+1) \sqrt{\zeta}}\,\, , \, \zeta\to+\infty
\label{eq:sinh-expansion-strong}
\end{eqnarray}
where $x_+(N)$ is the root in (\ref{eq:alexander-roots}). For $N>0$ this Borel function has a pole on the positive real $\zeta$ axis.
\\
(2) For both positive and negative $N$, the expansion coefficients in (\ref{eq:sinh-expansion-strong}) are integers divided by powers of $N$:
\begin{eqnarray}
\frac{\left(x_+(N)^{2k+1}+x_+(N)^{-(2k+1)}\right)}{\left(x_+(N)+x_+(N)^{-1}\right)}
&=&\sqrt{\frac{N}{4N+1}} \left[ \left(\frac{\sqrt{4N+1}+1}{2\sqrt{N}}\right)^{2k+1}+\left(\frac{\sqrt{4N+1}-1}{2\sqrt{N}}\right)^{2k+1} \right]
\nonumber\\
&=& \frac{{\rm integer}_k}{N^k} \quad, \quad \forall k\geq 0
\label{largetintegrality}
\end{eqnarray}
These integers are related to multiplicities of linear recurrences with constant coefficients \cite{Beu80}, generalizing the familiar Fibonacci bisection case for $N=1$. 
\end{mylem}
Whereas we expect the integraliy in \eqref{eq:weak-thooft-twist}--\eqref{eq:weak-coefficients} to have an explanation in terms of partial topological twist ($Q$-cohomology) of 3d $\mathcal{N}=2$ theory $T [S^3_0 (K)]$, the integrality in \eqref{largetintegrality} is more tricky. This is due to the fact that the variable $t$ which we keep fixed in the double-scaling limit is produced from two kinds of couplings: $\hbar$ that defines 3d background, and the surgery coefficient that defines 3-manifold $M_3$. These two belong to the different sides of 3d-3d correspondence, making $t$ a natural (finite) parameter in the 6d theory. From this, it is a priori unclear whether the large-$t$ expansion is supposed to have a dual weakly-coupled description in terms of QFT$_d$ for some $d<6$. We hope to shed light on this question in future work.

\begin{mylem}
\label{lem:twist-borel-strongB}
\begin{enumerate}
    \item 
When $N>0$ (i.e., for positive twist knots $K_{N>0}$, for which $x_\pm(N)$ are real, with $x_+(N)>1$: recall (\ref{eq:alexander-roots})), we have the following 
Fourier expansion:
\begin{eqnarray}
\frac{\sin(\sqrt{\zeta})}{\cos(2\sqrt{\zeta})-\frac{2N+1}{2N}}
= - 2\sum_{k=0}^\infty \frac{x_+(N)^{-(2k+1)}}{\left(x_+(N)+x_+(N)^{-1}\right)} \sin\left((2k+1)\sqrt{\zeta}\right)
\label{eq:sine-expansion-strong1}
\end{eqnarray}
\item
When $N<0$ (i.e., for negative twist knots $K_{N<0}$, for which $x_\pm(N)$ is a pure phase:  recall (\ref{eq:alexander-roots})), we have the following 
Fourier expansion: 
\begin{eqnarray}
\frac{\sin(\sqrt{\zeta})}{\cos(2\sqrt{\zeta})-\frac{2N+1}{2N}}
= -\sum_{k=0}^\infty \frac{\left(x_+(N)^{2k+1}+x_+(N)^{-(2k+1)}\right)}{\left(x_+(N)+x_+(N)^{-1}\right)} \sin\left((2k+1)\sqrt{\zeta}\right)
\label{eq:sine-expansion-strong2}
\end{eqnarray}
In this case there are poles on the $\zeta>0$ axis.
\end{enumerate}
\end{mylem}
\begin{mylem}
\label{lem:strong-integrals}
For ${\rm Re}(t)>0$ and $k\geq 0$ we have the following integrals:
\begin{eqnarray}
\int_0^\infty d\zeta \, e^{-\zeta/t}\, e^{-(2k+1)\sqrt{\zeta}}
&=&\frac{(2k+1)}{4} t^{3/2}\, e^{\left(k+\frac{1}{2}\right)^2 t}\,  \Gamma\left(-\frac{1}{2}, \left(k+\frac{1}{2}\right)^2 t\right) 
\label{eq:exp-t}
\\
&\hskip -8cm \sim \hskip -4cm & \hskip -2cm \frac{4}{\sqrt{\pi}\, (2k+1)^2}\sum_{n=0}^\infty \left(\frac{-4}{(2k+1)^2 \, t}\right)^n \Gamma\left(n+\frac{3}{2}\right) \, , \, t\to +\infty
\label{eq:exp-t-2}
\\
\int_0^\infty d\zeta \, e^{-\zeta/t}\, \sin\left((2k+1)\sqrt{\zeta}\right)
&=&\frac{(2k+1)}{2} \sqrt{\pi}\, t^{3/2}\, e^{-\left(k+\frac{1}{2}\right)^2 t}
\label{eq:sine-t}
\end{eqnarray}
\end{mylem}

These Lemmas specify the different large $t$ behaviors of the Chern-Simons partition function in the double-scaling limit. For example, for positive surgery for positive hyperbolic twist knots, (\ref{eq:exact-thooft-plus}),  (\ref{eq:sine-expansion-strong1}) and (\ref{eq:sine-t}) imply a pure transseries exponentially decaying behavior at large $t$:
\begin{eqnarray}
T_{N>0}^{(+)}(t)&=& -\frac{1}{N \sqrt{\pi}\, t^{3/2}} \int_0^\infty d\zeta \, e^{-\zeta/t} \frac{\sin(\sqrt{\zeta})}{\cos(2\sqrt{\zeta})-\frac{2N+1}{2N}}
\\
&=& \frac{1}{\sqrt{N(4N+1)}}\sum_{k=0}^\infty \left(x_-(N)\right)^{2k+1} (2k+1) e^{-\left(k+\frac{1}{2}\right)^2 t}
\quad, \quad t\to+\infty 
\label{eq:pos-pos-large-t}
\end{eqnarray}
This should be contrasted with the asymptotic power-law small $t$ behavior in (\ref{eq:weak-thooft-twist}).

In contrast, for negative surgery for negative hyperbolic twist knots, (\ref{eq:exact-thooft-minus}),  (\ref{eq:sine-expansion-strong2}) and (\ref{eq:exp-t}) imply an asymptotic expansion behavior at large $t$:
\begin{eqnarray}
T_{N<0}^{(-)}(t)&=& -\frac{1}{N \sqrt{\pi}\, t^{3/2}} \int_0^\infty d\zeta \, e^{-\zeta/t} \frac{\sinh(\sqrt{\zeta})}{\cosh(2\sqrt{\zeta})-\frac{2N+1}{2N}}
\\
&\sim & \frac{2}{\sqrt{N(4N+1)}\, \sqrt{\pi}\, t^{3/2}} 
\sum_{n=0}^\infty \left(\frac{-4}{t}\right)^n \Gamma\left(n+\frac{3}{2}\right) \left[ 
{\rm Li}_{2n+2}(x_+(N)) \right.
\nonumber
\\
&&\left. 
-{\rm Li}_{2n+2}(-x_+(N)) +{\rm Li}_{2n+2}(x_-(N))-{\rm Li}_{2n+2}(x_-(N))\right]
\, , \, t\to+\infty 
\label{eq:neg-neg-large-t}
\end{eqnarray}
Compare this with the asymptotic small $t$ behavior in (\ref{eq:weak-thooft-twist}), which can be re-expressed using Lerch duality in a form that makes the $t\to\frac{\pi^2}{t}$ transformation more clear:
\begin{eqnarray}
\hskip 1cm T_{N<0}^{(-)}(t)
&\sim&  -\frac{{\rm sign}(N)}{\sqrt{\pi N(4N+1)}} \sum_{n=0}^\infty \left(- t\right)^n \frac{\Gamma\left(n+\frac{3}{2}\right)}{\left(2\pi\right)^{2n+2}}
\times 
\nonumber
\\
&&
\left[
\zeta\left(2n+2, \frac{\log(x_+(N))}{2\pi i}\right) -\zeta\left(2n+2, \frac{1}{2}-\frac{\log(x_+(N))}{2\pi i}\right) 
 \right.
\label{eq:neg-neg-small-t-b}
\\
&&\left. 
+\zeta\left(2n+2, 1-\frac{\log(x_+(N))}{2\pi i}\right) 
-\zeta\left(2n+2,\frac{1}{2}+ \frac{\log(x_+(N))}{2\pi i}\right)
\right] \, , \, t\to 0^+
\nonumber
\end{eqnarray}
The other cases, for which the sign of the surgery and the twist knot label differ can be analyzed using the analytic continuation methods of Section \ref{sec:other-side}.

It would be interesting to explore the double-scaling limit for more general knots and, in particular, verify whether the curious integrality properties observed in Theorem~\ref{thm:smallttwist} and in Lemma~\ref{lem:twist-borel-strong} hold true for other families of knots. We leave this to future work.

\section*{Acknowledgements}

Thanks to J{\o}rgen Andersen, Tomoyuki Arakawa, Anna Beliakova, Kathrin Bringmann, John Chae, Miranda C. N. Cheng, Francesco Costantino, Daniele Dorigoni, Boris Feigin, Dennis Gaitsgory, Chris Howls, Mrunmay Jagadale, Mikhail Khovanov, Maxim Kontsevich, Piotr Kucharski, Marcos Marino, Eleanor Mcspirit, Antun Milas, Cris Negron, Sunghyuk Park, Davide Passaro, Du Pei, Pavel Putrov, Nicolai Reshetikhin, Raphael Rouquier, Lev Rozansky, David Sauzin, Yan Soibelman, Shoma Sugimoto, Josef Svoboda, Ahn Tran, Cumrun Vafa and Andr\'e Voros for discussions, ideas, help, advice, support and inspiration that have greatly benefited this project.

The work of OC is supported in part by the U.S. National Science Foundation, Division of Mathematical Sciences, Award NSF DMS-2206241.
The work of GD is supported in part by the U.S. Department of Energy, Office of High Energy Physics, Award  DE-SC0010339.
The work of AG and SG is supported in part by a Simons Collaboration Grant on New Structures in Low-Dimensional Topology, by the NSF grant DMS-2245099, and by the U.S. Department of Energy, Office of Science, Office of High Energy Physics, under Award No. DE-SC0011632.

The authors OC, GD and SG would like to thank the Isaac Newton Institute for Mathematical Sciences, Cambridge, for support and hospitality during the programme ``Applicable resurgent asymptotics: towards a universal theory'', where work on this paper was undertaken. This work was supported by EPSRC grant no EP/R014604/1.

\appendix
\section{Approaching $0$-surgeries via $\hbar$ and $1/r$ expansions}
\label{sec:MMR}

In this appendix we extend the analysis of Section \ref{sec:small-cs-values} to obtain the $\frac{1}{r}$ corrections and to extract the asymptotics. 

\subsection{Higher order corrections}

We start with the $\hbar$ expansion \cite{GM} (see also \cite{MM,Roz1}):
    \begin{align*}
        F_K(x, \hbar) & = \frac{(x^{\frac{1}{2}} - x^{-\frac{1}{2}})}{\Delta_K(x)} + \frac{(x^{\frac{1}{2}} - x^{-\frac{1}{2}}) R_1(K; x)}{\Delta_K(x)^{3}}\hbar + \frac{(x^{\frac{1}{2}} - x^{-\frac{1}{2}}) R_2(K; x)}{\Delta_K(x)^{5}}\hbar^2 +  \cdots
        \\ & = \frac{1}{2}\sum_{{m \ge 1 \atop \text{odd}}} (f_{m,0} + f_{m,1} \hbar + f_{m,2}\hbar^2 + \cdots) (x^{\frac{m}{2}} - x^{-\frac{m}{2}})
    \end{align*}
    as well as the usual surgery expansion for surgeries on knot complements
    \begin{equation} \label{eq:r surg formula}
        \widehat Z \big(S^3_{-1/r} (K) \big)
        \; = \; \frac{1}{2(q - 1)}\sum_{m=1}^{\infty}  q^{\frac{(mr-1)^2}{4r}} (q^m-1) f_m (q).
    \end{equation}
    We again want to analyse the double scaling limit
    \begin{equation}
        \hbar \to 0 \,, \qquad r \to \infty \,, \qquad t := r \hbar = \text{fixed}
    \end{equation}
    As $\hbar = \frac{t}{r}$ and so we can expand \eqref{eq:r surg formula} as a series in $t$ and $\frac{1}{r}$ to get
    \begin{equation}
        e^{\frac{t}{4r^2}} \sum_{m = 1}^{\infty} e^{\frac{m^2 t}{4}} \left(\frac{m}{2} f_{m, 0} + \left(-\frac{m f_{m, 0}}{4} + \frac{m}{2} f_{m, 1}\right)\frac{t}{r} + \left(\frac{m f_{m, 0}}{24} - \frac{m f_{m, 1}}{4} + \frac{m}{2} f_{m, 2} + \frac{m^3 f_{m, 0}}{48}\right)\frac{t^2}{r^2} + O\left(\frac{1}{r^3}\right) \right)
    \end{equation}
    We can evaluate these series individually to get a series of the form
    \[
        G_0(K; t) + \frac{1}{r}G_1(K; t) + \frac{1}{r^2}G_2(K; t) + \cdots = \sum_{n = 0}^{\infty} b_{n, 0}t^n + \frac{1}{r}\sum_{n = 0}^{\infty} b_{n, 1}t^n + \frac{1}{r^2}\sum_{n = 0}^{\infty} b_{n, 2}t^n + \cdots
    \]
    We can come up with a quick an easy cookbook to compute these coefficients. The key lemma is the following
    \begin{mylem}
        Let $F$ be a function with expansion
        \[
            F(x) = \frac{1}{2}\sum_{m = 1}^{\infty} f_m (x^{\frac{m}{2}} - x^{-\frac{m}{2}})
        \]
        with $f_m = 0$ for $m$ even then
        \begin{equation} \label{eq: Surgery and Borel}
             \sum_{m} e^{\frac{m^2 t}{4}}\left(\frac{m}{2}\right) f_m = \frac{-1}{\sqrt{\pi t}} B^{-1}F(e^{2\sqrt{t}})
        \end{equation}
         where $B^{-1}$ is the inverse of the usual Borel transform which means that we expand $F(e^{2\sqrt{t}})$ as a series in $t$ and then replace $t^{\frac{n}{2}}$ by $\Gamma(1 + \frac{n}{2})t^{\frac{n}{2}}$.
    \end{mylem}
    \begin{proof}
        Observe that by commuting derivatives and the sum, when $n$ is odd,
        \[
            \sum_{m = 1}^{\infty} \left(\frac{m}{2}\right)^n f_m = \frac{\partial^n F}{\partial \log(x)^n} \Big|_{x = 1}.
        \]
        If we define $\tilde{F}(u) = F(e^u)$ then this is exactly the $2n + 1$ taylor coefficient of $\tilde{F}(u)$ expanded around $u = 0$. This allows us to rewrite the left hand side of Equation \eqref{eq: Surgery and Borel} as
        \begin{align*}
            \sum_{m} e^{\frac{m^2 t}{4}}\left(\frac{m}{2}\right) f_m & = \sum_{n = 0}^{\infty} \frac{t^n}{n!} \sum_{m = 1}^{\infty} \left(\frac{m}{2}\right)^{2n + 1} f_m
            \\ & = \sum_{n = 0}^{\infty} \frac{t^n}{n!} \tilde{F}^{(2n + 1)}(0).
        \end{align*}
        We can similarly expand the right hand side of \eqref{eq: Surgery and Borel} as
        \begin{align*}
            B^{-1}F(e^{2\sqrt{t}}) & = B^{-1} \tilde{F}(2\sqrt{t})
            \\ & = \sum_{n = 0}^{\infty} \frac{2^n t^{\frac{n}{2}} \Gamma(\frac{n}{2} + 1)}{n!} \tilde{F}^{n}(0)
        \end{align*}
        Next, use the symmetry $\tilde{F}(u) = -\tilde{F}(-u)$ to conclude that $\tilde{F}^{n}(0) = 0$ for $n$ even. Additionally, when $n = 2m + 1$ is odd, $\Gamma(m + \frac{3}{2}) = \frac{(2m + 1)!!}{2^{m + 1}} \sqrt{\pi}$. Combining these we find
        \begin{align*}
            B^{-1}F(e^{2\sqrt{t}}) & = \sqrt{\pi}\sum_{m = 0}^{\infty} \frac{2^{m} t^{m + \frac{1}{2}} }{(2n)!!} \tilde{F}^{2m + 1}(0)
            \\ & = \sqrt{t \pi}\sum_{m = 0}^{\infty} \frac{t^{m + \frac{1}{2}}}{(n)!} \tilde{F}^{2m + 1}(0).
        \end{align*}
    \end{proof}
    As a simple corollary from this we also find that
    \[
        \frac{\partial^n}{\partial t^n} G(t) = \sum_{m} e^{\frac{m^2 t}{4}} \left(\frac{m}{2}\right)^{2m + 1} f_m.
    \]
    This analysis shows that these functions $G_i(t)$ are determined entirely by the MMR expansion. Defining
    \[
        \G{R}_n(K; t) = \frac{-1}{\sqrt{\pi t}} B^{-1}\left(\frac{(e^{\sqrt{t}} - e^{-\sqrt{t}})R_n(K; e^{2\sqrt{t}})}{\Delta_K(e^{2\sqrt{t}})^{2n + 1}}\right)
    \]
    we see that\footnote{Note that $R_0(K, x) = 1$}
    \begin{align*}
        G_0(K; t) & = \G{R}_0(K; t) \\
        G_1(K; t) & = \G{R}_1(K; t)t - \frac{1}{2}\G{R}_0(K; t)t \\
        G_2(K; t) & = \G{R}_2(K; t)t^2 - \frac{1}{2}\G{R}_1(K; t)t^2 + \frac{1}{4}\G{R}_0(K; t)t + \frac{1}{12}\G{R}_0(K; t)t^2 + \frac{1}{6} \left(\frac{\partial}{\partial t} \G{R}_0(K; t) \right)t^2
    \end{align*}
    Focusing on the $4_1$ case we have
    \begin{align*}
        R_0(4_1; t) & = 1 \\
        R_1(4_1; t) & = 0 \\
        R_2(4_1; t) & = x^2 - 4x + 5 - 4 x^{-1} + x^{-2}
    \end{align*}
    and so we can compute
    \begin{align*}
        \G{R}_0(4_1; t) & = 1 + \frac{25}{4} t + \frac{2161}{32} t^2 + \frac{391 945}{384} t^3 + O(t^4)\\
        \G{R}_1(4_1; t) & = 0 \\
        \G{R}_2(4_1; t) & = -1 - \frac{121}{4} t - \frac{28081}{32} t^2 - \frac{10 628 521 }{384} t^3 + O(t^4)
    \end{align*}
    which means that
    \begin{align*}
        G_0(4_1; t) & = 1 + \frac{25}{4} t + \frac{2161}{32} t^2 + \frac{391 945}{384} t^3 + O(t^4)\\
        G_1(4_1; t) & = -\frac{t}{2} - \frac{25}{8} t^2 - \frac{2161}{64} t^3 - \frac{391 945}{768} t^4 + O(t^5)\\
        G_2(4_1; t) & = \frac{t}{4} + \frac{27}{16}t^2 + \frac{1237}{128}t^3 - \frac{163 409}{1536}t^4 + O(t^5)
    \end{align*}
    and this exactly matches what we computed previously.
    
\subsection{Asymptotic behaviour}

There is an interesting puzzle presented by the analysis in the previous section. As we have seen elsewhere in this paper, when we look at the perturbative expansion
        \[
            \widehat Z \big(S^3_{-1/r} (K), \hbar\big) = \sum_{n = 0}^{\infty} a_n \hbar^n
        \]
        the coefficients grow as
        \[
            a_n \sim S\frac{\Gamma(n + \frac{3}{2})}{(\text{radius})^n}
        \]
        where $\text{radius}$ is the minimal value of $\frac{CS(\alpha)}{4\pi^2}$ as $\alpha$ ranges over irreducible flat connections. At the same time, the previous section shows that
        \[
            \widehat Z \big(S^3_{-1/r} (K), \hbar\big) = \sum_{n = 0}^{\infty} \frac{1}{r^n} G_n(K; r\hbar)
        \]
        where each $G_n$ is a sum of the $\G{R}_i$'s and their derivatives. The puzzle is that directly from the definition of the $\G{R}_i$, all of their radii of convergence (of their Borel transforms) are equal and given by $\Log(x^*)^2$ for an appropriate root $x^*$ of the Alexander polynomial $\Delta_K(x^2)$. This implies that all finite sums $\sum_{n = 0}^{m} \frac{1}{r^n} G_n(K; r\hbar)$ have this same radius of convergence which would seem to suggest that this should be the radius of convergence of $B\widehat Z \big(S^3_{-1/r} (K), \hbar\big)$ which is certainly not what we observe.
        
        The resolution of this puzzle is subtle, so we go through it with care in the $4_1$ case. First observe that while the radii of convergence of the $\G{R}_i$ are all equal, the asymptotic behaviour of the series is subtly different. We have
        \begin{align*}
            (\G{R}_0)_n & \sim \frac{1}{\sqrt{5\pi} \log(\frac{1 + \sqrt{5}}{2})^2} \frac{\Gamma(n + \frac{3}{2})}{\log(\frac{1 + \sqrt{5}}{2})^{2n}} = S_0 \frac{\Gamma(n + \frac{3}{2})}{d^n}
            \\ (\G{R}_2)_n & \sim \frac{-n^3}{75 \sqrt{\pi} \log(\frac{1 + \sqrt{5}}{2})^5} \frac{\Gamma(n + \frac{3}{2})}{\log(\frac{1 + \sqrt{5}}{2})^{2n}} = -n^3S_2 \frac{\Gamma(n + \frac{3}{2})}{d^n}
        \end{align*}
        Note the appearance of the $n^3$. This appears because the pole at $\log(\frac{1 + \sqrt{5}}{2})^2$ has order $4$ which occurs as $(x - 3 + x^{-1})$ divides $R_2(x)$. Next observe that if $(F(t))_n \sim n^k S \frac{\Gamma(n + \frac{3}{2})}{d^n}$ then
        \[
            (t^i \frac{\partial^i}{\partial t^j} F(t))_n \sim n^{k + 2j - i} d^{i - j} S \frac{\Gamma(n + \frac{3}{2})}{d^n}.
        \]
        This allows us to read off that
        \[
            (G_2)_n \sim (\G{R}_2)_n
        \]
        and hence that
        \[
            (G_0 + \frac{1}{r}G_1 + \frac{1}{r^2}G_2) \sim \frac{\Gamma(n + \frac{3}{2})}{d^n}\left(S_0 - \frac{d}{2nr} S_0 - \frac{n d^2 S_2}{r^2}\right) \sim \frac{S_0\Gamma(n + \frac{3}{2})}{d^n}\left(1 - \frac{n d^2 S_2}{S_0 r^2}\right) 
        \]
        If we fix an $r \in \m{Z}$, we find that the large order behaviour of this finite sum of $G_i$'s has radius of convergence $d$ as expected. However if we treat $\frac{1}{r}$ as an infinitesimal quantity, then
        \[
            \frac{1 - \frac{n d^2 S_2}{S_0 r^2}}{d^n} \sim \frac{1}{(d + \frac{d^3 S_2}{S_0 r^2})^n}
        \]
        and so the radius of convergence acquires a shift! Indeed this shift is exactly given by
        \[
            \frac{d^3 S_2}{S_0 r^2} = \frac{\sqrt{5}\log(\frac{1 + \sqrt{5}}{2})^3}{75} \frac{1}{r^2}
        \]
        which is exactly the $O(\frac{1}{r^2})$ correction to the CS value in \eqref{eq: Shifted CS value 41} which was computed via studying the A-polynomial.

\section{Torsion Polynomials for Twist Knot Surgeries} 
\label{app: Torsion Polynomial, Twist Knot Surgeries}

Recall from Section \ref{sec:twisted-Alexander} we have the general formula for Adjoint Torsion

\[
    \tau^{adj}_{M^3_{p/r}(K)}(\rho) = \frac{p \ \frac{y_{5_2}}{x} \frac{dx}{dy_{K}}\tau^{adj}_{K, [l]}(\rho) + r\ \tau^{adj}_{K, [l]}(\rho)}{\tau^{adj}_{S^1, K}}
\]

This is easy to evaluate for a large class of knots and surgery coefficients and is much easier to automate than the computations giving Chern-Simons values. Hence we briefly test and refine the predictions/observations made in Section \ref{sec: Observations, Torsions and CS values}. In particular we investigate further how close torsions are to algebraic integers, the constant coefficient of the torsion polynomial and factoring of the torsion polynomial. This sort of general structure has been studied previously for the standard torsion \cite{Kit2016}, but not the adjoint one.

\subsection{Torsion polynomials for the figure-eight surgeries}

Let us start by focusing on the figure eight knot. In Table \ref{tab: Torsion Polynomials Fig 8} we compute a large number of Torsion polynomials for various small surgeries.

\begin{itemize}
    \item Denoting the torsion polynomial as
    \[
        \sigma^{adj}\left(K, \frac{p}{r}; t\right) = \sum_{i = 0}^{n} a_i t^i
    \]
    The sum of inverse torsions is the ratio $\frac{a_1}{a_0}$. Hence in Table $\ref{tab: Torsion Polynomials Fig 8}$ we see that the sum of inverse torsions is non $0$ for integral surgeries and $\frac{5}{2}, \frac{7}{2}$ but is always integral. Additionally, if we hold $p$ constant and increase $r$ we find that the sum quickly returns to $0$. Hence a make the following conjecture:
    \begin{myconj}
        Fix a knot $K$ let $\tau_{\frac{p}{r}}$ denote the adjoint torsions for the $\frac{p}{r}$ surgery. Then the sum of inverse torsions is always integral and for any fixed $p$, for $|r|$ sufficienty large the sum of inverse torsions is $0$.
    \end{myconj}
    \item Looking at the leading coefficients, we also see that torsions are not always algebraic integers up to powers of $2$. However there does appear to be a constant $C_p$ depending on $K, p$ such that for all $r$, the leading term of the torsion polynomial is $2^i C_p t^j$ with $i \leq j$.
\end{itemize}

\begin{table}
        \centering
        \renewcommand{\arraystretch}{1.5}
        \begin{tabular}{ m{2cm} | m{15cm} }
        Surgery Coefficient & \multicolumn{1}{c}{Torsion Polynomial $\sigma^{adj}\left(4_1, -\frac{p}{r}; t\right)$} \\ \hline
        \shortstack{$1$} & $-49 + 98t - 56t^2 + 8t^3$ \\  \hline
        \shortstack{$\frac{1}{2}$} & $-7215127 + 2828784t^2 - 417832t^3 - 272624t^4 + 83296t^5 - 7168t^6 + 128t^7$ \\ \hline
        \shortstack{$\frac{1}{3}$} & $-18526263694969 + 4555881355272t^2 + 79840648144t^3 - 387642368768t^4  -4086926944t^5 + 13272051392t^6 - 370245888t^7 - 164618240t^8 + 14550528t^9 - 387072t^{10} + 2048t^{11}$ \\  \hline 
        \shortstack{$\frac{1}{4}$} & $-260317390949631266159 + 47813638030084915680t^2 + 63968552791767792t^3 - 3331851853538685632t^4 - 18418344332124448t^5 + 110148430262003200t^6 + 1172148087271296t^7 - 1774195301982464t^8 - 17706680725504t^9 + 12657826505728t^{10} - 82127824896t^{11} - 35044376576t^{12} + 1287880704t^{13} - 14680064t^{14} + 32768t^{15}$ \\  \hline 
        \shortstack{$2$} & $5 - 10 t + 4 t^2$ \\ \hline 
        \shortstack{$\frac{2}{3}$} & $753079 - 95554 t^2 - 1527 t^3 + 3453 t^4 - 234 t^5 + 4 t^6$ \\ \hline 
        \shortstack{$3$} & $1 - 12 t + 4 t^2$ \\ \hline 
        \shortstack{$\frac{3}{2}$} & $365263 - 249220 t^2 + 25584 t^3 + 53280 t^4 - 15840 t^5 + 1024 t^6$ \\ \hline 
        \shortstack{$\frac{4}{3}$} & $-156896 + 10240 t^2 + 2336 t^3 - 704 t^4 + 32 t^5$ \\ \hline 
        \shortstack{$5$} & $283 - 4528 t + 32020 t^2 - 16272 t^3 + 1936 t^4$ \\ \hline 
        \shortstack{$\frac{5}{2}$} & $-3998639 + 7997278 t - 292496 t^2 + 2013296 t^3 - 8580192 t^4 + 3740416 t^5 - 450048 t^6 + 15488 t^7$ \\ \hline 
        \shortstack{$\frac{5}{3}$} & $-13892528391821 + 3761182474536 t^2 + 320652615912 t^3 - 588943066640 t^4 + 21394528864 t^5 + 14259307648 t^6 + 8918661632 t^7 - 3628587264 t^8 + 433233408 t^9 - 18388992 t^{10} + 247808 t^{11}$ \\ \hline 
        \shortstack{$6$} & $-59 + 354 t - 772 t^2 + 256 t^3$ \\ \hline 
        \shortstack{$\frac{6}{5}$} & $2561632524397 - 206780347415 t^2 + 10795739215 t^3 + 4922885241 t^4 - 507071145 t^5 - 1061595 t^6 - 312559 t^7 + 261798 t^8 - 15860 t^9 + 256 t^{10}$ \\ \hline 
        \shortstack{$7$} & $50173 + 100346 t - 308856 t^2 - 221184 t^3 + 1355904 t^4 - 642496 t^5 + 53824 t^6$ \\ \hline 
        \shortstack{$\frac{7}{2}$} & $-3685907 + 140064466 t - 222192900 t^2 + 246279536 t^3 - 139637424 t^4 + 36192128 t^5 - 3865600 t^6 + 107648 t^7$ \\ \hline 
        \shortstack{$\frac{7}{3}$} & $-11259183017377 + 13780004616964 t^2 - 1460730231592 t^3 - 4038035497728 t^4 + 1804028110400 t^5 - 647662872576 t^6 + 215832738560 t^7 - 40414661888 t^8 + 3669160960 t^9 - 144599040 t^{10} + 1722368 t^{11}$ \\ \hline 
        \end{tabular}
        \caption{Torsion Polynomials for the Figure Eight Knot}
        \label{tab: Torsion Polynomials Fig 8}
    \end{table}

\subsection{Factoring of torsion polynomials} \label{sec: Special Surgeries}

In general, factoring of the torsion polynomial is controlled by the $A$ polynomial. Fixing a surgery coefficient $\frac{p}{r}$, we the polynomial $A_K(t^r,t^{-p})$ controls the intersections of $A(x, y)$ and $y^rx^p = 1$. After removing the components corresponding to $x, y = \pm 1$, we are left with a polynomial $A_{K, \frac{p}{r}}(t)$ whose roots correspond pairwise to irreducible flat connections on $M^3_{\frac{p}{r}}(K)$. If $A_{K, \frac{p}{r}}(t)$ factors, the torsion polynomial will also factor correspondingly.

On occasion however, the torsion polynomial factors further. In particular we see this in \eqref{eq:52n-torsion-polynomial}, for the $-\frac{1}{2}$ on the $K_2 = 5_2$ knot. As we see below this is the first example of a family of surgeries on twist knots where one factor of the torsion poly is a perfect square.

\begin{align*}
    \sigma^{adj}\left(K_{2}, -\frac{1}{2}; t\right) & = \big(-5 + t\big)^2 \big(164 - 28t + t^2\big)^2
    \\ & \big(44241255 + 32803272t + 695124t^2 - 2966904t^3 + 386592t^4 - 13152t^5 + 64t^6\big)
    \\ \sigma^{adj}\left(K_{-2}, \frac{1}{2}; t\right) & = (254983 + 32828t - 2976t^2 - 368t^3 + 16t^4)^2
    \\ &  (-4320863531 + 1112586392t + 69657084t^2 + 206142952t^3 - 44685360t^4
    \\ & \quad \quad \quad + 2593056t^5 - 48832t^6 + 128t^7)
    \\ \sigma^{adj}\left(K_{3}, -\frac{1}{3}; t\right) & = \big(334153383044420071 - 18349887725325090t - 3913729948936348t^2
    \\ &  + 185036724185680t^3 + 37720618951440t^4 - 4394629374624t^5
    \\ &  + 204381542080t^6 - 5034238336t^7 + 67439360t^8 - 445440t^9 + 1024t^{10})^2
    \\ & \quad (293324796050981989 + 32215607249284620t - 11504848303637416t^2
    \\ &  - 2462247201108184t^3 + 35739935743664t^4 + 41253374845120t^5
    \\ &  + 2132995639616t^6 - 84897722112t^7 + 856537600t^8 - 2677760t^9 + 1024t^{10})
     \\ \sigma^{adj}\left(K_{-3}, \frac{1}{3}; t\right) & = \big(-10 + t\big)^2 \big(-19505543083830541912 - 3045381269445479596t
     \\ & \quad  - 169076190701696918t^2 + 11929094279439271t^3 + 792736409895796t^4
     \\ & \quad  + 425924792048t^5 - 846181769360t^6 - 7749071520t^7
     \\ & \quad  + 358319136t^8 + 3097472t^9 - 85760t^{10} + 256t^{11}\big)^2
     \\ & \quad (-121354714396578397745 + 13623041676142110996t
     \\ & \quad \quad   + 2675608055331699792t^2 + 51100649359048448t^3
     \\ & \quad \quad   - 69736245006519296t^4 + 1570776432089952t^5
     \\ & \quad \quad   + 7459030688256t^6 + 27276645239296t^7
     \\ & \quad \quad  - 587342424832t^8 + 3806090240t^9 - 8184832t^{10} + 2048t^{11})
\end{align*}

\section{Mellin transform approach}
\label{app:mellin}

In this Appendix we note that the small and large $t$ asymptotics of the Mordell-like integrals can alternatively be analyzed with Mellin transforms, a widely used approach to asymptotics \cite{FGD95,Zag06}. Instead of using the Fourier transform identity (\ref{eq:fourier}) for the Gaussian term $e^{-u^2/t}$, we use instead a Mellin representation:
\begin{eqnarray}
e^{-u^2/t}=\frac{1}{2\pi i}\int_{c-i\infty}^{c+i\infty} u^{-s}\, \frac{1}{2} t^{s/2}\, \Gamma\left(\frac{s}{2}\right) \, ds
\label{eq:mellin-identity}
\end{eqnarray}
where $0<c<1$. 
This follows because the Mellin transform of $e^{-u^2/t}$ is:
\begin{eqnarray}
\mathcal M[e^{-u^2/t}](s):=\int_0^\infty du\, u^{s-1}\, e^{-u^2/t}=\frac{1}{2} t^{s/2}\, \Gamma\left(\frac{s}{2}\right) 
\label{eq:mellin}
\end{eqnarray}
This Mellin representation has the nice property that it efficiently encodes both the small and large $t$ expansions. Consider an integral of Mordell form
\begin{eqnarray}
\int_0^\infty du\,  e^{-u^2/t}\, \frac{z\, e^{-u}}{1\mp z\, e^{-u}}
\label{eq:general}
\end{eqnarray}
Using the Mellin transform
\begin{eqnarray}
\mathcal M\left[\frac{z\, e^{-u}}{1\mp z\, e^{-u}} \right](s) = \pm \Gamma\left(s\right){\rm Li}_s(\pm z)
\label{eq:mellin-bose}
\end{eqnarray}
by convolution we have:
\begin{eqnarray}
\int_0^\infty du\,  e^{-u^2/t}\, \frac{z\, e^{-u}}{1\mp z\, e^{-u}}=\frac{\pm 1}{2\pi i} \int_{c-i\infty}^{c+i\infty} t^{(1-s)/2} \frac{1}{2} \Gamma\left(\frac{1-s}{2}\right) \Gamma(s) {\rm Li}_s(\pm z)
\label{eq:general2}
\end{eqnarray}
where $0<c<1$. 
The $\Gamma\left(\frac{1-s}{2}\right)$ factor has simple poles at $s=2n+1$ ($n=0, 1, 2, ...$), while the $\Gamma(s)$ factor has simple poles at  $s=-n$ ($n=0, 1, 2, ...$). Closing the contour around the negative $s$ axis generates the small $t$ expansion:\begin{eqnarray}
\int_0^\infty du\,  e^{-u^2/t}\, \frac{z\, e^{-u}}{1\mp z\, e^{-u}}\sim \pm \sum_{n=0}^\infty (-1)^n \frac{ \Gamma\left(\frac{1+n}{2}\right)}{\Gamma(n+1)}  {\rm Li}_{-n}(\pm z) t^{(n+1)/2}
\label{eq:general-small-t}
\end{eqnarray}
On the other hand, closing the contour around the positive $s$ axis generates the large $t$ expansion:
\begin{eqnarray}
\int_0^\infty du\,  e^{-u^2/t}\, \frac{z\, e^{-u}}{1\mp z\, e^{-u}}\sim \pm \sum_{n=0}^\infty (-1)^n \frac{ \Gamma\left(2n+1\right)}{\Gamma(n+1)}  {\rm Li}_{2n+1}(\pm z) \frac{1}{t^{n}}
\label{eq:general-large-t}
\end{eqnarray}
This explains the appearance of polylog factors with negative integer indices for the small $t$ expansions, and polylog factors with positive integer indices for the large $t$ expansions. Note that ${\rm Li}_{-n}(\pm z)$ grows factorially fast in magnitude, while ${\rm Li}_{2n+1}(\pm z)$ tends rapidly to 1 in magnitude, so both expansions are factorially divergent.

\bibliography{MCSV}
\bibliographystyle{abstract}
    
\end{document}